\documentclass[bpam,reqno]{ipart}
\usepackage[T1]{fontenc}
\usepackage[utf8]{inputenc}
\usepackage[english]{babel}
\usepackage{amsmath,amsthm,amssymb,latexsym,upref,enumerate,fouridx}
\usepackage{dutchcal}
\usepackage{mathrsfs,color}
\usepackage{centernot}
\usepackage[colorlinks,linkcolor=blue,citecolor=blue]{hyperref}

\allowdisplaybreaks

\numberwithin{equation}{section}

\theoremstyle{plain}
\newtheorem{thm}{Theorem}[section]
\newtheorem{lem}[thm]{Lemma}
\newtheorem{prop}[thm]{Proposition}

\theoremstyle{definition}
\newtheorem{Def}[thm]{Definition}

\theoremstyle{remark}
\newtheorem{rem}[thm]{Remark}

\newcommand{\ep}{\epsilon}

\newcommand{\tauac}{\acute{\tau}{}}
\newcommand{\gac}{\acute{g}{}}

\newcommand{\Qft}{\tilde{\mathfrak{Q}}{}}

\newcommand{\Jcch}{\check{\mathcal{J}}{}}

\newcommand{\chihu}{\underline{\hat{\chi}}{}}
\newcommand{\ghu}{\underline{\gh}{}}
\newcommand{\whu}{\underline{\wh}{}}
\newcommand{\vhu}{\underline{\vh}{}}

\newcommand{\zhu}{\underline{\zh}{}}
\newcommand{\hhu}{\underline{\hh}{}}
\newcommand{\Bhu}{\underline{\Bh}{}}
\newcommand{\Qhu}{\underline{\Qh}{}}
\newcommand{\Vhu}{\underline{\Vh}{}}

\newcommand{\Ghu}{\underline{\Gh}{}}
\newcommand{\Ttthu}{\underline{\hat{\Ttt}}{}}
\newcommand{\Gammahu}{\underline{\hat{\Gamma}}{}}

\newcommand{\Gammah}{\hat{\Gamma}{}}

\newcommand{\gchat}{\hat{\gc}{}}

\newcommand{\chih}{\hat{\chi}{}}

\newcommand{\taur}{\mathring{\tau}{}}
\newcommand{\chir}{\mathring{\chi}{}}

\newcommand{\taug}{\grave{\tau}{}}
\newcommand{\er}{\mathring{e}{}}
\newcommand{\gr}{\mathring{g}{}}
\newcommand{\Vr}{\mathring{V}{}}
\newcommand{\lr}{\mathring{l}{}}
\newcommand{\Jcr}{\mathring{\mathcal{J}}{}}
\newcommand{\vr}{\mathring{v}{}}

\newcommand{\rhor}{\mathring{\rho}{}}

\newcommand{\Kttt}{\tilde{\Ktt}{}}
\newcommand{\gttt}{\tilde{\gtt}{}}
\newcommand{\Nttt}{\tilde{\Ntt}{}}
\newcommand{\bttt}{\tilde{\btt}{}}

\newcommand{\ggr}{\grave{g}{}}
\newcommand{\taugr}{\grave{\tau}{}}

\newcommand{\udim}{\mc}

\newcommand{\Ttth}{\hat{\Ttt}{}}

\input Tdef.def

\renewcommand{\emph}[1]{\textit{#1}}

\begin{document}

\title[Past stability of FLRW Einstein-Euler-scalar solutions]{Past stability of FLRW solutions to the Einstein-Euler-scalar field equations and their big bang singularities}

% \author[F. Beyer]{Florian Beyer}
% \address{Dept of Mathematics and Statistics\\
% 730 Cumberland St\\
% University of Otago, Dunedin 9016\\ New Zealand}
% \email{fbeyer@maths.otago.ac.nz }

% \author[T.A. Oliynyk]{Todd A. Oliynyk}
% \address{School of Mathematical Sciences\\
% 9 Rainforest Walk\\
% Monash University, VIC 3800\\ Australia}
% \email{todd.oliynyk@monash.edu}

\begin{aug}
    \author{\fnms{Florian} \snm{Beyer}\ead[label=e1]{florian.beyer@otago.ac.nz}},
    \address{Dept of Mathematics and Statistics\\
      730 Cumberland St\\
      University of Otago, Dunedin 9016\\ New Zealand\\
      \printead{e1}}
    \and
    \author{\fnms{Todd A.} \snm{Oliynyk}\ead[label=e2]{todd.oliynyk@monash.edu}},
    \address{School of Mathematical Sciences\\
      9 Rainforest Walk\\
      Monash University, VIC 3800\\ Australia\\
             \printead{e2}}
\end{aug}

\begin{abstract}
\ We establish, in spacetime dimensions $n\geq 3$, the nonlinear stability in the contracting direction of Friedmann-Lema\^itre-Robertson-Walker (FLRW) solutions to the Einstein-Euler-scalar field equations with linear equations of state $P=c_s^2 \rho$ for sounds speeds $c_s$ satisfying $1/(n-1)<c_s^2 < 1$. We further show that nonlinear perturbations of the FLRW solutions are asymptotically pointwise Kasner and terminate in crushing, asymptotically velocity term dominated (AVTD) big bang singularities characterised by curvature blow-up.\newline
\begin{center}
   \textit{This work is dedicated to the memory of Robert Bartnik.} 
\end{center}
\end{abstract}

\maketitle

\section{Introduction}
\label{sec:introduction}

The mathematical definition of a \textit{cosmological spacetime} was introduced by Robert Bartnik in \cite{Bartnik:1988b}. This definition is one of many lasting contributions Robert made to the field of Mathematical Relativity \cite{Chrusciel_et_al:2021}, and it is particularly relevant to this article in which we analyse the nonlinear stability of a class of cosmological spacetimes containing big bang singularities.  Robert, throughout his career, touched and enriched the lives of many in the mathematical community. He was known for his unique creativity, profound intellect, sense of curiosity and love of learning. He will be missed by many to whom he was a friend, a colleague and a mentor. We dedicate this article to his memory.

Friedmann-Lema\^itre-Robertson-Walker (FLRW) spacetimes
and their perturbations are the foundation of the standard model of cosmology. However, until recently, the nonlinear stability in the contracting direction of the FLRW solutions was not well understood with the exception of perturbations within the class of homogeneous solutions. In the contracting direction,  the Penrose and Hawking singularity theorems \cite{hawkingLargeScaleStructure1973} guarantee that cosmological spacetimes, including nonlinear perturbations of FLRW solutions, will be geodesically incomplete for a large class of matter models and initial data sets, including highly anisotropic ones. The origin of the geodesics incompleteness is widely expected to be due to the
formation of curvature singularities, and it is an outstanding problem in mathematical cosmology to rigorously
establish the conditions under which this expectation is true and to understand the dynamical behaviour of
cosmological solutions near singularities.

The influential BKL-conjecture \cite{belinskii1970,lifshitz1963} posits that  cosmological singularities are \emph{generically} spacelike and oscillatory. While recent work on \emph{spikes} \cite{berger1993,coley2015, lim2008, lim2009, rendall2001} and \emph{weak null singularities} \cite{dafermos2017,luk2017} indicate that the BKL-conjecture is incomplete, it is still expected to be true under quite general conditions. However, it should be noted that the
only rigorous arguments supporting the BKL-conjecture are limited to the spatially homogeneous setting \cite{BeguinDutilleul:2023,beguin:2010,liebscher_et_al:2011,Ringstrom:2001}, while in the non-homogeneous setting, there are numerical studies that support the conjecture \cite{Andersson:2005, curtis2005, garfinkle2002,garfinkle2002a,garfinkle2004,garfinkle2007,Weaver:2001}.

Currently, there are no rigorous nonlinear stability results that apply to inhomogenous cosmologies with oscillatory spacelike singularities.  However, the situation improves considerably for cosmological spacetimes that exhibit \emph{asymptotically velocity term dominated} (AVTD) behaviour \cite{Eardley:1972,Isenberg:1990} near the singularity. By definition, AVTD singularities are a special type of big bang type singularities, see Section~\ref{sec:AVTDAPK} for details, that are spacelike and non-oscillatory.
AVTD behaviour has been shown to occur generically in classes of vacuum spacetimes with symmetries \cite{CIM1990,Isenberg:1990,Fournodavlos_et_al:2023,ringstrom2009a}, and for infinite dimensional families of cosmological spacetimes with prescribed asymptotics near the singularity in a variety of settings using Fuchsian methods \cite{ames2019,ames2013a,andersson2001,beyer2017,choquet-bruhat2006, choquet-bruhat2004,ChruscielKlinger:2015,Clausen2007,damour2002,Fournodavlos:2020,heinzle2012,isenberg1999,isenberg2002,kichenassamy1998,stahl2002}.

In recent years, remarkable progress has been made on rigorously establishing the past\footnote{We always choose our time orientation so that the contracting time direction corresponds to the past.} stability of FLRW solutions to the Einstein-\emph{scalar field equations} and their AVTD big bang singularities under \emph{generic perturbations without symmetries}. For these solutions, the minimally coupled scalar field is responsible for the resulting non-oscillatory AVTD dynamics. The first such FLRW big bang stability result was established in the seminal articles \cite{RodnianskiSpeck:2018b,RodnianskiSpeck:2018c}. It is worth noting that the role of the scalar field in \cite{RodnianskiSpeck:2018b,RodnianskiSpeck:2018c}  and the subsequent stability results \cite{FajmanUrban:2022,Fournodavlos_et_al:2023,Speck:2018} is, in four spacetime dimensions and without any symmetry assumptions\footnote{In high enough spacetime dimensions or under certain symmetry assumptions, the oscillatory behaviour near big bang singularities of solutions to vacuum Einstein equations is also suppressed; see \cite{demaret_non-oscillatory_1985,Fournodavlos_et_al:2023,RodnianskiSpeck:2021} for details}, to suppress the oscillatory behaviour of solutions near big bang singularities, which leads to AVTD behaviour. It should be noted that it is the AVTD nature of Einstein-scalar field big bang singularities that make their analysis analytically tractable in contrast to oscillatory singularities in the inhomogeneous setting.

The FLRW big bang stability results \cite{RodnianskiSpeck:2018b,RodnianskiSpeck:2018c}
have been significantly extended in the article \cite{Fournodavlos_et_al:2023} to apply to
nonlinear perturbations of the Kasner family of solutions solutions in the following settings and spacetime dimensions $n$:  Einstein-scalar field equations $(n\geq 4)$, the polarized $\mathbb{U}(1)$-symmetric vacuum Einstein equations $(n=4)$, and the vacuum Einstein equations $(n\geq 11)$. Remarkably, the big bang stability results established in \cite{Fournodavlos_et_al:2023} hold for  the full range of Kasner exponents where stable singularity formation is expected. We note also the articles 
\cite{ABIO:2022_Royal_Soc,ABIO:2022} where related Kasner big bang stability results are established in the polarised $\Tbb^2$-symmetric vacuum setting, the general framework developed by Ringstr\"{o}m for analysing cosmological spacetimes with big bang singularities \cite{Ringstrom:2021b,Ringstrom:2021a,Ringstrom:2022}, and the recent work \cite{Groeniger_et_al:2023} in which a large class of initial data is identified that greatly extends the data considered in \cite{Fournodavlos_et_al:2023}, and at the same time, leads to stable big bang formation for the Einstein-scalar field system with non-vanishing potentials.

One important question not answered by the big bang stability results \cite{FajmanUrban:2022,Fournodavlos_et_al:2023,RodnianskiSpeck:2018b,RodnianskiSpeck:2018c,Speck:2018} is that of local instability, namely, do local changes of the initial data on the initial hypersurfaces induce local changes on the big bang singular hypersurface. 
This can also be rephrased as a question of existence of particle horizons. The technical reason as to why the stability proofs from the articles \cite{FajmanUrban:2022,Fournodavlos_et_al:2023,RodnianskiSpeck:2018b,RodnianskiSpeck:2018c,Speck:2018} cannot directly answer this question is that they rely on foliating spacetime by spacelike hypersurfaces of constant mean curvature (CMC). In these articles, this foliation is used to define a time function via
$t=-(\text{tr}{K})^{-1}$ where $\text{tr}{K}$ is trace of the extrinsic curvature of hypersurfaces. The importance of this time function is that the singularity can be shown to occur uniformly along the hypersurface $t=0$. In this sense, this choice of time coordinate \textit{synchronizes} the singularity. This is important because it allows statements to be made about the behaviour of the physical fields as the singularity is approached, i.e.\ in the limit $t\searrow 0$, that are uniform across the whole singular surface. On the other hand, CMC foliations, by definition, are non-local. Because of this, any stability result that is derived using it will, a priori, be non-local in the sense that local changes in the initial data will lead to non-local changes in the solution at a later time. Without additional arguments, it would remain uncertain as to whether this non-locality is physical or pure gauge\footnote{In the sense that the non-locality is an artifact of the choice of time function that could be remedied, i.e.\ made local, by a different choice of time slicing.}.

To resolve the question of localisability, 
we introduced a new method for analysing the stability of big bang singularities for the Einstein-scalar field equations in the article \cite{BeyerOliynyk:2021} that is based on using the scalar field $\phi$ to define a time function $\tau$ via $\phi=\sqrt{\frac{n-1}{2(n-2)}}\ln(\tau)$ and a wave gauge to reduce the Einstein equations. The advantages of this method are twofold. First, because the gauge is hyperbolic, the method is inherently local, i.e. all fields propagate with a finite speed, and yields localisable results. Second, the method yields evolution equations, at least for initial data that is nearly FLRW on the initial hypersurface, that can be cast into a Fuchsian form for which it is possible prove existence of solutions globally to the past via an application of the existence theory for Fuchsian systems developed in the articles \cite{BOOS:2021,Oliynyk:CMP_2016}.
This new method was employed in \cite{BeyerOliynyk:2021} to establish the \textit{local (in space)} past stability of nonlinear perturbations of FLRW solutions to the Einstein-scalar field equations and their big bang singularities, which of course, implies a (global in space) past stability theorem similar to those established in \cite{FajmanUrban:2022,RodnianskiSpeck:2018b,RodnianskiSpeck:2018c,Speck:2018}. It is worth noting that the \emph{Fuchsian approach} to establishing the global existence of solutions to systems of hyperbolic equations is a very general method and has recently been employed to establish a variety of stability results in the following articles \cite{BeyerOliynyk:2020,FOW:2021,LeFlochWei:2021,LiuOliynyk:2018b,LiuOliynyk:2018a,LiuWei:2021,MarshallOliynyk:2023,Oliynyk:CMP_2016,Oliynyk:2021,OliynykOlvera:2021,Wei:2018}.

The main aim of this article is to investigate, using the approach developed in \cite{BeyerOliynyk:2021}, the nonlinear stability of FLRW big bang singularities where the gravitating matter includes a perfect fluid with linear equations of states $P=c_s^2 \rho$ in addition to a scalar field. Here, the speed of sound $c_s$ is considered a free parameter. It is expected heuristically as part of the so-called \emph{matter does not matter} hypothesis \cite{belinskii1970,lifshitz1963} that the perfect fluid should be negligible near the big bang provided the speed of sound $c_s$ is smaller than the speed of light.

Whether this is true in general or not is an open question however. On fixed Kasner backgrounds in spacetime dimension $n=4$, the stability of relativistic fluids with linear equations of states $P=c_s^2 \rho$ in the contracting direction was investigated by us in \cite{BeyerOliynyk:2020}. 
In that article, we established the past nonlinear stability of a large class of solutions, which includes perturbations of homogeneous solutions, in the neighborhood of Kasner big bang singularities for sound speeds satisfying $q_{max}<c_s^2<1$, where the lower bound $q_{max}$ equals the largest Kasner exponent and $q_{max}=1/3$ for the FLRW-Kasner spacetime (i.e.\ equality of the Kasner exponents) in four spacetime dimensions. In this article, we show that the fluid remains nonlinearly stable toward the big bang singularity for small perturbations of FLRW initial data in spacetime dimensions $n\geq 3$ over the sound speed range $1/(n-1)<c_s^2 < 1$, which coincides, for $n=4$, with the range from \cite{BeyerOliynyk:2020} when coupling to the Einstein-scalar field equations is included. 

\subsection{Conformal Einstein-Euler-scalar field equations}
The coupled Einstein-Euler-scalar field equations\footnote{See Section
  \ref{indexing} for our indexing conventions.} for a free scalar
field and a perfect fluid with a linear equation of state
\begin{equation*}
    %\label{eq:16}
    P=c_s^2\rho,\quad c_s\in (0,1),
  \end{equation*}
are
\begin{equation}
  \label{eq:EFE.0}
  \Gb_{ij}=2(\Tb_{ij}^{\text{SF}}+ \Tb_{ij}^{\text{Fl}}),\quad
  \nablab^i \Tb^{\text{SF}}_{ij}=0,\quad \nablab^i
  \Tb^{\text{Fl}}_{ij}=0,
\end{equation}
where $\nablab_i$ and $\Gb_{ij}$ are the covariant derivative and the
Einstein tensor, respectively, of the physical metric $\gb_{ij}$ and 
\begin{equation}
  \label{Tb-ij-def}
  \begin{split}
  \Tb_{ij}^{\text{SF}}=&\nablab_i\phi\nablab_j \phi-\frac 12\gb_{ij}\nablab_k\phi\nablab^k\phi,\\
  \Tb_{ij}^{\text{Fl}}=&{P_0}\Bigl(\frac{1+c_s^2}{c_s^2} \vb^{-2}
  \Vb_i\Vb_j+\gb_{ij} \Bigr)\vb^{-(1+c_s^2)/c_s^2},
  \end{split}
\end{equation}
for $P_0>0$,
are the energy momentum tensors of the scalar field and the fluid,
respectively. Here, we use the
Frauendiener-Walton formalism \cite{frauendiener2003,walton2005} to represent the fluid by a non-normalised vector field $\Vb^i$, and we
employ the notation
  \begin{equation}
    \label{eq:18jsdkfjkdsjfksdf}
    \Vb_i=\gb_{ij}\Vb^j,\quad \vb^2=-\Vb_i\Vb^i.
  \end{equation} 
The divergence free condition on each energy momentum tensor in \eqref{eq:EFE.0} implies the scalar field and
fluid matter equations 
\begin{align}
\Box_{\gb} \phi &=0, \label{ESF.2}\\
  \label{eq:AAA1}
    \bar\att^i{}_{j k}\nablab_i \Vb^k&=0,
\end{align}
where $\Box_{\gb} =\gb^{ij}\nablab_i\nablab_j$ is the wave operator and
\begin{equation}
  \label{eq:AAA2}
  \bar\att^i{}_{jk}
  =\frac{3 c_s^2+1}{c_s^2} \frac{\Vb_j \Vb_k\Vb^i}{\vb^2} +\Vb^i \gb_{jk}
  +2{\gb^i}_{(k} \Vb_{j)}.
\end{equation}

For later use, we write the Einstein equations, see
\eqref{eq:EFE.0}, as
\begin{equation}
  \Rb_{ij}=2\nablab_i\phi\nablab_j \phi+\bar\Ttt_{ij},
\label{ESF.1}
\end{equation}
where $\Rb_{ij}$ is the Ricci tensor of the metric $\gb_{ij}$ and
\begin{equation}
    \label{bTtt-ij-def}
    \bar\Ttt_{ij}=2P_0\Bigl(\frac{1+c_s^2}{c_s^2} \vb^{-2} \Vb_i\Vb_j+\frac{1-c_s^2}{(n-2)c_s^2}\gb_{ij}\Bigr)\vb^{-(1+c_s^2)/c_s^2}.
  \end{equation}
Here, the dimensionless parameter $c_s$ is the \emph{speed of
  sound} and the positive constant $P_0>0$ has the dimension of pressure.  The 
physical fluid pressure $P$,
density $\rho$ and normalised fluid $n$-velocity $\ub^i$ can be calculated from $\Vb^i$ via the expressions
\begin{equation}
  \label{eq:physicsquantitiesfluid}
P=P_0 \vb^{-\frac {c_s^2+1}{c_s^2}},\quad
\rho=\frac{P_0}{c_s^2} \vb^{-\frac {c_s^2+1}{c_s^2}},\quad \ub^i=\frac{\Vb^i}{\vb}.
\end{equation}

Before we discuss the main results, we first reformulate the Einstein-Euler-scalar field equations  in a way that will be more favorable for the analysis carried out in this article. Following \cite{BeyerOliynyk:2021}, we replace the physical metric $\gb_{ij}$ with a \textit{conformal metric} $g_{ij}$ defined by
\begin{equation}\label{confmet}
\gb_{ij} = e^{2\Phi}g_{ij},
\end{equation}
where $\Phi$ is, for now, an unspecified scalar field, and for the remainder of the article, we assume that the spacetime dimension $n$ satisfies $n\geq 3$. Then, under the conformal transformation \eqref{confmet},
it is well known that the Ricci tensor transforms according to 
\begin{equation}\label{confRicci}
\Rb_{ij}=R_{ij}-(n-2)\nabla_i \nabla_j \Phi + (n-2)\nabla_i \Phi \nabla_j \Phi -(\Box_g \Phi + (n-2)|\nabla\Phi|^2_g)g_{ij}
\end{equation}
where $\nabla_i$ is the Levi-Civita connection of the conformal metric $g_{ij}$ and as above,
$\Box_g=g^{ij}\nabla_i\nabla_j$. For use below, we recall that the connection coefficents of the
metrics $\gb_{ij}$ and $g_{ij}$ are related by
\begin{equation}\label{confChrist}
\Gammab_{i}{}^k{}_j-\Gamma_{i}{}^k{}_j=g^{kl}(g_{il}\nabla_j \Phi + g_{jl}\nabla_i\Phi -g_{ij}\nabla_l\Phi).
\end{equation} Now, using \eqref{confRicci}, we can express the Einstein equations \eqref{ESF.1} as
\begin{equation}\label{confESFA}
  \begin{split}
-2R_{ij}=&-2(n-2)\nabla_i \nabla_j \Phi + 2(n-2)\nabla_i \Phi \nabla_j \Phi \\
&-2(\Box_g \Phi + (n-2)|\nabla\Phi|^2_g)g_{ij}-4\nabla_{i}\phi\nabla_{j} \phi-2{\bar\Ttt}_{ij}.
  \end{split}
\end{equation}
Also, by introducing a Lorentzian background metric $\gc_{ij}$ and letting $\Dc_i$ and $\gamma_{i}{}^k{}_{j}$ denote the associated Levi-Civita
connection and connection coefficients, we can write the scalar field equation \eqref{ESF.2} as
\begin{equation*}
\gb^{ij}\Dc_{i}\Dc_{j}\phi-\gb^{ij}(\Gammab_{i}{}^k{}_j-\Gamma_i{}^k{}_j+\Cc_{i}{}^k{}_{j})\Dc_{k}\phi = 0,
\end{equation*}
where
\begin{equation}\label{Ccdef}
\Cc_{i}{}^k{}_j := \Gamma_{i}{}^k{}_j-\gamma_{i}{}^k{}_j=\frac{1}{2}g^{k l}\bigl(\Dc_i g_{j l}+\Dc_j g_{i l}-\Dc_l g_{ij}\bigr).
\end{equation}
It is then not difficult to verify using \eqref{confmet} and \eqref{confChrist} that the scalar field equation can be expressed as
\begin{equation} \label{confESFB}
g^{ij}\Dc_{i}\Dc_{j}\phi= X^k\Dc_{k}\phi -(n-2)g^{ij}\Dc_{i}\Phi\Dc_{j}\phi,
\end{equation}
where
\begin{equation} \label{Xdef}
X^k := g^{ij}\Cc_{i}{}^k{}_j=\frac{1}{2}g^{ij}g^{k l}\bigl(2\Dc_i g_{j l}-\Dc_l g_{ij}\bigr),
\end{equation}
or equivalently as
\begin{equation} \label{confESFC}
\Box_g \phi = -(n-2)\nabla^i\Phi\nabla_i\phi.
\end{equation}
Note that, here and below, all indices are raised and lowered using the conformal metric, e.g.\  $\nabla^k\Phi = g^{kl}\nabla_l\Phi$.

We proceed by fixing the scalar field $\Phi$ in the conformal transformation \eqref{confmet} up to a constant scaling factor $\lambda$ by 
\begin{equation}\label{gaugefixA}
\Phi =\lambda \phi.
\end{equation}
We also replace the scalar field $\phi$ with a scalar field $\tau$ defined via
\begin{equation} \label{taudef}
 \tau=e^{-\alpha\phi}  \quad \Longleftrightarrow \quad \phi = -\frac{1}{\alpha}\ln(\tau),
\end{equation}
where
\begin{equation}\label{alphafix}
\alpha= -\frac{2-\lambda^2(n-2)}{\lambda (n-2)}.
\end{equation}
With this choice of $\Phi$, we note that
\begin{equation*}
    e^{2\Phi}=\tau^{-2\lambda/\alpha}=\tau^{\frac{2\lambda^2 (n-2)}{2-\lambda^2(n-2)}},
\end{equation*}
and, by \eqref{confESFC}, that  $\Phi$ satisfies the wave equation
\begin{equation} \label{Phi-wave}
\Box_g \Phi +(n-2)\nabla^i\Phi\nabla_i\Phi=0.
\end{equation}
Using \eqref{gaugefixA}-\eqref{alphafix} to replace $\Phi$ and $\phi$ with $\tau$ in \eqref{confESFA} and \eqref{confESFC}, we see, with the help of \eqref{Phi-wave}, that
\begin{equation} \label{confESFAa}
R_{ij}=\Bigl(1-{\frac{2-\lambda^2 (n-2)(n-1)}{2-\lambda^2(n-2)}}\Bigr){\tau^{-1}}\nabla_i \nabla_j \tau
{+{\bar\Ttt}_{ij}}
\end{equation}
and
\begin{equation}
  \Box_g \tau = {\frac{2-\lambda^2(n-2)(n-1)}{2-\lambda^2(n-2)} \tau^{-1}\nabla^i\tau\nabla_i\tau}.  \label{confESFAaa}
\end{equation}
Choosing now
\begin{equation} \label{lambdafix}
 \lambda = \sqrt{\frac{2}{(n-2)(n-1)}}
\end{equation}
and defining conformal fluid variables $V^i$ and $v^2$ via
\begin{equation}\label{eq:conffluidrel}
V^i=\Vb^i \AND v^2=-V_i V^i,
\end{equation}
respectively, 
we find, with our convention of using the conformal metric to raise and lower indices of conformal quantities while using the physical metric to raise and lower indices of physical variables, that
\begin{equation*}
  \Vb_i=\gb_{ij} \Vb^j=e^{2\Phi} g_{ij}V^j =e^{2\Phi} V_{i} \AND \vb^2=-\Vb_i\Vb^i=e^{2\Phi} v^2,
\end{equation*}
and note that the relations
\begin{equation}
  \label{eq:alphaByRtt}
  \frac{2-\lambda^2 (n-2)(n-1)}{2-\lambda^2(n-2)}=0,\quad
  {\alpha=-{\lambda (n-2)}} \AND {e^{2\Phi}=\tau^{\frac {2}{n-2}}}
\end{equation}
hold.
Using these relations and the conformal variables $\{g_{ij},V^i,\tau\}$, a short calculation shows that we can express the tensor \eqref{bTtt-ij-def} as
\begin{equation}
  \label{eq:Tttconfphys}
  \Ttt_{ij}:=\bar{\Ttt}_{ij}=2P_0\Bigl(\frac{1+c_s^2}{c_s^2} v^{-2} V_iV_j+\frac{1-c_s^2}{(n-2)c_s^2}g_{ij}\Bigr) \tau^{\frac{c_s^2-1}{c_s^2(n-2)}}v^{-\frac{(1+c_s^2)}{c_s^2}},
\end{equation}
and the conformal Einstein equations \eqref{confESFAa} and  scalar field equation \eqref{confESFAaa} as
\begin{gather} \label{confESFAaN}
R_{ij}={\tau^{-1}}\nabla_i \nabla_j \tau
{+\Ttt_{ij}}
\intertext{and}
\label{confESFCbN}
  \Box_g \tau =0,
\end{gather}
respectively. 
Furthermore, with the help of \eqref{confChrist}, it is not difficult to verify that the Euler equations \eqref{eq:AAA1} can be expressed in terms of the conformal variables $\{g_{ij},V^i,\tau\}$ as
\begin{equation}
  \label{eq:AAACF1}
  \att^i{}_{j k}\nabla_i V^k
    =- \frac{1}{n-2}\tau^{-1}\att^i{}_{j k}(\delta^{k}_{i}\nabla_l\tau
    + \delta^{k}_{l}\nabla_i\tau -g^{km}g_{il}\nabla_m\tau)V^l
\end{equation}
where
\begin{equation}
  \label{eq:AAACF2}
  \att^i{}_{jk}
  =\tau^{\frac{2}{(n-2)}}\bar\att^i{}_{jk}=\frac{3 c_s^2+1}{c_s^2} \frac{V_j V_k V^i}{v^2} +V^i g_{jk}
    +2{\delta^i}_{(k} V_{j)}.
\end{equation}

Noting that $\att^i{}_{jk}$ satisfies 
\begin{equation*}
\att_{ijk}=\att_{kij},  
\end{equation*}
we find that
\begin{equation*}
  \begin{split}
   \att^i{}_{j k}(\delta^{k}_{i}\nabla_l\tau
    + \delta^{k}_{l}\nabla_i\tau -g^{km}g_{il}\nabla_m\tau)V^l
    = &\att^k{}_{j k} V^i \nabla_i\tau
    + (\att_{ijk} -\att_{kj i})V^k \nabla^i\tau\\
    =& \frac{c_s^2(n-1)-1}{c_s^2}V_jV^i\nabla_i\tau.
  \end{split}
\end{equation*}
With the help of this identity, we see,
after rearranging, that the equations
\eqref{confESFAaN}-\eqref{confESFCbN} and \eqref{eq:AAACF1} can be expressed as
\begin{align}
  G_{ij}&= \tau^{-1}\nabla_i\nabla_j \tau+2T_{ij}^{\text{Fl}},
          %\mathcal{T}_{ij},
\label{confESFAb} \\
  \Box_g \tau &=0, \label{confESFCb}\\
                \att^i{}_{j k}\nabla_i V^k
              &                             
                =- \frac{c_s^2(n-1)-1}{c_s^2(n-2)}\tau^{-1} V_j V^i\nabla_i\tau, \label{confESAc}
\end{align}
where 
\begin{equation*}% \label{Tc-def}
    T_{ij}^{\text{Fl}}= %\Tc_{ij}= 
   {P_0}\tau^{\frac{c_s^2-1}{c_s^2(n-2)}} v^{-\frac{1+c_s^2}{c_s^2}}\biggl(\frac{1+c_s^2}{c_s^2} v^{-2} V_iV_j+g_{ij}\biggr).
\end{equation*}
We will refer to these equations as the \textit{conformal Einstein-Euler-scalar field equations}. It follows from \eqref{confmet}, \eqref{gaugefixA}-\eqref{alphafix} and \eqref{lambdafix} that each solution $\{g_{ij},\tau, V^i\}$ of the conformal Einstein-scalar field equations yields a solution
\begin{equation}
  \label{eq:conf2phys}
  \biggl\{\gb_{ij}=\tau^{\frac{2}{n-2}} g_{ij},\quad\phi=\sqrt{\frac {n-1}{2(n-2)}}\ln(\tau),\quad \Vb^i=V^i\biggr\}
\end{equation}
of the \textit{physical Einstein-scalar field equations} \eqref{eq:EFE.0}. It is worth noting that since \eqref{eq:EFE.0} is invariant under the transformation $\phi\mapsto -\phi$, our sign convention for $\phi$ that follows from the choice of sign in \eqref{lambdafix} incurs no loss of generality.

\subsection{Explicit model solutions}

\subsubsection{Kasner-scalar field spacetimes}
\label{sec:Kasnerscalarfield}
In our conformal picture, the  \textit{Kasner-scalar field spacetimes}, which are solutions of the conformal Einstein-scalar field equations
\eqref{confESFAb}-\eqref{confESFCb} with $T_{ij}^{\text{Fl}}=0$,  are determined
by the conformal metric and scalar
field
\begin{equation}\label{gtau-Kasner}
\gbr = -t^{\rbr_0}dt\otimes dt + \sum_{\Lambda=1}^{n-1} t^{\rbr_\Lambda} dx^\Lambda \otimes dx^\Lambda
\AND
\taubr= t,
\end{equation}
respectively, which are defined on the spacetime manifold $M^{(K)}=\Rbb_{>0}\times \Tbb^{n-1}$; see Section \ref{indexing} for our coordinate and indexing conventions.
In the above expressions, the constants
$\rbr_\mu$ are called \emph{Kasner exponents} and are defined by
\begin{equation}
    \rbr_0 = \frac{1}{\Pbr}\sqrt{\frac{2(n-1)}{n-2}}-\frac{2(n-1)}{n-2} \AND
    \rbr_\Lambda =\frac{1}{\Pbr} \sqrt{\frac{2(n-1)}{n-2}}\qbr_\Lambda
    -\frac{2}{n-2}, \label{rmu-def}
\end{equation}
where $0<\Pbr\le \sqrt{(n-2)/(2(n-1))}$ and the $\qbr_\Lambda$ satisfy the \textit{Kasner relations} 
\begin{equation}\label{qLambda-def}
    \sum_{\Lambda=1}^{n-1}\qbr_\Lambda =1 \AND 
    \sum_{\Lambda=1}^{n-1}\qbr_\Lambda^2 = 1-2 \Pbr^2. 
\end{equation}
Using \eqref{gtau-Kasner}-\eqref{qLambda-def} as well as \eqref{confmet} and \eqref{gaugefixA}-\eqref{alphafix} to compute the curvature scalar invariants $\Rb_{\mu\nu}\Rb^{\mu\nu}$  and $\Rb=\gb^{\mu\nu}\Rb_{\mu\nu}$ of
the physical  metric $\gb_{\mu\nu}=t^{\frac {2}{n -2}}\gbr_{\mu\nu}$, it follows from the resulting expressions
\begin{equation*}
  % \label{eq:19}
  \Rb_{\mu\nu}\Rb^{\mu\nu}=\left(\frac{n-1}{n-2}\right)^2 t^{-4 \frac{n-1}{n-2}-2\rbr_0} \AND
  \Rb=-\frac{n-1}{n-2} t^{-2 \frac{n-1}{n-2}-\rbr_0},
\end{equation*}
that the \textit{Kasner big bang singularity} occurs along
the spacelike hypersurface $\{0\}\times \Tbb^{n-1}$. 

Noting from \eqref{rmu-def} and \eqref{qLambda-def} that
\begin{equation}
  \label{eq:r0rLambda}
 \sum_{\Lambda=1}^{n-1} \rbr_\Lambda=\rbr_0,
\end{equation}
we observe from  a short calculation involving \eqref{gtau-Kasner} 
that
\begin{equation*}
\Box_{\gbr}x^\gamma = \frac{1}{\sqrt{|\det(\gbr_{\alpha\beta})|}}\del{\mu}\Bigl(\sqrt{|\det(\gbr_{\alpha\beta})|}\gbr^{\mu\nu}\del{\nu}x^\gamma\Bigr) = 0. 
\end{equation*}
This shows the  $(x^\mu)$ are \textit{wave coordinates} and that 
\begin{equation} \label{Kasner-wave-gauge}
\gbr^{\mu\nu}\Gammabr_{\mu\nu}^\gamma = 0
\end{equation}
is satisfied, where, here, $\Gammabr_{\mu\nu}^\gamma$ denotes the Christoffel symbols of the conformal Kasner metric $\gbr_{\mu\nu}$. 
We further note via a straightforward calculation that the Kasner-scalar field solutions \eqref{gtau-Kasner} 
satisfy
\begin{equation} \label{FLRW-Lag}
\frac{1}{|\breve{\nabla}\taubr|_{\gbr}^2}\breve{\nabla}{}^\mu \taubr=\delta^\mu_0.
\end{equation}

On the $t=const$-surfaces, the lapse $\Ntt$ and the Weingarten map induced  by the conformal Kasner-scalar field metric $\gbr_{\mu\nu}$ are
\begin{equation}
  \label{eq:KSFlapse2ndFF}
  \Ntt=t^{\frac{\rbr_0}{2}} \AND (\Ktt_\Lambda{}^\Omega)=\frac 12t^{-1-\frac{\rbr_0}{2}}\diag(\rbr_1,\ldots,\rbr_{n-1}),
\end{equation}
respectively. The conformal mean curvature is therefore
\begin{equation}
  %\label{eq:4}
  \Ktt=\frac{\rbr_0}{2}t^{-1-\frac{\rbr_0}{2}},
\end{equation}
as a consequence of \eqref{eq:r0rLambda}, while the physical mean curvature can be shown to be
\begin{equation}
  %\label{eq:4}
  \bar\Ktt=\Bigl(\frac{n-1}{n-2}+\frac{\rbr_0}{2}\Bigr)t^{-1-\frac 1{n-2}-\frac{\rbr_0}{2}}.
\end{equation}

When we express a Kasner-scalar field solution with respect to the time coordinate
\[\tb=\frac{1}{\frac{n-1}{n-2}+\frac{\rbr_0}2} t^{\frac{n-1}{n-2}+\frac{\rbr_0}2}\]
and appropriately rescale all the spatial coordinates $(x^\Lambda)$, that the physical solution \eqref{eq:conf2phys} corresponding to the conformal metric \eqref{gtau-Kasner} takes the more conventional form
\begin{equation}\label{gtau-Kasnerphys}
  \begin{split}
\overline{\gbr} &= -d\bar t\otimes d\bar t + \sum_{\Lambda=1}^{n-1}
{\bar t}^{2\qbr_\Lambda} d\xb^\Lambda \otimes d\xb^\Lambda,
\\
\breve\phi&= \Pbr\ln(\bar t)+\Pbr\ln\Bigl(\frac{n-1}{n-2}+\frac{\rbr_0}2\Bigr),
  \end{split}
\end{equation}
where $\Pbr$ and $\qbr_\Lambda$ are all related to $\rbr_0$ and $\rbr_\Lambda$ by \eqref{rmu-def}. This shows, in particular, that the constant $\Pbr$ can be interpreted as the \emph{asymptotic scalar field strength}. Noting that the case $\Pbr=0$ is excluded by \eqref{rmu-def}, it follows that the special case of vacuum Kasner solutions is not covered by our conformal representation of the Kasner-scalar field solutions.

Kasner spacetimes where the constants $q_\Lambda$ are all the same coincide with FLRW spacetimes. In this situation, we have by \eqref{qLambda-def} that
\begin{equation*}
    \qbr_\Lambda = \frac{1}{n-1}
\end{equation*}
and that 
\[|\Pbr|=\sqrt{\frac{n-2}{2(n-1)}},\]
which we note saturates the inequality  $|\Pbr|\le \sqrt{(n-2)/(2(n-1))}$. By \eqref{gtau-Kasner} -- \eqref{rmu-def}, we then deduce that $\rbr_0=\rbr_\Lambda=0$ and that the conformal Kasner-scalar field solution simplifies to
\begin{align}\label{gtau-FLRW}
\gbr = -dt\otimes dt + \delta_{\Lambda\Omega}dx^\Lambda \otimes dx^\Omega
\AND
\taubr= t.
\end{align}

\subsubsection{FLRW-Euler-scalar field solution}
\label{sec:FLRWEulerSFExplSol}
The Kasner-scalar field family of solutions considered in the previous section do not involve a fluid.
In this section, we allow for coupling to a fluid but limit our considerations to FLRW solutions. These solutions to the conformal Einstein-Euler-scalar field equations \eqref{eq:Tttconfphys}-\eqref{eq:AAACF2}, which we refer to as \textit{FLRW-Euler-scalar field solutions}, are defined by  
\begin{align}
    \label{eq:FLRWEulerSFExplSol1}
  \gbr&=-\omega^{2(n-1)}dt\otimes dt+\omega^2\sum_{\Lambda=1}^{n-1}dx^\Lambda\otimes dx^\Lambda,\\
  \label{eq:FLRWEulerSFExplSol2}
  \taubr&=t,\\
  \label{eq:FLRWEulerSFExplSol3}
    \Vbr&=V_*^0\, \omega^{-(1-c_s^2)(n-1)}\,t^{-(1-(n-1)c_s^2)/(n-2)}\del{t},
  \end{align}
  where
  \begin{equation}
    \label{eq:FLRWEulerSFExplSol4}
    \omega=\left(1-\frac{n-2}{n-1}\frac{P_0}{c_s^2} (V_*^0 )^{-\frac{1+c_s^2}{c_s^2}} t^{\frac{n-1}{n-2}(1-c_s^2)}\right)^{-2/((n-1)(1-c_s^2))}
  \end{equation}
and $c_s$, $P_0$ and $V_*^0$ are 
any constants satisfying  $c_s\in (0,1)$, $P_0>0$ and $V_*^0>0$.  Each of these solutions is well-defined for $t\in(0,t_0)$,  where $t_0>0$ depends of the choice of constants $c_s$, $P_0$ and $V_*^0$, and by \eqref{eq:physicsquantitiesfluid}, its fluid density is determined via
\[\rho=\frac{P_0}{c_s^2} (V_*^0 )^{-\frac{1+c_s^2}{c_s^2}} \omega^{-(n-1)(1+c_s^2)} t^{-\frac{n-1}{n-2}(1+c_s^2)}.\]
Observing that $\omega\rightarrow 1$ as $t\searrow 0$, we define
  \begin{equation}
    \label{eq:fluidBG1}
    \rho_*=\frac{P_0}{c_s^2} (V_*^0 )^{-\frac{1+c_s^2}{c_s^2}},
  \end{equation}
  which allows us to express  $\omega$ and $\rho$ as
  \begin{equation}
    \label{eq:fluidBG2}
    \begin{split}
    %\label{eq:FLRWEulerSFExplSol4N}
    \omega&=\left(1-\frac{n-2}{n-1}\rho_* t^{\frac{n-1}{n-2}(1-c_s^2)}\right)^{-2/((n-1)(1-c_s^2))}\\
    \rho&=\rho_* \omega^{-(n-1)(1+c_s^2)}t^{-\frac{n-1}{n-2}(1+c_s^2)},
    \end{split}
  \end{equation}
respectively.

As can be easily verified from \eqref{eq:FLRWEulerSFExplSol1}-\eqref{eq:FLRWEulerSFExplSol2}, the FLRW-Euler-scalar field solutions satisfy the wave gauge condition \eqref{Kasner-wave-gauge} and \eqref{FLRW-Lag}. In terminology that will be introduced below, \eqref{FLRW-Lag} implies that the coordinates $(x^\mu)=(t,x^\Lambda)$ used to define the FLRW metric \eqref{gtau-FLRW} are \textit{Lagrangian}, while  \eqref{Kasner-wave-gauge} shows that these coordinates are \textit{wave coordinates}, i.e. $\Box_{\gbr}x^\mu =0$. Both of these gauge conditions play a pivotal role in our stability analysis.

Irrespective of the value of the parameter $V_*^0>0$ or equivalently $\rho_*\ge 0$, we observe from \eqref{eq:FLRWEulerSFExplSol1}-\eqref{eq:FLRWEulerSFExplSol2} and $\lim_{t\searrow 0}\omega=1$ that the metric and scalar field have the same limit as the FLRW-scalar field solution \eqref{gtau-FLRW} at $t=0$. This is a manifestation of the before-mentioned \emph{matter does not matter} hypothesis.
%It is important to point out that only ``non-extreme'' matter, e.g.\ a fluid with a speed of sound smaller than the speed of light, is expected to be negligible near big bang singularities. In contrast, scalar field matter\footnote{This is also true for stiff fluids where the speed of sound is equal to the speed of light.} is known to have a significant effect on the behaviour of solutions near big bang singularities and does fit within the matter does not matter paradigm. 

%It is an interesting consequence of our results here  that in the case of \emph{generic} solutions of the Einstein-Euler-scalar field system a  fluid  with a speed of sound smaller than the speed of light is negligible near the big bang actually only if the speed of sounds is sufficiently large, that is,   $1/(n-1)<c_s^2 < 1$.

As discussed above, the main result of this paper is to establish the past nonlinear stability of FLRW solutions to the Einstein-Euler-scalar field equations and their big bang singularities for sound speeds satisfying  $1/(n-1)<c_s^2 < 1$. In addition to this, we verify that the ``matter does not matter'' assertion holds in the sense that we prove that nonlinear perturbations of FLRW solutions are asymptotic, in a suitable sense, to solutions of the Einstein-scalar field equations. 
%A more thorough discussion of the stability and asymptotic behaviour of the perturbed FLRW solutions is provided in the following section.  
\subsection{AVTD and asymptotic pointwise Kasner behaviour}
\label{sec:AVTDAPK}

In this article, we analyze solutions of the Einstein-Euler-scalar field system that are \emph{close} to one of the FLRW solutions from Section~\ref{sec:FLRWEulerSFExplSol} and for which the sound speed lies in the range $1/(n-1)<c_s^2<1$. Even though the spacetimes that are generated by our stability results contain fluids and are generically spatially inhomogeneous without any symmetries, they do retain some of the asymptotic properties that characterize the (non-fluid) spatially homogeneous Kasner-scalar field solutions from Section~\ref{sec:Kasnerscalarfield}.
For example, we show, see Theorem \ref{glob-stab-thm}.(d), that spatial derivative terms, the so-called \emph{velocity terms} \cite{Eardley:1972,Isenberg:1990},  that appear in the dynamical equations become negligible at $t=0$ in comparison to time derivative terms. This behavior is known as \emph{asymptotically velocity term dominated} (AVTD) and it implies that the dynamical equations can be approximated by ODEs close to the big bang.
In agreement with \cite{Fournodavlos_et_al:2023}, we define that a solution to the conformal Einstein-Euler-scalar field system satisfies the AVTD property provided it satisfies the \textit{velocity term dominated} (VTD) equations that are obtained, up to an error term that is integrable in time near $t=0$, from the main evolution system by removing all spatial derivative terms and by normalising the time derivative terms. Here, by \textit{normalising the time derivatives}, we mean that the evolution equations are put into Fuchsian form as discussed in Sections~\ref{sec:FuchsianForm}~and~\ref{sec:proof_globstab_Fuchsian}.

Given that spatial inhomogeneities of AVTD solutions become irrelevant at the big bang, it is, perhaps, not surprising that they behave locally like spatially homogeneous solutions. More precisely, we show, see Theorem \ref{glob-stab-thm}.(d), that the solutions generated from our stability result satisfy the following \emph{asymptotic pointwise Kasner} property just as in the non-fluid case \cite{BeyerOliynyk:2021}.

\begin{Def}%[Asymptotic pointwise Kasner]
\label{def:APKasner}
Given a $C^2$-solution $(M=(0,t_0]\times U, g_{\mu\nu},\tau, V^\mu)$, $t_0>0$, of the conformal Einstein-Euler-scalar field equations,
where $U\subset \Tbb^{n-1}$ is open and  $(x^\mu)=(t,x^\Lambda)$ are coordinates on $M$ such that $t\in (0,t_0]$, $\tau=t$ and the $(x^\Lambda)$ are periodic coordinates on $\Tbb^{n-1}$, we say that the spacetime $(M, g_{\mu\nu},\tau, V^\mu)$ is \textit{asymptotically pointwise Kasner on $U$} provided there exists a continuous orthonormal frame $e_i=e_i^\mu\del{\mu}$ and a continuous spatial tensor field $\kf_{I}{}^J$ on $U$ such that the following hold:
\begin{enumerate}[(i)]
\item The spatial vector fields $e_I$ are tangential to the $t=const$-surfaces, i.e.\ $e_I=e^\Lambda_I \del{\Lambda}$, and
\begin{equation}
  \label{eq:asymptptwKasner2}
  \lim_{t\searrow 0} \bigl|2t\,\Ntt(t,x)\, \Ktt_{I}{}^J (t,x)-\kf_{I}{}^J(x)\bigr|=0
\end{equation}
for each $x\in U$, where $\Ntt$ is the lapse and $\Ktt_{I}{}^{J}$ is the Weingarten map induced on the $t=\mathrm{const}$-hypersurfaces by the conformal metric $g_{\mu\nu}$. 
\item The tensor field $\kf_I{}^J$ satisfies $\kf_I{}^I\geq 0$ and the \emph{Kasner relation}
\begin{equation}
  \label{eq:asymptptwKasner}
  (\kf_{I}{}^{I})^2 -
  \kf_{I}{}^J \kf_{J}{}^I 
  +4\kf_{I}{}^{I}=0
\end{equation}
everywhere on $U$. At each point $x\in U$, the symmetry of $\kf_I{}^J(x)$ guarantees that $\kf_I{}^J(x)$ has $n-1$ real eigenvalues, which we denote by $r_1(x),\ldots,r_{n-1}(x)$. We refer to these functions $r_1,\ldots,r_{n-1}$ on $U$ as the \emph{Kasner exponents}\footnote{Notice that it is more customary in the literature to call the quantities $q_1,\ldots, q_{n-1}$ defined in \eqref{eq:r0qIdef} \emph{Kasner exponents} and \eqref{eq:Kasnerrel} the \emph{Kasner relations}.}.
\end{enumerate}
\end{Def}
\noindent We refer the reader to \cite[\S1.3]{BeyerOliynyk:2021} for more information regarding the motivation and consequences of  Definition~\ref{def:APKasner}.

Given Kasner exponents $r_1,\ldots, r_{n-1}$ from the above definition, we define
\begin{equation}
  \label{eq:r0qIdef}
  r_0:=\sum_{\Lambda=1}^{n-1}r_\Lambda=\kf_{I}{}^{I},\quad 
  q_\Lambda:=\Ptt\sqrt{\frac{n-2}{2(n-1)}}\Bigl(r_\Lambda+\frac{2}{n-2}\Bigr),
\end{equation}
and
\begin{equation}
  \label{eq:pdef}
  \Ptt
:=\frac{\sqrt{{2(n-1)}(n-2)}}{{2(n-1)} +(n-2)\sum_{\Lambda=1}^{n-1}r_\Lambda},
\end{equation}
which we note is well-defined since $\kf_I{}^I=\sum_{\Lambda=1}^{n-1}r_\Lambda\geq 0$.
From these formulas, we deduce that the $q_\Lambda$ satisfy the ``standard'' Kasner relations
\begin{equation}\label{eq:Kasnerrel}
  \sum_{\Lambda=1}^{n-1}q_\Lambda =1 \AND 
  \sum_{\Lambda=1}^{n-1}q_\Lambda^2 = 1-2 \Ptt^2.
\end{equation}
As discussed in \cite[\S1.3]{BeyerOliynyk:2021}, the quantity $\Ptt$ can be interpreted as the asymptotic scalar field strength.

Assume now that there exist continuous positive functions  $\bfr$ and $\nu$ such that
\begin{equation}
  \label{eq:lapselimit}
  \Bigl|t^{-\frac{1}{2}r_0(x)}\Ntt (t,x)-\bfr(x)\Bigr|\lesssim t^{\nu(x)}
\end{equation}
for all $(t,x)\in (0,t_0]\times\Tbb^{n-1}$. This property is clearly satisfied by the Kasner-scalar field spacetimes with and without fluids, and we show that it continues to hold for the entire class of perturbations of the FLRW fluid solutions that are generated from our past stability result; see Theorem \ref{glob-stab-thm}.(b). Given \eqref{eq:lapselimit},  it then follows that solutions that satisfy the asymptotically pointwise Kasner
condition \eqref{eq:asymptptwKasner2} will behave pointwise in a manner that is similar to that of the  Kasner-scalar field solutions, c.f.~\eqref{eq:KSFlapse2ndFF}. Thus, the asymptotically pointwise Kasner condition provides the sense in which the pertrubed solutions behave asymptotically at each spatial point like a Kasner-scalar field spacetime.

The second fundamental form $\bar\Ktt_{\Lambda\Omega}$ induced on $t=const$-surfaces by the physical metric $\gb_{\mu\nu}$ is related to the one $\Ktt_{\Lambda\Omega}$ induced by the conformal metric $g_{\mu\nu}$ via
\[\bar\Ktt_\Lambda{}^{\Omega}=\frac 12 t^{-\frac {1}{n -2}} t^{-1}\Ntt^{-1}\Bigl(2 t\Ntt\Ktt _{\Lambda}{}^{\Omega}
    +\frac {2}{n-2}{\delta}_{\Lambda}{}^{\Omega}\Bigr).\]
Because of this, \eqref{eq:asymptptwKasner2} and \eqref{eq:lapselimit}  imply that both the mean curvatures associated with the physical and with the conformal metric diverge pointwise near $t=0$, except in the case of FLRW where $K_{\Lambda}{}^{\Lambda}$ vanishes while $\bar{K}_{\Lambda}{}^{\Lambda}$ diverges. Given suitable uniform bounds over the spatial domain $U$, which we prove hold for the perturbed FLRW solutions generated by our main stability result, see Theorem \ref{glob-stab-thm}.(d), asymptotically pointwise Kasner metrics will have \emph{crushing singularities} at $t=0$  in the language of \cite{eardley1979}.

\subsection{An informal statement of the main stability theorem}
% \label{sec:informal}
The main result of this article is that we establish
the nonlinear stability in the contracting direction of perturbations of the FLRW solution \eqref{eq:FLRWEulerSFExplSol1}-\eqref{eq:FLRWEulerSFExplSol4} to the Einstein-Euler-scalar field equations in $n\geq 3$ spacetime dimensions and for sound speeds satisfying $1/(n-1)<c_s^2 < 1$; see Theorem~\ref{glob-stab-thm} for the precise statement of our stability result. We also show that the perturbed FLRW solutions are asymptotically pointwise Kasner in the sense of Definition~\ref{def:APKasner}, and that they terminate in a big bang singularity, which rigorously confirms the \emph{matter does not matter} paradigm for these solutions. 

An informal statement of our stability result is given below in the following theorem. However, before stating it, we first discuss \textit{synchronized initial data}. Initial data that is prescribed on the hypersurface $\Sigma_{t_0}=\{t_0\}\times \Tbb^{n-1}$, $t_0>0$ will be said to be \textit{synchronized} if  $\tau=t_0$ on $\Sigma_{t_0}$. The purpose of this synchronization condition is to ensure that the big bang singularity occurs at $\tau=0$; see the discussion in Section \ref{temp-synch} for details.  As noted in Remark \ref{synch-rem}, no generality is lost from restricting our attention to synchronized initial data.

\begin{thm}[Past global stability of the FLRW solution of the Einstein-Euler-scalar field equations]
Solutions $\{g_{ij},\tau, V^i\}$ of the conformal Einstein-Euler-scalar field equations that are 
generated from sufficiently differentiable, synchronized initial data imposed on $\Sigma_{t_0}=\{t_0\}\times \Tbb^{n-1}$ that is suitably close to FLRW initial data exist on the spacetime region $M \cong (0,t_0]\times \Tbb^{n-1}$ provided $c_s^2\in (1/(n-1),1)$. Moreover, these solutions are asymptotically pointwise Kasner, the fluid is asymptotically comoving\footnote{That is, the spatial fluid vector vanishes at $\tau=0$.},
and the corresponding physical solutions $\{\gb_{ij},\phi,\Vb^i\}$ of the Einstein-Euler-scalar field equations are past timelike geodesically incomplete, terminate at a crushing big bang singularity at $\tau=0$ that is characterised by curvature blow-up, and are $C^2$-inextendible through the $\tau=0$ boundary of $M$.
\end{thm}

The restriction $1/(n-1)<c_s^2<1$ on the sound speed has been observed in earlier studies \cite{beyer2017,BeyerOliynyk:2020} in the spacetime dimension $n=4$. In the terminology of  \cite{beyer2017}, the condition $1/(n-1)<c_s^2<1$ is referred to as the \emph{subcritical regime} and in that article solutions to the Einstein-Euler equations with a Gowdy symmetry\footnote{It is worth noting here that the Gowdy symmetry is responsible for the suppression of the oscillatory behavior of the fields near the big bang singularity, and explains why in this work the authors are able to establish the existence of solutions with monotone behaviour near the singularity without needing to couple the system to a scalar field.} are constructed by specifying asymptotic initial data on the big bang singularity and generating solution from this data by solving a Fuchsian singular initial value problem. The results of \cite{beyer2017} show that there exist families of solutions to the Einstein-Euler equations with Gowdy symmetry that have big bang singularities. The behavior of these solutions near the big bang singularity is similar to the solutions we obtain in this article from our stability theorem. Interestingly in  \cite{beyer2017}, families of solutions, under additional assumptions beyond Gowdy symmetry, are also constructed for the coupled Einstein-Euler equations and for the Euler equations on fixed Kasner backgrounds in the \textit{critical} and \textit{supercritical} regimes, respectively, that correspond to the sound speeds $c_s^2=1/(n-1)$ and $0<  c_s^2<1/(n-1)$, respectively. 
What is strongly suggested by the results of \cite{beyer2017} is that the asymptotics of solutions to the Einstein-Euler-scalar field equations for sound speeds satisfying $0\leq c_s^2< 1/(n-1)$ will be very different to that of the perturbed solutions consider in this article, which satisfy the subcritial condition $1/(n-1)<c_s^2 < 1$.  As we explain in more detail in Remark~\ref{rem:borderlineextension}, it is not obvious for Einstein-Euler-scalar field equations considered here how the dynamics might change in the critical $c_s^2=1/(n-1)$ or supercritical  $0\leq c_s^2< 1/(n-1)$ regimes. We plan to explore this question in future work.

\begin{rem} \label{local-stab-rem}
It is worth mentioning here that Theorem~\ref{glob-stab-thm} should be interpreted as a past global in space stability theorem since it requires that the initial data be specified on the  entire closed hypersurface $\Sigma_{t_0}=\{t_0\}\times\Tbb^{n-1}$. It is straightforward to establish a local in space version of Theorem~\ref{glob-stab-thm} by adapting the proof of Theorem 11.1 from \cite{BeyerOliynyk:2021}. Doing so would yield the existence of solutions to the Einstein-Euler-scalar field equations on truncated cones domains that (i) are generated from initial data that is sufficiently close to FLRW on an open subset of $\Sigma_{t_0}$, (ii) are asymptotically pointwise Kasner, and  (iii) terminate in a crushing AVTD big bang singularity characterised by curvature blow-up. 
We leave the details of the proof to the interested reader. 
\end{rem}

\subsection{Overview of the proof of Theorem~\ref{glob-stab-thm}\label{overview}} 
The proof our main stability result, Theorem~\ref{glob-stab-thm}, is based on an adaptation of the proof of Theorem~10.1 from
\cite{BeyerOliynyk:2021}. As in \cite{BeyerOliynyk:2021}, our proof of Theorem~\ref{glob-stab-thm} relies on two different formulations of the reduced Einstein-Euler-scalar field equations, which are used
for distinct purposes. 

The first formulation, given by \eqref{tconf-ford-C.1}-\eqref{tconf-ford-C.9} below, is used to establish the local-in-time existence and uniqueness of solutions
to the reduced conformal Einstein-scalar field equations in a \textit{Lagrangian coordinate system} $(x^\mu)$, which is adapted to the vector field $\chi^\mu = (|\nabla \tau|^2_g)^{-1}\nabla^\mu\tau$,
as well as a continuation principle for these solutions. See Proposition~\ref{lag-exist-prop} for the precise statement of the local-in-time existence, uniqueness and continuation result. 

Here, Lagrangian coordinates mean that in the coordinate system $(x^\mu)$ the vector field $\chi^\mu$ is trivialized, that is, $\chi^\mu =\delta^\mu_0$.
The precise definition of the Lagrangian coordinates $(x^\mu)$ can be found in  
Section \ref{Lag-coordinates}. For initial data that satisfies the gravitational and wave gauge  constraints, the system \eqref{tconf-ford-C.1}-\eqref{tconf-ford-C.9} propagates both of these constraints
and determines solutions of the conformal Einstein-scalar field equations. An important point regarding the wave gauge constraint
$\frac{1}{2}g^{\gamma \lambda}(2 \Dc_\mu g_{\nu\lambda}-\Dc_\lambda g_{\mu\nu}) = 0$
is that the covariant derivative $\Dc_\mu$  is not determined by a fixed Minkowski metric in the Lagrangian coordinates $(x^\mu)$, and consequently, the Lagrangian coordinates $(x^\mu)$ are \emph{not} wave coordinates, that is, generically $\Box_g x^\mu \neq 0$. Instead, the covariant derivative $\Dc_\mu$ is computed with respect to the flat metric
$\gc_{\mu\nu} = \del{\mu}l^\alpha \eta_{\alpha\beta} \del{\nu}l^\beta$
where the Lagrangian map $l^\mu(x)$ is determined by a solution of the system \eqref{tconf-ford-C.1}-\eqref{tconf-ford-C.9}; see Section \ref{Lag-coordinates} for details. The primary role of the Lagrangian coordinates $(x^\mu)$ is to synchronize the singularity. In these coordinates, the scalar field $\tau$ coincides with the time coordinate, that is,
$\tau = t:= x^0$; see Section \ref{temp-synch} for details.

While the system \eqref{tconf-ford-C.1}-\eqref{tconf-ford-C.9} is useful for establishing the local-in-time existence of solutions to the reduced conformal Einstein-scalar field equations and the propagation of the wave gauge constraint $\frac{1}{2}g^{\gamma \lambda}(2 \Dc_\mu g_{\nu\lambda}-\Dc_\lambda g_{\mu\nu}) = 0$,
it is not useful for establishing global-in-time estimates that can be used in conjunction with the continuation principle to show that solutions can be continued from some starting time $t_0>0$ all the way down to the big bang singularity at $t=0$. The system that we do use to establish global-in-time estimates is formulated in terms of a frame $e_i=e_i^\mu\del{\mu}$, the connection coefficients $\gamma_i{}^k{}_j$
of the flat background metric $\gc_{\mu\nu} = \del{\mu}l^\alpha \eta_{\alpha\beta} \del{\nu}l^\beta$ relative to the frame $e_i$, i.e. $D_{e_i}e_j = \gamma_{i}{}^k{}_j e_j$,
and suitable combinations of the metric, scalar and fluid fields and their derivatives:  
\begin{equation*} %\label{fields-intro}
  \begin{split}
\{&g_{ijk}=\Dc_i g_{jk},g_{ijkl}=\Dc_{i}\Dc_j g_{kl},\tau_{ij}= \Dc_i\Dc_j \tau,\\ 
&\tau_{ijk}=\Dc_i \Dc_j\Dc_k \tau,W^k,\Ut^k_i=\Dc_i W^k\}, 
  \end{split}
\end{equation*}
where $g_{ij}=e_i^\mu g_{\mu\nu}e^\nu_j$
is the frame representation of the conformal metric,
$\tau=t$, and $W^k = \ftt^{-1}V^k$ with    $\ftt=\tau^{\frac{(n-1)c_s^2-1}{n-2}}\betat^{c_s^2}$ and $\betat=(-|\nabla\tau|^2_g)^{-\frac{1}{2}}$.

The first step toward deriving the second formulation of the reduced Einstein-Euler-scalar field equations that is used to derive global-in-time estimates is to fix the frame $e_i^\mu$, which we do in Section \ref{Fermi}, by first setting
$e_0^\mu = (-|\chi|_g^2)^{-\frac{1}{2}}\chi^\mu$, where $\chi^\mu=\delta^\mu_0$ since we are using Lagrangian coordinates $(x^\mu)$. The spatial frame vectors $e_I^\mu$ are then determined by using Fermi-Walker transport, which is defined by
\begin{equation*}
\nabla_{e_0}e_J = -\frac{g(\nabla_{e_0}e_0,e_J)}{g(e_0,e_0)} e_0,
\end{equation*}
to propagate initial data $e_I^\mu|_{t=t_0}=\er^\mu_I$ that is chosen so that the frame is orthonormal at $t=t_0$. 
The orthonormality of the frame is preserved by Fermi-Walker transport, which implies, in particular, that the frame metric satisfies 
$g_{ij}=\eta_{ij}$. Once the frame is determined via Fermi-Walker transport, then the back ground connection coefficients $\gamma_i{}^k{}_j$ are determined via a system of evolution equations derived from the vanishing of the background curvature and identities derived from the Fermi-Walker transport equation; see Section \ref{Fermi} for details. We note that our use of a Fermi-Walker transported spatial frame was inspired by the work of \cite{Fournodavlos_et_al:2023} in which Fermi-Walker transported spatial frames played an essential role in the proof of the stability results established there. 

With the frame fixed, a first order formulation of the reduced conformal Einstein-Euler-scalar fields equations in terms of the frame variables 
\begin{align*}
\kt&=(\kt_{IJ}) :=(g_{0IJ}-g_{I0J}-g_{J0I}),
\\
\betat &=(-|\nabla\tau|^2_g)^{-\frac{1}{2}},\\
\ellt&=(\ellt_{IjK}) := (g_{IjK}), &&  \\
\mt&=(\mt_{I})  := (g_{00M}),\\
\tau &= (\tau_{ij}), \\
\gt&=(\gt_{Ijkl}) := (g_{Ijkl}),\\    
  \taut&=(\taut_{Ijk}) := (\tau_{Ijk}), \label{taut-def}\\
  W&=(W^i),\\
  \Utt&=(\Utt^k_Q),
\intertext{and}
\psit&=(\psit_I{}^k{}_J):=(\gamma_I{}^k{}_J)  
\end{align*}
is derived in Sections \ref{sec:FirstOrderForm} and \ref{sec:chvar1}, see in particular, equations \eqref{for-G.1.S2}-\eqref{for-Euler.2.2}.
The metric combination $\kt_{ij}$, which it is related to the second fundamental form of the conformal metric for the $t=const$-hypersurfaces, c.f.\eqref{Ktt-def},
plays a pivotal role in our analysis.
The property that distinguishes $\kt_{ij}$, as far as the analysis is concerned, is we have \textit{no freedom} to rescale the normalized version
\begin{equation*}
k_{IJ} = t \betat \kt_{IJ}
\end{equation*}
by any power of $t$.
Our stability proof relies on showing that $k_{IJ}$ remains bounded as $t\searrow 0$, and in fact, we show that $2k_{IJ}$ converges as $t\searrow 0$ to a, in general, non-vanishing symmetric matrix $\kf_{IJ}$ satisfying $\kf_I{}^I\geq 0$ and $(\kf_I{}^I)^2-\kf_I{}^J\kf_J{}^I+4\kf_I{}^I=0$. On the other hand, there is slack in the remaining variables  in the sense that we can rescale them by certain positive powers of $t$. This freedom to rescale these variables is essential to our stability proof.

To complete the derivation of the second formulation of the reduced Ein\-stein-Eu\-ler-sca\-lar field equations, we introduce, in Section \ref{sec:chvar2}, the rescaled variables 
\begin{align*}
k&=(k_{IJ}) :=(t\betat \kt_{IJ}),
\\
  \beta &=  t^{\ep_0}\betat,\\
 \check\beta&=t^{\ep_3+(1+c_s^2)\ep_0}\beta^{-(1+c_s^2)},\\
\ell&=(\ell_{IjK}):= (t^{\ep_1}\ellt_{IjK}), \\
m&=(m_{I}) := (t^{\ep_1}\mt_I), \\
\xi&=(\xi_{ij}) := (t^{\ep_1-\ep_0}\tau_{ij}), \\
\psi &=(\psi_I{}^k{}_J) :=  (t^{\ep_1}\psit_I{}^k{}_J),\\
f&=(f_I^\Lambda) := (t^{\ep_2} e^\Lambda_I), \\
\gac&=(\gac_{Ijkl}) := (t^{1+\ep_1} \betat \gt_{Ijkl}), \\
\tauac &=(\tauac_{Ijk}):= (t^{\ep_0+2 \ep_1} \taut_{Ijk}),\\
  W&=(W^k),\\
  \acute{\Utt} &=(t^{\ep_4} \Utt^s_Q), 
\end{align*}
where the $\ep_0,\ldots,e_4$ are constants.
Expressing the first order system \eqref{for-G.1.S2}-\eqref{for-Euler.2.2} in terms of these rescaled variables, see Section \ref{sec:Fuch-form}, yields a Fuchsian system of equations of the form
\begin{equation} \label{Fuch-eqn-intro}
    A^0 \del{t}u +\frac{1}{t^{\ep_0+\ep_1}}A^\Lambda \del{\lambda}u = \frac{1}{t}\Ac \Pbb u + F,
\end{equation}
where
\begin{align*}
u = \bigl(k_{LM},m_M,\ell_{R0M},\ell_{RLM},\xi_{rl},&\beta,f^\Lambda_I,\psi_I{}^k{}_{J},\tauac_{Qjl}, \gac_{Qjlm}, \check\beta, W^s,{\acute{\Utt}^s_Q}\bigr)^{\tr} \\
& - \bigl(0,0,0,0,0, t^{\ep_0}, t^{\ep_2}\delta_{I}^{\Lambda},0,0,0,t^{\ep_3},V_*^0\delta_0^j,0\bigr)^{\tr}
\end{align*}
and $\Pbb$
is the projection matrix 
\begin{multline*}
 \Pbb = \diag\Bigl(0,\delta_{\Mt}^{ M},\delta_{\Rt}^{R}\delta_{\Mt}^{M},\delta_{\Rt}^{R} \delta_{\Lt}^{ L} \delta_{\Mt}^{M},\delta_{\rt}^{r}\delta_{\lt}^{l},1,\delta_{\It}^{I}\delta^{\Lambdat}_{\Lambda},\delta_{\It}^{I}\delta^{\kt}_{k} \delta_{\Jt}^{J},\\\delta_{\Qt}^{Q}\delta_{\jt}^{j}\delta_{\lt}^{l},\delta_{\Qt}^{Q}\delta_{\jt}^{j}\delta_{\lt}^{l} \delta_{\mt}^{m},\delta_{\tilde J J}\delta_{\jt}^{\Jt}\delta_j^J , \delta_{\tilde j j} \delta^{\Qt Q}\Bigr).
 \end{multline*}

The purpose of introducing the $t$-power weights determined by the constants $\ep_0,\ldots,e_4$ in the variables above is to obtain a matrix $\Ac$ where the eigenvalues of $\frac{1}{2}(\Ac+\Ac^{\tr})\Pbb$ are all non-negative. This is essential for our existence proof. At the same time, we have to ensure that no singular terms worse than $t^{-1}$ appear in \eqref{Fuch-eqn-intro}. Both of these requirements  
necessitates choosing the constants $\ep_0,\ldots,e_4$ and the square of the sound speed $c_s^2$ to satisfy the inequalities  \eqref{eq:epscond.N} and $\frac{1}{n-1}<c_s^2 <1$; see Remark \ref{rem:borderlineextension}, Lemma \ref{lem:sourceterm} and Lemma \ref{lem:posdef1} for details.  
It is also worth noting that the zero in the first diagonal component of $\Pbb$, which corresponds to the zero eigenvalue block of $\frac{1}{2}(\Ac+\Ac^{\tr})\Pbb$, is responsible for the convergence as $t\searrow 0$ of $2k_{IJ}$ to a, generally non-vanishing, matrix $\kf_{IJ}$. On the other hand, remaining eigenvalues of $\frac{1}{2}(\Ac+\Ac^{\tr})\Pbb$, which are all positive, lead to power law decay, i.e. $t^a$ with $a>0$, for the other variables where the decay rates\footnote{That is the $a$'s where there is a different $a$ for each of the different groups of variables.} are determined by the  eigenvalues.

The virtue of the Fuchsian formulation \eqref{Fuch-eqn-intro} is that we can appeal to the existence theory developed in the articles\footnote{The actual existence theory we apply is from \cite{BeyerOliynyk:2020}, which a slight generalization of the existence theory from \cite{BOOS:2021}. } \cite{BeyerOliynyk:2020,BOOS:2021} to conclude, for suitably small choice of initial data $u_0$ at $t=t_0>0$, that there exist a unique solution of \eqref{Fuch-eqn-intro} that is defined all the way down to $t=0$ and satisfies $u|_{t=t_0}=u_0$. The Fuchsian existence theory also yields energy and decay estimates that provide uniform control over the behaviour of solutions in the limit $t\searrow 0$. The precise statement of the global existence result for the Fuchsian equation \eqref{Fuch-eqn-intro} is given in Proposition \ref{prop:globalstability}.

On one hand, Proposition \ref{prop:globalstability} yields the existence of a unique solution on $(0,t_0]\times \Tbb^{n-1}$ to the Fuchsian equation \eqref{Fuch-eqn-intro} generated from initial data\footnote{Here, the initial data for \eqref{Fuch-eqn-intro} is assumed to be derived from initial data for the reduced conformal Einstein-scalar equations that satisfies the gravitational and wave gauge constraint equations.} $u|_{t=t_0} = u_0$ that is sufficiently close to FLRW initial data.
On the other hand, this same initial data generates, by 
Proposition \ref{lag-exist-prop}, a local-in-time solution to the the system 
\eqref{tconf-ford-C.1}-\eqref{tconf-ford-C.9} that, after solving the Fermi-Walker transport equations for the spatial frame fields, determines a solution of the Fuchsian equation \eqref{Fuch-eqn-intro}. By uniqueness, these two solutions must be the same. The energy estimates from Proposition \ref{prop:globalstability} then allows us to conclude via the continuation principle from Proposition \ref{lag-exist-prop} that the solution $u$ of the Fuchsian equation determines a solution of the conformal Einstein-scalar field equations on $(0,t_0]\times \Tbb^{n-1}$. Asymptotic properties of the solution to the conformal Einstein-scalar field equations are then deduced from the energy and decay estimates for $u$ from Proposition \ref{prop:globalstability}. This completes the overview of the major steps involved in our proof of the past global nonlinear stability of small perturbations of the FLRW solutions to the conformal Einstein-Euler-scalar field equations.
The precise statement of this result is presented in Theorem \ref{glob-stab-thm} and the proof can be found in Section  \ref{sec:proof_globstab}.  

\section{Preliminaries\label{prelim}}

% \subsection{Data availability statement}

% This article has no associated data.

\subsection{Coordinates, frames and indexing conventions\label{indexing}}
In the article, we will consider $n$-dimensional spacetime manifolds of the form
\begin{equation} \label{Mt1t0-def}
    M_{t_1,t_0}= (t_1,t_0]\times \Tbb^{n-1},
\end{equation}
where $t_0>0$, $0<t_1<t_0$, and $\Tbb^{n-1}$ is
the $(n-1)$-torus defined by
\begin{equation} \label{Tbb-def}
    \Tbb^{n-1} = [-L,L]^{n-1}/\sim
\end{equation} 
with $\sim$ the equivalence relation obtained
from identifying the sides of the box $[-L,L]^{n-1}\subset \Rbb^{n-1}$. On $M_{t_1,t_0}$,
we will always employ coordinates $(x^\mu)=(x^0,x^\Lambda)$
where the $(x^\Lambda)$ are periodic spatial coordinates
on $\Tbb^{n-1}$ and $x^0$ is a time coordinate
on the interval $(t_1,t_0]$. Lower case Greek letters, e.g. $\mu,\nu,\gamma$, will run from $0$ to $n-1$ and be used to label spacetime coordinate indices while upper case Greek letters, e.g. $\Lambda,\Omega,\Gamma$, will run from $1$ to $n-1$ and label spatial coordinate indices. Partial derivative with respect to the coordinates $(x^\mu)$ will be denoted by $\del{\mu} = \frac{\del{}\;}{\del{}x^\mu}$.
We will often use $t$ to denote the time coordinate $x^0$, that is, $t=x^0$, and use the notion $\del{t} = \del{0}$ for the partial derivative with respect to the coordinate $x^0$.

We will use frames $e_j= e_j^\mu \del{\mu}$ 
throughout this article. Lower case Latin letter, e.g. $i,j,k$, will be used to label frame indices and they will run from $0$ to $n-1$ while spatial frame indices will be labelled by upper case Latin letter, e.g. $I,J,K$, that run from $1$ to $n-1$.

\subsection{Inner-products and matrices}
Throughout this article, we denote the Euclidean inner-product by $\ipe{\xi}{\zeta} = \xi^{\tr} \zeta$, $\xi,\zeta \in \Rbb^N$, and use
$|\xi| = \sqrt{\ipe{\xi}{\xi}}$
to denote the Euclidean norm. The set of $N\times N$ matrices is denoted by  $\Mbb{N}$, and we  use $\Sbb{N}$ to denote the subspace of symmetric $N\times N$-matrices. 

Given $A\in \Mbb{N}$, we define
the operator norm $|A|_{\op}$ of $A$ via
\begin{equation*}
   |A|_{\op} = \sup_{\xi\in \Rbb^N_\times} \frac{|A\xi|}{|\xi|},
\end{equation*}
where $\Rbb^N_\times = \Rbb^N\setminus\{0\}$.
For any $A,B\in \Mbb{N}$, we will also employ the notation
\begin{equation*}
    A\leq B \quad \Longleftrightarrow \quad \xi^{\tr}A\xi \leq \xi^{\tr}B\xi, \quad \forall \; \xi \in \Rbb^N.
\end{equation*}

\subsection{Sobolev spaces and extension operators\label{Sobolev}}
The $W^{k,p}$, $k\in \Zbb_{\geq 0}$, norm of a map $u\in C^\infty(U,\Rbb^N)$ with $U\subset \Tbb^{n-1}$ open is defined by
\begin{equation*}
\norm{u}_{W^{k,p}(U)} = \begin{cases} \begin{displaystyle}\biggl( \sum_{0\leq |\Ic|\leq k} \int_U |D^{\Ic} u|^p \, d^{n-1} x\biggl)^{\frac{1}{p}}  \end{displaystyle} & \text{if $1\leq p < \infty $} \\
 \begin{displaystyle} \max_{0\leq |\Ic| \leq k}\sup_{x\in U}|D^{\Ic} u(x)|  \end{displaystyle} & \text{if $p=\infty$}
\end{cases},
\end{equation*}
where $\Ic=(\Ic_1,\ldots,\Ic_{n-1})\in \Zbb_{\geq 0}^{n-1}$ denotes a multi-index and we write
$D^\Ic = \del{1}^{\Ic_1}\del{2}^{\Ic_2}\cdots\del{n-1}^{\Ic_{n-1}}$.
The Sobolev space $W^{k,p}(U,\Rbb^N)$ is then defined to be the completion of $C^\infty(U,\Rbb^N)$ with respect to the norm
$\norm{\cdot}_{W^{k,p}(U)}$. When $N=1$ or the dimension $N$ is clear from the context, we will simplify notation and write $W^{k,p}(U)$ instead of $W^{k,p}(U,\Rbb^N)$, and we will employ the standard notation $H^k(U,\Rbb^N)=W^{k,2}(U,\Rbb^N)$ throughout.

% To each centred ball $\mathbb{B}_\rho\subset \Tbb^{n-1}$, $0<\rho<L$, we
% assign a (non-unique) total
% extension operator
% \begin{equation} \label{Ebb-def}
% \Ebb_\rho \: : \: H^k(\mathbb{B}_\rho,\Rbb^N)\longrightarrow H^k(\Tbb^{n-1},\Rbb^N), \qquad k \in \Zbb_{\geq 0},
% \end{equation}
% that satisfies
% \begin{equation}\label{Ebb-prop}
% \Ebb_\rho(u)\bigl|_{\mathbb{B}_\rho} = u, \AND
% \norm{\Ebb_{\rho}(u)}_{H^{k}(\Tbb^{n-1)}} \leq C\norm{u}_{H^k(\mathbb{B}_\rho)}
% \end{equation}
% for some constant $C=C(k,n,\rho)>0$ independent of $u\in H^k(\mathbb{B}_\rho)$. The existence of such an operator is established in
% \cite{AdamsFournier:2003}; see Theorems 5.21 and 5.22, and Remark 5.23 for details.

\subsection{Constants and inequalities}
We use the standard notation $a \lesssim b$
for inequalities of the form
$a \leq Cb$
in situations where the precise value or dependence on other quantities of the constant $C$ is not required.
On the other hand, when the dependence of the constant on other inequalities needs to be specified, for
example if the constant depends on the norm $\norm{u}_{L^\infty}$, we use the notation
$C=C(\norm{u}_{L^\infty})$.
Constants of this type will always be non-negative, non-decreasing, continuous functions of their argument.

We will also employ the order notation from \cite[\S 2.4]{BOOS:2021}. Since we are working with trivial bundles, we can define this notation as follows: Given maps
\begin{equation*}
  \begin{split}
f&\in C^0\bigl((0,t_0],C^\infty(B_R(\Rbb^n)\times B_R(\Rbb^m),\Rbb^p)\bigr),
\\
g&\in C^0\bigl((0,t_0],C^\infty(B_{R}(\Rbb^m),\Rbb^q)\bigr),
  \end{split}
\end{equation*}
where $t_0,R>0$ are positive constants,
we say that
\begin{equation*}
f(t,w,v) = \Ordc(g(t,v))  
\end{equation*}
if there exist a $\Rt \in (0,R)$ and a map
\begin{equation*}
\ft \in C^0\bigl((0,t_0],C^\infty\bigl(B_{\Rt}(\Rbb^n)\times B_{\Rt}(\Rbb^m),L(\Rbb^q,\Rbb^p)\bigr)\bigr)
\end{equation*}
such that
\begin{gather*}
f(t,w,v) = \tilde{f}(t,w,v)g(t,v),\quad
|\ft(t,w,v)| \leq 1 \AND
|D^s_{w,v}\ft(t,w,v)| \lesssim 1
\end{gather*}
for all $(t,w,v) \in  (0,t_0] \times B_{\Rt}(\Rbb^n) \times B_{\Rt}(\Rbb^m)$ and $s\geq 1$. For situations, where we want to bound $f(t,w,v)$ by $g(t,v)$ up to an undetermined constant of proportionality,
we define
\begin{equation*}
f(t,w,v) = \Ord(g(t,v))  
\end{equation*}
if there exist a $\Rt \in (0,R)$ and a map
\begin{equation*}
\ft \in C^0\bigl((0,t_0],C^\infty(B_R(\Rbb^n)\times B_R(\Rbb^m),L(\Rbb^q,\Rbb^p))\bigr)
\end{equation*}
such that 
\begin{gather*}
f(t,w,v) = \tilde{f}(t,w,v)g(t,v)\AND
|D^s_{w,v}\ft(t,w,v)| \lesssim 1
\end{gather*}
for all $(t,w,v) \in  (0,t_0] \times B_{\Rt}(\Rbb^n) \times B_{\Rt}(\Rbb^m)$ and $s\geq 0$.

\subsection{Curvature}
The curvature of tensor $\Rc_{ijk}{}^l$ of the background metric $\gc_{ij}$ is defined via
\begin{equation}\label{comm-Rc}
  [\Dc_i,\Dc_j]\alpha_k = \Rc_{ijk}{}^l\alpha_l
\end{equation}
for arbitrary $1$-forms $\alpha_l$.
This definition along with $\Rc_{ik}=\Rc_{ijk}{}^j$ for
the Ricci tensor fixes the curvature conventions that will be employed for all curvature tensors appearing in this article.   

\section{Reduced conformal field equations} In order to establish the existence of solutions to the conformal Einstein-Euler-scalar field equations, we need to replace the conformal Einstein equations \eqref{confESFAb} or equivalently \eqref{confESFAaN} with a gauge reduced version. In the following, we employ a (conformal) wave gauge defined by the constraint
\begin{equation}\label{wave-gauge}
X^k =0,
\end{equation}
where  $X^k$ is given above by \eqref{Xdef},
and consider the wave gauge reduced equations 
\begin{equation}\label{confESFF}
  \begin{split}
  -2R_{ij}+2\nabla_{(i} X_{j)}
  &=-\frac{2}{\tau}\nabla_i \nabla_j \tau
  -2\Ttt_{ij}\\
  &\oset{\eqref{Ccdef}}{=}-\frac{2}{\tau}\bigl(
  \Dc_i \Dc_j \tau - \Cc_i{}^k{}_j\Dc_k \tau  \bigr)
  {-2\Ttt_{ij}},
  \end{split}
\end{equation}
which we will refer to as the \textit{reduced conformal Einstein equations}. Recall that $\Ttt_{ij}$ is defined by \eqref{eq:Tttconfphys}.

For the moment, we \emph{assume} that the wave gauge constraint \eqref{wave-gauge} holds. Because we establish below in Proposition \ref{prop-ES-constr} that this wave gauge constraint propagates, we lose nothing by making this assumption. Now, by \eqref{Ccdef}, \eqref{Xdef} and \eqref{wave-gauge}, we observe that the conformal scalar field equation \eqref{confESFCb}
can be expressed as
\begin{equation} \label{confESFG}
  g^{ij}\Dc_{i}\Dc_{j}\tau=0.
\end{equation}
Further, using \eqref{Ccdef}, \eqref{eq:conffluidrel} and \eqref{eq:AAACF2}, we note that the conformal Euler equations \eqref{confESAc} can be expressed as
\begin{equation}
  \label{confESFH}
  \att^i{}_{j k}\Dc_i V^k
  =\Gt_{jsl}V^s V^l
\end{equation}
where 
\begin{equation}
  \label{eq:DefGjsl}
  \begin{split}
  \Gt_{jsl}=&\tau^{-1}
  \frac{(n-1)c_s^2-1}{c_s^2(n-2)} g_{j(s} \Dc_{l)}\tau \\    
  &- \frac{1}{2}\Bigl(
  \frac{3 c_s^2+1}{c_s^2} \frac{V^q V^p}{v^2}+g^{pq}\Bigr)  g_{j(s}\Dc_{l)} g_{p q} 
  -  \Dc_{(l} g_{s) j}.
  \end{split}
\end{equation}

Gathering \eqref{confESFF}, \eqref{confESFG} and \eqref{confESFH} together, we have
\begin{align}
-2R_{ij}+2\nabla_{(i} X_{j)}& =-\frac{2}{\tau}\bigl(
  \Dc_i \Dc_j \tau - \Cc_i{}^k{}_j\Dc_k \tau  \bigr)
  {-2\Ttt_{ij}}, 
\label{confESFFa}\\
  g^{ij}\Dc_{i}\Dc_{j}\tau    &=0, \label{confESFGa}\\
  \label{confESFHa}
  \att^i{}_{j k}\Dc_i V^k
  &=\Gt_{jsl}V^s V^l.
\end{align}
We will refer to these equations as the \textit{reduced conformal Einstein-Euler-scalar field equations}. For use below, we recall that the reduced Ricci tensor can be expressed as
\begin{equation} \label{red-Ricci}
-2R_{ij}+2\nabla_{(i} X_{j)}=g^{kl}\Dc_k \Dc_l g_{ij} + Q_{ij}
+2g^{kl}g_{m(i}\Rc_{j)kl}{}^m
\end{equation}
where 
\begin{equation}\label{Q-def}
  \begin{split}
Q_{ij} = \frac{1}{2}g^{kl}g^{mn}\Bigl(&\Dc_i g_{mk} \Dc_j g_{n l}
+2 \Dc_{n}g_{il}\Dc_{k}g_{jm} - 2\Dc_{l}g_{in} \Dc_{k}g_{jm}\\
&-2 \Dc_{l}g_{in}\Dc_j g_{mk} -2 \Dc_{i}g_{mk}\Dc_{l}g_{jn}\Bigr),
  \end{split}
\end{equation}
and, as above, $\Rc_{ijk}{}^l$ denotes the curvature tensor of the background metric $\gc_{ij}$. By differentiating \eqref{confESFGa} and employing the commutator formula
\[\Dc_k\Dc_i \Dc_j\tau-\Dc_i\Dc_j \Dc_k \tau
%= \Dc_k\Dc_i \Dc_j\tau-\Dc_i\Dc_k \Dc_j \tau
=\Rc_{kij}{}^l\Dc_l\tau,\]
we also note that
\begin{align} \label{confESFI} 
  g^{ij}\Dc_i\Dc_{j} \Dc_{k}\tau
  =& g^{il}g^{jm}\Dc_kg_{lm}  \Dc_{i} \Dc_{j}\tau
     -g^{ij}\Rc_{kij}{}^l \Dc_l \tau.
\end{align}

\section{Choice of background metric}
The background metric $\gc_{ij}$ is thus far arbitrary. Since the conformal FLRW metric in \eqref{eq:FLRWEulerSFExplSol1} -- \eqref{eq:FLRWEulerSFExplSol4} is flat in leading order at $t=0$ and we are interested in nonlinear perturbations of this solution, we are motivated to restrict our attention to background metrics that are flat, which by definition, means that the curvature tensor vanishes, that is,
\begin{equation}\label{curvature}
\Rc_{ijk}{}^l = 0.
\end{equation}
By the commutator formula \eqref{comm-Rc}, the vanishing of the curvature implies that
\begin{equation}\label{commutator}
[\Dc_i,\Dc_j] =0.
\end{equation}

\section{Local existence and continuation in Lagrangian coordinates}
\label{sec:locexist_cont}

Thus far, we have expressed the reduced conformal Einstein-Euler-scalar field equations \eqref{confESFFa}-\eqref{confESFHa} in an arbitrary frame. We now turn to establishing a local-in-time existence and uniqueness result for solutions to these equations along with a continuation principle. We do so following the same approach as in \cite[\S5]{BeyerOliynyk:2021}, which involves formulating the conformal Einstein-Euler-scalar field equations as a first order system of equations in a coordinate frame and then solving it in Lagrangian coordinates. Because the modifications required to adapt the local-in-time existence and uniqueness theory from \cite[\S5]{BeyerOliynyk:2021} for the Einstein-scalar field equations  to allow for coupling with the Euler equations are straightforward, most of the proofs in this section will be omitted and we refer the interested reader to \cite[\S5]{BeyerOliynyk:2021} for the details.

Following \cite[\S5]{BeyerOliynyk:2021}, we fix the coordinate frame by introducing coordinates
$(\xh^\mu)=(\xh^0,\xh^\Lambda)$ on a spacetime $M_{t_1,t_0}$ of the form
\eqref{Mt1t0-def},
and we assume that the components of the flat background metric $\gc$ in this coordinate system, denoted $\gchat_{\mu\nu}$, are given by
\begin{equation}\label{gct-def}
    \gchat_{\mu\nu}=\eta_{\mu\nu}:= -\delta_\mu^0\delta_\nu^0 + \delta_\mu^\Lambda \delta_\nu^\Gamma \delta_{\Gamma\Lambda}.
\end{equation}
In this coordinate frame, the Levi-Civita connection of the background metric coincides with partial differentiation with respect to the coordinates $(\xh^\mu)$, that is,
$\hat{\Dc}_\mu = \delh{\mu}$. 
Using this, we find, with the help of \eqref{Ccdef}, \eqref{eq:Tttconfphys},  \eqref{eq:DefGjsl}, \eqref{red-Ricci} and \eqref{curvature}, that the reduced conformal Einstein-Euler-scalar field equations \eqref{confESFFa}-\eqref{confESFHa} are given in the coordinates $(\xh^\mu)$ by
\begin{align}
\gh^{\alpha\beta}\delh{\alpha}\delh{\beta}\gh_{\mu\nu} + \Qh_{\mu\nu} &= -\frac{2}{\tauh}\bigl( \delh{\mu}\delh{\nu}\tauh - \Gammah^\gamma_{\mu\nu} \delh{\gamma}\tauh \bigr)-2\Ttth_{\mu\nu}, \label{tconfESF-A.1}\\
\gh^{\alpha\beta}\delh{\alpha}\delh{\beta}\tauh &= 0,  \label{tconfESF-A.2}\\
{\hat\att}^\gamma{}_{\mu\nu}\delh{\gamma}\Vh^\nu &= \Gh_{\mu\nu\gamma}\Vh^\nu \Vh^\gamma,
\label{tconfESF-A.3}
\end{align}
where $\tauh$ denotes the scalar field $\tau$ viewed as a function of the coordinates $(\xh^\mu)$, $\gh_{\mu\nu}$ are the components of the conformal metric $g$ with respect to the coordinates $(\xh^\mu)$, $\Gammah^\gamma_{\mu\nu} =
     \frac{1}{2}\gh^{\gamma\lambda}\bigl(\delh{\mu}\gh_{\nu\lambda} +\delh{\nu}\gh_{\mu\lambda}-\delh{\lambda}\gh_{\mu\nu}\bigr)$
are the Christoffel symbols of $\gh_{\mu\nu}$,
\begin{equation}
  \label{eq:defQh}
  \begin{split}
\Qh_{\mu \nu} = \frac{1}{2}\gh^{\alpha \beta}\gh^{\sigma \delta}\bigl(&\delh{\mu} \gh_{\sigma \alpha} \delh{\nu} \gh_{\delta \beta}
+2 \delh{\delta}\gh_{\mu \beta}\delh{\alpha}\gh_{\nu \sigma} - 2\delh{\beta}\gh_{\mu \delta} \delh{\alpha}\gh_{\nu \sigma}\\
&-2 \delh{\beta}\gh_{\mu \delta}\delh{\nu} \gh_{\sigma\alpha} -2 \delh{\mu }\gh_{\sigma\alpha}\delh{\beta}\gh_{\nu \delta}\bigr), 
  \end{split}
\end{equation}
\begin{equation} \label{eq:Ttth-def}
\Ttth_{\mu\nu}=2P_0\Bigl(\frac{1+c_s^2}{c_s^2} \Vh^{-2} \Vh_\mu\Vh_\nu+\frac{1-c_s^2}{(n-2)c_s^2}\gh_{\mu\nu}\Bigr) \tauh^{\frac{c_s^2-1}{c_s^2(n-2)}}\Vh^{-\frac{(1+c_s^2)}{c_s^2}},
\end{equation}
with $\Vh^2 =-\gh_{\mu\nu}\Vh^\mu \Vh^\nu$,
\begin{equation} \label{eq:Ah-def}
    {\hat\att}^\gamma{}_{\mu \nu}
  =\frac{3 c_s^2+1}{c_s^2} \frac{\Vh_\mu \Vh_\nu \Vh^\gamma}{\Vh^2} +\Vh^\gamma \gh_{\mu \nu}
    +2{\delta^\gamma}_{(\nu} \Vh_{\mu)}
\end{equation}
and
\begin{equation} \label{eq:Gh-def}
  \begin{split}
\Gh_{\mu \nu \gamma} = &-\tauh^{-1}
  \frac{(n-1)c_s^2-1}{c_s^2(n-2)} \gh_{\mu(\nu} \delh{\gamma)}\tauh    \\ 
  &- \frac{1}{2}\Bigl(
  \frac{3 c_s^2+1}{c_s^2} \frac{\Vh^\alpha \Vh^\beta}{\Vh^2}+\gh^{\alpha\beta}\Bigr)  \gh_{\mu(\nu}\delh{\gamma)} \gh_{\alpha \beta} 
  -  \delh{(\gamma} \gh_{\nu) \mu}.
  \end{split}
\end{equation}
We further observe from \eqref{Ccdef} and \eqref{Xdef} that the coordinate components of the wave gauge vector field $X$, denoted $\Xh^\gamma$, are given by
\begin{equation} \label{Xt-rep}
    \Xh^\gamma = \gh^{\mu\nu}\Gammah^\gamma_{\mu\nu}.
\end{equation}

\subsection{Initial data and constraint propagation}
On the initial hypersurface 
\begin{equation*}
\Sigma_{t_0}=\{t_0\}\times\Tbb^{n-1},
\end{equation*}
we specify the following initial data for the reduced conformal Einstein-Euler-scalar field  equations \eqref{tconfESF-A.1}-\eqref{tconfESF-A.3}:
\begin{align}
    \gh_{\mu\nu}\bigl|_{\Sigma_{t_0}} &= \gr_{\mu\nu}, \label{gt-idata} \\
    \delh{0}\gh_{\mu\nu}\bigl|_{\Sigma_{t_0}} &= \grave{g}_{\mu\nu}, \label{dt-gt-idata} \\
    \tauh\bigl|_{\Sigma_{t_0}} &= \taur, \label{taut-idata}\\
    \delh{0}\tauh \bigl|_{\Sigma_{t_0}} &= \taug,\label{dt-taut-idata}\\
    \Vh^\mu  \bigl|_{\Sigma_{t_0}} &= \Vr^\mu.
        \label{Vh-idata}
\end{align}
Since we want solutions of the reduced conformal Einstein-Euler-scalar field  equations \eqref{tconfESF-A.1}-\eqref{tconfESF-A.3} to also satisfy the conformal Einstein-Euler-scalar field equation \eqref{confESFAb}-\eqref{confESAc},
the initial data will need to be chosen so that it satisfies the constraint equations
\begin{align}
   \nh_\mu\Bigr( \Gh^{\mu\nu} - \frac{1}{\tauh}\nablah^\mu\nablah^\nu \tauh-2 \hat T^{\text{Fl}}{}^{\mu\nu}\Bigr)\Bigl|_{\Sigma_{t_0}} &=0, \!\!\quad \text{(gravitational constraints)}
  \label{grav-constr}\\
  \Xh^\mu\bigl|_{\Sigma_{t_0}} &=0, \!\!\quad \text{(wave gauge constraints)}
  \label{wave-constr}
\end{align}
where $\Gh^{\mu\nu}$ and $\nablah_\mu$ are the Einstein tensor and Levi-Civita connection of the conformal metric $\gh_{\mu\nu}$, respectively, and $\nh = d\xh^0$ (i.e. $\nh_\mu = \delta^0_\mu)$. As established in  
Proposition \ref{prop-ES-constr} below, solutions of the reduced conformal Einstein-Euler-scalar field equations that are generated from initial data satisfying both of these constraint equations will also solve the conformal Einstein-Euler-scalar field equations.

\begin{rem} \label{idata-rem}
The \emph{geometric initial data} on $\Sigma_{t_0}$ is  $\{\gtt,\Ktt,\taur,\taugr,\Vr\}$
where $\gtt=\gtt_{\Lambda\Omega}d\xh^\Lambda \otimes d\xh^\Omega$ is the \textit{spatial metric} and $\Ktt=\Ktt_{\Lambda\Omega}d\xh^\Lambda \otimes d\xh^\Omega$ is the \textit{second fundamental form}, which are determined from the initial data $\{\gr_{\mu\nu}, \grave{g}_{\mu\nu}, \taur, \taug,\Vr^\mu\}$ via
\begin{equation}
  \label{gtt-def1}
  \gtt_{\Lambda\Omega} = \gr_{\Lambda\Omega} \AND \Ktt_{\Lambda\Omega} = \frac{1}{2\Ntt}(\ggr_{\Lambda\Omega}-2\Dtt_{(\Lambda} \btt_{\Omega)}),
\end{equation}
respectively. Here, 
\begin{equation*}
  %\label{gtt-def2}
  \btt_\Lambda = \gr_{0\Lambda} \AND \Ntt^2 = -\gr_{00}+ \btt^\Lambda\btt_\Lambda
\end{equation*}
define the \textit{shift}
$\btt=\btt_\Lambda d\xh^\Lambda$ and  \textit{lapse} $\Ntt$, respectively,  $\Dtt_\Lambda$
denotes the Levi-Civita connection of the spatial metric $\gtt_{\Lambda\Omega}$, and we have used the inverse metric $\gtt^{\Lambda\Omega}$ of $\gtt_{\Lambda\Omega}$ to raise indices, e.g.\ $\btt^\Lambda = \gtt^{\Lambda\Omega}\btt_\Omega$. The importance of the geometric initial data is that it represents the physical part (i.e. non-gauge) of the initial data. Moreover, the gravitational constraint equations \eqref{grav-constr} can be formulated entirely in terms of the geometric initial data. On the other hand, it is always possible for a given choice of geometric initial data $\{\gtt,\Ktt,\taur,\taugr,\Vr\}$ to choose the remaining initial data so that the wave gauge constraints \eqref{wave-constr} are satisfied; see \cite[Rem.~5.1]{BeyerOliynyk:2021} for details.
\end{rem}

\begin{prop} \label{prop-ES-constr}
Suppose $\gh_{\mu\nu},\tauh\in C^3(M_{t_1,t_0})$
and $\Vh^\mu \in  C^1(M_{t_1,t_0})$ solve the reduced conformal Einstein-Euler-scalar field equations \eqref{tconfESF-A.1}-\eqref{tconfESF-A.3} and the constraints \eqref{grav-constr}-\eqref{wave-constr}, and let
\begin{equation*}
\hat T^{\text{Fl}}_{\mu\nu}= 
   {P_0}\tauh^{\frac{(c_s^2-1)}{c_s^2(n-2)}} \Vh^{-\frac{1+c_s^2}{c_s^2}}\biggl(\frac{1+c_s^2}{c_s^2} \Vh^{-2} \Vh_\mu\Vh_\nu+\gh_{\mu\nu}\biggr).
\end{equation*}
Then  $\gh_{\mu\nu},\tauh$ and $\Vh^\mu$ satisfy the conformal Einstein-Euler-scalar field equations 
\begin{equation} \label{tconfESF-Aa}
  \begin{split}
\Gh_{\mu\nu}&= \tauh^{-1}\nablah_\mu\nablah_\nu \tauh+2 \hat T^{\text{Fl}}_{\mu\nu}, \quad 
  \Box_{\gh} \tauh =0, \\
                {\hat\att}^\gamma{}_{\mu\nu }\nablah_\gamma \Vh^\nu
                &=- \frac{c_s^2(n-1)-1}{c_s^2(n-2)}\tauh^{-1} \Vh_\mu \Vh^\gamma\nablah_\gamma\tauh, 
  \end{split}
\end{equation}
and the wave gauge constraint $\Xh^\mu = 0$
in $M_{t_1,t_0}$.
\end{prop}
\begin{proof}
By a straightforward calculation, we find that the divergence of the tensor $\hat T^{\text{Fl}}_{\mu\nu}$ is given by
\begin{equation*}
\nablah^\mu \hat T^{\text{Fl}}_{\mu\nu} =    \frac{(c_s^2-1)}{c_s^2(n-2)}\tauh^{-1}\nablah^\mu\tauh \hat T^{\text{Fl}}_{\mu\nu}+  \frac{1+c_s^2}{2c_s^2}P_0\tauh^{\frac{(c_s^2-1)}{c_s^2(n-2)}} \Vh^{-\frac{3c_s^2+1}{c_s^2}}{\hat\att}^\gamma{}_{\nu \mu}\nablah_\gamma \Vh^\mu.
\end{equation*}
With the help of \eqref{tconfESF-A.3}, i.e. 
\begin{equation} \label{prop-ES-constr4}
{\hat\att}^\gamma{}_{\mu\nu }\nablah_\gamma \Vh^\nu=- \frac{c_s^2(n-1)-1}{c_s^2(n-2)}\tau^{-1} \Vh_\mu \Vh^\gamma\nablah_\gamma\tauh,
\end{equation} 
this becomes
\begin{align*}
    \nablah^\mu \hat T^{\text{Fl}}_{\mu\nu}=&    
    \frac{(c_s^2-1)}{c_s^2(n-2)}\tauh^{-1}\nablah^\mu\tauh \hat T^{\text{Fl}}_{\mu\nu}\\
    &- \frac{1+c_s^2}{2c_s^2}P_0\tauh^{\frac{(c_s^2-1)}{c_s^2(n-2)}} \Vh^{-\frac{3c_s^2+1}{c_s^2}}\frac{c_s^2(n-1)-1}{c_s^2(n-2)}\tauh^{-1} \Vh_\nu \Vh^\gamma\nablah_\gamma\tauh\\
    =& -\tauh^{-1}  {P_0}\tauh^{\frac{(c_s^2-1)}{c_s^2(n-2)}} \Vh^{-\frac{1+c_s^2}{c_s^2}}\biggl(\frac{1+c_s^2}{c_s^2} \Vh^{-2} \Vh_\mu\Vh_\nu+\frac{1-c^2_s}{c_s^2(n-2)}\gh_{\mu\nu}\biggr)\nablah^\mu \tauh\\
    %&= -\tauh^{-1}\nabla^\mu\tauh\biggl(\hat{\Tc}_{\mu\nu}+\frac{1}{2-n}\gh^{\alpha\beta}\hat{\Tc}_{\alpha\beta}\gh_{\mu\nu}\biggr), \\
    \oset{\eqref{eq:Ttth-def}}{=}& -\frac12\tauh^{-1}\nablah^\mu\tauh\hat{\Ttt}_{\mu\nu}.
\end{align*}
From this identity and \eqref{tconfESF-A.2}, i.e.\ 
\begin{equation} \label{prop-ES-constr3}
    \Box_{\gh}\tauh =-\Xh^\mu \nablah_\mu\tauh,
\end{equation}
we observe that
\begin{align}
    &\nablah_\mu\bigl(\tauh^{-1}\nablah^\mu\nablah^\nu \tauh + 2 \hat T^{\text{Fl}}{}^{\mu\nu}\bigr)\\
    &= -\tauh^{-2}\nablah_\mu\tauh \nablah^\mu\nablah^\nu \tauh + \tauh^{-1}\nablah_\mu \nablah^\mu\nablah^\nu  \tauh-\tauh^{-1}\nabla_\mu\tauh\hat{\Ttt}{}^{\mu\nu}\notag\\
    &=\tauh^{-1} \nablah^\nu \nablah_\mu \nablah^\mu \tauh +\tauh^{-1}\Rh^{\nu\mu}\nablah_\mu\tauh
    -\tauh^{-1}\nabla_\mu\tauh\bigl(\tauh^{-1}\nablah^\mu\nablah^\nu \tauh+\hat{\Ttt}{}^{\mu\nu}\bigr)\notag \\
    &=
   -\tauh^{-1} \nablah^\nu(\Xh^\mu \nablah_\mu\tauh)
    -\tauh^{-1}\nabla_\mu\tauh\bigl(-R^{\mu\nu}+\tauh^{-1}\nablah^\mu\nablah^\nu \tauh+\hat{\Ttt}{}^{\mu\nu}\bigr)\label{prop-ES-constr1}.
\end{align}

Next, expressing \eqref{tconfESF-A.1} as
\begin{equation} \label{prop-ES-constr2}
    -2\Rh_{\mu\nu} + 2\nablah_{(\mu}\Xh_{\nu)} = -2\tauh^{-1}\nablah_\mu\nablah_\nu \tauh-2\hat{\Ttt}_{\mu\nu},
\end{equation}
we have, after rearranging, that 
\begin{equation*}
    \Gh^{\mu\nu}=\nablah^{(\mu}\Xh^{\nu)} -\frac{1}{2}\nablah_\alpha \Xh^\alpha \gh^{\mu\nu}+\frac{1}{2}\Xh^\alpha \nablah_\alpha \ln(\tauh) \gh^{\mu\nu}+\tauh^{-1}\nablah^\mu\nablat^\nu \tauh + 2 \hat T^{\text{Fl}}{}^{\mu\nu},
\end{equation*}
where in deriving this we have again used
\eqref{prop-ES-constr3}. Taking the divergence of this expression, we find with the help of the second contracted Bianchi identity $\nablah_\mu \Gh^{\mu\nu}=0$ and \eqref{prop-ES-constr1}-\eqref{prop-ES-constr2}, that 
\begin{equation*}
  \begin{split}
0=&\nablah_\mu\nablah^{(\mu}\Xh^{\nu)} -\frac{1}{2}\nablah^\nu\nablah_\mu \Xh^\mu +\frac{1}{2}\nabla^\nu\bigl(\Xh^\mu \nablah_\mu \ln(\tauh)\bigr)  -\tauh^{-1} \nablah^\nu(\Xh^\mu \nablah_\mu\tauh) \\
&-\tauh^{-1}\nabla_\mu\tauh \nablah^{(\mu}\Xh^{\nu)}. 
  \end{split}
\end{equation*}
Re-expressing this as
\begin{equation*}
  \begin{split}
0=&\frac{1}{2}\Box_{\gh}\Xh^\nu+\frac{1}{2}\Rh^\nu{}_\mu \Xh^\mu +\frac{1}{2}\nabla^\nu\bigl(\Xh^\mu \nablah_\mu \ln(\tauh)\bigr)  -\tauh^{-1} \nablah^\nu(\Xh^\mu \nablah_\mu\tauh) \\
&-\tauh^{-1}\nabla_\mu\tauh \nablah^{(\mu}\Xh^{\nu)},
  \end{split}
\end{equation*}
we see that $\Xh^\mu$ satisfies a linear wave equation on $M_{t_0,t_1}$. 
Since the constraints \eqref{grav-constr}-\eqref{wave-constr} imply by a well known argument, e.g. see \cite[\S 10.2]{Wald:1994}, that $\Xh^\mu$ and $\delh{0}\Xh^\mu$ vanish in $\Sigma_{t_0}$, we conclude from the uniqueness of solutions to linear wave equations that $\Xh^\mu$ must vanish in $M_{t_1,t_0}$. By \eqref{prop-ES-constr4},  \eqref{prop-ES-constr3} and \eqref{prop-ES-constr2}, it then follows that the triple $\{\gh_{\mu\nu},\tauh,\Vh^\mu\}$ solves the conformal Einstein-Euler-scalar field equations \eqref{tconfESF-Aa} in $M_{t_1,t_0}$, which completes the proof.
\end{proof}

\subsection{First order formulation\label{fof}}
Following \cite[\S5.3]{BeyerOliynyk:2021}, we introduce first order variables
\begin{gather}
    \hh_{\beta\mu\nu} = \delh{\beta}\gh_{\mu\nu}, \quad
    \vh_{\mu} = \delh{\mu}\tauh \AND
    \wh_{\mu\nu} = \delh{\mu}\delh{\nu}\tauh, \label{tvars}
\end{gather}
and define a vector field $\chih^\mu$ via
\begin{equation}
  \chih^\mu = \frac{1}{|\vh|_{\gh}^2} \vh^{\mu}, \label{chit-def}
\end{equation}
which we note satisfies $\chih(\tauh)=1$.
We will always assume that $\vh^\mu$ is timelike (i.e.\ $|\vh|_{\gh}^2 <0$) in order to ensure that $\hat{\chi}^\mu$ remains well defined and timelike. 

Carrying out the same computations as in \cite[\S5.3]{BeyerOliynyk:2021}, it is straightforward to verify that the conformal Einstein-Euler-scalar field equations \eqref{tconfESF-A.1}-\eqref{tconfESF-A.3} can be cast in the first order form
\begin{align}
   \Bh^{\lambda\beta\alpha} \delh{\alpha}\hh_{\beta\mu\nu}&= \chih^\lambda \Bigr(\Qh_{\mu\nu}+\frac{2}{\tauh}\bigl( \wh_{(\mu \nu)} - \Gammah^\gamma_{\mu\nu} \vh_\gamma\bigl)+2\Ttth_{\mu\nu} \Bigr), \label{tconf-ford-B.1}\\
     \Bh^{\lambda\beta\alpha}\delh{\alpha}\wh_{\beta\mu}
     & =-\chih^\lambda\gh^{\alpha \sigma}\gh^{\beta\delta}\hh_{\mu\sigma\delta}\wh_{\alpha\beta}, \label{tconf-ford-B.2}\\
      \Bh^{\lambda\beta\alpha}\delh{\alpha}\zh_{\beta} &=0, \label{tconf-ford-B.3} \\
      \chih^\alpha \delh{\alpha}\gh_{\mu\nu}&= \chih^\alpha \hh_{\alpha\mu\nu}, \label{tconf-ford-B.4}\\
     \chih^\alpha \delh{\alpha} \vh_{\mu}&= \chih^\alpha \wh_{\alpha\mu},  \label{tconf-ford-B.5}\\
     \chih^\alpha \delh{\alpha} \tauh&= \chih^\alpha \zh_{\alpha},  \label{tconf-ford-B.6}\\
     {\hat\att}^\gamma{}_{\mu\nu}\delh{\gamma}\Vh^\nu &= \Gh_{\mu\nu\gamma}\Vh^\nu \Vh^\gamma, \label{tconf-ford-B.7}
\end{align}
where 
\begin{equation}\label{Bt-def}
 \Bh^{\lambda\beta\alpha}=    -\chih^\lambda\gh^{\beta\alpha} -\chih^\beta \gh^{\lambda\alpha} 
  + \gh^{\lambda\beta}\chih^\alpha
\end{equation}
and $\zh_\mu$ should be interpreted as being the derivative of $\tauh$, i.e. 
\begin{equation} \label{zh-def}
    \zh_\mu = \delh{\mu}\tauh.
\end{equation}

\subsection{Lagrangian coordinates\label{Lag-coordinates}}
Following \cite{BeyerOliynyk:2021}, we introduce Lagrangian coordinates $(x^\mu)$ adapted to the vector field $\chih^\alpha$ and consider a transformed version of the system \eqref{tconf-ford-B.1}-\eqref{tconf-ford-B.7}. The primary purpose of doing so is that the Lagrangian coordinates allow us to use the scalar field $\tau$ as a time coordinate, which synchronizes the singularity; see Section \ref{temp-synch} below for details.

The Lagrangian coordinates $(x^\mu)$ are defined via the map
\begin{equation} \label{lemBa.1a}
\xh^\mu = l^\mu(x) := \Gc_{x^0-t_0}^\mu(t_0,x^\Lambda), \quad \forall\, (x^0,x^\Lambda)\in M_{t_1,t_0},
\end{equation}
where $\Gc_s(\xh^\lambda) = (\Gc^\mu_s(\xh^\lambda))$
denotes the flow map of $\chih^\mu$, i.e. 
\begin{equation*}
\frac{d\;}{ds} \Gc^\mu_s(\xh^\lambda) = \chih^\mu(\Gc_s(\xh^\lambda)) \AND
\Gc^\mu_{0}(\xh^\lambda) = \xh^\mu.
\end{equation*}
We note that the map $\ell^\mu$ defines a diffeomorphism
\begin{equation*}
    l \: : \: M_{t_1,t_0} \longrightarrow l(M_{t_1,t_0})\subset M_{-\infty,t_0}
\end{equation*}
that satisfies $l(\Sigma_{t_0})=\Sigma_{t_0}$ so long as the vector field $\chih^\mu$ does not vanish and remains sufficiently regular. 
In line with our coordinate conventions, we often use 
$t=x^0$
to denote the Lagrangian time coordinate. 

\begin{rem}
As we show below, by formulating the reduced conformal Ein\-stein-sca\-lar field equations in Lagrangian coordinates, the Lagrangian map $l^\mu$ becomes an additional unknown field that needs to be solved for. The local-in-time existence theory developed in Proposition \ref{lag-exist-prop} will then guarantee that $l^\mu$ exists and is well-defined on a spacetime region of the form $M_{t_1,t_0}$ for $t_1$ sufficiently close to $t_0$. Moreover, the continuation principle from Proposition \ref{lag-exist-prop} will ensure that $l^\mu$ can be extended to domains of the form  $M_{t_1^*,t_0}$ with $t_1^*<t_1$ provided the full solution to the reduced conformal Einstein-scalar field equations in Lagrangian coordinates satisfies appropriate bounds. In this way, the local-in-time existence and continuation theory from Proposition \ref{lag-exist-prop} determines the domain of definition of the Lagrangian map $l^\mu$. 
\end{rem}

By definition, $l^\mu$ solves 
the IVP
\begin{align}
    \del{0}l^\mu &= \chihu^\mu, \label{l-ev.1} \\
    l^\mu(t_0,x^\Lambda) &= \delta^\mu_0 t_0 +\delta^\mu_\Lambda x^\Lambda, \label{l-ev.2} 
\end{align}
where, here and below, we use the notation
\begin{equation}\label{fu-def}
    \underline{f}=f\circ l
\end{equation}
to denote the pull-back of scalars by the Lagrangian map $l$.
In the following, symbols without a ``hat'' will denote the geometric pull-back by the Lagrangian map $l$. For example,
\begin{equation}\label{chi-lag}
\chi^\mu = \Jcch_\nu^\mu \chihu^\nu \oset{\eqref{l-ev.1}}{=} \Jcch_\nu^\mu \Jc_0^\nu =\delta^\mu_0,
\end{equation}
\begin{equation} \label{tau-lag}
\tau = \underline{\tauh},
\end{equation}
\begin{equation} \label{g-lag}
    g_{\mu\nu} = \Jc^\alpha_\nu \Jc^\beta_\mu\ghu_{\alpha\beta},
\end{equation}
and
\begin{equation} \label{V-lag}
V^\mu = \Jcch^\mu_\nu \Vhu^\nu,
\end{equation}
where
\begin{equation} \label{Jc-def}
\Jc^\mu_\nu = \del{\nu}l^\mu
\end{equation}
is the Jacobian matrix of the map $l^\mu$ and 
\begin{equation} \label{Jcch-def}
    (\Jcch^\mu_\nu):= (\Jc^\mu_\nu)^{-1}
\end{equation}    
is its inverse. 

Recalling from \cite[\S5.4]{BeyerOliynyk:2021} that 
the Jacobian matrix $\Jc^\mu_\nu$ satisfies the
IVP \begin{align}
    \del{0}\Jc^\mu_\nu &= \Jc^\lambda_\nu \Jsc_\lambda^\mu, \label{J-ev.1} \\
    \Jc^\mu_\nu(t_0,x^\Lambda) &= \delta^0_\nu \chih^\mu(t_0,x^\Lambda)+\delta_\nu^\Lambda\delta^\mu_\Lambda, \label{J-ev.2} 
\end{align}
where
\begin{equation} \label{Jsc-def}
  \begin{split}
    \Jsc_\lambda^\mu =\frac{1}{|\vhu|^2_{\ghu}} \biggl(&
   \ghu^{\mu\sigma}\whu_{\lambda\sigma}-\ghu^{\mu\tau}\ghu^{\sigma\omega} \hhu_{\lambda\tau\omega}\vhu_\sigma \\
   &-\frac{1}{|\vhu|^2_{\ghu}}\bigl(-\ghu^{\alpha\tau}\ghu^{\beta\omega} \hhu_{\lambda\tau\omega}\vhu_\alpha\vhu_\beta+2\ghu^{\alpha\beta}\vhu_\alpha \wh_{\lambda\beta}\bigr)\ghu^{\sigma\mu}\vhu_\sigma \biggr),
  \end{split}
\end{equation}
we see from transforming  \eqref{tconf-ford-B.1}-\eqref{tconf-ford-B.7} into Lagrangian coordinates and combining it with 
\eqref{l-ev.1} and \eqref{J-ev.1} that the system
\begin{align}
   \Bhu^{\lambda\beta\alpha} \Jcch_\alpha^\gamma \del{\gamma}\hhu_{\beta\mu\nu}&= \Jc_0^\lambda \Bigr(\Qhu_{\mu\nu}+\frac{2}{\tau}\bigl( \whu_{(\mu \nu)} - \Gammahu^\gamma_{\mu\nu} \vhu_\gamma\bigl)+2\Ttthu_{\mu\nu} \Bigr), \label{tconf-ford-C.1}\\
     \Bhu^{\lambda\beta\alpha} \Jcch_\alpha^\gamma \del{\gamma}\whu_{\beta\mu},
     &=-\Jc^\lambda_0\ghu^{\alpha \sigma}\ghu^{\beta\delta}\hhu_{\mu\sigma\delta}\whu_{\alpha\beta}, \label{tconf-ford-C.2}\\
      \Bhu^{\lambda\beta\alpha} \Jcch_\alpha^\gamma \del{\gamma}\zhu_{\beta} &=0, \label{tconf-ford-C.3} \\
       \del{0}\ghu_{\mu\nu}&= \Jc_0^\alpha \hhu_{\alpha\mu\nu}, \label{tconf-ford-C.4}\\
     \del{0} \vhu_{\mu}&= \Jc_0^\alpha \whu_{\alpha\mu},  \label{tconf-ford-C.5}\\
     \del{0} \tau&= \Jc^\alpha_0\zhu_{\alpha},  \label{tconf-ford-C.6}\\
    {\underline{\hat\att}}^\alpha_{\mu\nu}\Jcch^\gamma_\alpha \del{\gamma} \Vhu^\nu &= \Ghu_{\mu\nu\gamma}\Vhu^\nu\Vhu^\gamma,\label{tconf-ford-C.7}\\
     \del{0} \Jc^\mu_\nu &=  \Jc^\lambda_\nu \Jsc^\mu_\lambda,\label{tconf-ford-C.8}\\
    \del{0}l^\mu &= \chihu^\mu, \label{tconf-ford-C.9}\end{align}
defines a first order Lagrangian formulation of the reduced conformal Ein\-stein-Eu\-ler-sca\-lar field equations in the variables $$\{\hhu_{\beta\mu\nu},\whu_{\beta\mu},\zhu_{\beta},\ghu_{\mu\nu},\vhu_\mu,\tau,\Vhu^\mu,\Jc^\mu_\nu,l^\mu\}.$$ It is worth noting that compared to the Lagrangian system considered in \cite[\S5.4]{BeyerOliynyk:2021}, the differences are the additional term $2\Ttthu_{\mu\nu}$ in \eqref{tconf-ford-C.1} and equation \eqref{tconf-ford-C.7} for the Lagrangian fluid variable $\Vhu^\mu$. As far as establishing the local-in-time existence of solutions, the term $2\Ttthu_{\mu\nu}$ in \eqref{tconf-ford-C.1} causes no additional difficulties as it is a non-principal term that depends smoothly on the Lagrangian variables. Furthermore, the evolution equations \eqref{tconf-ford-C.7} is straightforward to handle since it is manifestly symmetric hyperbolic as long as the Lagrangian coordinates do not breakdown (i.e.\ $\Jc^\mu_\nu$ remains non-singular) and the vector field $\Vhu^\mu$ remains timelike. 

\subsection{Lagrangian initial data}
With the help of \eqref{tvars}, \eqref{zh-def}, \eqref{l-ev.2}, \eqref{Jc-def} and \eqref{J-ev.2}, we observe that the reduced conformal Einstein-Euler-scalar field initial data
\eqref{gt-idata}-\eqref{Vh-idata} generates the following
initial data for the Lagrangian representation \eqref{tconf-ford-B.1}-\eqref{tconf-ford-B.7}:
\begin{align}
    l^\mu \bigl|_{\Sigma_{t_0}} &= \lr^\mu, \label{l-idata}\\
    \Jc^\mu_\nu \bigl|_{\Sigma_{t_0}} &= \Jcr^\mu_\nu, \label{Jc-idata}\\
    \Vhu^\mu\bigl|_{\Sigma_{t_0}} &= \Vr^\mu, \label{Vhu-idata}\\
    \tau\bigl|_{\Sigma_{t_0}} &= \taur, \label{tauh-idata} \\ 
    \vhu_\mu\bigl|_{\Sigma_{t_0}} &= \delta^0_\mu \taug +\delta_\mu^\Lambda \del{\Lambda} \taur, \label{vhu-idata} \\
    \whu_{\Lambda\Omega}\bigl|_{\Sigma_{t_0}} &= \del{\Lambda}\del{\Omega} \taur, \label{whu-idata-1}\\
    \whu_{0\Omega}\bigl|_{\Sigma_{t_0}} &= \del{\Omega} \taug, \label{whu-idata-2}\\
     \whu_{\Lambda 0}\bigl|_{\Sigma_{t_0}} &= \del{\Lambda} \taug, \label{whu-idata-3}\\
      \whu_{0 0}\bigl|_{\Sigma_{t_0}} &= -\frac{1}{\gr^{00}}\bigl(2\gr^{0\Lambda}\del{\Lambda}\taug+
      \gr^{\Lambda\Omega}\del{\Lambda}\del{\Omega}\taur \bigr), \label{whu-idata-4}\\
    \zhu_\mu\bigl|_{\Sigma_{t_0}} &= \delta^0_\mu \taug +\delta_\mu^\Lambda \del{\Lambda} \taur, \label{zhu-idata} \\
    \ghu_{\mu\nu}\bigl|_{\Sigma_{t_0}} &= \gr_{\mu\nu}, \label{ghu-idata} \\
    \hhu_{\alpha\mu\nu}\bigl|_{\Sigma_{t_0}} &= \delta_\alpha^0 \grave{g}_{\mu\nu}+\delta_\alpha^\Lambda\del{\Lambda}\gr_{\mu\nu}, \label{hhu-idata}
\intertext{where}
    \lr^\mu &= \delta^\mu_0 t_0 + \delta^\mu_\Lambda x^\Lambda, \label{lr-def} \\
    \vr^\mu &= \gr^{\mu\nu}(\delta^0_\nu \taug +\delta_\nu^\Lambda\del{\Lambda}\taur),\label{vr-def}\\
    \chir^\mu &= \frac{1}{|\vr|_{\gr}^2}\vr^{\mu}, \label{chir-def}\\
    \Jcr^\mu_\nu &= \delta^0_\nu \chir^\mu + \delta^\Lambda_\nu \delta^\mu_\Lambda, \label{Jcr-def}
\end{align}
and 
\begin{equation}
    (\Jcr^{-1})^\mu_\nu = \frac{1}{\chir^0}\delta^0_\nu(\delta^\mu_0 -\delta^{\mu}_\Lambda \chir^\Lambda) + \delta^\Lambda_\nu \delta^\mu_\Lambda. \label{Jcr-inv}
\end{equation}

\begin{rem}\label{rem-Lag-idata}
On the initial hypersurface $\Sigma_{t_0}=\{t_0\}\times \Tbb^{n-1}$,  our choice of initial data
implies that 
\begin{equation} \label{V-idata}
    V^\mu\bigl|_{\Sigma_{t_0}} = (\Jcr^{-1})^\mu_\nu \Vr^\nu,
\end{equation}
\begin{equation} \label{dt-tau-idata}
    \del{0}\tau\bigl|_{\Sigma_{t_0}}  = \chir^\mu \vr_\mu =1, 
\end{equation}
\begin{equation}\label{g-idata}
 g_{\mu\nu} \bigl|_{\Sigma_{t_0}}  
 = \Jcr^\alpha_\mu \gr_{\alpha\beta}\Jcr^\beta_\nu,
\end{equation}
\begin{equation} \label{dt-g-idata}
    \del{0}g_{\mu\nu} \bigl|_{\Sigma_{t_0}} = 
     \Jcr^\gamma_0\Jcr^\mu_\alpha \Jcr^\beta_\nu \bigl(\delta_\gamma^0 \grave{g}_{\alpha\beta}+\delta_\gamma^\Lambda\del{\Lambda}\gr_{\alpha\beta}\bigr) +  \gr_{\alpha\beta}\bigl(
     \Jcr_\mu^\lambda \mathring{\Jsc}{}^\alpha_\lambda\Jcr^\beta_\nu + \Jcr^\alpha_\mu \Jcr_\nu^\lambda \mathring{\Jsc}{}^\beta_\lambda\bigr)
\end{equation}
and
\begin{equation} \label{chi-idata}
    \chi^\mu \bigl|_{\Sigma_{t_0}} % = (\Jcr^{-1})^\mu_\nu \chir^\nu 
    \overset{\eqref{chi-lag}}{=} \delta^\mu_0,
\end{equation}
where $\mathring{\Jsc}{}^\mu_\nu= \Jsc^\mu_\nu|_{\Sigma_{t_0}}$.
\end{rem}

\begin{rem} \label{FLRW-idata-rem-A}
By \eqref{Kasner-wave-gauge} and \eqref{gtau-FLRW}, it is clear that the FLRW solutions \eqref{eq:FLRWEulerSFExplSol1}- \eqref{eq:FLRWEulerSFExplSol4} 
determine the initial data
\begin{equation*}
  \begin{split}
    &\bigl\{g_{\mu\nu}\bigl|_{\Sigma_{t_0}},\del{0}g_{\mu\nu}\bigl|_{\Sigma_{t_0}}, \tau\bigl|_{\Sigma_{t_0}},\del{0}\tau\bigl|_{\Sigma_{t_0}},V^\mu\bigl|_{\Sigma_{t_0}}\bigr\}_{\textrm{FLRW}}\\
    =&
    \bigl\{\breve g_{\mu\nu}(t_0),\del{t}\breve g_{\mu\nu}(t_0),t_0,1,\breve V^\mu(t_0)\bigr\}
  \end{split}
\end{equation*}
on $\Sigma_{t_0}$ and this initial data satisfies both the gravitational and wave gauge constraints. Furthermore, since these solutions satisfy \eqref{FLRW-Lag}, they are already in the Lagrangian representation.
\end{rem}

\subsection{Local-in-time existence}
We are now ready to state, in the following proposition, a local-in-time existence and uniqueness result and continuation principle for solutions of the system \eqref{tconf-ford-C.1}-\eqref{tconf-ford-C.9}. We omit the proof since it is essentially the same as the proof of Proposition 5.5 from  
\cite{BeyerOliynyk:2021}. 

\begin{prop} \label{lag-exist-prop}
Suppose $k>(n-1)/2+1$, $t_0>0$, and that the initial data $\Vr^\mu\in H^k(\Tbb^{n-1},\Rbb^n)$, $\taur\in H^{k+2}(\Tbb^{n-1})$, $\taugr\in H^{k+1}(\Tbb^{n-1})$, $\gr_{\mu\nu}\in H^{k+1}(\Tbb^{n-1},\Sbb{n})$ and $\ggr_{\mu\nu}\in H^{k}(\Tbb^{n-1},\Sbb{n})$ is chosen so that 
the inequalities $|\Vr|_{\gr}^2<0$, $\det(\gr_{\mu\nu})<0$ and $|\vr|_{\gr}^2 <0$
are satisfied where $\vr^\mu$ is defined by \eqref{vr-def}.
Then there exists a $t_1<t_0$ and a unique solution
\begin{equation}
  \label{eq:Wreg}
\Wsc \in \bigcap_{j=0}^{k}C^j\bigl((t_1,t_0], H^{k-j}(\Tbb^{n-1})\bigr),
\end{equation}
where
\begin{equation}
  \label{eq:Wdef}
    \Wsc=(\hhu_{\beta\mu\nu},\whu_{\beta\nu}, \zhu_\beta,\ghu_{\mu\nu},\vhu_\mu,\tau,\Vhu^\mu,\Jc^\mu_\nu,\ell^\mu ),
\end{equation}
on $M_{t_1,t_0}$, to the IVP consisting of the evolution
equations \eqref{tconf-ford-C.1}-\eqref{tconf-ford-C.9} and
the initial conditions \eqref{l-idata}-\eqref{hhu-idata}. Moreover, the following properties hold:
\begin{enumerate}[(a)]
\item Letting $\Wsc_0 = \Wsc|_{\Sigma_{t_0}} \in H^{k}(\Tbb^{n-1})$
denote the initial data, there exists for each $t_*\in (t_1,t_0)$ a $\delta>0$ such that if ${\tilde\Wsc}_0 \in  H^{k}(\Tbb^{n-1})$ satisfies $\norm{{\tilde\Wsc}_0-\Wsc_0}_{H^k(\Tbb^{n-1})}<\delta$,    
then there exists a unique solution 
${\tilde\Wsc} \in \bigcap_{j=0}^{k}C^j\bigl((t_*,t_0], H^{k-j}(\Tbb^{n-1})\bigr)$
of the evolution equations \eqref{tconf-ford-C.1}-\eqref{tconf-ford-C.9} on $M_{t_*,t_0}$  that agrees with the initial data ${\tilde\Wsc}_0$ on the initial hypersurface $\Sigma_{t_0}$. 
\item The relations 
\begin{equation} \label{eq:Lag-constraints}
  \begin{split}
\del{\alpha}\ghu_{\mu\nu}&=\Jc^\beta_\alpha \hhu_{\beta\mu\nu},\quad  \del{\alpha}\vhu_{\mu}=\Jc^\beta_\alpha \whu_{\beta\mu}, \quad \del{\alpha}\tau=\Jc^\beta_\alpha \zhu_{\beta},\\
\vhu_\mu&=\zhu_\mu,\quad \Jc^\mu_\nu = \del{\nu} l^\mu,    
  \end{split}
\end{equation}
hold in $M_{t_1,t_0}$.
\item The triple $\{g_{\mu\nu}=\Jc_\mu^\alpha\ghu_{\alpha\beta}
\Jc_{\nu}^\beta,\tau,V^\mu = \Jcch^\mu_\nu \Vhu^\nu\}$
determines a solution 
of the reduced conformal Einstein-Euler-scalar field equations
\begin{equation} \label{lag-redeqns}
  \begin{split}
    -2R_{\mu\nu}+2\nabla_{(\mu} X_{\nu)}&=-\frac{2}{\tau}\nabla_\mu \nabla_\nu \tau-2\Ttt_{ij}, \quad g^{\alpha\beta}\Dc_\alpha\Dc_\beta \tau=0,\\ 
    \att^i{}_{j k}\nabla_i V^k
                &=- \frac{c_s^2(n-1)-1}{c_s^2(n-2)}\tau^{-1} V_j V^i\nabla_i\tau,
  \end{split}
\end{equation}
on $M_{t_1,t_0}$ that satisfies the initial conditions \eqref{V-idata}, \eqref{tauh-idata} and \eqref{dt-tau-idata}-\eqref{dt-g-idata}, where 
\begin{equation*}
  \begin{split}
    \Ttt_{ij}&=2P_0\Bigl(\frac{1+c_s^2}{c_s^2} V^{-2} V_iV_j+\frac{1-c_s^2}{(n-2)c_s^2}g_{ij}\Bigr) \tau^{\frac{c_s^2-1}{c_s^2(n-2)}}V^{-\frac{(1+c_s^2)}{c_s^2}},\\
    X^\gamma &= \frac{1}{2}g^{\mu\nu}g^{\gamma\lambda}
    (2\Dc_{\mu}g_{\nu\lambda}-\Dc_\lambda g_{\mu\nu}),
  \end{split}
\end{equation*}
and $\Dc_\mu$ is the Levi-Civita connection of the flat metric
$\gc_{\mu\nu}=\Jc_{\mu}^\alpha \eta_{\alpha\beta} \Jc_{\nu}^\beta$
on $M_{t_1,t_0}$. 
\item The scalar field $\tau$ is given by
\begin{equation} \label{tau-synch} 
    \tau = t-t_0 + \taur  
\end{equation}
in $M_{t_1,t_0}$ while the vector field
\begin{equation} \label{chi-def}
    \chi^\mu = \frac{1}{|\nabla\tau|^2_g}\nabla^\mu\tau 
\end{equation}
satisfies 
\begin{equation} \label{Lagrangian}
\chi^\mu = \delta^\mu_0
\end{equation}
in $M_{t_1,t_0}$.
\item If the initial data $\{\gr_{\mu\nu},\ggr_{\mu\nu},\taur,\taugr,\Vr^\mu\}$ also satisfies the
constraint equations \eqref{grav-constr}-\eqref{wave-constr}
on $\Sigma_{t_0}$, then the triple $\{g_{\mu\nu},\tau,V^\mu\}$ solves
the conformal Einstein-Euler-scalar field equations
\begin{equation} 
  \begin{split}
G^{\mu\nu} &= \frac{1}{\tau}\nabla^\mu\nabla^\nu \tau + 2 T^{\text{Fl},\mu\nu}, \quad
  \Box_{g}\tau = 0, \\
  \att^\mu{}_{\nu \rho}\nabla_\mu V^\rho
                &=- \frac{c_s^2(n-1)-1}{c_s^2(n-2)}\tau^{-1} V_\nu V^\mu\nabla_\mu\tau,  
  \end{split}
  \label{lag-confeqns}
\end{equation}
and satisfies the wave gauge constraint 
\begin{equation} \label{lag-wave-gauge}
    X^\gamma =\frac{1}{2}g^{\mu\nu}g^{\gamma\lambda}(2\Dc_\mu g_{\nu\lambda}-\Dc_{\lambda}g_{\mu\nu} )=0 
\end{equation}
in $M_{t_1,t_0}$.
\item If 
  \begin{equation}
    \label{eq:cont_crit1}
\max\biggl\{\sup_{M_{t_1,t_0}}\!\det(g_{\mu\nu}), \sup_{M_{t_1,t_0}}\!|\nabla\tau|_{g}^2,\sup_{M_{t_1,t_0}}\!|V|_{g}^2\biggr\} <0
\end{equation}
and
\begin{align}
   \sup_{t_1<t<t_0}\Bigl(&\norm{g_{\mu\nu}(t)}_{W^{2,\infty}(\Tbb^{n-1})}+\norm{\del{t}g_{\mu\nu}(t)}_{W^{1,\infty}(\Tbb^{n-1})}\notag\\ &+\norm{V^\mu(t)}_{W^{1,\infty}(\Tbb^{n-1})+\norm{\Dc_\nu \chi^\lambda(t)}_{W^{2,\infty}(\Tbb^{n-1})}}\notag\\
   &+\norm{\del{t}(\Dc_\nu \chi^\lambda)(t)}_{W^{1,\infty}(\Tbb^{n-1})}\Bigr)<\infty, 
  \label{eq:cont_crit2}
\end{align}
then there exists a $t_1^*<t_1$ such that the solution $\Wsc$ can be uniquely continued to the time interval $(t_1^*,t_0]$.
\end{enumerate}
\end{prop}

\subsection{Temporal synchronization of the singularity\label{temp-synch}}
The temporal synchronization of a big bang singularity requires the introduction of a time coordinate whose level set at a particular time, say $0$, coincides with the spacelike singular hypersurface that defines the singularity. In \cite{BeyerOliynyk:2021}, the scalar field $\tau$ was employed as a time coordinate to synchronize the big bang singularity at $\tau=0$ for the Einstein-scalar field equations. Here, we again use the scalar field as a time coordinate $\tau$ to synchronize the singularity at $\tau=0$, and it is this use of $\tau$ as time that is the main motivation for the use of Lagrangian coordinates in the previous section. As is clear from Proposition \ref{lag-exist-prop}.(d), $\tau$ will coincide with the Lagrangian time $t$ if and only if the initial data $\taur=t_0$ on the initial hypersurface $t=t_0$.  Now, in general, we cannot assume that $\taur $ is constant on the initial hypersurface if we want our results to apply to an open set of geometric initial data. To remedy this, we proceed as in   \cite{BeyerOliynyk:2021}; namely, 
if $\tau$ is not constant on the initial hypersurface $\Sigma_{t_0}$, but is close to constant, say $\tau = t_0 + \rhor$ in $\Sigma_{t_0}$ with $\rhor$ a sufficiently small function, then we evolve $\tau$ for a short amount of time to obtain a solution $\{g_{\mu\nu},\tau,V^\mu\}$ of the conformal Einstein-Euler-scalar field equations on $M_{t_1,t_0}$ for some $t_1<t_0$ with $t_1$ close to $t_0$. We then find a level surface of $\tau^{-1}(t^*_0)$
for some $t^*_0\in (t_1,t_0)$ that satisfies $\tau^{-1}(t^*_0) \subset (t_1,t_0)\times \Tbb^{n-1}$ and $\tau^{-1}(t^*_0)\cong \Tbb^{n-1}$. 
By replacing $\Sigma_{t_0}$ with $\tau^{-1}(t^*_0)$, we obtain a hypersurface $\tau^{-1}(t^*_0)\cong \Tbb^{n-1}$ on which
$\tau$ is constant as desired. This construction is made precise in the following proposition. We omit the proof since it follows from a straightforward modification of the proof of Proposition 5.6 from  \cite{BeyerOliynyk:2021}.

\begin{prop} \label{synch-prop}
Suppose $k>(n-1)/2+1$, $t_0>0$, the initial data $\taur=t_0+\rhor$, $\rhor\in H^{k+2}(\Tbb^{n-1})$, $\taugr\in H^{k+1}(\Tbb^{n-1})$, $\Vr^\mu\in H^k(\Tbb^{n-1},\Rbb^n)$, $\gr_{\mu\nu}\in H^{k+1}(\Tbb^{n-1},\Sbb{n})$ and $\ggr_{\mu\nu}\in H^{k}(\Tbb^{n-1},\Sbb{n})$ is chosen so that 
the inequalities $\det(\gr_{\mu\nu})<0$, $|\Vr|_{\gr}^2<0$ and $|\vr|_{\gr}^2 <0$ hold
and the constraint equations \eqref{grav-constr}-\eqref{wave-constr} are satisfied, and
let $\{\gtt_{\Lambda\Gamma},\Ktt_{\Lambda\Gamma},\taur,\taugr,\Vr^{\mu}\}$
denote the geometric initial data on $\Sigma_{t_0}=\{t_0\}\times \Tbb^{n-1}$ that is determined from the initial data $\{\gr_{\mu\nu},\ggr_{\mu\nu},\taur=t_0+\rhor,\taugr,\Vr^\mu\}$ via \eqref{gtt-def1}.
Then for any $\tilde{\delta}>0$, there exists a $\delta>0$ and times $t^*_1<t^*_0<t_0$ such that if $\norm{\rhor}_{H^{k+2}(\Tbb^{n-1})}<\delta$, then:
\begin{enumerate}[(a)]
    \item The solution
$\Wsc$ from Proposition \ref{lag-exist-prop} to the IVP consisting of the evolution
equations \eqref{tconf-ford-C.1}-\eqref{tconf-ford-C.9} and
the initial conditions \eqref{l-idata}-\eqref{hhu-idata} exists on $M_{t_1^*,t_0}$.
\item The triple $\{g_{\mu\nu}=\Jc_\mu^\alpha\ghu_{\alpha\beta}\Jc_\nu^\beta,\tau,V^\mu = \Jcch^\mu_\nu \Vhu^\nu\}$
determined from the solution $\Wsc$ via \eqref{eq:Wdef} defines a solution to the conformal Einstein-Euler-scalar field equations \eqref{lag-confeqns}.
\item The map
\begin{equation*}
    \Psi \: : \: (t_1^*,t_0)\times \Tbb^{n-1}  \longrightarrow  \Rbb \times \Tbb^{n-1}\: :\: (t,x)\longmapsto (\ttl,\xt) = (t+\rhor(x),x)
\end{equation*}
defines a diffeomorphism onto its image and the push-forward
$ \{ \gt_{\mu\nu} = (\Psi_*g)_{\mu\nu},\taut=\Psi_*\tau,\Vt^\mu = (\Psi_* V)^\mu\}$ of the solution $\{g_{\mu\nu},\tau,V^\mu\}$ by this map 
determines geometric initial data 
$\{\gttt_{\Lambda\Sigma},\Kttt_{\Lambda\Sigma},\mathring{\taut},\grave{\taut}, \mathring{\Vt}{}^\mu\}$ on the hypersurface $\Sigma_{t^*_0}=\{t^*_0\}\times\Tbb^{n-1}$ satisfying $\mathring{\taut}=t^*_0$ and
\begin{align*}
&\norm{\gttt_{\Lambda\Sigma}-\gtt_{\Lambda\Sigma}}_{H^{k+1}(\Tbb^{n-1})}+\norm{\Kttt_{\Lambda\Sigma} -\Ktt_{\Lambda\Sigma}}_{H^{k}(\Tbb^{n-1})}\\
&\!\quad +\norm{t^*_0-\taur}_{H^{k+2}(\Tbb^{n-1})}+\norm{\grave{\taut}-\taugr}_{H^{k+1}(\Tbb^{n-1})}
+\norm{\mathring{\Vt}{}^\mu-\Vr^\mu}_{H^{k}(\Tbb^{n-1})}<\tilde{\delta}.    
\end{align*}
\end{enumerate}
\end{prop}

\begin{rem} \label{synch-rem}
Given the geometric initial data 
$\{\gttt_{\Lambda\Sigma},\Kttt_{\Lambda\Sigma},\mathring{\taut},\grave{\taut},\mathring{\Vt}{}^\mu\}$ on $\Sigma_{t^*_0}$ from Proposition \ref{synch-prop}, we can always solve the wave gauge constraint on $\Sigma_{t^*_0}$ by an appropriate choice of the free initial data\footnote{It is worthwhile noting that this choice of free initial data will, in general, be different from the lapse-shift pair computed from restricting the conformal Einstein-Euler-scalar field solution $\{\gt_{\mu\nu},\taut,\Vt^\mu\}$ from Proposition \ref{synch-prop} to $\Sigma_{t^*_0}$.} $\{\Nttt,\bttt_\Lambda,\dot{\Nttt},\dot{\bttt}{}_\Lambda\}$. Because of this, we lose no generality, as far as our stability results are concerned, by assuming that the initial data \eqref{gt-idata}-\eqref{Vh-idata} satisfies the gravitational and wave gauge constraint equations
\eqref{grav-constr}-\eqref{wave-constr}
along with the synchronization condition $\taur = t_0$ on the initial hypersurface $\Sigma_{t_0}$,
which by \eqref{tau-synch} implies that
\begin{equation} \label{tau-fix}
\tau = t \quad \text{in $M_{t_1,t_0}$}.
\end{equation}
\end{rem}

\section{A Fermi-Walker transported frame \label{Fermi}}

In the calculations carried out in this section, we will assume that $\{g_{\mu\nu},\tau,V^\mu\}$ is a solution, which is guaranteed to exist by Proposition \ref{lag-exist-prop}, of the conformal Einstein-scalar field equations \eqref{lag-confeqns} in the Lagragian coordinates $(x^\mu)$ that satisfies the wave gauge constraint \eqref{lag-wave-gauge} along with the slicing condition \eqref{tau-fix}. A difficulty with the Lagrangian coordinate representation is that it is not suitable for obtaining estimates that are well behaved near the big bang singularity, which is located at $t=0$ in these coordinates. In order to obtain estimates that are well behaved in the limit $t\searrow 0$, we will instead use a frame representation of the conformal metric given by \begin{equation*}% \label{gij-def}
g_{ij} = e_i^\mu g_{\mu\nu} e^\nu_j,
\end{equation*}
where $e_i = e^\mu_i \del{\mu}$ is a frame that is to be determined following the same approach as in \cite[\S6]{BeyerOliynyk:2021}, which, in turn, was inspired by  \cite{Fournodavlos_et_al:2023}. In the following, we take all frame indices as being expressed relative to this frame, and so in particular, the frame components
$e_i^j$ of the frame vector fields $e_i$ are
\begin{equation} \label{frame-def}
e_i^j = \delta^j_i.
\end{equation}

To proceed, we need to fix the frame. We do this by first fixing $e_0$ via
\begin{equation}\label{e0-def}
e_0 = (-|\chi|_g^2)^{-\frac{1}{2}}\chi \oset{\eqref{chi-def}}{=}-\betat \nabla \tau
\end{equation}
where $\betat$ is defined by
\begin{equation}\label{alpha-def}
    \betat = (-|\nabla\tau|^2_g)^{-\frac{1}{2}},
  \end{equation}
  which means that the coordinate components satisfy
  \begin{equation}\label{e0-mu}
e_0^\mu = \betat^{-1}\chi^\mu =\betat^{-1}\delta^\mu_0.
\end{equation}

It follows that $e_0$ is normalized according to
\begin{equation*} 
g_{00} = g(e_0,e_0) =-1.
\end{equation*}
We then complete $e_0$ to a frame by propagating the spatial vector fields $e_J$ using Fermi-Walker transport, which is defined by
\begin{equation} \label{Fermi-A}
\nabla_{e_0}e_J = -\frac{g(\nabla_{e_0}e_0,e_J)}{g(e_0,e_0)} e_0.
\end{equation}
As demonstrated in \cite{BeyerOliynyk:2021}, this, after an appropriate choice of the spatial frame on the initial hypersurface, yields a frame that is  orthonormal with respect to the conformal metric, that is,
\begin{equation}\label{orthog}
g_{ij}=g(e_i,e_j) = \eta_{ij}:= -\delta_i^0\delta_j^0+ \delta_i^I\delta_j^J\delta_{IJ}.
\end{equation}

By definition, the connection coefficients $\Gamma_i{}^k{}_j$ of the metric $g_{ij}$ and $\gamma_i{}^k{}_j$ of the metric $\gc_{ij}$ are determined by
\begin{equation} \label{Gamma-def}
\nabla_{e_i} e_j = \Gamma_{i}{}^k{}_{j} e_k \AND \Dc_{e_i} e_j = \gamma_{i}{}^k{}_{j} e_k,
\end{equation}
respectively.
These connection coefficients are related by \eqref{Ccdef}, and is was established in \cite{BeyerOliynyk:2021} that
the connection coefficients 
$\gamma_{0}{}^k{}_j$, $\gamma_I{}^0{}_0$ and $\gamma_I{}^K{}_0$ can be expressed as
\begin{align}
\gamma_{0}{}^k{}_j  &
=
-\betat\delta^i_0\bigl(\delta^0_jg^{lk}+\delta^l_j\delta^k_0\bigr)
(\Dc_{i}\Dc_l\tau-\Cc_i{}^p{}_l\Dc_p\tau)  -\Cc_{0}{}^k{}_j,\label{gamma-0kj}
\end{align}
\begin{equation}\label{gamma-I00-IK0}
  \gamma_I{}^0{}_0 = \frac{1}{2}\delta^i_I\delta_0^j\delta^k_0\Dc_i g_{jk} \AND
   \gamma_I{}^K{}_0 =-  \delta^{Kl} \delta_I^i\delta^j_0\Dc_i g_{jl}+ \eta^{KL}\gamma_I{}^0{}_L,
\end{equation}
respectively.

In the following, we view \eqref{gamma-0kj} and \eqref{gamma-I00-IK0} as determining the background
connection coefficients $\gamma_{0}{}^k{}_j$, $ \gamma_I{}^0{}_0$
and  $\gamma_I{}^K{}_0$.  The remaining background connection coefficients
$\gamma_{I}{}^k{}_J$ will be determined by a transport equation. To derive the transport equation, we recall that the background curvature is determined via
\begin{equation*}
\Rc_{ijk}{}^l = e_{j}(\gamma_i{}^l{}_k)-e_{i}(\gamma_j{}^l{}_k)+\gamma_i{}^m{}_k\gamma_j{}^l{}_m-\gamma_j{}^m{}_k\gamma_i{}^l{}_m-(\gamma_j{}^m{}_i-\gamma_i{}^m{}_j)\gamma_m{}^l{}_k \oset{\eqref{curvature}}{=}0,
\end{equation*}
which yields, see \cite{BeyerOliynyk:2021} for details,
\begin{equation}\label{gamma-Ikj-A}
 \del{t}\gamma_I{}^k{}_J=\betat\bigl(e_I(\gamma_0{}^k{}_J)-\gamma_I{}^l{}_J\gamma_0{}^k{}_l+\gamma_0{}^l{}_J\gamma_I{}^k{}_l+(\gamma_0{}^l{}_I-\gamma_I{}^l{}_0)\gamma_l{}^k{}_J\bigr).
\end{equation}
From \eqref{e0-mu} and
\begin{equation} \label{e0I-fix}
    e_I^0=e_I(t) = e_I(\tau) = g(e_I,\nabla \tau) = \betat^{-1}g(e_0,e_I) = 0,
\end{equation}
we observe that the evolution of the spatial frame $e_I$ is governed by
\begin{equation} \label{[e0,eI]-C}
\del{t}e_I^\Lambda=\betat (\gamma_0{}^J{}_I - \gamma_{I}{}^J{}_0)
e^\Lambda_J.
\end{equation}
Finally, for use below, we note that the derivative of \eqref{alpha-def} is
given by 
\begin{equation}\label{grad-alpha}
    e_i(\betat)=\nabla_i\betat =\betat^3
    \nabla^k\tau \nabla_i\nabla_k \tau =
    -\betat^2 \delta^k_0 (\Dc_i\Dc_k\tau
    -\Cc_i{}^l{}_k \Dc_l\tau)
\end{equation}
where in deriving the third equality we 
used \eqref{Ccdef}, \eqref{frame-def} and \eqref{e0-def}. 

\section{First order form}
\label{sec:FirstOrderForm}

We now turn to the construction of a frame representation for the reduced conformal Einstein-Euler-scalar field equations in first order form. Our ultimate goal is to derive a Fuchsian formulation of the evolution equations, which can be used to establish the existence of solutions up to the big bang singularity at $t=0$. The derivations presented in
this section follow closely those from 
 \cite[\S7]{BeyerOliynyk:2021}. 
 
\subsection{Primary fields}
Following \cite[\S7]{BeyerOliynyk:2021}, 
we begin by deriving a set of first order evolution equations for the \emph{primary fields}
\begin{equation}\label{p-fields}
g_{ijk}=\Dc_i g_{jk}, \quad \tau_{ij}=\Dc_i\Dc_j \tau \AND W^k=\frac{1}{\ftt} V^k,
\end{equation}
where 
\begin{equation}
  \label{eq:Eulerfchoice}
    \ftt=\tau^{\frac{(n-1)c_s^2-1}{n-2}}\betat^{c_s^2}.
\end{equation}
Once that is done, we then derive evolution equations for the differentiated fields
\begin{equation}\label{d-fields}
    g_{ijkl}=\Dc_ig_{jkl},\quad \tau_{ijk}=\Dc_i\tau_{jk} \AND \Utt^s_q=\Dc_q W^s.
\end{equation}
The above choice for $\ftt$ turns out to be the optimal choice for our
analysis and is also consistent with the explicit model solution \eqref{eq:FLRWEulerSFExplSol1} -- \eqref{eq:FLRWEulerSFExplSol3} for which $\betat=a^{n-1}$ and $\del{t}=\betat e_0$ and  therefore $W^0=V_*^0=const$ and $W^I=0$.

 In terms of the variables \eqref{p-fields} 
the tensor $\Ttt_{ij}$ defined by \eqref{eq:Tttconfphys} becomes
% \begin{equation}
%   %\label{eq:TttconfphysW}
%   \Ttt_{ij}=P_0\Bigl(\frac{1+c_s^2}{c_s^2} w^{-2} W_iW_j+\frac{1-c_s^2}{(n-2)c_s^2}g_{ij}\Bigr) \tau^{-\frac{1-c_s^2}{c_s^2(n-2)}}\ftt^{-\frac{1+c_s^2}{c_s^2}}w^{-\frac{1+c_s^2}{c_s^2}},
% \end{equation}
\begin{equation}
  \label{eq:TttconfphysW}
  \Ttt_{ij}=2P_0\Bigl(\frac{1+c_s^2}{c_s^2} \frac{W_iW_j}{w^2}+\frac{1-c_s^2}{(n-2)c_s^2}g_{ij}\Bigr) \tau^{-\frac{(n-3)+(n-1)c_s^2}{n-2}}\betat^{-(1+c_s^2)} w^{-\frac{1+c_s^2}{c_s^2}}
\end{equation}
where
\begin{equation}
  \label{eq:defw1234}
w^2=-W_iW^i=\ftt^{-2} v^2.
\end{equation}
% and
% \begin{equation}
%   \label{eq:defhatbeta}
%   \hat\beta=\betat^{-(1+c_s^2)}.
% \end{equation}
Using \eqref{red-Ricci} and \eqref{tau-fix}, we express
\eqref{confESFF} as
\begin{equation}\label{confESSFJ}
g^{kj}\Dc_k \Dc_j g_{lm} =-\frac{2}{t}\bigl(
\Dc_l \Dc_m \tau - \Cc_l{}^p{}_m\Dc_p \tau  \bigr)- Q_{lm}
{-2\Ttt_{lm}},
\end{equation}
where by \eqref{Ccdef},  \eqref{frame-def}, \eqref{e0-def} and \eqref{orthog}, we
note that
\begin{equation}\label{grad-tau}
    \Dc_i\tau= -\betat^{-1}g_{ij}e^j_0 = \betat^{-1}\delta_i^0 
\end{equation}
and
\begin{equation}\label{Cc-dt}
\Cc_l{}^p{}_m\Dc_p \tau 
=-\frac{1}{2}\betat^{-1}e_0^j\bigl(\Dc_l g_{jm}
+\Dc_m g_{jl} - \Dc_j g_{lm}\bigr).
\end{equation}
Multiplying \eqref{confESSFJ} by $-e_0^i$ yields
\begin{equation} \label{for-A}
-e^i_0 g^{kj}\Dc_k \Dc_j g_{lm} =\frac{2}{t}e^i_0
\bigl(
\Dc_l \Dc_m \tau - \Cc_l{}^p{}_m\Dc_p \tau  \bigr) + e^i_0 Q_{lm}
{+2e^i_0\Ttt_{lm}}.
\end{equation}
But $-g^{ik}e^j_0\Dc_k\Dc_j g_{lm} + g^{ij}e^k_0 \Dc_k \Dc_j g_{lm} 
=0$ by \eqref{commutator},
and so, adding this to \eqref{for-A} yields with the help of \eqref{frame-def}, \eqref{orthog} and \eqref{Cc-dt}, the first order formulation of the frame formulation of the reduced conformal Einstein equations given by
\begin{equation}\label{for-B}
  \begin{split}
B^{ijk}\Dc_kg_{jlm} =&
\frac{1}{t}\betat^{-1}\delta^i_0\delta^j_0(g_{ljm}+g_{mjl}-g_{jlm})\\
&+\frac{2}{t}\delta^i_0\tau_{lm}+ \delta^i_0 Q_{lm}{+2\delta^i_0\Ttt_{lm}}
  \end{split}
\end{equation}
where
\begin{align}
B^{ijk} &= -\delta^i_0 \eta^{jk}-\delta^j_0\eta^{ik}+ \eta^{ij}\delta^k_0.\label{B-def}
\end{align}
Using similar arguments, it is not difficult to verify that the conformal scalar field equation
\eqref{confESFI} can be written in first order
form as 
\begin{align}\label{for-D}
  B^{ijk}\Dc_k \tau_{jl} =& -\delta^i_0 \eta^{pr}\eta^{qs}g_{lrs}  \tau_{pq}.
\end{align}

To proceed, we view \eqref{for-B} and \eqref{for-D} as transport equations for $g_{jlm}$ and  $\tau_{jl}$ by expressing them as
\begin{equation}
  \begin{split}
    \del{t}g_{rlm} =&
                 \frac{1}{t}\delta^0_r\delta^j_0(g_{ljm}+g_{mjl}-g_{jlm})
                 +\frac{2}{t}\delta_r^0\betat
                 \tau_{lm}\\
                 &-\delta_{ri}B^{ijK}\betat g_{Kjlm}+ \delta_{ri}\betat Q^i_{lm}{+2 \delta_{ri}\delta^i_0 \betat\Ttt_{lm}}
  \end{split}
  \label{for-F.1}
\end{equation}
  {and}
  \begin{equation}
    %\begin{split}
                   \del{t}\tau_{rl} =
                                       \delta_{ri}\betat J^i_l-\delta_{ri}B^{ijK}\betat\tau_{Kjl},
 \label{for-F.2}
\end{equation}
respectively, where in deriving these expression we have employed \eqref{e0-mu}, \eqref{d-fields} and
\eqref{B-def} and set
\begin{align}
  Q^i_{lm} &= \delta^i_0 Q_{lm}+\delta^{ij}\bigl(  \gamma_0{}^p{}_j  g_{plm} +\gamma_0{}^p{}_l  g_{jpm} +\gamma_0{}^p{}_m  g_{jlp}\bigr)
             \label{Qilm-def}
\intertext{and}
             J^i_l&=-\delta^i_0\eta^{kp}\eta^{jq}g_{lpq}\tau_{kj}+\delta^{ij}\bigl(  \gamma_0{}^p{}_j  \tau_{pl} +\gamma_0{}^p{}_l  \tau_{jp}\bigr).\label{J-def}
\end{align}

\begin{rem}
As in \cite{BeyerOliynyk:2021}, we do not employ equation \eqref{for-F.1} with $l=m=0$ in any of our subsequent arguments. This is possible because the wave gauge condition \eqref{wave-gauge} allows us to express $g_{i00}$ in terms of the other metric variables $g_{ilM}$. To see why this is the case, we can use \eqref{Xdef}, \eqref{orthog}, and \eqref{p-fields} to rewrite the wave gauge constraint \eqref{wave-gauge} as 
\begin{equation*}
  \begin{split}
2X_0 &= -g_{000}+\delta^{JK}(2g_{JK0}-g_{0JK})=0, \\
2 X_I &= -(2g_{00I}-g_{I00})+\delta^{JK}(2g_{JKI}-g_{IJK})=0.
  \end{split}
\end{equation*}
After rearranging and utilizing the symmetry $g_{IJ0}=g_{I0J}$, this can be expressed as
\begin{equation} \label{gi00}
  \begin{split}
g_{000}&= -\delta^{JK}(g_{0JK}-2g_{J0K}),\\ 
g_{I00}&= 2g_{00I}-\delta^{JK}(2g_{JKI}-g_{IJK}),
  \end{split}
\end{equation}
which confirms the assertion.
\end{rem}

Separating \eqref{for-F.1} into the components $(r,l,m)=(0,0,M)$ and $(r,l,m)=(R,0,M)$, $(r,l,m)=(0,L,M)$ and
$(r,l,m)=(R,L,M)$, we obtain,
with the help of \eqref{gi00}, the equations
\begin{align}
  \del{t}g_{00M} =& \frac{1}{t}\bigl(2g_{00M}-\delta^{IJ}(2g_{IJM}-g_{MIJ})\bigr)
                   +\frac{2}{t}\betat\tau_{0M} - B^{0jK}\betat g_{Kj0M}\notag\\
&+\betat Q^0_{0M}{+2 \betat\Ttt_{0M}}, \label{for-G.1}\\
\del{t}g_{R0M} =& -\delta_{RI}\betat B^{IjK}g_{Kj0M}+\delta_{RI}\betat Q^I_{0M}, \label{for-G.2}\\
  \del{t}g_{0LM} =& \frac{1}{t}(g_{L0M}+g_{M0L}-g_{0LM})
                   +\frac{2}{t}\betat\tau_{LM}\notag\\
                   &- B^{0jK}\betat g_{KjLM}
                   +\betat Q^0_{LM}
  {+2 \betat\Ttt_{LM}}\label{for-G.3}
\intertext{and}
\del{t}g_{RLM} =& -\delta_{RI}\betat B^{IjK}g_{KjLM}+ \delta_{RI}\betat Q^I_{LM}. \label{for-G.4}
\end{align}
Rather than using \eqref{for-G.3}, we follow \cite{BeyerOliynyk:2021} and use \eqref{for-G.2} to rewrite it as
\begin{align}
\del{t}(g_{0LM}-g_{L0M}-g_{M0L}) = & -\frac{1}{t}(g_{0LM}-g_{L0M}-g_{M0L})\notag\\
&+\frac{2}{t}\betat\tau_{LM}+S_{LM} {+2 \betat\Ttt_{LM}}\label{for-H}
\end{align}
where
\begin{align}
    S_{LM}=&-\betat B^{0jK}g_{KjLM}
+\betat Q^0_{LM} +\delta_{MI}\betat B^{IjK}g_{Kj0L} -\delta_{MI}\betat Q^I_{0L}\notag\\
&+\delta_{LI}\betat B^{IjK}g_{Kj0M}-\delta_{LI}\betat Q^I_{0M}.\label{S-def}
\end{align}
The metric combination $g_{0LM}-g_{L0M}-g_{M0L}$, which appears above in \eqref{for-H},
plays a dominant role in our analysis, and
it is related to the second fundamental form $\Ktt_{IJ}$ associated to the $t=const$-hypersurfaces and the conformal metric $g$ via the formula
\begin{equation}\label{Ktt-def}
\Ktt_{IJ} = \nabla_K \ntt_{(J}\delta_{I)}^K
  {=} \frac{1}{2}(g_{0IJ}-g_{I0J}-g_{J0I}) +\gamma_{(I}{}^0{}_{J)}.
\end{equation}
Here, $\ntt_i= -(-|\nabla t|_g^2)^{-\frac{1}{2}} e_i(t)$ is the conormal to the $t=const$-hypersurfaces
and in deriving the second equality in the above formula, we
employed \eqref{Ccdef} and the identity
$\nabla_I n_J = g(\nabla_{e_I} e_0, e_J) = -g(e_0,\nabla_{e_I}e_J) =\Gamma_{I}{}^{0}{}_J$, which holds by \eqref{alpha-def}, \eqref{e0-mu}
and \eqref{orthog}.

For use below, we observe by  \eqref{grad-alpha}, \eqref{Cc-dt} and
\eqref{gi00} that $\Dc_q\betat$  can be expressed as
\begin{align}
\Dc_q\betat %=-\Bigl(\betat^{2}\tau_{q0}+\frac{1}{2}\betat g_{q00}\Bigr)
%= e_q(\betat)
=&-\delta_q^0\Bigl(\betat^{2}\tau_{00}-\frac{1}{2}\betat \delta^{JK}(g_{0JK}-2g_{J0K}) \Bigr)\notag\\
&-\delta_q^P\Bigl(\betat^{2}\tau_{P0}+\frac{1}{2}\betat\bigl(2g_{00P}-\delta^{JK}(2g_{JKP}-g_{PJK})\bigr)\Bigr).
\label{grad-alpha-A}
\end{align}
Now, in terms of the variables $W^s$ defined in \eqref{p-fields}, the conformal Euler equations \eqref{confESFHa} take the form
\begin{equation}
  \label{confESFHa.1}
  {a^i}_{j k}\Dc_i W^k =G_{jsl}W^sW^l
\end{equation}
where
\begin{align}
  \label{confESFHa.2}
   &a^i{}_{jk} :=\frac 1{\ftt} \att^i{}_{jk}
    =\frac{3 c_s^2+1}{c_s^2} \frac{W_j W_k W^i}{w^2} +W^i g_{jk}
                +2{\delta^i}_{(k} W_{j)},\\
  &G_{jsl}:=\Gt_{jsl}+\Bigl(\frac{c_s^2+1}{c_s^2} \delta^p_{(l} g_{s)j}-\delta^p_jg_{sl}\Bigr) \frac{\Dc_p \ftt}{\ftt}\notag\\
  =&\betat^{-1}t^{-1}
       \frac{(n-1)c_s^2-1}{n-2} \Bigl(       
       \delta_{(l}^0  \eta_{s)j}-\delta_j^0\eta_{sl}
             \Bigr) \notag\\   
  &- \frac{1}{2}\Bigl(
  \frac{3 c_s^2+1}{c_s^2} \frac{W^q W^p}{w^2}+\eta^{pq}\Bigr)  \eta_{j(s} g_{l) p q} 
             -  g_{(l s) j}\notag\\
              &+\betat^{-1}\Bigl(\frac{c_s^2+1}{c_s^2} \delta^0_{(l} \eta_{s)j}-\delta^0_j\eta_{sl}\Bigr) \Bigl(
       \frac{c_s^2}{2} \betat \delta^{JK}(g_{0JK}-2g_{J0K})  -c_s^2\betat^{2}\tau_{00}\Bigr)\notag\\
  &-c_s^2\Bigl(\frac{c_s^2+1}{c_s^2} \delta^P_{(l} \eta_{s)j}-\delta^P_j\eta_{sl}\Bigr)\betat\tau_{P0}\notag\\
  &-c_s^2\Bigl(\frac{c_s^2+1}{c_s^2} \delta^P_{(l} \eta_{s)j}-\delta^P_j\eta_{sl}\Bigr)\frac{1}{2}\bigl(2g_{00P}-\delta^{JK}(2g_{JKP}-g_{PJK})\Bigr),
  \label{confESFHa.3}
\end{align}
and in deriving the above formulas for
$a^i{}_{jk}$ and $G_{jsl}$, we employed 
\eqref{eq:AAACF2}, \eqref{eq:DefGjsl}, \eqref{orthog}, \eqref{p-fields},
\eqref{grad-tau}, \eqref{grad-alpha-A} and 
\begin{align*}
    \frac{\Dc_i \ftt}{\ftt}\oset{\eqref{eq:Eulerfchoice}}{=}&
                        \frac{(n-1)c_s^2-1}{n-2}\frac{\Dc_i\tau}{\tau} +c_s^2\frac{\Dc_i\betat}{\betat}.
\end{align*}
Setting $i=0$ in \eqref{confESFHa.2}, we observe that
\begin{align}
      a^0{}_{jk}=&\delta_j^{0}\delta_k^{0}\Bigl(\frac{1}{c_s^2} 
      +\frac{3 c_s^2+1}{c_s^2} \frac{W^IW_I}{w^2}\Bigr) W^{0}\notag\\
      &+2\delta_{(j}^{0}\delta_{k)}^{K}\Bigl(-\frac{2 c_s^2+1}{c_s^2} -\frac{3 c_s^2+1}{c_s^2} \frac{W^IW_I}{w^2}\Bigr) W_{K}\notag\\
      &+\delta_j^{J}\delta_k^{K}\Bigl(
      \delta_{JK}
      +\frac{3 c_s^2+1}{c_s^2} \frac{W_{J} W_{K}}{w^2}\Bigr)    W^0,
      \label{eq:B0fluid}
\end{align}
  where in deriving this we have used
  \begin{equation}
    \label{eq:WQdr}
    \frac{(W^0)^2}{w^2}=1+\frac{\delta_{IJ}W^IW^J}{w^2}=1+\frac{W^IW_I}{w^2},
  \end{equation}
  which is a consequence of \eqref{orthog} and \eqref{eq:defw1234}. Similarly, setting $i=I$ in \eqref{confESFHa.2} yields
  \begin{align}
    a^I{}_{jk}=&\delta_j^{0}\delta_k^{0}\Bigl(
    \frac{2 c_s^2+1}{c_s^2} 
    +\frac{3 c_s^2+1}{c_s^2} \frac{W^LW_L}{w^2}
    \Bigr)W^I\notag\\
    &-2\delta_{(j}^{0}\delta_{k)}^{K}\Bigl(
    {\delta^I}_{K}+
    \frac{3 c_s^2+1}{c_s^2} \frac{W^I W_{K}}{w^2}    
    \Bigr) W^{0} \notag\\
    &+{\delta^J}_{j} {\delta ^{K}}_{k}\Bigl(
    \delta^{LI}\delta_{JK}
    +2{\delta^I}_{(K}{\delta^L}_{J)}+
    \frac{3 c_s^2+1}{c_s^2} \frac{W_{J} W_{K}}{w^2}\delta^{LI}
    \Bigr)W_L.
    \label{eq:BIfluid}
  \end{align}

\subsection{Differentiated fields}
We now turn to deriving evolution equations for the differentiated fields \eqref{d-fields}. Here again, we closely follow the approach taken in \cite{BeyerOliynyk:2021}. Applying $\Dc_q$ to \eqref{confESSFJ}, we find with the
help of \eqref{commutator}, \eqref{tau-fix}, \eqref{frame-def}, \eqref{orthog}, \eqref{p-fields}-\eqref{d-fields} and \eqref{grad-tau}-\eqref{Cc-dt} that
\begin{align}
  g^{kj}\Dc_k \Dc_q \Dc_j g_{lm}
  =&-\frac{1}{t}\betat^{-1} e_0^j (\Dc_q\Dc_l g_{jm}
  +\Dc_q\Dc_m g_{jl}-\Dc_q\Dc_j g_{lm})\notag\\
  &-\frac{2}{t}\Dc_q\Dc_l\Dc_m\tau -P_{qlm}{-2 \Dc_q\Ttt_{lm}},\label{confESSFL}
\end{align}
where
\begin{align}
P_{qlm}=&  \frac{1}{t}\bigl(\Dc_q(\betat^{-1}) \delta_0^j+\betat^{-1}  \gamma_q{}^j{}_0\bigr)  (g_{ljm}+g_{mjl}-
g_{jlm}) - \eta^{kr}\eta^{js} g_{qrs}g_{kjlm} \notag\\
 &+ \Dc_q Q_{lm}
+\biggl(- \frac{1}{t^2}\betat^{-1} \delta_0^j (g_{ljm}+g_{mjl}-
g_{jlm})-\frac{2}{t^2}
 \tau_{lm}\biggr)\betat^{-1}\delta_q^0 \label{P-def}
\end{align}
and in deriving \eqref{P-def} we have used
\begin{equation} \label{De0}
 \Dc_q e^j_0=\Dc_q \delta^j_0 = \gamma_q{}^j{}_0.
\end{equation}
In \eqref{P-def}, $\Dc_q(\betat^{-1})=-\betat^{-2}\Dc_q(\betat)$ is to
be computed using \eqref{grad-alpha-A}.

Next, by \eqref{commutator}, we have
\begin{equation*}
-g^{ik}e^j_0\Dc_k\Dc_q\Dc_j g_{lm} + g^{ij}e^k_0 \Dc_k \Dc_q \Dc_j g_{lm} =0,
\end{equation*}
and so adding this to \eqref{confESSFL} yields the
first order equation
\begin{align}
  B^{ijk}\betat\Dc_k g_{qjlm}
  =& \frac{1}{t} \delta_0^i\delta_0^j  (g_{qljm}+ g_{qmjl}-g_{qjlm})\notag\\
  &+\frac{2}{t}\delta^i_0\betat\tau_{qlm}+ \delta^i_0 \betat P_{qlm}{   +2 \delta_0^i \betat\Dc_q\Ttt_{lm}},
  \label{for-I}
\end{align}
for the differentiated fields $g_{qjlm}$, where $B^{ijk}$ is as defined above by \eqref{B-def}. 
Similarly, applying $\Dc_q$ to the scalar field equation \eqref{confESFI} yields the first
order equation
\begin{align} 
  B^{ijk}\betat\Dc_k\tau_{qjl}
  =&  \betat K^i_{ql}\label{for-J}
\end{align}
where
\begin{align} 
K^i_{ql}=& \delta^{i}_0\bigl(
-\eta^{jr}\eta^{ks}g_{qrs}\tau_{kjl}
-\bigl(\eta^{jr}\eta^{ks}g_{qlrs}-2\eta^{jm}\eta^{rn}\eta^{ks}g_{qmn}
g_{lrs}\bigr)\tau_{jk}\notag\\ 
&- \eta^{jr}\eta^{ks}g_{lrs}\tau_{qjk}\bigr).
\label{K-def}
\end{align}

\begin{rem}
As in \cite{BeyerOliynyk:2021}, we do not use the $q=0$ component of the equations
\eqref{for-I} and \eqref{for-J}. Instead, we obtain $g_{0jlm}$ and
$\tau_{0jl}$ from the equations \eqref{for-B} and \eqref{for-D}. Indeed, a straightforward calculation reveals that
\begin{align}
  g_{0rlm} =& \frac{1}{t}\betat^{-1}\delta_r^0\delta_0^j(g_{ljm}+g_{mjl}-g_{jlm})
             +\frac{2}{t}\delta_r^0\tau_{lm} - \delta_{ri}B^{ijK}g_{Kjlm}
             \notag\\
             &+\delta_r^0 Q_{lm}+2 \delta^0_r\Ttt_{lm}\label{for-K.1}
\intertext{and}
             \tau_{0rl} =&
                          -\delta_{ri}B^{ijK}\tau_{Kjl}
                           -\delta_r^0 \eta^{jp}\eta^{kq} g_{lpq} \tau_{kj}. \label{for-K.2}
\end{align}
\end{rem}

To complete the derivation of evolution equations for the differentiated fields, 
we apply $\Dc_q$ to the conformal Euler equations \eqref{confESFHa.1} to get
\begin{equation}
  \label{eq:EulerDer}
  {a^i}_{j k}\Dc_i \Dc_q W^k =\Dc_q(G_{jsl}W^sW^l)-\Dc_q{a^i}_{j k}\Dc_i W^k,
\end{equation}
where in deriving this we have employed \eqref{commutator}. From the definition \eqref{d-fields} of $U^s_q$ and the conformal Euler equations \eqref{confESFHa.1}, we observe that
\begin{equation}
  \label{eq:Um0}
  \Utt^m_0=\Dc_0W^m=-(a^{0})^{-1}{}^{mj}a^I{}_{jk}\Utt^k_I+(a^{0})^{-1}{}^{mj}G_{jsl}W^sW^l
\end{equation}
where $(a^{0})^{-1}{}^{mj}$ is the inverse of the matrix $a^0{}_{jk}$, that is,
\begin{equation}
  \label{eq:defa0inv}
  (a^{0})^{-1}{}^{mj}a^0{}_{jk}=\delta^m_k.
\end{equation}
This inverse is well-defined and depends smoothly on $W^s$ provided $w^2=-W_s W^s >0$.
Due to the relation \eqref{eq:Um0}, we only
need to consider the differentiated fluid variables $U^s_Q$. Setting $q=Q$ in \eqref{eq:EulerDer} and exploiting the relations \eqref{frame-def}, \eqref{Gamma-def} and \eqref{eq:Um0}, we arrive at the following evolution
equations for the differentiataed fluid variables
$U^k_Q$:
\begin{align}  
  {a^i}_{j k}\Dc_i \Utt_Q^k
  =&\Dc_Q(G_{jsl}W^sW^l)
     -(\Dc_Q{a^p}_{j k}-{a^i}_{j k}\gamma_i{}^p{}_Q) \Utt_p ^k\notag\\
  =&\Dc_Q(G_{jsl}W^sW^l)
     +\bigl(\Dc_Q{a^0}_{j m}-{a^0}_{j m}\gamma_0{}^0{}_Q \notag \\
     &-{a^I}_{j m}\gamma_I{}^0{}_Q\bigr) (a^{0})^{-1}{}^{mn}\Bigl( a^I{}_{nk}\Utt^k_I-G_{nsl}W^sW^l\Bigr) \notag \\
     &-\bigl(\Dc_Q{a^P}_{j k}-{a^0}_{j k}\gamma_0{}^P{}_Q -{a^I}_{j k}\gamma_I{}^P{}_Q\bigr) \Utt_P ^k. \label{eq:EulerDer2}
  % =&\Dc_Q(G_{jsl}W^sW^l)\\
  %    &+\bigl(\Dc_Q{a^0}_{j m}+{a^0}_{j m}\betat \tau_{0Q} -{a^I}_{j m}\gamma_I{}^0{}_Q\bigr) (a^{0})^{-1}{}^{mn}\Bigl( a^I{}_{nk}\Utt^k_I-G_{nsl}W^sW^l\Bigr)\\
  %    &-\bigl(\Dc_Q{a^P}_{j k}-{a^0}_{j k}(-\frac{1}{2}\delta^{KL}(g_{0JL}+g_{J0L}-g_{L0J})) -{a^I}_{j k}\gamma_I{}^P{}_Q\bigr) \Utt_P ^k\\
\end{align}
The quantity $\Dc_Q(G_{jsl}W^sW^l)$ in \eqref{eq:EulerDer} is understood as
\begin{equation}
  \label{eq:DcGWW}
  \Dc_Q(G_{jsl}W^sW^l)=\Dc_QG_{jsl}W^sW^l+G_{jsl} {{\Utt}}_Q^sW^l+G_{jsl}W^s{{\Utt}}_Q^l.
\end{equation}
We note that $\Dc_QG_{jsl}$ can be expressed
in terms of the variables \eqref{p-fields} and \eqref{d-fields} by differentiating \eqref{confESFHa.3}. Similarly,  $\Dc_Q{a^i}_{j k}$ can be expressed in terms of the variables \eqref{p-fields} and \eqref{d-fields} by differentiating \eqref{confESFHa.3}.

\subsection{Frame and connection coefficient transport equations}
As observed in \cite[\S7]{BeyerOliynyk:2021}, 
 $\betat$ satisfies the evolution equation
\begin{equation}\label{for-L}
\del{t}\betat 
=-\betat^{3}\tau_{00}+\frac{1}{2}\betat^2 \delta^{JK}(g_{0JK}-2g_{J0K}),
\end{equation} 
which is be obtained from setting $q=0$ in \eqref{grad-alpha-A} and using \eqref{e0-mu}.
We also observe from \eqref{Ccdef}, \eqref{frame-def}, \eqref{orthog}, \eqref{gamma-0kj},
 \eqref{p-fields}-\eqref{d-fields} and \eqref{Cc-dt}
that the connection coefficients
$\gamma_0{}^k{}_j$ can be expressed as
\begin{equation}\label{gamma-0kj-A}
    \gamma{}_0{}^k{}_j
    =-(\delta^0_j\eta^{kl}+\delta^l_j\delta^k_0)\Bigl(\betat\tau_{0l}+\frac{1}{2}g_{l00} \Bigr)-\frac{1}{2}\eta^{kl}(g_{0jl}+g_{j0l}-g_{l0j}).
\end{equation}
Setting $(k,j)=(0,0)$, $(k,j)=(K,j)$, $(k,j)=(0,J)$
and $(k,j)=(K,J)$
in the above formula for $\gamma_0{}^k{}_j$ gives
\begin{align}
     \gamma{}_0{}^0{}_0 =&
     \frac{1}{2}g_{000}\oset{\eqref{gi00}}{=}-\frac{1}{2}g^{PQ}(g_{0PQ}-2g_{P0Q}),\label{gamma-000}\\
     \gamma{}_0{}^K{}_0 =& -\delta^{KL}(\betat \tau_{0L}+g_{00L}), \label{gamma-0K0}\\
     \gamma_0{}^0{}_J=&-\betat \tau_{0J}, \label{gamma-00J}
     \intertext{and}
      \gamma_0{}^K{}_J=& -\frac{1}{2}\delta^{KL}(g_{0JL}+g_{J0L}-g_{L0J}). \label{gamma-0KJ}
\end{align}
We further observe from \eqref{gamma-I00-IK0}, \eqref{p-fields} and \eqref{gi00} that the connection coefficients $\gamma_I{}^0{}_0$ and $\gamma_{I}{}^J{}_0$ can be expressed as
\begin{align}
    \gamma_I{}^0{}_0 &= g_{00I}-\frac{1}{2}\delta^{JK}(2 g_{JKI}-g_{IJK}) \label{gamma-I00}
    \intertext{and}
    \gamma_I{}^J{}_0 &= -\delta^{JK}g_{I0K} + \delta^{JK}\gamma_{I}{}^0{}_K, \label{gamma-IJ0}
\end{align}
respectively. 
We note here that the above expressions for the connection coefficients allows us to write the differentiated conformal Euler equations \eqref{eq:EulerDer2} as
\begin{align}
  &{a^i}_{j k}\Dc_i \Utt_Q^k\notag\\
  =&\Dc_Q(G_{jsl}W^sW^l)\notag\\
     &+\bigl(\Dc_Q{a^0}_{j m}+{a^0}_{j m}\betat \tau_{0Q} -{a^I}_{j m}\gamma_I{}^0{}_Q\bigr) (a^{0})^{-1}{}^{mn}\Bigl( a^I{}_{nk}\Utt^k_I-G_{nsl}W^sW^l\Bigr)\notag\\
     &-\bigl(\Dc_Q{a^P}_{j k}+\frac{1}{2}{a^0}_{j k}\delta^{PL}(g_{0QL}+g_{Q0L}-g_{L0Q}) -{a^I}_{j k}\gamma_I{}^P{}_Q\bigr) \Utt_P ^k.
     \label{eq:EulerDer3}
\end{align}

Next, applying $e_I$ to \eqref{gamma-0kj-A}, we note that
\begin{align}
    e_I(\gamma{}_0{}^k{}_j)
    =&-(\delta^0_j\eta^{kl}+\delta^l_j\delta^k_0)\Bigl( e_I(\betat)\tau_{0l}+ \betat e_I(\tau_{0l})+\frac{1}{2}e_I(g_{l00}) \Bigr)\notag\\
    &-\frac{1}{2}\eta^{kl}(e_I(g_{0jl})+e_I(g_{j0l})-e_I(g_{l0j})).
    \label{gamma-0kj-B}
\end{align}
With the help of \eqref{grad-alpha-A}, \eqref{gamma-000}-\eqref{gamma-0kj-B},   
it then follows from \eqref{gamma-Ikj-A} and
\eqref{[e0,eI]-C} that frame components $e^\Lambda_I$ and the connection coefficients
$\gamma_I{}^k{}_J$ satisfy the evolution
equations 
\begin{align}
\del{t}e_I^\Lambda=&-\betat \Bigl(\frac{1}{2}\delta^{JK}(g_{0IK}-g_{I0K}-g_{K0I}) + \delta^{JK}\gamma_{I}{}^0{}_K \Bigr)
e^\Lambda_J \label{for-M.1}
\end{align}
{and}
\begin{align}
\del{t}\gamma_I{}^k{}_J=
-\betat\Bigl(&\delta^k_0\Bigl(\betat e_I(\tau_{0J})+\frac{1}{2}e_I(g_{J00})\Bigr)\notag\\
&+\frac{1}{2}\eta^{kl}\bigl(e_I(g_{0Jl})+e_I(g_{J0l})-e_I(g_{l0J})\bigr)\Bigr)+L_I{}^k{}_J, \label{for-M.2}
\end{align}
respectively,
where
\begin{align}
    L_I{}^k{}_J = \betat\Bigl(\delta^k_0\Bigl(&\betat^2 \tau_{I0}+\frac{1}{2}\betat \bigl(2g_{00I}-\delta^{KL}(2g_{KLI}-g_{IKL})\bigr)\Bigr)\tau_{0J}\label{L-def}\\ &-\gamma_I{}^l{}_J\gamma_0{}^k{}_l+\gamma_0{}^l{}_J\gamma_I{}^k{}_l+ (\gamma_0{}^l{}_I-\gamma_I{}^l{}_0)\gamma_l{}^k{}_J\Bigr). \notag
\end{align}

\section{Nonlinear decompositions}
\label{sec:NonlinDecomp}

In this section, we state and prove a number of lemmas that characterize the structure of the
nonlinear terms appearing in the first order equations that were derived in the previous section. In order to state the lemmas, we first
make the following definitions:
\begin{align}
\kt&=(\kt_{IJ}) :=(g_{0IJ}-g_{I0J}-g_{J0I}),
\label{kt-def} \\
\ellt&=(\ellt_{IjK}) := (g_{IjK}), &&  \label{ellt-def}\\
\mt&=(\mt_{I})  := (g_{00M}), \label{mt-def}\\
\tau &= (\tau_{ij}), \label{tau-def}\\
\gt&=(\gt_{Ijkl}) := (g_{Ijkl}), \label{gt-def}\\    
  \taut&=(\taut_{Ijk}) := (\tau_{Ijk}), \label{taut-def}\\
  W&=(W^i), \label{W-def}\\
  \Ttt&=(\Ttt_{ij}) \label{Ttt-def}\\
\intertext{and}
\psit&=(\psit_I{}^k{}_J):=(\gamma_I{}^k{}_J).   \label{psit-def}
\end{align}
Even though we are using $\tau$ in \eqref{tau-def} to denote
the collection of derivatives $\tau_{ij}=\Dc_i\Dc_j\tau$, no ambiguities will arise due to the slicing condition \eqref{tau-fix} that allows us to use the coordinate time $t$ to denote the scalar field $\tau$.

We employ the $*$-notation from \cite[\S8]{BeyerOliynyk:2021} for multilinear maps. This notation is useful for analyzing multilinear maps where it is not necessary to keep track of the constant coefficients that define the maps. For example,
$\kt*\ellt$ denotes tensor fields of the form
\begin{align*}
[\kt*\ellt]_{ij}= C_{ij}^{KLMpQ}\kt_{KL}\ellt_{MpQ}
\end{align*}
where the coefficients $C_{ij}^{KLMpQ}$ are constants. We also use the notation
\begin{equation*}
    (1+\betat)\kt*\ellt= \kt*\ellt +\betat (\kt*\ellt)
\end{equation*}
where on the right hand side the two $\kt*\ellt$ terms correspond, in general, to distinct bilinear maps, e.g.
\begin{equation*}
    [(1+\betat)\kt*\ellt]_{ij}:=  C_{ij}^{KLMpQ}\kt_{KL}\ellt_{MpQ}+\betat  \tilde{C}_{ij}^{KLMpQ}\kt_{KL}\ellt_{MpQ}.
\end{equation*}
More generally, the $*$ will function as a product that distributes over sums of terms where the terms
$\lambda \kt*\ellt$, $\lambda \in \Rbb\setminus\{0\}$, and $\ellt *\kt$ are identified. For example,
\begin{gather*}
    \ellt*(\kt+\mt)= (\kt+\mt)*\ellt= \ellt*\kt+\ellt*\mt, \\
    (\mt+\psit)*(\ellt+\kt) = \mt*\ellt + \mt*\kt + \psit*\ellt + \psit*\kt 
\intertext{and}
    (\kt+\ellt)*(\kt+\ellt)= \kt*\kt + \kt*\ellt + \ellt*\ellt.
\end{gather*}

\subsection{Structure lemmas}
The proofs of the first four lemmas, Lemmas~\ref{Qij-lem}, \ref{Qijk-lem}, \ref{lem-cov-exp} and \ref{lem-JK-exp}, can be found in \cite[\S8]{BeyerOliynyk:2021}. 
The final lemma, Lemma~\ref{lem-P-exp}, is a slight extension of \cite[Lemma~8.4]{BeyerOliynyk:2021} that accounts for the additional terms that appear due to fluid coupling.

\begin{lem} \label{Qij-lem}
\begin{equation}
  \begin{split}
Q_{00}&=  -\frac{3}{2}\bigl(\delta^{RS}\kt_{RS}\bigr)^2 + \frac{1}{2} \delta^{PQ}\delta^{RS}\kt_{PR} \kt_{QS}+
\Qf, \quad
Q_{0M}= \Qf_M, \\
Q_{LM}&= \delta^{RS}\kt_{LR} \kt_{MS}+
\Qf_{LM}\label{Qij-lem.1}
  \end{split}
\end{equation}
where  $\Qf = (\kt+\ellt+\mt)*(\ellt+\mt)$.
\end{lem}

\begin{lem}\label{Qijk-lem}
\begin{equation}
  \begin{split}
Q^0_{0M}&= \Qft_M, \quad Q^I_{0M}=\Qft^I_M, \quad
Q^I_{LM}=  \Qft^I_{LM},\\
Q^0_{LM}&= -\frac{1}{2}\delta^{RS}\kt_{RS}\kt_{LM}+ \Qft^0_{LM} \label{Qijk-lem.1}
  \end{split}
\end{equation}
where $\Qft =\betat(\kt+\ellt+\mt)*\tau +(\kt+\ellt+\mt)*(\ellt+\mt)$.
\end{lem}

\begin{lem} \label{lem-cov-exp}
\begin{align}
     e_I(\tau_{jk})&=\tau_{Ijk} + \tf_{Ijk}, \label{lem-cov-exp.1}\\
     e_I(g_{jkl})&= g_{Ijkl}+\gf_{Ijkl},
     \label{lem-cov-exp.2}
     \intertext{and}
     L_I{}^k{}_J&=\betat\Lf_I{}^k{}_J\label{lem-cov-exp.3}
\end{align}
where
\begin{equation} \label{gf-def}
\tf= (\psit+\ellt+\mt)*\tau, \quad \gf= (\psit+\ellt+\mt)*(\kt+\ellt +\mt) 
\end{equation}
and
\begin{equation} \label{Lf-def}
     \Lf=  (\betat\tau+\psit+\kt+\ellt +\mt)*(\betat \tau+\psit+\ellt+\mt).
\end{equation}
\end{lem}

We note that \eqref{lem-cov-exp.1}-\eqref{lem-cov-exp.2} allows us to express the evolution equation
\eqref{for-M.2} for the background connection coefficients $\gamma_I{}^k{}_J$
as 
\begin{align}
    \del{t}\gamma_I{}^k{}_J=&
-\delta^k_0\Bigl(\betat^2 \tau_{I0J}+\frac{1}{2}\betat g_{IJ00}\Bigr)-\frac{1}{2}\eta^{kl}\bigl(\betat g_{I0Jl}+\betat g_{IJ0l}-\betat g_{Il0J}\bigr)\notag\\
&+\betat \Lf_I{}^k{}_J\label{for-N},
\end{align}
where $\Lf$ is of the form \eqref{Lf-def}.
Using \eqref{kt-def} and \eqref{psit-def}, we also note that the transport equations
\eqref{for-L} and \eqref{for-M.1} can be written as
\begin{align}
\del{t}\betat &= -\betat^3 \tau_{00}+\frac{1}{2} \delta^{JK}\betat^2\kt_{JK} \label{for-O.1}
\intertext{and}
\del{t}e^\Lambda_I &= - \Bigl(\frac{1}{2}\delta^{JL}\betat\kt_{IL}+\delta^{JK}\betat\psit_{I}{}^0{}_K\Bigr)e^\Lambda_J, \label{for-O.2}
\end{align}
respectively. For use below, we record that
$ e_I(\betat)$ can, using \eqref{grad-alpha-A} and
\eqref{ellt-def}-\eqref{mt-def}, be expressed as
\begin{equation} \label{grad-alpha-B}
    e_I(\betat)= -\betat^2 \tau_{I0}-\frac{1}{2}\betat\bigl( 2 \mt_I -\delta^{JK}(2\ellt_{JKI}-\ellt_{IJK})\bigr).
  \end{equation}

  \begin{lem}\label{lem-JK-exp}
\begin{equation*}
  \betat J^j_l=   \bigl[(\betat^2 \tau + \betat\kt +\betat\ellt +\betat \mt)*\tau\bigr]^j_l,\quad
  \betat K^i_{Ql} = \Kf^i_{Ql},
\end{equation*}
where
\begin{equation*}
  \Kf = (\betat\kt+\betat\ellt+\betat\mt)*\taut + \betat\gt*\tau + (\ellt+\mt)*(\betat\kt+\betat\ellt+\betat\mt)*\tau.
\end{equation*}
\end{lem}

\begin{lem}\label{lem-P-exp}
\begin{align}
   \betat P_{Qlm} =& \Pf_{Qlm}{+\left[\betat (\ellt+\mt)*\Ttt\right]_{Qlm}},\label{lem-P-exp.1} \\
   \delta^q_Q \betat \Dc_k g_{qjlm} =&  \delta_k^0\del{t}g_{Qjlm} +\delta_k^K\betat e_K(g_{Qjlm})+ \Gf_{Qkjlm}\notag\\
    &+\left[\betat(\betat \tau +\psit) *\Ttt\right]_{Qkjlm}, \label{lem-P-exp.2}
   \intertext{and}
   \delta^q_Q \betat \Dc_k \tau_{qjl}=&
                                       \delta_k^0\del{t}\tau_{Qjl} +\delta_k^K\betat e_K(\tau_{Qjl})+ \Tf_{Qkjl}\label{lem-P-exp.3} 
\end{align}
where
\begin{align}
\Pf =& \frac{1}{t}(\betat\tau + \psit + \ellt +\mt)*(\kt+\ellt+\mt) \notag\\
&+ \betat(\kt+\ellt+\mt) *\bigl( \gt+
      (\ellt+\mt)*(\kt+\ellt+\mt)\bigl), \label{Pf-def}\\
    \Gf=&(\betat \tau +\psit) *\Bigl( \frac{1}{t}\bigl(\betat \tau+\kt+\ellt+\mt)  
     +\betat(\kt+\ellt+\mt)*(\kt+\ellt+\mt)\Bigr) \notag\\
     &\qquad
     +\betat \bigl(\betat \tau+\psit+\kt+\ellt+\mt)*\gt  
    \label{Gf-def}
    \intertext{and}
    \Tf=&\betat (\betat \tau+\psit+ \kt+ \ellt +\mt)*\taut\notag\\
         &+\betat (\betat \tau+\psit +\ellt +\mt)*(\kt + \ellt +\mt)*\tau.\label{Tf-def}
\end{align}
\end{lem}
\begin{proof}
Differentiating \eqref{Q-def}, we find, by \eqref{orthog} and \eqref{p-fields}-\eqref{d-fields}, that
\begin{equation} \label{Dq-Qij}
    \Dc_q Q_{ij} =\Qtt^1_{qij}+ \Qtt^2_{qij}
\end{equation}
where 
\begin{align*}
\Qtt^{1}_{qij} = \frac{1}{2}\eta^{kl}\eta^{mn}\Bigl(&g_{qimk} g_{jn l}+g_{imk}g_{qjnl}
+2 g_{qnil}g_{kjm} +2g_{nil}g_{qkjm}\\ 
&- 2g_{qlin} g_{kjm}- 2g_{lin}g_{qkjm}
-2g_{qlin}g_{jmk}-2g_{lin}g_{qjmk}\\
& -2 g_{qimk}g_{ljn}-2g_{imk}g_{qljn}\Bigr)
\end{align*}
and
\begin{align*}
    \Qtt^2_{qij} =  -\frac{1}{2}(\eta^{kr}\eta^{ls}g^{mn}+& \eta^{kl} \eta^{mr}\eta^{ns}\bigr)g_{qrs}\Bigl(g_{imk}g_{jn l}
+2 g_{nil}g_{kjm} \\
&- 2g_{lin}g_{kjm}
-2g_{lin}g_{jmk} -2g_{imk}g_{ljn}\Bigr).
\end{align*}
Setting $q=Q$ in \eqref{Dq-Qij} gives
\begin{equation} \label{Dq-exp}
    \delta^q_Q\Dc_q Q_{ij} =\bigl[(\kt+\ellt +\mt)*\bigl( \gt+(\ellt+\mt)*(\kt+\ellt +\mt) \bigr)\bigr]_{Qij},
\end{equation}
where in deriving this we have employed  \eqref{gi00} and \eqref{kt-def}-\eqref{gt-def}.

To proceed, we let
\begin{equation*}
    \Ptt_{Qlm} =
 \frac{1}{t}\bigl(e_Q(\betat^{-1}) \delta_0^j+\betat^{-1}  \gamma_Q{}^j{}_0\bigr)  (g_{ljm}+g_{mjl}-
g_{jlm}) -\eta^{kr}\eta^{js}g_{Qrs}g_{kjlm},
\end{equation*}
which we note, using \eqref{gi00} and \eqref{ellt-def}-\eqref{mt-def}, can be expressed as
\begin{align}
    \Ptt_{Qlm} =&
 \frac{1}{t}\bigl(e_Q(\betat^{-1})+\betat^{-1}  \gamma_Q{}^0{}_0\bigr)  (g_{l0m}+g_{m0l}-
g_{0lm})\notag\\
 &+\frac{1}{t}\betat^{-1}  \gamma_Q{}^J{}_0 (g_{lJm}+g_{mJl}-
g_{Jlm})\notag \\
& -\bigl(2\mt_Q-\delta^{RS}(2\ellt_{RSQ}-\ellt_{QRS})\bigr)g_{00lm}+\eta^{JS} \ellt_{Q0S}g_{0Jlm}\notag\\
&- \delta^{KR}\eta^{js} \ellt_{QsR}g_{Kjlm}.  \label{Ptt-def}
\end{align}
Now, by \eqref{gi00}, \eqref{gamma-I00}-\eqref{gamma-IJ0},  \eqref{kt-def}-\eqref{mt-def}, \eqref{psit-def} and  
\eqref{grad-alpha-B}, we have
\begin{align}
&\frac{1}{t}\bigl(e_Q(\betat^{-1})+\betat^{-1}  \gamma_Q{}^0{}_0\bigr)  (g_{l0m}+g_{m0l}-
g_{0lm})\notag\\
&+\frac{1}{t}\betat^{-1}  \gamma_Q{}^J{}_0 (g_{lJm}+g_{mJl}-
g_{Jlm})\notag \\
=& \frac{1}{t}\betat^{-1}[(\betat \tau+\psit+\ellt+\mt)*(\kt+\ellt+\mt)]_{Qlm}, \label{Ptt-exp-A}
\end{align}
while it is clear from \eqref{for-K.1} and \eqref{Qij-lem.1} that
\begin{align} 
  g_{0rlm} =& \frac{1}{t}\betat^{-1}\delta_r^0(g_{l0m}+g_{m0l}-g_{0lm})
               +\frac{2}{t}\delta_r^0\tau_{lm} - \delta_{ri}B^{ijK}g_{Kjlm}
                           \notag   \\
& +[(\kt+\ellt+\mt)*(\kt+\ellt+\mt)]_{rlm}{+2 \delta^0_r\Ttt_{lm}}. \label{g0rlm-exp}
\end{align}
By \eqref{Dq-exp}-\eqref{g0rlm-exp}, we then, with the help of the definitions \eqref{kt-def}-\eqref{mt-def}, deduce from \eqref{P-def} the validity of the expansion \eqref{lem-P-exp.1}.

Next, by the definition of the covariant derivative $\Dc_q$, we have
\begin{align*}
    \delta^q_Q \betat \Dc_k g_{qjlm} =& \betat e_k(g_{Qjlm})-\betat \gamma_k{}^p{}_{Q} g_{pjlm}
    -\betat \gamma_k{}^p{}_{j} g_{Qplm}\\
     &-\betat \gamma_k{}^p{}_{l} g_{Qjpm}
      -\betat \gamma_k{}^p{}_{m} g_{Qjlp}\\
     =& \delta_k^0\del{t}g_{Qjlm} +\delta_k^K\betat e_K(g_{Qjlm}) -\betat \gamma_k{}^0{}_{Q} g_{0jlm}\\
    &\!\!\!\!-\Bigl(\betat \gamma_k{}^P{}_{Q} g_{Pjlm}
    +\betat \gamma_k{}^p{}_{j} g_{Qplm}
     +\betat \gamma_k{}^p{}_{l} g_{Qjpm}
      +\betat \gamma_k{}^p{}_{m} g_{Qjlp}\Bigr),
\end{align*}
where in deriving the second equality we used \eqref{e0-mu}. The expansion \eqref{lem-P-exp.2} is then a straightforward consequence of \eqref{gamma-000}-\eqref{gamma-0KJ}, \eqref{g0rlm-exp} and the definitions \eqref{kt-def}-\eqref{psit-def}. Furthermore, by employing similar arguments, it is not difficult, with the help of \eqref{for-K.2}, to verify that the expansion \eqref{lem-P-exp.3} also holds. 
\end{proof}

\section{Fuchsian form}
\label{sec:FuchsianForm}
We collect together the following evolution equations from Section \ref{sec:FirstOrderForm}: 
\begin{align}
  \del{t}g_{00M} =& \frac{1}{t}\bigl(2g_{00M}-\delta^{IJ}(2g_{IJM}-g_{MIJ})\bigr)+\frac{2}{t}\betat\tau_{0M}  \notag\\
  &- B^{0jK}\betat g_{Kj0M}+\betat Q^0_{0M}+2 \betat{\Ttt}_{0M},\label{for-G.1.S}\\
\del{t}g_{R0M} =& -\delta_{RI} B^{IjK}\betat g_{Kj0M}+\delta_{RI}\betat Q^I_{0M}, \label{for-G.2.S}\\
  \del{t}(g_{0LM}-g_{L0M}-g_{M0L}) =&  -\frac{1}{t}(g_{0LM}-g_{L0M}-g_{M0L}) +\frac{2}{t}\betat\tau_{(LM)}\notag\\
  &+S_{(LM)}{+2 \betat{\Ttt}_{(LM)}}, \label{for-H.S}\\
\del{t}g_{RLM} =& -\delta_{RI} B^{IjK}\betat g_{KjLM}+
                 \delta_{RI}\betat Q^I_{LM}, \label{for-G.4.S}\\
\del{t}\tau_{rl} =&
 \delta_{ri}\betat J^i_l-B^{ijK}\betat\tau_{Kj(l}\delta_{r)i}
                   , \label{for-F.2.S}\\
\label{for-I.S}
B^{ijk}\betat\Dc_k g_{qjlm} =& \frac{1}{t}
                               \delta_0^i\delta_0^j  (g_{qljm}+
                               g_{qmjl}-g_{qjlm})
                               \notag\\                 
      &+\frac{2}{t}\delta^i_0\betat\tau_{q(lm)}+ \delta^i_0 \betat P_{q(lm)}{+2 \delta^i_0 \betat\Dc_q{\Ttt}_{lm}},\\
  B^{ijk}\betat\Dc_k\tau_{qjl} =&  \betat K^i_{ql}, \label{for-J.S}\\   
\label{for-L.S}
\del{t}\betat 
  =&-\betat^{3}\tau_{00}+\frac{1}{2}\betat^2 \delta^{JK}(g_{0JK}-2g_{J0K}),\\
  % \label{for-Lhat.S}
  % \del{t}\hat\beta =& \Bigl(-\frac{1+c_s^2}{2} \betat\delta^{JK}(g_{0JK}-2g_{J0K})+(1+c_s^2) \betat^2 \tau_{00}\Bigr) \hat\beta\\
\del{t}e_I^\Lambda=&-\betat \Bigl(\frac{1}{2}\delta^{JK}(g_{0IK}-g_{I0K}-g_{K0I})\notag\\ 
&\qquad+ \delta^{JK}\gamma_{I}{}^0{}_K \Bigr)
e^\Lambda_J, \label{for-M.1.S}\\
\del{t}\gamma_I{}^k{}_J=&
                          -\betat\Bigl(\delta^k_0\Bigl(\betat e_I(\tau_{0J})+\frac{1}{2}e_I(g_{J00})\Bigr)\notag \\
                          &+\frac{1}{2}\eta^{kl}\bigl(e_I(g_{0Jl})+e_I(g_{J0l})-e_I(g_{l0J})\bigr)\Bigr) \notag\\
                  &+L_I{}^k{}_J, \label{for-M.2.S}\\
  {a^i}_{j k}\betat\Dc_i W^k =&\betat G_{jsl}W^sW^l, \label{for-Euler.1}\\
  {a^i}_{j k}\betat\Dc_i {\Utt}^k_Q=&\betat \Dc_Q(G_{jsl}W^sW^l)\notag\\
  &+\betat\bigl(\Dc_Q{a^0}_{j m}+{a^0}_{j m}\betat \tau_{0Q} \notag \\
  &\!\quad-{a^I}_{j m}\gamma_I{}^0{}_Q\bigr) (a^{0})^{-1}{}^{mn}\Bigl( a^I{}_{nk}\Utt^k_I-G_{nsl}W^sW^l\Bigr)\notag\\
  &-\betat\Bigl(\Dc_Q{a^P}_{j k}
  +\frac{1}{2}{a^0}_{j k}\delta^{PL}(g_{0QL}+g_{Q0L}-g_{L0Q})\notag\\
  &\!\quad-{a^I}_{j k}\gamma_I{}^P{}_Q\Bigr) \Utt_P ^k,   \label{for-Euler.2}
\end{align}
where
\begin{align}
B^{ijk} &= -\delta^i_0 \eta^{jk}-\delta^j_0\eta^{ik}+ \eta^{ij}\delta^k_0,\label{B-def.S}
\end{align}
$Q^i_{lm}$ is determined by \eqref{Q-def} and \eqref{Qilm-def}, $S_{LM}$ and $J^i_l$ are given
by \eqref{S-def} and \eqref{J-def}, respectively, $g_{000}$ and $g_{I00}$ are determined by \eqref{gi00}, $\tau_i=-\betat^{-1}g_{0i}$ according to \eqref{e0-def}, $P_{qlm}$ is determined by \eqref{P-def} and \eqref{grad-alpha-A}, $K^i_{ql}$ is given by \eqref{K-def}, $g_{0rlm}$ and $\tau_{0rl}$ are determined by \eqref{for-K.1}-\eqref{for-K.2}, $L_I{}^k{}_J$ is given by \eqref{L-def}, the background connection coefficients  $\gamma{}_0{}^0{}_0$, $\gamma{}_0{}^K{}_0$, $\gamma_0{}^0{}_J$, $\gamma_0{}^K{}_J$, $\gamma_I{}^0{}_0$ and $\gamma_I{}^J{}_0$  are determined by \eqref{gamma-000}-\eqref{gamma-IJ0}, $a^i{}_{jk}$ is defined by \eqref{confESFHa.2}, $ (a^{0})^{-1}{}^{mp}$ is determined by \eqref{eq:defa0inv}, $G_{jsl}$ is defined by \eqref{confESFHa.3}, $\Ttt_{ij}$ in given by \eqref{eq:TttconfphysW}, and $\Dc_Q(G_{jsl}W^sW^l)$ and $D_Q{a^i}_{j k}$ are the maps that are to be understood as per the discussion below \eqref{eq:DcGWW}. 

Together, the equations \eqref{for-G.1.S}-\eqref{for-Euler.2} 
comprise our first order formulation
of the frame representation of the reduced conformal Einstein-Euler-scalar field equations. In this section, we transform these equations into a Fuchsian form 
that is suitable for establishing the existence of solutions all the way to the singularity.  

 \begin{rem}
\label{rem:symmetry}
 We will not always assume that solutions of the system \eqref{for-G.1.S}-\eqref{for-Euler.2} arise from solutions of the conformal Einstein-Euler-scalar field equations. We will also consider general solutions of this system. For these general solutions, we will require that $g_{0LM}-g_{L0M}-g_{M0L}$
 and $g_{qjlm}$ are symmetric in $L,M$ and $l,m$, respectively. This condition would naturally hold for solutions derived from the conformal Einstein-Euler-scalar field equations. To ensure this symmetry holds for general solutions, we have symmetrized equations \eqref{for-H.S} and \eqref{for-I.S} in the indices $L,M$ and $l,m$, respectively. This symmetrization guarantees that the solutions to \eqref{for-G.1.S}-\eqref{for-Euler.2} will exhibit the desired symmetry provided the initial data also satisfies this condition.
 \end{rem}

 \subsection{Initial data\label{frame-idata}}
 Before continuing on with transforming the system \eqref{for-G.1.S}-\eqref{for-Euler.2} into Fuchsian form, we will, in this section, describe how initial data $\{\gr_{\mu\nu},\ggr_{\mu\nu},\taur=t_0,\taugr,\Vr^\mu\}$ for the reduced conformal Einstein-Euler-scalar field equations, see \eqref{gt-idata}-\eqref{Vh-idata}, along with a choice of spatial frame initial data
\begin{equation} \label{frame-idata-A}
    e_I^\mu\bigl|_{\Sigma_{t_0}} = \er^\mu_I
\end{equation}
determines initial data for the first order system \eqref{for-G.1.S}-\eqref{for-Euler.2}. In the following, we will not, a priori, assume that the initial data $\{\gr_{\mu\nu},\ggr_{\mu\nu},\taur=t_0,\taugr,\Vr^\mu\}$
satisfies either the gravitational or wave gauge constraints, but it will, or course, be necessary to do so in order to formulate and prove the main past stability result.

We  know from Remark \ref{rem-Lag-idata} that the initial data set $\{\gr_{\mu\nu},\ggr_{\mu\nu},\taur=t_0,\taugr,\Vr^\mu\}$ determines a corresponding Lagrangian initial data set
\begin{equation} \label{Lag-idata-set-A}
\bigl\{g_{\mu\nu}\bigl|_{\Sigma_{t_0}},\del{0}g_{\mu\nu}\bigl|_{\Sigma_{t_0}}, \tau\bigl|_{\Sigma_{t_0}}=t_0,\del{0}\tau\bigl|_{\Sigma_{t_0}}=1, V^{\mu}\bigl|_{\Sigma_{t_0}}\bigr\}
\end{equation}
via \eqref{l-idata}--\eqref{chi-idata}. This initial data, with the help of the reduced conformal Einstein-Euler-scalar field equations \eqref{lag-redeqns}  and first derivatives thereof evaluated on $\Sigma_{t_0}$, uniquely determines the following higher order partial derivatives on the initial hypersurface:
\begin{equation} \label{Lag-idata-set-B}
  \begin{split}
\bigl\{&\del{\beta}g_{\mu\nu}\bigl|_{\Sigma_{t_0}},\del{\alpha}\del{\beta}g_{\mu\nu}\bigl|_{\Sigma_{t_0}},\del{\alpha}\tau\bigl|_{\Sigma_{t_0}}=\delta^0_\alpha,\del{\alpha}\del{\beta}\tau\bigl|_{\Sigma_{t_0}},\\
 &\qquad\del{\alpha}\del{\beta}\del{\gamma}\tau\bigl|_{\Sigma_{t_0}}, \del{\alpha}V^{\mu}\bigl|_{\Sigma_{t_0}}\bigr\}.
  \end{split}
\end{equation}
We further observe that the flat metric $\gc_{\mu\nu}=\del{\mu}l^\alpha \eta_{\alpha\beta}\del{\nu} l^\beta$ and its first and second order partial derivatives on the initial hypersurface, that is,
\begin{equation} \label{Lag-idata-set-C}
\bigl\{\gc_{\mu\nu}\bigl|_{\Sigma_{t_0}},\del{\alpha}\gc_{\mu\nu}\bigl|_{\Sigma_{t_0}},\del{\alpha}\del{\beta}\gc_{\mu\nu}\bigl|_{\Sigma_{t_0}} \bigr\},
\end{equation}
are uniquely determined from the initial data set $\{\gr_{\mu\nu},\ggr_{\mu\nu},\taur=t_0,\taugr\}$ by
\eqref{gct-def}, \eqref{Jc-def}, \eqref{vr-def}-\eqref{Jcr-def} and the
reduced conformal Einstein-scalar field equations, especially \eqref{J-ev.1}, and first derivatives thereof.
Taken together, \eqref{Lag-idata-set-A}-\eqref{Lag-idata-set-C} determine the following geometric fields on $\Sigma_{t_0}$:
\begin{equation} \label{Lag-idata-set-D}
  \begin{split}
\bigl\{&g_{\mu\nu}\bigl|_{\Sigma_{t_0}},\Dc_{\beta}g_{\mu\nu}\bigl|_{\Sigma_{t_0}},\Dc_{\alpha}\Dc_{\beta}g_{\mu\nu}\bigl|_{\Sigma_{t_0}},\tau\bigl|_{\Sigma_{t_0}}=t_0,\Dc_{\alpha}\tau\bigl|_{\Sigma_{t_0}}=\delta^0_\alpha,\\
&\quad\Dc_{\alpha}\Dc_{\beta}\tau\bigl|_{\Sigma_{t_0}},\Dc_{\alpha}\Dc_{\beta}\Dc_{\gamma}\tau\bigl|_{\Sigma_{t_0}}, V^{\mu}\bigl|_{\Sigma_{t_0}}, \Dc_{\alpha}V^{\mu}\bigl|_{\Sigma_{t_0}} \bigr\}.
  \end{split}
\end{equation}

Next, by \eqref{chi-idata} and \eqref{e0-def}, the coordinate components of the frame vector $e_0$ on the initial hypersurface are given by 
\begin{equation} \label{frame-idata-E}
e_0^\mu \bigl|_{\Sigma_{t_0}} = (-g_{00}|_{\Sigma_{t_0}} )^{-\frac{1}{2}}\delta^\mu_0,   
\end{equation}
while by \eqref{e0-def} we have 
\begin{equation} \label{frame-idata-F}
\betat\bigl|_{\Sigma_{t_0}} = (- g^{00}\bigl|_{\Sigma_{t_0}})^{-\frac{1}{2}}.
\end{equation}
Using \eqref{frame-idata-E}, we can employ a Gram-Schmidt algorithm to select (non-unique) spatial frame initial data \eqref{frame-idata-A} so that that the frame metric
$g_{ij} = e_i^\mu g_{\mu\nu} e^\nu_j$
satisfies
\begin{equation*}% \label{frame-idata-G}
    g_{ij}\bigl|_{\Sigma_{t_0}}=\eta_{ij}
\end{equation*}
on $\Sigma_{t_0}$. We note also that the initial data $\omega^j_\nu \bigl|_{\Sigma_{t_0}}$ for the coordinate components of the dual frame $\omega^j$ is determined by
\begin{equation*}
  \omega^j_\nu \bigl|_{\Sigma_{t_0}} e_i^\nu \bigl|_{\Sigma_{t_0}}=\delta^j_i.
\end{equation*}
Then, with the help of the relations
\begin{gather*}
\Dc_i g_{jk} = e_i^\beta e_j^\mu e_k^\nu \Dc_{\beta}g_{\mu\nu}, \quad \Dc_i\Dc_j g_{kl} = e_i^\alpha e_j^\beta e_k^\mu e_l^\nu \Dc_{\alpha}\Dc_{\beta}g_{\mu\nu}, \\
\Dc_i \Dc_{j}\tau = e_i^\beta e_j^\beta \Dc_{\beta}\Dc_{\beta}g_{\mu\nu}, \quad \Dc_i\Dc_j \Dc_k \tau = e_i^\alpha e_j^\beta e_k^\gamma \Dc_{\alpha}\Dc_{\beta}\Dc_{\gamma}\tau,\\
V^i=\omega^i_\mu V^\mu,\quad \Dc_i V^j=\omega^j_\beta e_i^\alpha\Dc_\alpha V^\beta,
\end{gather*}
it follows from the definitions  \eqref{p-fields}-\eqref{eq:Eulerfchoice} that \eqref{frame-idata-A} and \eqref{Lag-idata-set-D} determine the fields
$g_{ijk}$, $g_{ijkl}$, $\tau_{ij}$, $\tau_{ijk}$, $W^i$ and $\Utt^s_q$ on
the initial hypersurface, that is,
\begin{equation}\label{frame-idata-H}
    \bigl\{g_{ijk}\bigl|_{\Sigma_{t_0}},g_{ijkl}\bigl|_{\Sigma_{t_0}},\tau_{ij}\bigl|_{\Sigma_{t_0}},\tau_{ijk}\bigl|_{\Sigma_{t_0}}, W^i \bigl|_{\Sigma_{t_0}}, \Utt^s_q \bigl|_{\Sigma_{t_0}}\bigr\}.
\end{equation}

Using the fact that the frame $e_j^\mu$ is orthonormal with respect to the metric $g_{\mu\nu}$, it follows from a straightforward calculation that 
$\Gamma_{I}{}^0{}_J|_{\Sigma|_{t_0}} = \Ktt_{\Lambda\Omega}\er_I^\Lambda \er_J^\Omega$,
where $\Ktt=\Ktt_{\Lambda\Omega}dx^\Lambda\otimes dx^\Omega$ is the second fundamental form, c.f.\ \eqref{gtt-def1}, determined from the initial data $\{g_{\mu\nu}|_{\Sigma_{t_0}},\del{0}g_{\mu\nu}|_{\Sigma_{t_0}}\}$. From this expression and \eqref{Ccdef}, we deduce that
\begin{equation} \label{frame-idata-J}
    \gamma_I{}^0{}_J\bigl|_{\Sigma_{t_0}} = \Ktt_{\Lambda\Omega}\er_I^\Lambda \er_J^\Omega +\frac{1}{2}(g_{IJ0}+g_{JI0}-g_{0IJ})\bigl|_{\Sigma_{t_0}}.
\end{equation}
We also note that the connection coefficients 
$\Gamma_I{}^K{}_J\bigl|_{\Sigma_{t_0}}$ on the initial hypersurface are determined completely by the spatial orthonormal frame $\ett_I = \er^\Lambda_I \del{\Lambda}$
and the spatial metric
$\gtt=\gr_{\Lambda\Omega}dx^\Lambda \otimes dx^\Omega$, c.f.\ \eqref{gtt-def1}, on
$\Sigma_{t_0}$ according to
$\Gamma_I{}^K{}_J|_{\Sigma_{t_0}} = \delta^{KL}\gtt(\Dtt_{\ett_I}\ett_J,\ett_L)$,
where $\Dtt$ is the Levi-Civita connection of $\gtt$.
We can then use this to determine 
$\gamma_I{}^K{}_J\bigl|_{\Sigma_{t_0}}$ via
\eqref{Ccdef} to get
\begin{equation}  \label{frame-idata-K}
    \gamma_I{}^K{}_J \bigl|_{\Sigma_{t_0}} = \delta^{KL}\gtt(\Dtt_{\ett_I}\ett_J,\ett_L)
    - \frac{1}{2}\delta^{KL}(g_{IJL}+g_{JIL}-g_{LIJ})\bigl|_{\Sigma_{t_0}}.
\end{equation}
Together, \eqref{frame-idata-A}, \eqref{frame-idata-F},  \eqref{frame-idata-H},  \eqref{frame-idata-J} and \eqref{frame-idata-K} determine initial data on $\Sigma_{t_0}$
for the system \eqref{for-G.1.S}-\eqref{for-Euler.2}.

\subsection{Change of variables}
\label{sec:chvar1}

Using the variable definitions \eqref{kt-def}-\eqref{psit-def} along with \eqref{S-def} and \eqref{gamma-000}-\eqref{gamma-IJ0}, we can express the first order system \eqref{for-G.1.S}-\eqref{for-Euler.2} as
\begin{align}
  \del{t}\mt_{M} =& \frac{1}{t}\bigl(2 \mt_{M}-\delta^{IJ}(2 \ellt_{IJM}-\ellt_{MIJ})\bigr)\notag \\
  &+\frac{2}{t}\betat\tau_{0M} - B^{0jK}\betat \gt_{Kj0M}
                    +\betat \Qft_M
  +2 \betat{\Ttt}_{0M},  \label{for-G.1.S2}\\
\del{t}\ellt_{R0M} =& -\delta_{RI} B^{IjK}\betat\gt_{Kj0M}+\delta_{RI}\betat \Qft^I_M, \label{for-G.2.S2}\\
\del{t}\kt_{LM} =&  -\frac{1}{t}\kt_{LM} +\frac{2}{t}\betat\tau_{(LM)}- B^{0jK}\betat\gt_{Kj(LM)}
-\frac{1}{2}\delta^{RS}\betat\kt_{RS}\kt_{LM}\notag \\
&
%+B^{IjK}\betat\gt_{Kj0(L}\delta_{M)I}
+ \betat\Qft^0_{(LM)} + 2B^{IjK}\betat g_{Kj0(M}\delta_{L)I} -2\betat \Qft^I_{(L} \delta_{M)I}
%-\betat \Qft^I_{M}\delta_{LI}, 
 {+2 \betat{\Ttt}_{(LM)}}, \label{for-H.S2}\\
\del{t}\ellt_{RLM} =& -\delta_{RI} B^{IjK}\betat\gt_{KjLM}+
                 \delta_{RI}\betat \Qft^I_{LM}, \label{for-G.4.S2}\\
\del{t}\tau_{rl} =&
 \delta_{ri}\betat J^i_l-B^{ijK}\betat\taut_{Kj(l}\delta_{r)i}
                   , \label{for-F.2.S2}
\end{align}
\begin{align}
\delta^{ij}\del{t}\gt_{Qjlm} +B^{ijK}\betat e_K^\Lambda\partial_\Lambda\gt_{Qjlm}&=\frac{1}{t} \delta_0^i\delta_0^j  (\gt_{Qljm}+ \gt_{Qmjl}-\gt_{Qjlm})\notag \\
&+\frac{2}{t}\delta^i_0
 \betat\taut_{Q(lm)}+ \delta^i_0 \Pf_{Q(lm)}-B^{ijk}\Gf_{Qkj(lm)}\notag \\
& +2\delta_0^i\betat\Dc_Q\Ttt_{lm} \notag -B^{ijk}\left[\betat(\betat \tau +\psit) * \Ttt\right]_{Qkjlm}\notag \\
 &+\delta^i_0\left[\betat (\ellt+\mt)*
  \Ttt \right]_{Qlm},\label{for-I.S2}\\
  \delta^{ij}\del{t}\taut_{Qjl} +B^{ijK}\betat e_K^\Lambda\partial_\Lambda \taut_{Qjl}=&  \betat K^i_{Ql}-B^{ijk}\Tf_{Qkjl} \label{for-J.S2}
\end{align}
\begin{align}
  \del{t}\betat =& -\betat^3 \tau_{00}+\frac{1}{2} \delta^{JK}\betat^2\kt_{JK}, \label{for-O.1.S2}\\
  % \del{t}\hat\beta =& \Bigl(-\frac{1+c_s^2}{2} \betat\delta^{JK}\kt_{JK}+(1+c_s^2) \betat^2 \tau_{00}\Bigr) \hat\beta \label{for-Ohat.1.S2}\\
\del{t}e^\Lambda_I =& - \Bigl(\frac{1}{2}\delta^{JL}\betat\kt_{IL}+\delta^{JK}\betat\psit_{I}{}^0{}_K\Bigr)e^\Lambda_J,\label{for-M.1.S2}\\
\del{t}\psit_I{}^k{}_J=&
                         -\delta^k_0\Bigl(\betat^2 \taut_{I0J}+\frac{1}{2}\betat \gt_{IJ00}\Bigr)\notag \\
                         &-\frac{1}{2}\eta^{kl}\bigl(\betat \gt_{I0Jl}+\betat \gt_{IJ0l}-\betat \gt_{Il0J}\bigr)
                  +\betat \Lf_I{}^k{}_J,\label{for-N.S2}
\end{align}
\begin{align}
  {a^0}_{j k}\del{t} W^k+{a^I}_{j k}\betat e_I^\Lambda\del{\Lambda} W^k =&\betat G_{jsl}W^sW^l- \betat {a^i}_{j k}\gamma_i{}^k{}_lW^l, \label{for-Euler.1.2}\\
  {a^0}_{j k}\del{t} {\Utt}^k_Q
  +{a^I}_{j k}\betat e_I^\Lambda\del{\Lambda} {\Utt}^k_Q
  =&\Dc_Q(\betat G_{jsl}W^sW^l)-\Dc_Q\betat G_{jsl}W^sW^l \notag \\
  &+\betat (a^{0})^{-1}{}^{mn}\bigl(\Dc_Q{a^0}_{j m}+2{a^0}_{j m}\betat \tau_{0Q}
                   \notag \\
                   &-2{a^I}_{j m}\psit_I{}^0{}_Q\bigr) \Bigl( a^I{}_{nk}\Utt^k_I-G_{nsl}W^sW^l\Bigr)\notag\\
                  &-\betat\bigl(\Dc_Q{a^P}_{j k}
                    +\frac{1}{2}{a^0}_{j k}\delta^{PL}\kt_{QL}
                    \notag \\
                    &+{a^0}_{j k}\delta^{PL}\ellt_{Q0L}
                    -{a^I}_{j k}\psit_I{}^P{}_Q\bigr) \Utt_P ^k \notag \\
                    &-\betat{a^i}_{j k}\gamma_i{}^k{}_l {\Utt}^l_Q+\betat{a^i}_{j k}\gamma_i{}^L{}_Q {\Utt}^k_L.\label{for-Euler.2.2}
\end{align}
Here, $\Lf_I{}^k{}_J$ is given by \eqref{lem-cov-exp.3}, $\Pf_{Qlm}$, $\Gf_{Qkjlm}$, and $\Tf_{Qkjl}$ are given by \eqref{lem-P-exp.1}-\eqref{lem-P-exp.3}, and $\Qft_M$, $\Qft^I_M$, $\Qft^0_{LM}$ and $\Qft^I_{LM}$ are determined by \eqref{Qijk-lem.1}, and we note by \eqref{B-def.S} that
\begin{equation*}
  B^{ij0}=\delta^{ij}, \quad
B^{ijK} = -\delta^i_0 \eta^{jK}-\delta^j_0\eta^{iK}.
\end{equation*}
In deriving these equations, we have employed
\eqref{confESFHa.2} , \eqref{confESFHa.3} and \eqref{eq:defa0inv},  and
have repeatedly used \eqref{eq:Um0} to express ${\Utt}^k_0$ in terms of the other variables and thereby remove it from the system of equations.

Expressing $G_{jsl}$ in \eqref{for-Euler.1.2} in terms of the variables in \eqref{kt-def}-\eqref{psit-def}, we find, using \eqref{gi00},  \eqref{confESFHa.3} and \eqref{eq:WQdr} while retaining some of the original variable definitions in order to prevent the expression from becoming too lengthy, that
\begin{align}
G_{jsl}
  =&\betat^{-1}t^{-1}
       \frac{(n-1)c_s^2-1}{n-2} \Bigl(       
       \delta_{(l}^0  \eta_{s)j}-\delta_j^0\eta_{sl}
       \Bigr)\notag\\
    &+\frac 12\Bigl(\frac{c_s^4+2c_s^2+1}{c_s^2} \delta_{(l}^0\eta_{s)j}-c_s^2\delta_j^0\eta_{sl}+2 \delta_s^{0}\delta_j^{0}\delta_{l}^{0}\Bigr)
        \delta^{JK}\kt_{JK} \notag \\
        &- \delta_{(l}^{0}\delta_{s)}^{S}\delta_j^{J}\kt_{SJ}
        +\frac{1}{2}     
        \frac{3 c_s^2+1}{c_s^2} \frac{W^SW_S}{w^{2}}\eta_{j(s}\delta_{l)}^{0} \delta^{JK}\kt_{JK}   \notag\\
        &- \frac{1}{2}     
      \frac{3 c_s^2+1}{c_s^2} \frac{W^Q W^P}{w^{2}} \kt_{P Q}\delta_{(l}^0\eta_{s)j}       - \delta_{(l}^{0}\delta_{s)}^{S}\delta_j^{J}(\ellt_{S0J}+\ellt_{J0S})
  \notag \\
  &- \Bigl(\frac{3 c_s^2+1}{c_s^2} \frac{W^Q W^P}{w^{2}}+g^{PQ}\Bigr) \ellt_{P0Q}\delta_{(l}^0\eta_{s)j}- \frac{1}{2}     
        \Bigl(\frac{3 c_s^2+1}{c_s^2} \frac{W^Q W^P}{w^{2}} \notag \\
        &+g^{PQ}\Bigr) \ellt_{LP Q}\delta_{(l}^L\eta_{s)j}- 2\delta_{l}^{0}\delta_{(s}^{0}\delta_{j)}^{J}\mt_{J}
        - \delta_{(l}^{L}\delta_{s)}^{s'}\delta_j^{j'}g_{Ls'j'} \notag \\
    &- \frac{1}{2}     
        \Bigl(\frac{3 c_s^2+1}{c_s^2} \frac{(W^0)^2}{w^{2}}-1\Bigr) g_{L 0 0}\delta_{(l}^{L}\eta_{s)j}- \frac{3 c_s^2+1}{c_s^2} \frac{W^0 W^P}{w^{2}} \eta_{j(s} g_{l)0 P}
     \notag \\
     &-c_s^2\Bigl(\frac{c_s^2+1}{c_s^2} \delta_{(l}^0  \eta_{s)j}-\delta_j^0\eta_{sl}\Bigr)
    \betat \tau_{00}-c_s^2\Bigl(\frac{c_s^2+1}{c_s^2} \delta_{(l}^P  \eta_{s)j}-\delta_j^P\eta_{sl}\Bigr)\Bigl(\betat\tau_{P0}\notag \\
    &+\frac{1}{2}\bigl(2\mt_{P}-\delta^{JK}(2\ellt_{JKP}-\ellt_{PJK})\bigr)\Bigr).  \label{eq:Gjsl.2}
  \end{align}
Additionally, we can use \eqref{gamma-000}-\eqref{gamma-IJ0}, \eqref{kt-def}-\eqref{psit-def} and \eqref{eq:B0fluid} to show that
\begin{align}
  {a^i}_{j k}\gamma_i{}^k{}_l
  =& -\frac 12\delta_j^{0}\Bigl(\frac{1}{c_s^2} 
      +\frac{3 c_s^2+1}{c_s^2} \frac{W^IW_I}{w^2}\Bigr) W^{0}\delta^{PQ} \kt_{PQ}\delta_l^0 \notag \\
      &+\frac 12\delta_{j}^{K}\Bigl(\frac{2 c_s^2+1}{c_s^2} +\frac{3 c_s^2+1}{c_s^2} \frac{W^IW_I}{w^2}\Bigr) W_{K}\delta^{PQ} \kt_{PQ}\delta_l^0\notag\\
   &+\frac{1}{2}\delta_{j}^{0}\Bigl(\frac{2 c_s^2+1}{c_s^2} +\frac{3 c_s^2+1}{c_s^2} \frac{W^IW_I}{w^2}\Bigr) W_{K}\delta^{KP}\kt_{LP}\delta_l^L
     \notag \\
     &- \frac{1}{2}\delta_j^{J}\Bigl(
      \delta_{JK}
      +\frac{3 c_s^2+1}{c_s^2} \frac{W_{J} W_{K}}{w^2}\Bigr)    W^0\delta^{KP}\kt_{LP}\delta_l^L\notag\\
    &-{a^0}_{j K}\delta^{KP}\ellt_{P0L}\delta_l^L-{a^0}_{j 0} \betat \tau_{0L}\delta_l^L
    -{a^0}_{jK}\delta^{KP}(\betat \tau_{0P}+\mt_{P})\delta_l^0    
    \notag \\
    &+{a^I}_{j 0}(\mt_{I}-\frac{1}{2}\delta^{JK}(2 \ellt_{JKI}-\ellt_{IJK}) )\delta_l^0+{a^I}_{j 0}\psit_I{}^0{}_L\delta_l^L \notag\\
    &-{a^I}_{j K} (\delta^{KP}\ellt_{I0P} - \delta^{KP}\psit_{I}{}^0{}_P)\delta_l^0
      +{a^I}_{j K}\psit_I{}^K{}_L\delta_l^L,  \label{eq:fluidgamma.1}
 \intertext{and}
  {a^i}_{j k}\gamma_i{}^L{}_Q
   &=-\frac{1}{2}\delta_j^{0}\delta_k^{0}\Bigl(\frac{1}{c_s^2} 
      +\frac{3 c_s^2+1}{c_s^2} \frac{W^IW_I}{w^2}\Bigr) W^{0}\delta^{LP}\kt_{QP}\notag \\
      &+\delta_{(j}^{0}\delta_{k)}^{K}\Bigl(\frac{2 c_s^2+1}{c_s^2} +\frac{3 c_s^2+1}{c_s^2} \frac{W^IW_I}{w^2}\Bigr) W_{K}\delta^{LP}\kt_{QP}\notag\\
      &-\frac{1}{2}\delta_j^{J}\delta_k^{K}\Bigl(
      \delta_{JK}
      +\frac{3 c_s^2+1}{c_s^2} \frac{W_{J} W_{K}}{w^2}\Bigr)    W^0 \delta^{LP}\kt_{QP}
     \notag \\
     &-{a^0}_{j k}\delta^{LP}\ellt_{Q0P} +{a^I}_{j k}\psit_I{}^L{}_Q. \label{eq:fluidgamma.3}
\end{align}

\begin{rem} \label{rem:borderlineextension}
 Inspecting \eqref{eq:Gjsl.2} reveals that it is unlikely that our results can be extended to values of $c_s^2$ outside the regime $(1/(n-1),1)$ including the interesting borderline case $c_s^2=1/(n-1)$. The reason for this as follows. As we will see below, the first three lines of \eqref{eq:Gjsl.2} together with certain terms from \eqref{eq:fluidgamma.1} and \eqref{eq:fluidgamma.3} constitute the main singular time-non-integrable terms that govern the dynamics of the fluid near $t=0$. If $c_s^2$ is lies in the interval $(1/(n-1),1)$, then the first term in \eqref{eq:Gjsl.2} is strictly positive and therefore dominates the terms in lines two and three, which, being proportional to $\kt_{JK}$,  are small near the FLRW solution. This will be critical for our past stability proof. Since the latter terms have no definite signs, however, this will not work  if $c_s^2=1/(n-1)$, in which case the first term in \eqref{eq:Gjsl.2} vanishes and therefore cannot be used to control the other terms, and the problem becomes even worse if $c_s^2<1/(n-1)$. 
\end{rem}

As observed in \cite {BeyerOliynyk:2021}, it is necessary to weight the second derivatives of the metric, i.e.\ $\gt_{Qjlm}$, with $\betat$. 
Using \eqref{grad-alpha-A} and \eqref{for-O.1.S2}, it is
not difficult to see that we can write the evolution
equation \eqref{for-I.S2} in terms of the weighted derivatives $\betat \gt_{Qjlm}$
as
\begin{multline}
  \delta^{ij}\del{t}(\betat \gt_{Qjlm}) +B^{ijK}\betat e_K^\Lambda\partial_\Lambda(\betat \gt_{Qjlm}) \\
  = \frac{1}{t} \delta_0^i\delta_0^j  (\betat \gt_{Qljm}+ \betat \gt_{Qmjl}-\betat \gt_{Qjlm}) +\frac{1}{2}\delta^{ij} \delta^{JK}\betat\kt_{JK}\betat\gt_{Qjlm} \\
  +\Hf^i_{Qlm} {+2\delta_0^i\betat^2\Dc_Q\Ttt_{lm}+\left[\betat^2 (\betat \tau +\psit+\ellt+\mt)*
 {\Ttt} \right]^i_{Qlm}}\label{for-I.S2a}
\end{multline}
where we have set
\begin{align}
    \Hf^i_{Qlm}=& -\delta^{ij}\betat^2 \tau_{00}\betat\gt_{Qjlm} + B^{ijK}\Bigl(-\betat^2 \tau_{K0}-\frac{1}{2}\betat\bigl( 2 \mt_K  
   \notag\\
   &-\delta^{JI}(2\ellt_{JIK}
   -\ellt_{KJI})\bigr)\Bigr)\betat \gt_{Qjlm} \notag \\
   &+\frac{2}{t}\delta^i_0
\betat^2\taut_{Q(lm)} 
 + \delta^i_0 \betat\Pf_{Q(lm)}-B^{ijk}\betat\Gf_{Qkj(lm)}.
 \label{Hf-def}
\end{align}
As we note in \cite {BeyerOliynyk:2021}, we will also need to replace $\kt_{LM}$ with the $\betat$
weighted variables $\betat\kt_{LM}$. With the help of \eqref{for-O.1.S2}, we can write the evolution equation
\eqref{for-H.S2} in terms of $\betat\kt_{LM}$ as
\begin{align}
    \del{t}(\betat\kt_{LM}) =& -\frac{1}{t}\betat\kt_{LM}-\betat^2 \tau_{00} \betat\kt_{LM}  +\frac{2}{t}\betat^2\tau_{(LM)}- \betat B^{0jK}\betat\gt_{Kj(LM)} \notag \\
    &+2\betat B^{IjK}\betat\gt_{Kj0(L}\delta_{M)I} 
%+\betat\delta_{LI} B^{IjK}\betat g_{Kj0M} 
      +\Mf_{LM} {+2 \betat^2{\Ttt}_{(LM)}} \label{for-H.S2a}
\end{align}
where
\begin{align} 
    \Mf_{LM} =&  \betat^2\Qft^0_{(LM)}  -2\betat^2 \Qft^I_{(L}\delta_{M)I}.
    %-\delta_{LI}\betat^2 \Qft^I_M. 
    \label{Kf-def}
\end{align}
For use below, we note from Lemmas \ref{Qijk-lem} and
\ref{lem-P-exp} that \eqref{Hf-def} and \eqref{Kf-def} can
be expressed using the $*$-notation as
\begin{align*}
    \Hf =&   (\betat^2 \tau+ \betat\psit+\betat\kt+ \betat \ellt +\betat \mt)*\betat\gt  +\frac{1}{t}\betat^2\taut \notag \\
& + (\betat\tau +\psit+\ellt+\mt) *
(\betat\kt+\betat\ellt+\betat\mt)*(\betat\kt+\betat\ellt+\betat\mt) \notag\\
& +\frac{1}{t}(\betat \tau +\psit+\ellt+\mt) *(\betat\kt+\betat\ellt+\betat\mt)+\frac{1}{t}(\betat \tau +\psit) *\betat^2 \tau   %\label{Hf-exp}
    \intertext{and}
    \Mf &= (\betat\kt+\betat\ellt+\betat\mt)*(\betat^2\tau+\betat\ellt+\betat\mt),%\label{Kf-exp}
\end{align*}
respectively.

\subsection{Rescaled variables}
\label{sec:chvar2}
The next step in the transformation to Fuchsian form involves the introduction of the following rescaled variables:
\begin{align}
k&=(k_{IJ}) :=(t\betat \kt_{IJ})
\label{k-def}\\
  \beta &=  t^{\ep_0}\betat, \label{beta-def}\\
  % \check\beta &= t^{\ep_3}\hat \beta= t^{\ep_3}\betat^{-(1+c_s^2)}, \label{betacheck-def}\\
\ell&=(\ell_{IjK}):= (t^{\ep_1}\ellt_{IjK}), \label{ell-def}\\
m&=(m_{I}) := (t^{\ep_1}\mt_I), \label{m-def} \\
\xi&=(\xi_{ij}) := (t^{\ep_1-\ep_0}\tau_{ij}), \label{xi-def}\\
\psi &=(\psi_I{}^k{}_J) :=  (t^{\ep_1}\psit_I{}^k{}_J),\label{psi-def}\\
f&=(f_I^\Lambda) := (t^{\ep_2} e^\Lambda_I), \label{f-def}\\
\gac&=(\gac_{Ijkl}) := (t^{1+\ep_1} \betat \gt_{Ijkl}), \label{gac-def}\\
  \tauac &=(\tauac_{Ijk}):= (t^{\ep_0+2 \ep_1} \taut_{Ijk}), \label{tauac-def}\\
  \acute{\Utt} &=(t^{\ep_4} \Utt^s_Q), \label{Uac-def}
\end{align}
while we continue to use the non-rescaled variable $W$, see \eqref{W-def} and \eqref{p-fields}. At this point, the constants  $\ep_0$, $\ep_1$, $\ep_2$, $\ep_4>0$ are arbitrary numbers satisfying
\begin{equation}
  \label{eq:epscond.1}
  0<\ep_0<\ep_1, \quad 3\ep_0+\ep_1<1, \quad 0<\ep_2, \quad \ep_0+\ep_2<1,
  % \quad \ep_3>0,
  \quad \ep_4>0,
\end{equation}
which we note imply that $\ep_0+\ep_1<3\ep_0+\ep_1<1$. 
For reasons that become evident in the proof of
Lemma~\ref{lem:sourceterm} below, we now also define
\begin{equation}
  \label{betacheck-defN}
  \check\beta=t^{\ep_3+(1+c_s^2)\ep_0}\beta^{-(1+c_s^2)}
\end{equation}
for $\ep_3>0$
and add $\check\beta$ to the list of variables, which, together with \eqref{k-def}-\eqref{Uac-def} and $W$, then takes the form\footnote{In line with Remark~\ref{rem:symmetry}, we always assume that $U$ is defined with $k_{LM}=k_{(LM)}$ and $\gac_{Qjlm}=\gac_{Qj(lm)}$.}
\begin{multline}
U =\bigl(k_{LM},m_M,\ell_{R0M},\ell_{RLM},\xi_{rl},\beta,  \\
f^\Lambda_I,\psi_I{}^k{}_{J},\tauac_{Qjl}, \gac_{Qjlm}, \check\beta, W^s,{\acute{\Utt}^s_Q}\bigr)^{\tr}.\label{U-def}
\end{multline}
The new quantity $\check\beta$ satisfies the evolution equation
\begin{align}
  \del{t}\check\beta=&(\ep_3+(1+c_s^2)\ep_0)t^{-1}\check\beta-(1+c_s^2)\beta^{-1}\check\beta\del{t}\beta
   \notag \\
   =& t^{-1}\Bigl(\ep_3 -\frac{1}{2} (1+c_s^2)\delta^{JK}k_{JK}\Bigr)\check\beta +(1+c_s^2) t^{-\ep_0-\ep_1}\beta^2 \xi_{00}\check\beta. \label{beta-hat-ev}
\end{align}

In all appearances of the undifferentiated energy momentum tensor $\Ttt_{ij}$ in the equations \eqref{for-G.1.S2}-\eqref{for-Euler.2.2}, we replace $\betat^{-(1+c_s^2)}$, which appears in $\Ttt_{ij}$, see \eqref{eq:TttconfphysW},
with $t^{-\ep_3}\check\beta$. Similarly, in the derivatives $\Dc_Q \Tt_{ij}$ of the energy momentum tensor appearing in
\eqref{for-G.1.S2}-\eqref{for-Euler.2.2}, we replace $\Dc_Q(\betat^{-(1+c_s^2)})$ with $-(1+c_s^2) t^{-\ep_3}\check\beta\betat^{-1}\Dc_Q\betat$ by exploiting \eqref{grad-alpha-A} and we replace any remaining $\betat^{-(1+c_s^2)}$ terms with $t^{-\ep_3}\check\beta$. 

As discussed in the introduction, the purpose of using $t$-powers with so far unspecified exponents $\ep_1$, \ldots, $\ep_4$  to rescale the variables as in \eqref{k-def}-\eqref{Uac-def} and  \eqref{betacheck-defN} is to generate positivity among the $t^{-1}$-terms in the resulting evolution equations below. This turns out to be crucial for the Fuchsian analysis at the core of the proof of our main result.

\subsection{Fuchsian formulation\label{sec:Fuch-form}}
It is now straightforward to verify from
the first order equations \eqref{for-G.1.S2}-\eqref{for-Euler.2.2}, \eqref{for-I.S2a}, \eqref{for-H.S2a} and \eqref{beta-hat-ev} that $U$ satisfies the following symmetric hyperbolic Fuchsian equation:
\begin{equation} \label{Fuch-ev-A}
A^0\del{t}U + \frac{1}{t^{\ep_0+\ep_2}}A^\Lambda \del{\Lambda} U =
\frac{1}{t}\Ac\Pbb U + F
\end{equation}
where
%{\setlength{\mathindent}{0pt}
\begin{align} \label{A0-def}
    A^0 &=\small \diag\bigl(A^0_G,A^0_F\bigr),\\
  A_G^0& =\small \diag\bigl(\delta^{\Lt L}\delta^{\Mt M},\delta^{\Mt M},\delta^{\Rt R}\delta^{\Mt M},\delta^{\Rt R} \delta^{\Lt L} \delta^{\Mt M},\notag \\
  &\hspace{1.8cm}\delta^{\rt r}\delta^{\lt l},1,\delta^{\It I}\delta_{\Lambdat \Lambda},\delta^{\It I}\delta_{\kt k} \delta^{\Jt J},\delta^{\Qt Q}\delta^{\jt j}\delta^{\lt l},\delta^{\Qt Q}\delta^{\jt j}\delta^{\lt l} \delta^{\mt m}\bigr),  \label{A0G-def} \\
  \label{A0F-def}
  A_F^0 &=\small \diag\bigl(1, a^0_{\jt j}, a^0_{\jt j}\delta^{\Qt Q}\bigr),\\
  \label{ALambda-def}
  A^\Lambda &= \diag\bigl( A^\Lambda_G,A^\Lambda_F\bigr),\\
  \label{ALambdaG-def}
  A_G^\Lambda &= \diag\bigl( 0,0,0,0,0,0,0,0,\delta^{\Qt Q}\delta^{\lt l}B^{\jt j K}\beta f^\Lambda_K,\delta^{\Qt Q}\delta^{\lt l}\delta^{\mt m}B^{\jt j K}\beta f^\Lambda_K\bigr),\\
  \label{ALambdaF-def}
  A_F^\Lambda &= \diag\bigl(0, {a^I}_{\jt j}\beta f_I^\Lambda, {a^I}_{\jt j}\beta f_I^\Lambda \delta^{\Qt Q}\bigr),\\
  \label{Pbb-def}
  \Pbb &= \diag\Bigl(0,\delta_{\Mt}^{M},\delta_{\Rt}^{R}\delta_{\Mt}^{M},\delta_{\Rt}^{R} \delta_{\Lt}^{ L} \delta_{\Mt}^{M},\delta_{\rt}^{r}\delta_{\lt}^{l},1,\delta_{\It}^{I}\delta^{\Lambdat}_{\Lambda},\notag \\
  &\hspace{1.8cm}\delta_{\It}^{I}\delta^{\kt}_{k}\delta_{\Jt}^{J},\delta_{\Qt}^{Q}\delta_{\jt}^{j}\delta_{\lt}^{l},\delta_{\Qt}^{Q}\delta_{\jt}^{j}\delta_{\lt}^{l} \delta_{\mt}^{m},\delta_{\tilde J J}\delta_{\jt}^{\Jt}\delta_j^J , \delta_{\tilde j j} \delta^{\Qt Q}\Bigr),\\
  \label{Ac-def}
  \Ac &= \diag\bigl( \Ac_G, \Ac_F\bigr),\\
  \label{AcG-def}
    \Ac_G &= 
 \small{ \begin{pmatrix}
    \Ac_{1\,1} & 0 & 0 & 0 & 0 & 0 & 0 & 0 & 0  & 0 \\
    0 &\Ac_{2\,2} & 0 & \Ac_{2\,4} & 0 & 0 & 0 & 0 & 0  & \Ac_{2\,10} \\
    0 & 0 & \Ac_{3\,3} & 0 & 0 & 0 & 0 & 0 & 0  & \Ac_{3\,10} \\
    0 & 0 & 0 & \Ac_{4\,4} & 0 & 0 & 0 & 0 & 0  & \Ac_{4\,10}\\
    0 & 0 & 0 & 0 & \Ac_{5\,5} & 0 & 0 & 0 & 0  & 0 \\
    0 & 0 & 0 & 0 & 0 & \Ac_{6\,6} & 0 & 0 & 0  & 0 \\
    0 & 0 & 0 & 0 & 0 & 0 & \Ac_{7\,7} & 0 & 0  & 0 \\
    0 & 0 & 0 & 0 & 0 & 0 & 0 & \Ac_{8\,8} & 0  & \Ac_{8\,10} \\
    0 & 0 & 0 & 0 & 0 & 0 & 0 & 0 &\Ac_{9\,9} & 0 \\
    0 & 0 & 0 & 0 & 0 & 0 & 0 & 0 & 0  & \Ac_{10\,10}    
  \end{pmatrix}},\\
  \label{AcF-def}
  \Ac_F &= \diag\bigl(\Ac_{11\,11}, \Ac_{12\,12}, \Ac_{13\,13}\bigr),
\end{align}
the non-zero diagonal and off-diagonal 
blocks of $\Ac$ are given by
\begin{gather}
    \Ac_{1\,1}=\delta^{\Lt L}\delta^{\Mt M},\quad \Ac_{2\,2}=(2+\ep_1)\delta^{\Mt M}, \quad
    \Ac_{3\, 3} = \ep_1\delta^{\Rt R}\delta^{\Mt M}, \label{Ac-diag-1}\\
    \Ac_{4\, 4} =  \ep_1\delta^{\Rt R}\delta^{\Lt L}\delta^{\Mt M}, \quad  
    \Ac_{5\, 5} = (\ep_1-\ep_0)\delta^{\rt r}\delta^{\lt l},\quad
     \Ac_{6\, 6} =  \ep_0, \label{Ac-diag-2}\\
     \Ac_{7\, 7} = \ep_2\delta^{\It I}\delta_{\Lambdat\Lambda},\quad
     \Ac_{8\, 8} = \ep_1\delta^{\It I}\delta_{\kt k}\delta^{\Jt J},\quad
     \Ac_{9\, 9} = (\ep_0+2\ep_1)\delta^{\Qt Q}\delta^{\jt j} \delta^{\lt l}, \label{Ac-diag-3},
     \end{gather}
     \begin{multline}
     \Ac_{10\, 10} = (1+\ep_1)\delta^{\Qt Q}\delta^{\jt j} \delta^{\lt l}\delta^{\mt m}\\
     + \delta^{\Qt Q}\delta^{\jt}_0\bigl( \delta_0^l\delta^{\lt j}\delta^{\mt m}+\delta_0^l \delta^{\lt m}\delta^{\mt j}-\delta_0^j \delta^{\lt l}\delta^{\mt m}\bigr),\quad \Ac_{11\, 11} =  \ep_3, \label{Ac-diag-4}
     \end{multline}
\begin{gather}
     \Ac_{12\, 12}=\frac{(n-1)c_s^2-1}{n-2}W^0 \bigl(c_s^{-2}\delta_{\jt}^0\delta_j^0+\delta_{\Jt J}\delta_{\jt}^{\Jt}\delta_{j}^J\bigr), \label{Ac-diag-5}\\
     A_{13\,13}= \Bigr(\ep_4 a^0_{s\tilde j} +
  W^0\frac{(n-1)c_s^2-1}{n-2} \Bigl(   
  \eta_{s\jt}+\delta_{s}^0  \delta_{\jt}^0
  \Bigr)\Bigr)\delta^{\Qt Q}, \label{Ac-diag-6}
\end{gather}
and
\begin{gather}
    \Ac_{2\,4} = \delta^{\Mt R}\delta^{LM}-2\delta^{RL}\delta^{\Mt M}, \quad 
    \Ac_{2\, 10}=-B^{0jQ}\delta_0^l\delta^{\Mt m}, \label{Ac-offdiag-1} \\
    \Ac_{3\, 10}=-B^{\Rt jQ}\delta^l_0\delta^{\Mt m}, \quad 
    \Ac_{4\, 10} =  -B^{\Rt jQ}\delta^{\Lt l}\delta^{\Mt m}, \label{Ac-offdiag-2} \\
     \Ac_{8\, 10} =
     \frac{1}{2}\delta^{\It Q}\bigl(-\delta_{\kt}^0\delta^{\Jt j}\delta^l_0\delta^m_0-\delta_{\kt k}\bigl(\delta^j_0\eta^{k m}\delta^{\Jt l}+\eta^{km}\delta^{\Jt j}\delta^l_0-\eta^{kj}\delta^{\Jt m}\delta^l_0\bigr)\bigr), \label{Ac-offdiag-3}
\end{gather}
respectively,
and
\begin{equation}\label{F-def}
F=\bigl(F_1,F_2,F_3,F_4,F_5,F_6,F_7,F_8,F_9,F_{10},F_{11},F_{12},F_{13}\bigr)^{\tr}
\end{equation}
with
\begin{align}
F_1 =&\delta^{\Lt L}\delta^{\Mt M}\Bigl({t^{-1}\Rtt k_{LM}} -t^{- \ep_0-\ep_1}\beta^2 \xi_{00} k_{LM} \notag \\
&+2t^{- \ep_0-\ep_1}\beta^2\xi_{(LM)}-t^{-\ep_0-\ep_1} \beta B^{0jK}\gac_{Kj(LM)} \notag\\
    &+t^{-\ep_0-\ep_1}2\beta B^{IjK}\gac_{Kj0(L}\delta_{M)I}  
%+t^{-\ep_0-\ep_1}\beta\delta_{LI} B^{IjK}\gac_{Kj0M} 
+t\Mf_{LM}{+2 t^{1-2\ep_0}\beta^2{\Ttt}_{(LM)}}\Bigr), \label{F1-def}\\
F_2 =& t^{-1}2\delta^{\Mt M}\beta\xi_{0M}
+ \delta^{\Mt M}(t^{\ep_1}\betat\Qft_M{+2 t^{\ep_1-\ep_0}\beta{\Ttt}_{0M}}), \label{F2-def}\\
F_3 =& \delta^{\Mt M}t^{\ep_1}\betat \Qft^{\Rt}_M, \label{F3-def}\\
F_4 =& \delta^{\Lt L}\delta^{\Mt M}t^{\ep_1}\betat \Qft^{\Rt}_{LM}, \label{F4-def}\\
F_5 =& \delta^{\lt l}
       t^{\ep_1-\ep_0}\betat J^{\rt}_l-t^{-\ep_1-3\ep_0}\delta^{\lt l}\delta^{\rt r} B^{i jK} \beta\tauac_{Kj(l}\delta_{r)i}
      ,\label{F5-def}\\
 F_6 =& t^{-1}\frac{1}{2}\delta^{JK}k_{JK}\beta -t^{-\ep_0-\ep_1}\beta^3 \xi_{00},\label{F6-def}\\
F_7 =& t^{-1}\delta^{\It I}\delta_{\Lambdat \Lambda}\Bigl(- \frac{1}{2}\delta^{JL}k_{IL}-\delta^{JK} t^{1-\ep_0-\ep_1}\beta\psi_{I}{}^0{}_K\Bigr)f^\Lambda_J,\label{F7-def}\\
F_8 =&
-\delta^{\It I}\delta_{\kt}^0\delta^{\Jt J}t^{-\ep_1-3\ep_0}\beta^2 \tauac_{I0J}+ \delta^{\It I}\delta_{\kt k} \delta^{\Jt J}t^{\ep_1}\betat \Lf_I{}^k{}_J,\label{F8-def}\\
  F_9=& \delta^{\Qt Q}\delta^{\lt l}\bigl(t^{\ep_0+2 \ep_1}\betat K^{\jt}_{Ql}-B^{\jt jk}t^{\ep_0+2 \ep_1}\Tf_{Qkjl}\bigr) \label{F9-10}\\
  F_{10}=& \delta^{\Qt Q}\delta^{\lt l}\delta^{\mt m} t^{1+\ep_1}\Hf^{\jt}_{Qlm}
  +\frac{1}{2t} \delta^{j\jt} \delta^{\Qt Q}\delta^{\lt l}\delta^{\mt m} \delta^{JK}k_{JK}\gac_{Qjlm}\label{F10-def}\\ 
  &+2\delta^{\Qt Q}\delta^{\lt l}\delta^{\mt m} t^{1+\ep_1-2\ep_0}\delta_0^i\beta^2\Dc_Q\Ttt_{lm}
 \notag \\
 &+\delta^{\Qt Q}\delta^{\lt l}\delta^{\mt m} t^{1-2\ep_0}\left[\beta^2 (\beta \xi +\psi+l+m)*
    {\Ttt} \right]^i_{Qlm},\notag\\
  F_{11}=& t^{-1}\Bigl(-\frac{1+c_s^2}{2} \delta^{JK}k_{JK}+(1+c_s^2) \betat^2 t\tau_{00}\Bigr) \hat\beta,\label{F11-def}\\
  F_{12}=& \delta_{\jt} ^j \Bigl(\betat G_{jsl}W^sW^l- \betat {a^i}_{j k}\gamma_i{}^k{}_lW^l\Bigr)
           \notag \\
           & -t^{-1}\frac{(n-1)c_s^2-1}{n-2} W^0\delta_{\jt}^{\Jt}\delta_l^J\delta_{\Jt J}W^l,\label{F12-def}\\
  F_{13}=& \delta_{\jt} ^j \delta^{\Qt Q} \Bigl(
           t^{\ep_4} \Dc_Q(\betat G_{jsl}W^sW^l)-t^{-1} 
  \frac{(n-1)c_s^2-1}{n-2} \Bigl(       
  \eta_{sj}+\delta_{s}^0  \delta_{j}^0
  \Bigr)\,{\acute{\Utt}}^s_Q\,W^0
          \notag \\
          &-t^{\ep_4} \betat^{-1}\Dc_Q\betat \betat G_{jsl}W^sW^l +\bigl(\Dc_Q{a^0}_{j m}+2t^{-\ep_1}{a^0}_{j m}\beta \xi_{0Q}
                    \notag \\
        &-2t^{-\ep_1}{a^I}_{j m}\psi_I{}^0{}_Q\bigr) (a^{0})^{-1}{}^{mn}\Bigl( t^{-\ep_0}\beta a^I{}_{nk}\acute{\Utt}^k_I-t^{\ep_4}\betat G_{nsl}W^sW^l\Bigr)\notag\\
                  &-\bigl(t^{-\ep_0}\beta \Dc_Q{a^P}_{j k}
                    +\frac{1}{2} t^{-1}{a^0}_{j k}\delta^{PL}k_{QL}
                    +t^{-\ep_0-\ep_1}\beta{a^0}_{j k}\delta^{PL}\ell_{Q0L}
                    \notag \\
                    &-t^{-\ep_0-\ep_1}\beta {a^I}_{j k}\psi_I{}^P{}_Q\bigr) \acute{\Utt}_P ^k-\betat{a^i}_{j k}\gamma_i{}^k{}_l {\acute{\Utt}}^l_Q
       +\betat {a^i}_{j k}\gamma_i{}^L{}_Q {\acute{\Utt}}^k_L\Bigr).\label{F13-def}
\end{align}
In the above expressions, we recall the $G_{jsl}$, $a^i{}_{jk}\gamma_i{}^k{}_l$ and $a^i{}_{jk}\gamma_i{}^L{}_Q$ are determined by  \eqref{eq:Gjsl.2}, \eqref{eq:fluidgamma.1} and \eqref{eq:fluidgamma.3}, respectively. 
We also point out that the kernel of the matrix $\Pbb$ in \eqref{Pbb-def} is spanned by the variables $k_{IJ}$ and $W^0$; this property of $\Pbb$ will play a decisive role in our past stability proof.

\subsection{FLRW background solution}

In this article, we are interested in nonlinear perturbations of the FLRW solution \eqref{eq:FLRWEulerSFExplSol1}-\eqref{eq:FLRWEulerSFExplSol4} of the conformal
 Einstein-Euler-scalar field system that are parameterised by $(P_0,c_s^2,V_*^0)\in (0,\infty)\times (0,1)\times (0,\infty)$. For this family of solutions, the variables \eqref{k-def}-\eqref{Uac-def} take the form 
\begin{align}
  \label{eq:UbreveFirst}
  \breve k_{IJ}&=2\frac{t \omega'}{\omega},\quad
  \breve\beta =t^{\ep_0} \omega^{n-1},\quad
  \breve\ell_{IjK}=0,\quad
  \breve m_I=0,\\
  \breve\xi_{ij}&=(n-1)^2\frac{t^{\ep_1-\ep_0} \omega'}{\omega}\omega^{-2(n-1)}\delta_i^0\delta_j^0,\\
  \breve\psi_I{}^k{}_J &=0,\quad
  \breve f_I^\Lambda=t^{\ep_2} \omega^{-1} \delta_I^\Lambda,\quad
                         \breve\gac_{Ijkl}=0,\\
  \label{eq:UbreveLast}
  \breve\tauac_{Ijk}&=0,\quad
  \breve W^s=V^0_*\delta_0^s,\quad
  \breve{\acute{\Utt}}_Q^s =0,
\end{align}
where $\ep_0$, \ldots, $\ep_4$, are constants, which for now may be chosen arbitrarily,
$\omega$ is defined by \eqref{eq:FLRWEulerSFExplSol4}, and $\omega'=\frac{d\omega}{dt}$. It can be verified that the corresponding vector $\breve U$ defined by \eqref{U-def} constitutes a solution of the Fuchsian system \eqref{Fuch-ev-A} for any $t>0$ for which the function $\omega$ is well-defined.

\begin{lem}
  \label{lem:bgpertsol}
Suppose $P_0>0$, $V_*^0>0$, $c_s^2\in (0,1)$, $\ep_0$, \ldots, $\ep_4$ satisfy
  \begin{equation}
    \label{eq:epsiloncond.0}
    \begin{split}
    0<\ep_0<\frac{1-\ep_1}3,\quad 1- &\frac{n-1}{2(n-2)}(1-c_s^2)<\ep_1<1, \\
    0<\ep_2,\quad 0&<\ep_3,\quad 0<\ep_4,
    \end{split}
  \end{equation}
  and let $\breve U$ be the solution of the Fuchsian system \eqref{Fuch-ev-A} that 
  is defined by \eqref{eq:UbreveFirst}-\eqref{eq:UbreveLast} and \eqref{U-def}, which corresponds to the FLRW solution of the conformal Einstein-Euler-scalar field system. Then there exists a positive constant $c>0$ such that
  \begin{equation}
    \label{eq:Ubgdef}
    \mathring U(t)
    =\bigl(0,0,0,0,0, t^{\ep_0}, t^{\ep_2}\delta_{I}^{\Lambda},0,0,0,t^{\ep_3},V_*^0\delta_0^j,0\bigr)
  \end{equation}
  satisfies
  \begin{equation}\label{Ubr-Ur}
    \breve U(t)-\mathring U(t)=\Ord(t^c)
\end{equation} as $t\searrow 0$.
\end{lem}
\noindent The significance of this lemma is that it will allow us to interpret
\begin{equation} \label{u-def}
  u=U-\mathring U
\end{equation}
as a nonlinear perturbation of a FLRW solution \eqref{eq:FLRWEulerSFExplSol1}-\eqref{eq:FLRWEulerSFExplSol4} near $t=0$; this property will play an important role in our stability proof.
The proof of the above lemma is straightforward and is easily seen to follow from the inequalities 
\begin{multline*} \frac{(n-1)(1-c_s^2)}{n-2}-1+\ep_1-\ep_0
  >\frac{(n-1)(1-c_s^2)}{n-2}-1+\ep_1-\frac{1-\ep_1}3
 \\
 =\frac{(n-1)(1-c_s^2)}{n-2}-\frac 43 +\frac 43\ep_1
  >\frac{(n-1)(1-c_s^2)}{n-2}-\frac 43\\
  +\frac 43\Bigl(1- \frac{n-1}{2(n-2)}(1-c_s^2)\Bigr)
=\frac{(n-1)(1-c_s^2)}{3(n-2)}>0,
\end{multline*}
which hold by \eqref{eq:epsiloncond.0}, the expressions \eqref{eq:UbreveFirst}-\eqref{eq:UbreveLast}, and (see \eqref{eq:FLRWEulerSFExplSol4}) the expansions
\begin{equation}
  \label{eq:FLRWexpansions}
  \omega=1+\Ord\bigl(t^{\frac{n-1}{n-2}(1-c_s^2)}\bigr) \AND
  \omega'=\Ord\bigl(t^{\frac{n-1}{n-2}(1-c_s^2)-1}\bigr).
\end{equation}

In the following arguments, the positive constants $V_*^0$ and $P_0$ that appear in the FLRW solution \eqref{eq:FLRWEulerSFExplSol1}-\eqref{eq:FLRWEulerSFExplSol4} can be chosen arbitrarily, but will be considered as fixed to some particular value throughout. We will also use
\begin{equation*}
\mc = n^4+2 n^3-3 n^2+2n+2 
\end{equation*}
to denote the dimension of the vector $u$.

Now, it is not difficult, using the definitions \eqref{A0-def}-\eqref{F13-def} and \eqref{eq:Ubgdef}-\eqref{u-def}, to verify that the system \eqref{Fuch-ev-A} can be expressed as
\begin{equation} \label{Fuch-ev-A2}
A^0(u)\del{t}u + \frac{1}{t^{\ep_0+\ep_2}}A^\Lambda(t,u) \del{\Lambda} u =
\frac{1}{t}\Ac(u)\Pbb u + \tilde{\Ftt}(t)+ \Ftt(t,u)
\end{equation}
where $A^0(u)$ and $\Ac(u)$ denote $A^0(\mathring U(t)+u)$ and $\Ac(\mathring U(t)+u)$, respectively\footnote{This notation makes sense because $A^0(\mathring U(t)+u)$ and $\Ac(\mathring U(t)+u)$ do not depend on $t$ as can be easily verified from \eqref{A0-def}-\eqref{Ac-def} and \eqref{eq:Ubgdef}.},
$A^\Lambda(t,u)$ denotes $A^\Lambda(\mathring U(t)+u)$, and the source terms $\Ftt(t,u)$ and $\tilde{\Ftt}(t)$ are defined by
  \begin{equation}
    \label{eq:Ftt}
    \Ftt(t,u)=F(t,\mathring U(t)+u)-F(t,\mathring U(t))%\quad\Rightarrow\quad \Ftt(t,0)=0,
  \end{equation}
  and
\begin{equation}
\tilde{\Ftt}(t) %=F(t,\mathring U(t))
=\begin{pmatrix}2\delta^{\Lt L}\delta^{\Mt M} P_0\frac{1-c_s^2}{(n-2)c_s^2} t^{-1+\frac{(n-1)(1-c_s^2)}{n-2}} (V_*^0)^{-(1+c_s^2)/c_s^2}\delta_{LM}\\ 0\\\vdots \\ 0
\end{pmatrix},\label{eq:Ftttilde}
\end{equation}
   respectively. In obtaining \eqref{eq:Ftttilde}, we have employed \eqref{eq:TttconfphysW}, \eqref{eq:Gjsl.2}-\eqref{eq:fluidgamma.1} and  \eqref{eq:Ubgdef}. It is important to note that the approximate solution $\mathring U$, defined above by \eqref{eq:Ubgdef}, satisfies 
   \[A^0(\mathring U(t)+u)\del{t}\mathring U(t)-\frac 1t \Ac(\mathring U(t)+u)\Pbb \mathring U(t)=0.\]

 The system \eqref{Fuch-ev-A2} is the Fuchsian formulation of the reduced conformal Einstein-Euler-scalar field that we will use below to establish the past stability of the FLRW solutions  \eqref{eq:FLRWEulerSFExplSol1}-\eqref{eq:FLRWEulerSFExplSol4} and their big-bang singularities.   

\subsection{Source term expansion}
\label{sec:sourcetermexpansion}
In this section, we establish a number of properties of the source term maps $\Ftt$ and $\tilde\Ftt$ defined above by \eqref{eq:Ftt}-\eqref{eq:Ftttilde} that will be needed below in the proof of Theorem \ref{glob-stab-thm}.

\begin{lem}
  \label{lem:sourceterm}
  Suppose $V_*^0>0$, $P_0>0$, $T_0>0$, $c_s^2\in (0,1)$, and  $\ep_0,\ldots, \ep_4,\tilde{\ep}$ are constants that satisfy the inequalities
  \begin{equation}
    \label{eq:epscond.N}
    \begin{split}
      1- \frac{n-1}{2(n-2)}(1-c_s^2)<\ep_1<1,\\
      % \quad \ep_1<\ep_4<\min\{1,\ep_1+ \frac{n-1}{3(n-2)}(1-c_s^2)\},
      \ep_1<\ep_4<1,\\
      0<\ep_0<\min\{1-\ep_4,(1-\ep_1)/3\},\\
      0<\ep_2<1-\ep_0,\\
      0<\ep_3<  \frac{n-1}{3(n-2)}(1-c_s^2),
    \end{split}
  \end{equation}  
and
\begin{multline}
  \max\biggl\{1 -\Bigl(\frac{n-1}{n-2}(1-c_s^2)-\ep_3+\ep_1-\ep_0-1\Bigr),3\ep_0+\ep_1, \\
  1-\ep_0,1-\ep_2,1-(\ep_4-\ep_1), \ep_0+\ep_4\biggr\}\leq \tilde{\ep}<1.   \label{eq:epstildecond}
\end{multline}
  Then 
  there exists a constant $R>0$ and maps  
  \begin{gather}
    \label{eq:kdkjdsfkjsd309}
    H\in
    C^0\bigl([0,T_0],C^\infty\bigl(B_R(\Rbb^{\udim}),\Rbb^{\udim}\bigr)\bigr),\\
    \label{eq:kdkjdsfkjsd309.2}
    \Htt \in C^\infty\bigl(B_R(\Rbb^{\udim}),\Mbb{\udim}\bigr),\\
    \label{eq:kdkjdsfkjsd309.3}
  \hat\Htt \in C^\infty\bigl(B_R(\Rbb^{\udim}),
  \Rbb^{\udim}\bigr)\bigr),
\end{gather}
satisfying 
  \begin{gather}    
    \label{eq:SourcetermVan2}
    \Htt(0) = 0,\quad H(t,0)=0,\\
    \label{eq:Gtprops}
    [\Htt,\Pbb]=0 \AND \hat\Htt=\Ord(\Pbb u\otimes\Pbb u),
  \end{gather}
  for all $u\in B_R(\Rbb^{\udim})$ and $t\in (0,T_0]$,
  such that $\Ftt(t,u)$ can be expressed as
  \begin{equation}
    \label{eq:SourcetermVan1}
    \Ftt(t,u)=\frac{1}{t}\Htt(u)\Pbb u+\frac 1t\Pbb^{\perp}
    \hat\Htt(u)+\frac{1}{t^{\tilde\epsilon}} H(t,u).
  \end{equation}
\end{lem}

Before considering the proof of the above lemma, we first make some observations that are easily verified and will provide a strategy for proving the lemma.  We begin by noting that the conditions \eqref{eq:epscond.N} on the constants $c_s^2$ and $\ep_0,\ldots,\ep_4$ imply that $\ep_0,\ldots,\ep_4$ satisfy \eqref{eq:epscond.1} and that there exists a constant $\tilde\ep$ satisfying \eqref{eq:epstildecond}. In particular,
it is straightforward to verify that the following choice of constants
  \begin{equation}
  \label{eq:epsfluidchoice}
  \ep_0=\frac A8,\quad \ep_1=1-\frac A2+\sigma,\quad \ep_2=\frac 12-\frac A{16},\quad \ep_3=\frac A8,\quad \ep_4=1-\frac {5A}{16},
\end{equation}
where
\begin{equation}
  \label{eq:defA}
  A=\frac{n-1}{n-2}(1-c_s^2)
\end{equation}
and $\sigma>0$ is chosen sufficiently small,
satisfies \eqref{eq:epscond.N} for any $n\ge 3$ and $c_s^2\in (0,1)$. Indeed, this is the choice we employ in the proof of Theorem~\ref{glob-stab-thm} as discussed in Section~\ref{sec:proof_globstab}, and we  note, with this choice, that  \eqref{eq:epstildecond} simplifies to
\begin{equation*}
  %\label{eq:2}
  1-\frac A8 <\tilde\epsilon<1.
\end{equation*}

In order to prove Lemma~\ref{lem:sourceterm}, it is useful to observe that any component of the map $F$, see \eqref{F-def} and \eqref{F1-def}-\eqref{F13-def}, can be expanded in the schematic form 
\begin{equation}\label{F-expansion}
  \begin{split}
    \mathtt f=\sum_{\ell=1}^{N^{(0)}}& t^{-1}f^{(0)}_\ell\Bigl(U, w^{-1}, (W^0)^{-1}, w^{-\frac{1+c_s^2}{c_s^2}}, (\det (a^0{}_{ij}))^{-1}\Bigr)\\
    &+\sum_{\ell=1}^{N^{(1)}} t^{-\sigma_\ell}f^{(1)}_\ell\Bigl(U, w^{-1}, (W^0)^{-1}, w^{-\frac{1+c_s^2}{c_s^2}}, (\det a^0{}_{ij})^{-1}\Bigr),
  \end{split}
\end{equation}
where $N^{(0)}$ and $N^{(1)}$ are positive integers, each exponent $\sigma_\ell$ depends on $n$, $c_s^2$ and $\ep_1$,\ldots, $\ep_4$, and where each map $f_\ell^{(0)}$ and $f_\ell^{(1)}$ is a product of a constant, non-negative integer powers  of $w^{-1}$, $(W^0)^{-1}$, $w^{-\frac{1+c_s^2}{c_s^2}}$, $(\det(a^0{}_{ij}))^{-1}$, and components of $U$.
We will refer to terms of the form $t^{-\sigma_\ell}f^{(1)}_{\ell}$ where $\sigma_\ell <1$ as \emph{time-integrable} and
terms of the form $t^{-1}f^{(0)}_\ell$ as \emph{time-non-integrable}.

For use below, we note from \eqref{U-def}, \eqref{eq:Ubgdef} and \eqref{u-def} that $W^i$
can be expressed as\footnote{Here and in the following, subscripts, e.g. $u_\lambda$, on any one of the vectors $U$, $\mathring{U}$ or $u$ refer to the position in the slots/blocks of these vectors. So $u_{13}$ would refer to the last slot of $u$.} 
\begin{equation} \label{Wi-exp}
    \quad W^i=(\mathring U(t)+u)_{12}^i,
  \end{equation}
and from \eqref{eq:B0fluid}, \eqref{Pbb-def}, \eqref {u-def} and \eqref{eq:Ubgdef} that
\begin{equation}
  \label{eq:a0exp}
  a^0{}_{jk}(u)=(V_*^0+u_{12}^0)\Bigl(\frac{1}{c_s^2}\delta_j^{0}\delta_k^{0} 
  +\delta_{JK}\delta_j^{J}\delta_k^{K}
  \Bigr)
 +\Ord(\Pbb u)
\end{equation}
and
\begin{equation}
  \label{eq:deta0exp}
  \det(a^0{}_{jk})=\frac 1{c_s^2}(V_*^0+u_{12}^0)
 +\Ord(\Pbb u).
\end{equation}
The following lemma, which will be needed in the proof of Lemma \ref{lem:sourceterm}, will be used to analyse the  time-integral component of \eqref{F-expansion}. We omit the proof since it is straightforward to establish using \eqref{eq:defw1234}, \eqref{Wi-exp} and \eqref{eq:deta0exp}.
\begin{lem}
  \label{lem:TIsourceterm}
  Suppose $T_0>0$, $V_*^0>0$, $\ep_0,\ep_2,\ep_3>0$ and let
  \begin{equation*}
    \mathtt f(t,U,z_1,z_2,z_3,z_4)=\sum_{\ell=1}^{N} t^{-\sigma_\ell}f_\ell(U, z_1, z_2,z_3,z_4)
  \end{equation*}
  where $N$ is a non-negative integer, $\sigma_\ell\in\Rbb$, $\ell=1,\ldots,N$, and each map $f_\ell$ is a product of non-negative integer powers of $z_1,\ldots,z_4$ and components of $U$. Then the exists a constant $R>0$ such that the map
  \begin{equation*}
    %\label{eq:5}
    \mathtt\ft(t,u) = \mathtt f\Bigl(t, \mathring U(t)+u, w^{-1}, (W^0)^{(-1)}, w^{-\frac{1+c_s^2}{c_s^2}}, (\det (a^0{}_{ij}))^{-1}\Bigr)
  \end{equation*}
  is well defined for $(t,u)\in (0,T_0]\times B_R(\Rbb^{\udim})$ and
  can be expanded as
  \begin{equation*}
    %\label{eq:TIdecomp2}
    \mathtt\ft(t,u)=\sum_{\ell=1}^{\tilde N} t^{-\tilde\sigma_\ell}\tilde f_\ell(u)
  \end{equation*}
  for some positive integer $\tilde N$, constants $\tilde\sigma_\ell$, $\ell=1,\ldots,\Nt$, and maps
 $\tilde f^{(1)}_\ell$, $\ell=1,\ldots,\Nt$ that are smooth on  $B_R(\Rbb^{\udim})$. Moreover, 
  \begin{equation*}
    \max_{\ell\in\{1,\ldots,\tilde N\}} \tilde\sigma_\ell = \max_{\ell\in\{1,\ldots, N\}} \sigma_\ell,
  \end{equation*}
  and the map
  \begin{equation*}
    %\label{eq:9}
    h(t,u)=t^{\tilde\ep}\mathtt\ft(t,u)
  \end{equation*}
  defines an element of $C^0\bigl([0,T_0],C^\infty\bigl(B_R(\Rbb^{\udim})\bigr)\bigr)$ provided
  \begin{equation*}
    %\label{eq:10}
    \tilde\ep\ge \max_{\ell\in\{1,\ldots, N^{(1)}\}} \sigma_\ell.
  \end{equation*}
\end{lem}

Before proceeding, we establish the following technical lemma that will be employed in the proof of Lemma~\ref{lem:sourceterm}. 

\begin{lem}
  \label{lem:DerivTerm}
Let $\eta=(\eta_{ij})$, $\Omega_1$ be a non-empty open subset of $\Rbb^{\udim}$, $\Omega_2$ be an open subset of $\Rbb^{n^2}$ that contains $\eta$, and suppose that $\ftt(U,g)$  is map smooth real-valued map for $(U,g)\in \Omega_1\times \Omega_2$
and is independent of $(U_7)^\Lambda_I=f^\Lambda_I$, $(U_8)_I{}^k{}_{J}=\psi_I{}^k{}_{J}$, $(U_9)_{Qjl}=\tauac_{Qjl}$, $(U_{10})_{Qjlm}=\gac_{Qjlm}$ and $(U_{13})^s_Q=\acute{\Utt}^s_Q$. Then the derivative
  \begin{equation*}
    \mathtt f_Q =\Dc_Q (\mathtt f(U,g))\bigl|_{g=\eta}=e_Q (\mathtt f(U,g))\bigl|_{g=\eta}
  \end{equation*}
can be expanded via the chain rule as
  \begin{equation}
    \label{eq:DerDecomCR}
    \mathtt f_Q=
    \sum_{\Sc=1}^{13} \ftt_{\Sc}^{(0)}(U) \Dc_Q U_{\Sc}    
    + \mathtt f^{(1)ij}(U) g_{Qij}
  \end{equation}
  where $\ftt_{\Sc}^{(0)}\in C^\infty(\Omega_1,\Rbb)$ and $\mathtt f^{(1)ij}\in C^\infty(\Omega_1,\Rbb)$. 
  In particular, for any $T_0>0$, the map $\mathtt f_Q$ can understood as a map defined on $(0,T_0]\times\Omega_1$ given by
\begin{align}
  \mathtt f_Q 
  =&
   \; t^{-\ep_1} \ftt_{1}^{(0),KL}\Bigl(
    (t^{\ep_1}\beta^{-1}\Dc_Q\beta) k_{KL}         
    +(t^{\ep_1}\gamma_Q{}^0{}_0) k_{KL}
    +\psi_Q{}^I{}_K k_{IL}\notag \\
    &+\psi_Q{}^I{}_L k_{KI} +\gac_{Q0KL}-\gac_{QK0L}-\gac_{QL0K}
    \Bigr)\notag \\
    &+t^{-\ep_1} \ftt_{2}^{(0),M}\Bigl((t^{\ep_1}\gamma_Q{}^i{}_0)( t^{\ep_1} g_{i0M})
    + (t^{\ep_1}\gamma_Q{}^0{}_0) m_{M}         \notag \\
    &+(t^{\ep_1}\gamma_Q{}^I{}_0) (\ell_{I0M}+\ell_{M0I})
    + \psi_Q{}^I{}_M m_I\Bigr)+t^{-\ep_1} \ftt_{3}^{(0),R0M}\Bigl(\psi_Q{}^i{}_R( t^{\ep_1} g_{i0M})
       \notag \\
       &+ (t^{\ep_1}\gamma_Q{}^0{}_0) m_{M}         
         + (t^{\ep_1}\gamma_Q{}^I{}_0) (\ell_{I0M}+\ell_{M0I})
    + \psi_Q{}^I{}_M m_I\Bigr) \notag\\
  &+t^{-\ep_1} \ftt_{4}^{(0),RLM}\Bigl(\psi_Q{}^i{}_R( t^{\ep_1} g_{i0M})
         + \psi_Q{}^0{}_L m_{M}
         + \psi_Q{}^I{}_L (\ell_{I0M}+\ell_{M0I})
   \notag \\
   &+ \psi_Q{}^I{}_M m_I\Bigr) +t^{-\ep_1} \ftt_{5}^{(0),rl}\Bigl( (t^{\ep_1}\gamma_Q{}^i{}_r)\xi_{il}
         + (t^{\ep_1}\gamma_Q{}^i{}_l)\xi_{ri}\Bigr)\notag\\
  &+ t^{-\ep_1} \Bigl(\ftt_{6}^{(0)} \beta (t^{\ep_1}\beta^{-1}D_Q\beta)
  -(1+c_s^2) \ftt_{11}^{(0)}\check\beta (t^{\ep_1}\beta^{-1}D_Q\beta)
  \notag \\
  &+\ftt_{12,i}^{(0)} (t^{\ep_1}\gamma_Q{}^i{}_j) W^j
  + \mathtt f^{(1),ij} (t^{\ep_1} g_{Qij})\Bigr) + t^{-\ep_4}\ftt_{12,i}^{(0)}\Utt^i_Q \notag\\
  &
         +t^{1-\ep_0-2\ep_1}\ftt_{1}^{(0),KL}\beta \Bigl(
         (t^{\ep_1}\gamma_Q{}^I{}_0) (\ell_{IKL}-\ell_{KIL}-\ell_{LIK})
         - \psi_Q{}^0{}_K (t^{\ep_1}g_{L00})
         \notag \\
&- \psi_Q{}^0{}_L (t^{\ep_1} g_{K00})
      \Bigr) + t^{-1+\ep_0}\ftt_{2}^{(0),M}\beta^{-1} \Bigl(\gac_{Q00M}-\psi_Q{}^0{}_M \delta^{JK}k_{JK}\notag \\
&+ (t^{\ep_1}\gamma_Q{}^I{}_0) k_{IM}\Bigr) + t^{-1+\ep_0}\ftt_{3}^{(0),R0M}\beta^{-1} \Bigl(\gac_{QR0M}
         + (t^{\ep_1}\gamma_Q{}^I{}_0) k_{IM}
      \notag \\
&-\psi_Q{}^0{}_M \delta^{JK}k_{JK}\Bigr) + t^{-1+\ep_0}\ftt_{4}^{(0),RLM}\beta^{-1} \Bigl(\gac_{QRLM}
         +\beta^{-1} \psi_Q{}^I{}_L) k_{IM}
     \notag \\
&-\beta^{-1}\psi_Q{}^0{}_M \delta^{JK}k_{JK}\Bigr)+ t^{-\ep_1-2\ep_0}\ftt_{5}^{(0),rl}\xi_{Qrl},    \label{eq:DerDecom}
\end{align}
  where $t^{\ep_1}g_{Qij}$ and $t^{\ep_1}\gamma_Q{}^i{}_j$ can be replaced by a constant coefficient, linear combination of $m_Q$, $\ell_{QiJ}$ and $\psi_Q{}^i{}_J$, and  
  \begin{equation}
    \label{grad-alpha-A.N}
    t^{\ep_1}\beta^{-1}D_Q\beta=-\beta\xi_{Q0} -m_{Q}+\frac{1}{2}\delta^{JK}(2\ell_{JKQ}-\ell_{QJK}).
  \end{equation}
\end{lem}
\noindent We omit the proof of this lemma since it follows directly from the relations \eqref{p-fields}, \eqref{d-fields}, \eqref{gi00}, \eqref{grad-alpha-A},
\eqref{gamma-I00}-\eqref{gamma-IJ0},
  \eqref{kt-def}-\eqref{W-def}, \eqref{psit-def},
  \eqref{k-def}-\eqref{Uac-def},
  \eqref{betacheck-defN} and \eqref{U-def}.

\begin{proof}[Proof of Lemma~\ref{lem:sourceterm}]  As discussed above, the map $F$, defined by \eqref{F-def} and \eqref{F1-def}-\eqref{F13-def}, can be decomposed, non-uniquely, as \eqref{F-expansion}. Here, we take the first sum in \eqref{F-expansion}
to be composed of all terms in $F$ that are \emph{explicitly}
proportional to $t^{-1}$, while all remaining terms, which are proportional to other
powers of $t$ (in particular those depending on $n$,
$c_s^2$, $\ep_1$,\ldots, $\ep_4$), make up the second sum. Over the course of this proof, we will establish that all exponents $\sigma_\ell$ in the second sum are smaller than one. It follows from Lemma~\ref{lem:TIsourceterm} that the
second sum in \eqref{F-expansion} in our expansion of $F$ can be written as
\begin{align}
  \sum_{\ell=1}^{N^{(1)}} t^{-\sigma_\ell}f^{(1)}_\ell\Bigl(U, w^{-1}, (W^0)^{-1}, w^{-\frac{1+c_s^2}{c_s^2}},  &\det(a^0{}_{ij})^{-1}\Bigr)\notag \\
  &=\sum_{\ell=1}^{\tilde N^{(1)}} t^{-\tilde\sigma^{(1)}_\ell}\tilde f^{(1)}_\ell(u) \label{F-expansion2TI}
  \end{align}
  with 
  \[\max_{\ell\in\{1,\ldots,\tilde N^{(1)}\}} \tilde\sigma^{(1)}_\ell = \max_{\ell\in\{1,\ldots, N^{(1)}\}} \sigma_\ell\]
  and $\tilde f^{(1)}_\ell\in C^\infty\bigl(B_R(\Rbb^{\udim})\bigr)$, $\ell=1,\ldots,\tilde N^{(1)},$ provided $R>0$ is chosen
sufficiently small. It is also a consequence of Lemma~\ref{lem:TIsourceterm} together with \eqref{U-def}, \eqref{eq:Ubgdef} and \eqref{u-def} and the fact that
$\ep_0,\ldots,\ep_4$ and $V_*^0$ are all positive that the first sum \eqref{F-expansion} in our expansion of $F$ can be expressed as
  \begin{align}
    \sum_{\ell=1}^{N^{(0)}} t^{-1}f^{(0)}_\ell\Bigl(U, w^{-1},& (W^0)^{-1}, w^{-\frac{1+c_s^2}{c_s^2}}, \det (a^0{}_{ij})^{-1}\Bigr) \notag\\
    &=\sum_{\ell=1}^{\tilde N^{(0)}} t^{-1}\tilde f^{(0)}_\ell(u)
    +\sum_{\ell=1}^{\tilde N^{(2)}} t^{-\tilde\sigma^{(2)}_\ell}\tilde f^{(2)}_\ell(u) \label{F-expansion2TNI}
\end{align}
where each $\tilde\sigma^{(2)}_\ell$ is smaller than
$1$, and $\tilde f^{(0)}_\ell, \tilde f^{(2)}_{\tilde\ell} \in
C^\infty\bigl(B_R(\Rbb^{\udim})\bigr)$ for $\ell=1,\ldots,\Nt^{(0)}$ and $\tilde{\ell}=1,\ldots,\Nt^{(2)}$, so long as $R>0$ is chosen
sufficiently small. A close inspection of \eqref{F1-def}-\eqref{F13-def} together with \eqref{U-def}, \eqref{eq:Ubgdef}, \eqref{u-def} and \eqref{eq:deta0exp} reveals, in addition, that
  \begin{equation*}
    f^{(0)}_\ell\Bigl(\mathring U, (V^0_*)^{-1}, (V^0_*)^{-1},(V^0_*)^{-\frac{1+c_s^2}{c_s^2}}, c_s^2(V_*^0)^{-1}\Bigr)=0
  \end{equation*}
  from which we deduce
  \begin{equation*}
    \tilde f^{(0)}_\ell(0)=\tilde f^{(2)}_{\tilde{\ell}}(0)=0, \quad  \ell=1,\ldots,\Nt^{(0)}, \; \tilde{\ell}=1,\ldots,\Nt^{(2)}.
  \end{equation*}
  By \eqref{eq:Ftt}, {and the expansions \eqref{F-expansion2TI} and \eqref{F-expansion2TNI} for $F$}, we observe that the map $H$ defined by \eqref{eq:SourcetermVan1}
  can be written as
  \begin{equation}
    \label{eq:decompH}
    H(t,u)
    =\sum_{\ell=1}^{\tilde N^{(1)}} t^{\tilde\ep-\tilde\sigma^{(1)}_\ell} \Bigl(\tilde f^{(1)}_\ell(u)-\tilde f^{(1)}_\ell(0)\Bigr) +\sum_{\ell=1}^{\tilde N^{(2)}} t^{\tilde\ep-\tilde\sigma^{(2)}_\ell}\tilde f^{(2)}_\ell(u).
  \end{equation}

Assuming for the moment that all constants $\sigma_\ell$ in \eqref{F-expansion2TI} are smaller than one, that is, that \eqref{eq:decompH} represents the time-integrable part of $F$, we observe that it would follow from Lemma~\ref{lem:TIsourceterm} that the map  $H$ would statisfy
  \eqref{eq:kdkjdsfkjsd309} and \eqref{eq:SourcetermVan2}
  provided 
  we choose
  \begin{equation}
    \label{eq:sourcetildeep}
    \max\{\tilde\ep^{(1)},\tilde\ep^{(2)}\}\le \tilde\ep<1,\quad
    \ep^{(1)}
    = \max_{\ell\in\{1,\ldots, N^{(1)}\}} \sigma_\ell,\quad
   \ep^{(2)}=\max_{\ell\in \tilde N^{(2)}} \tilde\sigma^{(2)}_\ell.
  \end{equation}
  It would also follow from  \eqref{eq:SourcetermVan1} and \eqref{F-expansion2TNI} that
  \begin{equation}
    \label{eq:hatHtt}
    %\Htt(u)\Pbb u =\sum_{\ell=1}^{\tilde N^{(0)}} \Pbb \tilde f^{(0)}_\ell(u)
    \sum_{\ell=1}^{\tilde N^{(0)}} \Pbb^\perp \tilde f^{(0)}_\ell(u)=\hat\Htt(u) ,
  \end{equation}
  and $\hat\Htt(u)$ defined this way would satisfy \eqref{eq:kdkjdsfkjsd309.3} and \eqref{eq:Gtprops} if each term of $\Pbb^\perp \tilde f^{(0)}_\ell(u)$ has at least two factors involving components of $\Pbb u$.
  Moreover, if each term of $\Pbb \tilde f^{(0)}_\ell(u)$ has at least one factor involving a component of $\Pbb u$ and one factor involving a component of $u$, then it is clear that there would exist a map $\Htt(u)$ satisfying \eqref{eq:kdkjdsfkjsd309.2}, \eqref{eq:SourcetermVan2},  
  \begin{equation}
    \label{eq:Htt}
    \sum_{\ell=1}^{\tilde N^{(0)}} \Pbb \tilde f^{(0)}_\ell(u)=\Htt(u)\Pbb u\AND \Pbb^\perp\Htt(u)=\Htt(u)\Pbb^\perp=0.
  \end{equation}
 The calculation 
  \[\Htt(u)\Pbb-\Pbb\Htt(u)=\Pbb(\Htt(u)\Pbb-\Htt(u))=-\Pbb\Htt(u)\Pbb^\perp=0\]
  shows that the map $\Htt(u)$ would also satisfy \eqref{eq:Gtprops}.

To proceed, we now turn to matching the components \eqref{F1-def}-\eqref{F13-def} of $F$ to terms of the type that appear in the sums \eqref{F-expansion2TI} and  \eqref{F-expansion2TNI}, and showing that the exponents $\sigma_\ell$ that appear in  \eqref{F-expansion2TI} are all less than $1$. 
%For the time-integrable terms, that is, terms proportional to $t^{-\sigma_\ell}$, $\sigma_\ell<1$, we will verify that the largest corresponding value of $\sigma_\ell$ is compatible with the condition for $\tilde\ep$ in \eqref{eq:epstildecond}. If a term is time-non-integrable, this is, involves a factor $t^{-1}$, we will verify that its contribution to the maps $\Htt$ or $\hat\Htt$ via \eqref{F-expansion2TNI} and \eqref{eq:HtthatHtt}  satisfies \eqref{eq:kdkjdsfkjsd309.2}-\eqref{eq:kdkjdsfkjsd309.3} and \eqref{eq:SourcetermVan2}-\eqref{eq:Gtprops}.
We begin by considering the components $F_1$ to $F_{10}$ of $F$ defined by
\eqref{F1-def}-\eqref{F10-def}. 
Here, we observe that these components can be decomposed into two groups of terms: those that depend on the fluid variables and those that do not, where we note that only $F_1$, $F_2$ and $F_{10}$ have terms that depend on the fluid variables via the expressions $t^{1-2\ep_0}\beta^2{\Ttt}_{(LM)}$, $t^{\ep_1-\ep_0}\beta{\Ttt}_{0M}$, $t^{1-2\ep_0}\beta^2 (\beta \xi +\psi+l+m)*
  {\Ttt}$ and $t^{1+\ep_1-2\ep_0}\beta^2\Dc_Q\Ttt_{lm}$.
It has been established previously in \cite[\S9.6]{BeyerOliynyk:2021} that the fluid  independent terms in $F_1$ to $F_{10}$ satisfy the properties asserted by Lemma~\ref{lem:sourceterm}. Consequently, to complete our analysis of the terms $F_1$ to $F_{10}$, we need only consider the terms
\begin{align}
  \label{eq:fluidterm.1}
  t^{1-2\ep_0}\beta^2{\Ttt}_{(LM)}
  &=t^{-1+\frac{n-1}{n-2}(1-c_s^2)-\ep_3-2\ep_0}
     \beta^2\hat\Ttt_{(LM)},\\
  \label{eq:fluidterm.2}
  t^{\ep_1-\ep_0}\beta{\Ttt}_{0M}  
  &=t^{-1+\frac{n-1}{n-2}(1-c_s^2)-\ep_3+\ep_1-\ep_0-1}
     \beta \hat\Ttt_{0M}
     \intertext{and}
  t^{1-2\ep_0}\beta^2 (\beta \xi& +\psi+l+m)*
  {\Ttt}
  \notag \\
  &= t^{-1+\frac{n-1}{n-2}(1-c_s^2)-\ep_3-2\ep_0}
     \beta^2 (\beta \xi +\psi+l+m)* \hat\Ttt, \label{eq:fluidterm.3.N}
\end{align}
where
\begin{align}
  \label{eq:fluidterm382304}
  \hat\Ttt_{ij}&=t^{\frac{(n-3)+(n-1)c_s^2}{n-2}+\ep_3}\Ttt_{ij}\notag \\
  &\oset{\eqref{eq:TttconfphysW}}{=}2P_0\Bigl(\frac{1+c_s^2}{c_s^2} w^{-2} W_iW_j+\frac{1-c_s^2}{(n-2)c_s^2}g_{ij}\Bigr) \check\beta w^{-\frac{1+c_s^2}{c_s^2}},
\end{align}
and the derivative
\begin{align}
  &t^{1+\ep_1-2\ep_0}\beta^2\Dc_Q\Ttt_{lm}
  =t^{-1+\frac{n-1}{n-2}(1-c_s^2)-\ep_3-2\ep_0+\ep_1}
     \beta^2 \Dc_Q\hat\Ttt_{lm} \notag\\
     &\qquad =t^{-1+\frac{n-1}{n-2}(1-c_s^2)-\ep_3-2\ep_0+\ep_1}
        \beta^2 \bigl(e_Q(\hat\Ttt_{lm})-\gamma_{Q}{}^i{}_l \hat\Ttt_{im}
        -\gamma_{Q}{}^i{}_m \hat\Ttt_{li}\bigr).\label{eq:fluidterm.3}
\end{align}
Noting that the inequalities
\begin{align*}\frac{n-1}{n-2}(1-c_s^2)-\ep_3-2\ep_0
  &>\frac{n-1}{n-2}(1-c_s^2)-\ep_3+\frac 13\ep_1-\frac 13-\ep_0\\   
  &>\frac{n-1}{n-2}(1-c_s^2)-\ep_3+\ep_1-\ep_0-1
  \end{align*}
and
\begin{align*}\frac{n-1}{n-2}(1-c_s^2)-\ep_3+\ep_1-\ep_0-1
  &>\frac{n-1}{n-2}(1-c_s^2)-\ep_3+\ep_1-\frac 13+\frac{\ep_1}3-1\\
  &>0
\end{align*}
hold due to \eqref{eq:epscond.N}, it follows that each $t$-power to the right of the equal signs in \eqref{eq:fluidterm.1}-\eqref{eq:fluidterm.3.N} is greater than negative one. Consequently, each of the maps \eqref{eq:fluidterm.1}-\eqref{eq:fluidterm.3.N} will be time-integrable, and therefore part of the expansion \eqref{F-expansion2TI}, which in turn, will imply that it is part of the map $H$ by \eqref{eq:decompH}. Moreover, 
since $\tilde\ep$ satisfies \eqref{eq:epstildecond} by assumption, it also satisfies \eqref{eq:sourcetildeep} from which it follows that this part of the map $H$ satisfies \eqref{eq:kdkjdsfkjsd309} and \eqref{eq:SourcetermVan2}.

Now, to analyse the derivative \eqref{eq:fluidterm.3}, we observe that each component of the map $\hat\Ttt_{ij}$ satisfies the requirements for the map $\mathtt f$ in Lemma~\ref{lem:DerivTerm},
and consequently 
 the derivative \eqref{eq:fluidterm.3} can be expressed using \eqref{eq:DerDecom} where we notice that the terms corresponding to $\mathtt f_{\Sc}^{[1]}$ in the chain rule \eqref{eq:DerDecomCR} vanish for all $\Sc=1,\ldots,13$ except for $\Sc=11$ and $\Sc=12$. Given that $t^{\ep_1}\gamma_Q{}^i{}_j$ can be replaced by a constant coefficient, linear combination of $m_Q$, $\ell_{QiJ}$ and $\psi_Q{}^i{}_J$ because of \eqref{gamma-I00} and \eqref{gamma-IJ0}, the time-integrability of  the map \eqref{eq:fluidterm.3} is then an immediate consequence of \eqref{eq:epscond.N}. Thus \eqref{eq:fluidterm.3} determines a part of the expansion \eqref{F-expansion2TI} and therefore a part of the map $H$ according to \eqref{eq:decompH}. Again  since \eqref{eq:sourcetildeep} holds on account of $\tilde\ep$ satisfying \eqref{eq:epstildecond},  this part of the $H$ map satisfies \eqref{eq:kdkjdsfkjsd309} and \eqref{eq:SourcetermVan2}.
This completes the analysis of the components $F_1$, \ldots, $F_{10}$.

We now turn to analysing the components $F_{11}$, $F_{12}$ and $F_{13}$ in \eqref{F11-def}-\eqref{F13-def}. We will not consider the component $F_{11}$ further because it is straightforward to verify that it satisfies the stated properties. Moving on to $F_{12}$, we see from  \eqref{eq:WQdr}, \eqref{eq:Gjsl.2}, \eqref{eq:fluidgamma.1}, \eqref{k-def}-\eqref{Uac-def} and \eqref{F12-def} that
all terms proportional to $t^{-1}$ in $F_{12}$
can be combined into the expression
  \begin{align}    
       -\delta_{\jt}^0\Biggl[&
         \frac{(n-1)c_s^2-1}{n-2} W^QW_Q
       +\frac 12c_s^2
       \delta^{PQ}k_{PQ} W_L W^L       
       \notag \\
       &-\Bigl(\frac 12-\frac{3 c_s^2+1}{c_s^2} \frac{W^IW_I}{w^2}\Bigr) k_{PQ} W^PW^Q
       \Biggr]\notag \\
       &
        -\frac 12\delta_{\jt}^J W^0\Bigl[ k_{SJ} -
          c_s^2  
          \delta^{PQ}k_{PQ}  \delta_{SJ}        
          \Bigr]W^S.  \label{eq:F12TNI}
  \end{align}
This part of $F_{12}$ generates an expansion of the form \eqref{F-expansion2TNI} where it can be checked
using  \eqref{U-def}, \eqref{eq:Ubgdef} and \eqref{u-def} that that no contribution corresponding to the second sum on the right hand side of \eqref{F-expansion2TNI} is generated. Moreover, using \eqref{Pbb-def}, it can be verified that the terms in the first line of \eqref{eq:F12TNI} yield contributions only to the map $\hat\Htt$, see \eqref{eq:hatHtt}, which satisfy \eqref{eq:Gtprops} as a consequence of \eqref{U-def}, \eqref{u-def}, \eqref{eq:Ubgdef} and \eqref{Pbb-def} on the one hand. On the other hand, we can use the same relationships together with \eqref{eq:Htt} to establish that the second line of \eqref{eq:F12TNI}
yields well-defined constributions only to the map $\Htt$ which are consistent with \eqref{eq:kdkjdsfkjsd309.2}-\eqref{eq:Gtprops}.
This covers all time-non-integrable terms of $F_{12}$.

Considering now all the remaining terms from $F_{12}$, which, as we will show are time-integrable, a lengthy calculation involving \eqref{eq:Gjsl.2} and \eqref{eq:fluidgamma.1} reveals that these terms can be expressed as
  \begin{align*}
    \delta_{\jt} ^j\betat &\Bigl[- \delta_{(l}^{0}\delta_{s)}^{S}\delta_j^{J}(\ellt_{S0J}+\ellt_{J0S})-      
\Bigl(\frac{3 c_s^2+1}{c_s^2}\frac{W^Q W^P}{w^{2}}+g^{PQ}\Bigr) \ellt_{P0Q}\delta_{(l}^0\eta_{s)j}\\
    &- \frac{1}{2}     
        \Bigl(\frac{3 c_s^2+1}{c_s^2} \frac{W^Q W^P}{w^{2}}+g^{PQ}\Bigr) \ellt_{LP Q}\delta_{(l}^L\eta_{s)j}- 2\delta_{l}^{0}\delta_{(s}^{0}\delta_{j)}^{J}\mt_{J}
       \\
       &- \delta_{(l}^{L}\delta_{s)}^{s'}\delta_j^{j'}g_{Ls'j'} - \frac{1}{2}     
        \Bigl(\frac{3 c_s^2+1}{c_s^2} \frac{(W^0)^2}{w^{2}}-1\Bigr) g_{L 0 0}\delta_{(l}^{L}\eta_{s)j}\\
    &- \frac{3 c_s^2+1}{c_s^2} \frac{W^0 W^P}{w^{2}} \eta_{j(s} g_{l)0 P}
    -c_s^2\Bigl(\frac{c_s^2+1}{c_s^2} \delta_{(l}^0  \eta_{s)j} -\delta_j^0\eta_{sl}\Bigr)
    \betat \tau_{00}\\ 
    &-c_s^2\betat^{-1} \Bigl(\frac{c_s^2+1}{c_s^2} \delta_{(l}^P  \eta_{s)j}\notag -\delta_j^P\eta_{sl}\Bigr)\Bigl(\betat^{2}\tau_{P0}+\frac{1}{2}\betat\bigl(2\mt_{P}-\delta^{JK}(2\ellt_{JKP} \\
    &-\ellt_{PJK})\bigr)\Bigr)\Bigr] W^sW^l - \delta_{\jt} ^j\betat \Bigl[-{a^0}_{j K}\delta^{KP}\ellt_{P0L}\delta_l^L-{a^0}_{j 0} \betat \tau_{0L}\delta_l^L
    \\
    &-{a^0}_{jK}\delta^{KP}(\betat \tau_{0P}+\mt_{P})\delta_l^0    
    +{a^I}_{j 0}(\mt_{I}-\frac{1}{2}\delta^{JK}(2 \ellt_{JKI}-\ellt_{IJK}) )\delta_l^0\notag\\
    &+{a^I}_{j 0}\psit_I{}^0{}_L\delta_l^L
    -{a^I}_{j K} (\delta^{KP}\ellt_{I0P} - \delta^{KP}\psit_{I}{}^0{}_P)\delta_l^0
      +{a^I}_{j K}\psit_I{}^K{}_L\delta_l^L\Bigr]W^l.
  \end{align*}
It then follows from \eqref{k-def}-\eqref{Uac-def}, \eqref{U-def}, \eqref{eq:Ubgdef}, \eqref{u-def}  and \eqref{eq:epscond.N} that these can be expanded as \eqref{F-expansion2TI} where all exponents $\sigma_\ell$ are smaller than $1$, and consequently, it contributes to the map $H$. As above, since \eqref{eq:sourcetildeep} holds due to  $\tilde\ep$ satisfying \eqref{eq:epstildecond} by assumption, this contribution to the map $H$ will satisfy \eqref{eq:kdkjdsfkjsd309} and \eqref{eq:SourcetermVan2}.

Finally, to complete the proof, we consider $F_{13}$, see \eqref{F13-def}, and observe that it can be written as
  \begin{align}   
      F_{13}= &\delta_{\jt} ^j \delta^{\Qt Q} \Biggl[     
          \bigl(2t^{-\ep_1}{a^0}_{j m}\beta \xi_{0Q}
        -2t^{-\ep_1}{a^I}_{j m}\psi_I{}^0{}_Q\bigr) (a^{0})^{-1}{}^{mn}\Bigl( t^{-\ep_0}\beta a^I{}_{nk}\acute{\Utt}^k_I\notag \\
        &-t^{\ep_4-\ep_0}\beta G_{nsl}W^sW^l\Bigr)  
 -t^{-\ep_0-\ep_1}\beta{a^0}_{j k}\delta^{PL}\ell_{Q0L}\acute{\Utt}_P ^k
 \notag \\
 &+t^{-\ep_0-\ep_1}\beta {a^I}_{j k}\psi_I{}^P{}_Q \acute{\Utt}_P ^k
-t^{\ep_4} (\beta^{-1}\Dc_Q\beta) (\betat G_{jsl}W^sW^l) \notag \\
&+ t^{-\ep_0}\beta \Dc_Q{a^0}_{j m}(a^{0})^{-1}{}^{mn} a^I{}_{nk}\acute{\Utt}^k_I          
           -t^{-\ep_0}\beta \Dc_Q{a^P}_{j k}\acute{\Utt}_P ^k \notag \\
    %%%
      &-\frac{1}{2} t^{-1}{a^0}_{j k}\delta^{PL}k_{QL}\acute{\Utt}_P ^k \notag -t^{-\ep_0}\beta{a^i}_{j k}\gamma_i{}^k{}_l {\acute{\Utt}}^l_Q
                      +t^{-\ep_0}\beta {a^i}_{j k}\gamma_i{}^L{}_Q {\acute{\Utt}}^k_L \notag \\
           & -t^{\ep_4-\ep_0}\Dc_Q{a^0}_{j m}(a^{0})^{-1}{}^{mn}(\beta
             G_{nsl}W^sW^l)
      +t^{\ep_4-\ep_0} \Dc_Q(\beta G_{jsl}W^sW^l)\notag \\
&-t^{-1} 
  \frac{(n-1)c_s^2-1}{n-2} \Bigl(       
  \eta_{sj}+\delta_{s}^0  \delta_{j}^0
  \Bigr)\,{\acute{\Utt}}^s_Q\,W^0
                    \Biggr].  \label{eq:F13new.5} 
  \end{align}
  Now, due to \eqref{eq:epscond.N}, we have that  $\ep_4>\ep_1$, $\ep_0+\ep_1<1$, $\ep_4>\ep_0$ and $\ep_1>\ep_0$, and so, we see from  \eqref{eq:B0fluid}, \eqref{eq:BIfluid}, \eqref{eq:Gjsl.2}, \eqref{k-def}-\eqref{Uac-def} and \eqref{grad-alpha-A.N}  that the terms
   \begin{align*} 
      \delta_{\jt} ^j \delta^{\Qt Q} \Biggl[      
      &
          \bigl(2t^{-\ep_1}{a^0}_{j m}\beta \xi_{0Q}
        -2t^{-\ep_1}{a^I}_{j m}\psi_I{}^0{}_Q\bigr) (a^{0})^{-1}{}^{mn}\Bigl( t^{-\ep_0}\beta a^I{}_{nk}\acute{\Utt}^k_I\\
        &-t^{\ep_4-\ep_0}\beta G_{nsl}W^sW^l\Bigr)    
                    -t^{-\ep_0-\ep_1}\beta{a^0}_{j k}\delta^{PL}\ell_{Q0L}\acute{\Utt}_P ^k
                          \notag \\
&+t^{-\ep_0-\ep_1}\beta {a^I}_{j k}\psi_I{}^P{}_Q \acute{\Utt}_P ^k
             -t^{\ep_4} (\beta^{-1}\Dc_Q\beta) (\betat G_{jsl}W^sW^l) \\  
           &+ t^{-\ep_0}\beta \Dc_Q{a^0}_{j m}(a^{0})^{-1}{}^{mn} a^I{}_{nk}\acute{\Utt}^k_I          
           -t^{-\ep_0}\beta \Dc_Q{a^P}_{j k}\acute{\Utt}_P ^k \Biggl]
  \end{align*}
  from \eqref{eq:F13new.5} can be expanded as \eqref{F-expansion2TI} with exponents $\sigma_\ell$ that are all smaller than $1$. Consequently, this part of $F_{13}$ is time-integrable, and therefore contributes to the map $H$. Again, since \eqref{eq:sourcetildeep} holds due to  $\tilde\ep$ satisfying \eqref{eq:epstildecond} by assumption, this contribution to the map $H$ will satisfy \eqref{eq:kdkjdsfkjsd309} and \eqref{eq:SourcetermVan2}.

Next, it is clear from  \eqref{U-def}, \eqref{eq:Ubgdef}
and \eqref{u-def}
that the term 
\begin{equation*}
\frac{1}{2}\delta_{\jt} ^j \delta^{\Qt Q}  t^{-1}{a^0}_{j k}\delta^{PL}k_{QL}\acute{\Utt}_P ^k
\end{equation*}
from \eqref{eq:F13new.5} can be expanded
as \eqref{F-expansion2TNI} without the second sum appearing on the right hand side, 
and consequently it is time-non-integrable. From the definition \eqref{Pbb-def}, it is also easy to verify that this term contributes to the map $\Htt$ in consistency with \eqref{eq:Htt} and that this contribution satisfies the properties \eqref{eq:kdkjdsfkjsd309.2}, \eqref{eq:SourcetermVan2} and \eqref{eq:Gtprops} as a consequence of
\eqref{Pbb-def},  \eqref{U-def},\eqref{eq:Ubgdef} and \eqref{u-def}.

Inspecting \eqref{eq:fluidgamma.1} and \eqref{eq:fluidgamma.3}, it is not difficult to verify using the variable definitions \eqref{k-def}-\eqref{Uac-def} that the term
\begin{equation*}
\delta_{\jt} ^j \delta^{\Qt Q}\Bigl[-t^{-\ep_0}\beta{a^i}_{j k}\gamma_i{}^k{}_l {\acute{\Utt}}^l_Q
                      +t^{-\ep_0}\beta {a^i}_{j k}\gamma_i{}^L{}_Q {\acute{\Utt}}^k_L\Bigr]
\end{equation*}
from \eqref{eq:F13new.5} contains both time-integrable and time-non-integrable terms contributing to the map $H$ and the map $\Htt$, respectively. Furthermore, since  $\tilde\ep$ satisfies  \eqref{eq:epstildecond} by assumption, which in turn implies that \eqref{eq:sourcetildeep} holds, these contribution to the maps $H$ and $\Htt$ are easily seen to satisfy the properties \eqref{eq:kdkjdsfkjsd309}-\eqref{eq:kdkjdsfkjsd309.2}. Additionally, from \eqref{Pbb-def}, \eqref{U-def}, \eqref{eq:Ubgdef} and \eqref{u-def}, we observe that these contributions to the maps
$H$ and $\Htt$
satisfy \eqref{eq:SourcetermVan2}-\eqref{eq:Gtprops}.

Next, we consider the term
\begin{equation*}
\delta_{\jt} ^j \delta^{\Qt Q}t^{\ep_4-\ep_0}\Dc_Q{a^0}_{j m}(a^{0})^{-1}{}^{mn}(\beta
             G_{nsl}W^sW^l)
\end{equation*}
from \eqref{eq:F13new.5}. With the help of  \eqref{eq:B0fluid}, \eqref{eq:Gjsl.2}, \eqref{k-def}-\eqref{Uac-def}, \eqref{eq:deta0exp}, and Lemma~\ref{lem:DerivTerm}, in particular \eqref{eq:DerDecom}, it then follows that the above expression can be separated in a collection of time-integrable and time-non-integrable terms where the former contribute to the map $H$ and the latter to $\Htt$.
Regarding the terms that contribute to $H$, it can be verified as above these terms will satisfy the properties \eqref{eq:kdkjdsfkjsd309} and \eqref{eq:SourcetermVan2} as a consequence of  \eqref{eq:sourcetildeep}, which holds since $\tilde\ep$ satisfies \eqref{eq:epstildecond} by assumption. On the other hand, the terms that contribute to $\Htt$ are, thanks to \eqref{Pbb-def}, \eqref{U-def}, \eqref{eq:Ubgdef} and \eqref{u-def}, easily shown to satisfies the properties
\eqref{eq:kdkjdsfkjsd309.2} \eqref{eq:SourcetermVan2}-\eqref{eq:Gtprops} with the possible exception 
of the term
\[-t^{-1}(a^{[1]}_{12,p})^{0}_{j m}(a^{0})^{-1}{}^{mn}\acute{\Utt}^p_Q \frac{(n-1)c_s^2-1}{n-2} \Bigl(       
  \delta_{(l}^0  \eta_{s)n}-\delta_n^0\eta_{sl}
  \Bigr) W^sW^l,\]
  which originates from the first term in \eqref{eq:Gjsl.2} after being expanded using \eqref{eq:DerDecom}.
The problem with this term is that it is not clear that its contribution to the $\Htt$ maps satisfies \eqref{eq:SourcetermVan2}.
However,
a straightforward calculation demonstrates that this term can be expressed as
\[-t^{-1}(a^{[1]}_{12,p})^{0}_{j m} (a^{0})^{-1}{}^{mn}\acute{\Utt}^p_Q \frac{(n-1)c_s^2-1}{n-2} \Bigl( \delta_n^N\delta_{NS}W^{(0}W^{S)}-\delta_n^0 W^SW_S
  \Bigr).\]
From this expression and the definitions \eqref{U-def}, \eqref{eq:Ubgdef} and \eqref{u-def},  it is then clear that this term is of the required form.

To complete the proof, we analyse the last
term
\begin{align} 
\delta_{\jt} ^j \delta^{\Qt Q}\Bigl[t^{\ep_4-\ep_0} &\Dc_Q(\beta G_{jsl}W^sW^l) \notag \\
&-t^{-1} 
  \frac{(n-1)c_s^2-1}{n-2} \Bigl(       
  \eta_{sj}+\delta_{s}^0  \delta_{j}^0
  \Bigr)\,{\acute{\Utt}}^s_Q\,W^0
                    \Bigr]\label{F13-last-term}
\end{align}
in \eqref{eq:F13new.5} starting with  $t^{\ep_4-\ep_0}\Dc_Q(\beta G_{jsl}W^sW^l)$.
Using \eqref{k-def}-\eqref{Uac-def}, we can express $G_{jsl}$, see \eqref{eq:Gjsl.2}, as
\begin{equation}
  \label{eq:FluidGbSplit}
  G_{jsl}=t^{-1}\betat^{-1}\hat G_{jsl}+t^{-\ep_1}\check G_{jsl}
\end{equation}
where $\hat G_{jsl}$ consists of all the time-non-integrable terms that are obtained from the first three lines of \eqref{eq:Gjsl.2} while all the remaining terms, which are time-integrable, are denoted by $\check G_{jsl}$.

Differentiating \eqref{eq:FluidGbSplit}, we find that
\begin{align}
  t^{\ep_4}\Dc_Q&(\betat G_{jsl}W^sW^l)
  =  t^{-1+\ep_4}e_Q (\hat G_{jsl})\,W^sW^l
     +2t^{-1} \hat G_{jsl}\,{\acute{\Utt}}^s_QW^l \notag \\
  %\label{eq:skcmoevmvppp3.2}
  &-t^{-1+\ep_4-\ep_1}(t^{\ep_1}\gamma_Q{}^i{}_j) \hat G_{isl}\,W^sW^l
    -t^{-1+\ep_4-\ep_1}(t^{\ep_1}\gamma_Q{}^i{}_s) \hat G_{jil}\,W^sW^l
   \notag \\
   &-t^{-1+\ep_4-\ep_1}(t^{\ep_1}\gamma_Q{}^i{}_l) \hat G_{jsi}\,W^sW^l +2t^{-\ep_1-\ep_0}\beta \check G_{jsl}\,{\acute{\Utt}}^s_QW^l
    \notag \\
    &+t^{\ep_4-\ep_1}\Dc_Q (\betat \check G_{jsl})\,W^sW^l-t^{\ep_4-2\ep_1}(t^{\ep_1}\gamma_Q{}^i{}_j) (\betat \check G_{isl})\,W^sW^l \notag \\
    &-t^{\ep_4-2\ep_1}(t^{\ep_1}\gamma_Q{}^i{}_s) (\betat \check G_{jis})\,W^sW^l.  \label{eq:skcmoevmvppp3.4}
\end{align}
Taking into account \eqref{eq:epscond.N} and the fact that $t^{\ep_1}\gamma_Q{}^i{}_j$ can be replaced by a constant coefficient, linear combination of $m_Q$, $\ell_{QiJ}$ and $\psi_Q{}^i{}_J$ as a consequence of \eqref{gamma-I00} and \eqref{gamma-IJ0},
it then straightforward, with the help of Lemma~\ref{lem:DerivTerm}, to verify that all of the terms that last three lines of \eqref{eq:skcmoevmvppp3.4} are time-integrable
and can be expanded as \eqref{F-expansion2TI}
with exponents $\sigma_\ell$ less than one. In particular, these terms contribute to the map $H$ and, for the same reasons as noted above, will satisfy the properties
\eqref{eq:kdkjdsfkjsd309} and \eqref{eq:SourcetermVan2}.

Similarly, it is also straightforward to verify, again using  \eqref{eq:epscond.N} and Lemma~\ref{lem:DerivTerm}, that the terms on the first line of the right hand side of \eqref{eq:skcmoevmvppp3.4} are time-integrable
and can be expanded as \eqref{F-expansion2TI}
with exponents $\sigma_\ell$ less than one
except for the terms 
\begin{align}
  t^{-1} (\hat G_{jsl})_{12,i}^{(0)}\acute{\Utt}^i_Q \,W^sW^l&
    + 2 t^{-1} \hat G_{jsl}\,{\acute{\Utt}}^s_QW^l
    =
    2 t^{-1} (\hat G_{j0L})_{12,i}^{(0)}\acute{\Utt}^i_Q \,W^0W^L
      \notag \\
    &+t^{-1} (\hat G_{jSL})_{12,i}^{(0)}\acute{\Utt}^i_Q \,W^SW^L 
      + 2 t^{-1}\hat G_{jsL}\,{\acute{\Utt}}^s_QW^L
      \notag \\
    &+t^{-1} (\hat G_{j00})_{12,i}^{(0)}\acute{\Utt}^i_Q \,(W^0)^2 
    + 2 t^{-1}\hat G_{js0}\,{\acute{\Utt}}^s_QW^0.  \label{eq:dkjask330303030}
\end{align}
Moreover, the time-integrable terms contribute to $H$ and satisfy \eqref{eq:kdkjdsfkjsd309} and \eqref{eq:SourcetermVan2}.
We futher note that the map $(\hat G_{jsl})_{12,i}^{(0)}$ that appears in \eqref{eq:dkjask330303030} corresponds to the map $\ftt_{12}^{(0)}$ from Lemma~\ref{lem:DerivTerm} that is generated from the derivative $e_Q(\hat{G}_{jsl})$ in the first term on the right hand side of \eqref{eq:skcmoevmvppp3.4}. 

Due to the $t^{-1}$ singular terms in \eqref{eq:dkjask330303030}, this expression can be expanded as \eqref{F-expansion2TNI}. As can be readily verified using \eqref{U-def}, \eqref{eq:Ubgdef} and \eqref{u-def}, the expansion only contains terms of the type that appear in the first sum on the right hand side of  \eqref{F-expansion2TNI}, that is, the second sum is absent. In order to verify that this yields a contribution to the map $\Htt$,
c.f.~\eqref{eq:Htt}, that satisfies \eqref{eq:kdkjdsfkjsd309.2}, \eqref{eq:SourcetermVan2} and \eqref{eq:Gtprops} we must establish that each term in \eqref{eq:dkjask330303030} has one factor of $\Pbb u$ (which is obvious by \eqref{U-def}, \eqref{Pbb-def}, \eqref{eq:Ubgdef} and \eqref{u-def}) in addition to a factor of $u$.
%
% By \eqref{eq:HtthatHtt}, this sum would contribute to the map $\Htt$, and by \cnote{red}{\eqref{Uac-def}} and \eqref{eq:dkjask330303030}, each term in the sum would contain at least one factor of  ${\acute{\Utt}}_Q^{i}$,  which we note by \eqref{U-def}, \eqref{Pbb-def}, \eqref{eq:Ubgdef} and \eqref{u-def} is equivalent to a component of $\Pbb u$.  
%
% In order to verify this contribution to the map $\Htt$ has the property \eqref{eq:SourcetermVan2}, we must show that each term in \eqref{eq:dkjask330303030} has another factor of $u$. 
This is,  by \eqref{U-def}, \eqref{eq:Ubgdef} and \eqref{u-def},  obvious for the first three terms on the right side of \eqref{eq:dkjask330303030}, but not for the last two terms. Indeed, 
a direct calculation involving \eqref{eq:Gjsl.2}, \eqref{k-def}-\eqref{Uac-def}, \eqref{eq:Ubgdef} and \eqref{u-def} reveals that these two terms yield precisely one term of the form
that does \emph{not} involve a product of $\Pbb u$ and $u$, given by
\begin{equation*} t^{-1} 
  \frac{(n-1)c_s^2-1}{n-2} \Bigl(       
  \eta_{sj}+\delta_{s}^0  \delta_{j}^0
  \Bigr)\,{\acute{\Utt}}^s_Q\,W^0.
  \end{equation*}
  However, this term precisely cancels the second term in \eqref{F13-last-term}. So in total, \eqref{eq:dkjask330303030} results in a contribution to the map $\Htt$ that satisfies all the required properties. 
This completes the proof of 
Lemma~\ref{lem:sourceterm}.
\end{proof}

\section{Past global FLRW stability \label{global-sec}}

In this section, we turn to establishing the past stability of the FLRW solutions \eqref{eq:FLRWEulerSFExplSol1}-\eqref{eq:FLRWEulerSFExplSol4} and their big bang singularities. The precise statement of our past stability result is given below in Theorem \ref{glob-stab-thm}.
The proof of this theorem is carried out in two steps. The first step, detailed in Section~\ref{sec:proof_globstab_Fuchsian}, involves establishing the stability of the solution $\breve u=\breve U-\mathring U$ of the Fuchsian equation \eqref{Fuch-ev-A2}, where $\breve U$, defined by \eqref{eq:UbreveFirst}-\eqref{eq:UbreveLast}, corresponds to the FLRW solution \eqref{eq:FLRWEulerSFExplSol1}-\eqref{eq:FLRWEulerSFExplSol4} and $\mathring U$ is defined by \eqref{eq:Ubgdef}. By \eqref{Ubr-Ur}, we
can choose $t_0>0$ small enough to ensure that $\breve u(t)$ remains arbitrarily close to $u\equiv 0$ on the time interval $(0,t_0]$. The stability of the trivial  $u\equiv 0$ to the Fuchsian equation \eqref{Fuch-ev-A2} therefore implies the stability of $\breve u$, which we note, in turn, implies the stability of $\breve U$. Section~\ref{sec:proof_globstab_Fuchsian} contains both the statement of the Fuchsian stability result for the trivial solution $u\equiv 0$, see Proposition~\ref{prop:globalstability}, as well as its proof.
In the second step, which is carried out in Section~\ref{sec:proof_globstab}, we use the stability result from Proposition~\ref{prop:globalstability} in conjunction with the local-in-time existence and continuation theory from Proposition \ref{lag-exist-prop} to complete the proof of Theorem \ref{glob-stab-thm}.

\subsection{The past global stability theorem}
\label{sec:glob-stab-thm}

Theorem \ref{glob-stab-thm} below establishes the nonlinear stability of the Einstein-Euler-scalar field FLRW  solutions \eqref{eq:FLRWEulerSFExplSol1}-\eqref{eq:FLRWEulerSFExplSol4} on $M_{0,t_0}=(0,t_0]\times \Tbb^{n-1}$ for some $t_0>0$
by guaranteeing that sufficiently small perturbations of FLRW initial data, see Remark \ref{FLRW-idata-rem-A}, 
which also satisfies the gravitational and wave gauge constraints as well as the synchronization condition $\tau|_{\Sigma_{t_0}}=t_0$, will generate solutions of conformal Einstein-Euler-scalar field equations on $M_{0,t_0}$ that are asymptotically Kasner in the sense of Definition \ref{def:APKasner} provided the speed of sound parameter $c_s^2$ is bounded by $1/(n-1)<c_s^2 < 1$. 
As discussed in Section \ref{temp-synch}, if the initial data does not satisfy the synchronization condition  $\tau|_{\Sigma_{t_0}}=t_0$, then it can be evolved for short amount of time so that it does, and consequently, we lose no generality by assuming that the initial data is synchronized.

\begin{thm}[Past global stability of the FLRW solution of the Einstein-Euler-scalar field system]\label{glob-stab-thm}

Suppose that $n\in\Zbb_{\ge 3}$, $k \in \Zbb_{>(n+3)/2}$, $P_0>0$, $V_*^0>0$, $c_s^2\in ( 1/(n-1),1)$, $\sigma>0$, and let $\{\breve g_{\mu\nu},\breve\tau,\breve V^\mu\}$ denote the FLRW solution \eqref{eq:FLRWEulerSFExplSol1}-\eqref{eq:FLRWEulerSFExplSol4} determined by the constants $n$, $P_0$, $V_*^0$ and $c_s^2$. Then there exist constants $\delta_0, t_0>0$ such that for every $\delta\in (0,\delta_0]$, $\gr_{\mu\nu}\in H^{k+2}(\Tbb^{n-1},\mathbb{S}_n)$, 
$\ggr_{\mu\nu}\in H^{k+1}(\Tbb^{n-1},\mathbb{S}_n)$, $\taur=t_0$,
$\taugr\in H^{k+2}(\Tbb^{n-1})$ and $\Vr^\mu\in H^{k+1}(\Tbb^{n-1})$
satisfying
\begin{align} 
    &\norm{\gr_{\mu\nu}-\breve g_{\mu\nu}(t_0)}_{H^{k+2}(\Tbb^{n-1})}
  +\norm{\ggr_{\mu\nu}-\del{t}\breve g_{\mu\nu}(t_0)}_{H^{k+2}(\Tbb^{n-1})}
  \notag \\ 
  &\hspace{1.5cm}+\norm{\taugr-1}_{H^{k+2}(\Tbb^{n-1})}
  +\norm{\Vr^\mu-\breve V^\mu(t_0)}_{H^{k+1}(\Tbb^{n-1})}  
  <\delta \label{glob-stab-thm-idata-A}
\end{align}
and the gravitational and wave gauge constraints \eqref{grav-constr}-\eqref{wave-constr}, there exists a unique classical solution $\Wsc\in C^1(M_{0,t_0})$, see \eqref{eq:Wdef}, 
of the system of evolution equations \eqref{tconf-ford-C.1}-\eqref{tconf-ford-C.9} on $M_{0,t_0}=(0,t_0]\times \Tbb^{n-1}$ with regularity
\begin{equation}
  \label{eq:WregFin}
\Wsc \in \bigcap_{j=0}^{k}C^j\bigl((0,t_0], H^{k-j}(\Tbb^{n-1})\bigr)
\end{equation} 
that satisfies the corresponding initial
conditions \eqref{l-idata}-\eqref{hhu-idata} on $\Sigma_{t_0}=\{t_0\}\times\Tbb^{n-1}$ and
the constraints \eqref{eq:Lag-constraints} in $M_{0,t_0}$.

\medskip

\noindent Moreover, the triple $\{g_{\mu\nu}=\del{\mu}l^\alpha\ghu_{\alpha\beta}\del{\nu}l^\beta,\tau=t, V^\mu=\Jcch^\mu_\nu \Vhu^\nu\}$, which is uniquely determined by $\Wsc$, defines a solution of the conformal Einstein-Euler-scalar field equations \eqref{lag-confeqns} on $M_{0,t_0}$ that satisfies the wave gauge constraint \eqref{lag-wave-gauge} and the following properties: 
\begin{enumerate}[(a)]
  \item Let $e_0^\mu=\betat^{-1}\delta_0^\mu$ with $\betat= (-g(dx^0,dx^0))^{-\frac{1}{2}}$, and $e^\mu_I$ be the unique solution of the Fermi-Walker transport equations \eqref{Fermi-A} with initial conditions $e^\mu_I |_{\Sigma_{t_0}}=\delta^\mu_\Lambda\er^\Lambda_I$ where the functions $\er^\Lambda_I \in H^k(\Tbb^{n-1})$ are chosen to satisfy $\norm{\er^\Lambda_I-(\omega|_{t=t_0})^{-1}\delta^\Lambda_I}_{H^{k}(\Tbb^{n-1})} < \delta$
and make the frame $e^\mu_i$ orthonormal on $\Sigma_{t_0}$.
Then $e^\mu_i$ is a well defined frame in $M_{0,t_0}$ that
satisfies
$e^0_I=0$ and $g_{ij}=\eta_{ij}$
where $g_{ij}=e_i^\mu g_{\mu\nu}e^\nu_j$. 
\item There exists a tensor field $\kf_{IJ}$ in $H^{k-1}(\Tbb^{n-1},\Sbb{n-1})$ satisfying 
\begin{equation*}
    \norm{\kf_{IJ}}_{H^{k-1}(\Tbb^{n-1})}\lesssim\delta+t_0^{\frac{n-1}{n-2}(1-c_s^2)}
\end{equation*} such that 
  \begin{align}
\norm{t\betat\Dc_0 g_{00}+\delta^{JK}\kf_{JK}}_{H^{k-1}(\mathbb T^{n-1})}&\lesssim
t^{\frac 18\frac{n-1}{n-2}(1-c_s^2)-\sigma}\notag \\
&\quad +t^{2 \bigl(\frac{(n-1)c_s^2-1}{n-2}-\sigma\bigr)}, \label{eq:glostab-est.First}\\
\norm{t\betat\Dc_0 g_{JK}-\kf_{JK}}_{H^{k-1}(\mathbb T^{n-1})}
     &\lesssim  t^{\frac 18\frac{n-1}{n-2}(1-c_s^2)-\sigma}\notag \\
     & +t^{2 \bigl(\frac{(n-1)c_s^2-1}{n-2}-\sigma\bigr)},  \label{eq:glostab-est.2}\\
 \norm{\Dc_I g_{00}}_{H^{k-1}(\mathbb T^{n-1})}+\norm{\Dc_0 g_{J0}}_{H^{k-1}(\mathbb T^{n-1})}\qquad & \notag \\
 +\norm{\Dc_I g_{J0}}_{H^{k-1}(\mathbb T^{n-1})}+\norm{\Dc_I g_{JK}}_{H^{k-1}(\mathbb T^{n-1})} &\lesssim t^{-1+\frac 58\frac{n-1}{n-2}(1-c_s^2)-2\sigma}\notag \\
&  + t^{-\frac 12\frac{n-1}{n-2}(1-c_s^2)-\sigma},  \label{eq:glostab-est.3}\\
\label{eq:glostab-est.4}
\norm{t\betat\Dc_I\Dc_j g_{kl}}_{H^{k-1}(\mathbb T^{n-1})}&\lesssim t^{-1+\frac 58\frac{n-1}{n-2}(1-c_s^2)-2\sigma}\notag \\
&+ t^{-\frac 12\frac{n-1}{n-2}(1-c_s^2)-\sigma},\\
\label{eq:glostab-est.5}
\norm{\Dc_i\Dc_j\tau }_{H^{k-1}(\mathbb T^{n-1})}&\lesssim t^{-1+\frac 34\frac{n-1}{n-2}(1-c_s^2)-2\sigma}\notag\\
&+ t^{-\frac 38\frac{n-1}{n-2}(1-c_s^2)-\sigma},\\
\label{eq:glostab-est.Last}
\norm{\Dc_I\Dc_j\Dc_k\tau }_{H^{k-1}(\mathbb T^{n-1})}&\lesssim t^{-2+\frac{n-1}{n-2}(1-c_s^2)-3\sigma}\notag \\
+& t^{-1-\frac 18\frac{n-1}{n-2}(1-c_s^2)-2\sigma},
  \end{align}
for all $t\in (0,t_0]$,
where all the
fields in these estimates are expressed in terms of the frame $e_i^\mu$ and the Levi-Civita connection $\Dc$ of the flat background metric
$\gc_{\mu\nu}=\del{\mu}l^\alpha \eta_{\alpha\beta} \del{\nu}l^\beta$. 
In addition, there exist a strictly positive function $\mathfrak{b}\in H^{k-1}(\Tbb^{n-1})$, a matrix $\ef^\Lambda_J\in H^{k-1}(\Tbb^{n-1},\Mbb{n-1})$, and a constant $C>0$ such that
  \begin{align}
       \Bnorm{t^{-\kf_{J}{}^J/2}\betat
    -\mathfrak{b}}_{H^{k-1}(\Tbb^{n-1})}
 & \lesssim t^{\frac 18\frac{n-1}{n-2}(1-c_s^2)-\sigma}\notag \\
 &\quad +t^{2 \bigl(\frac{(n-1)c_s^2-1}{n-2}-\sigma\bigr)} \label{eq:improvbetaestimate-glob}
\intertext{and}
  \Bnorm{\exp\Bigl(\frac 12\ln(t)\kf_{J}{}^I\Bigr) e_I^\Lambda-\ef_J^\Lambda}_{H^{k-1}(\Tbb^{n-1})}    
  &\lesssim t^{\frac 18\frac{n-1}{n-2}(1-c_s^2)-\sigma -C\delta} \notag \\
  &\quad+t^{2 \bigl(\frac{(n-1)c_s^2-1}{n-2}-\sigma\bigr) -C\delta} \label{eq:improvframeestimate}
\end{align}
 for all $t\in (0,t_0]$, where $\kf_{L}{}^J=\kf_{LM}\delta^{MJ}$.
\item The second fundamental form $\Ktt_{\Lambda\Omega}$ induced on the constant time surface $\Sigma_{t}=\{t\}\times\Tbb^{n-1}$  by $g_{\mu\nu}$ satisfies
\begin{equation}
\label{eq:2ndFF.est}
\Bnorm{2 t\betat \Ktt_{LJ}- \kf_{LJ}}_{H^{k-1}(\mathbb T^{n-1})}
\lesssim t^{\frac 18\frac{n-1}{n-2}(1-c_s^2)-\sigma}+t^{2 \bigl(\frac{(n-1)c_s^2-1}{n-2}-\sigma\bigr)}
\end{equation}
for all $t\in (0,t_0]$, while the lapse, shift and the spatial metric on $\Sigma_{t}$ are determined by $\Ntt=\betat$,
$\btt_\Lambda=0$, and $\gtt_{\Lambda\Omega}=g_{\Lambda\Omega}$, respectively.
\item The triple $\Bigl\{\gb_{\mu\nu}=t^{\frac {2}{n -2}}g_{\mu\nu},\phi=\sqrt{\frac{n-1}{2(n-2)}}\ln(t), \Vb^\mu=V^\mu\Bigr\}$ defines a solution of the physical Einstein-Euler-scalar field equations \eqref{ESF.1}-\eqref{eq:AAA1} on $M_{0,t_0}$ that exhibits AVTD behaviour and is asymptotically pointwise Kasner on $\Tbb^{n-1}$ with Kasner exponents $r_1(x),\ldots,r_{n-1}(x)$ determined by the eigenvalues of $\kf_{L}{}^J(x)$ for each $x\in \Tbb^{n-1}$. In particular, $\kf_{L}{}^L(x)\ge 0$ for all $x\in\Tbb^{n-1}$, and $\kf_{L}{}^L(x)=0$ for some $x\in \Tbb^{n-1}$ if and only if $r_1(x)=\ldots=r_{n-1}(x)=0$. The time $t=0$ represents a crushing singularity in the sense of \cite{eardley1979}. Furthermore, the function
  \begin{equation*}
%  \label{eq:pdef}
    \Ptt
=\frac{\sqrt{{2(n-1)}(n-2)}}{{2(n-1)} +(n-2) \kf_{L}{}^L}  
\end{equation*}
can be interpreted as the asymptotic scalar field strength in the sense of Section~\ref{sec:AVTDAPK}.
\item The physical solution $\bigl\{\gb_{\mu\nu},\phi,\Vb^\mu\bigr\}$ is past $C^2$ inextendible at $t=0$ and past timelike geodesically incomplete. The scalar curvature  $\Rb=\Rb_{\mu\nu}\gb^{\mu\nu}$ of the physical metric $\gb_{\mu\nu}$ 
satisfies
  \begin{align}
    \Bnorm{t^{2\frac {n-1}{n -2}+\kf_{J}{}^J}\Rb+\frac{n-1}{n-2}\mathfrak{b}^{-2}}_{H^{k-1}(\mathbb T^{n-1})}\lesssim& t^{\frac 18\frac{n-1}{n-2}(1-c_s^2)-\sigma -C\delta}\notag \\
    &+t^{2 \bigl(\frac{(n-1)c_s^2-1}{n-2}-\sigma\bigr) -C\delta},
     \label{eq:SptRicciEstimat-glob}
  \end{align}
and consequently, it blows up pointwise everywhere on the hypersurface $t=0$.
\item The fluid variables $\Vb^\mu$ of the physical solution $\bigl\{\gb_{\mu\nu},\phi, \Vb^\mu\bigr\}$ can be expressed as
\begin{equation}
  \label{eq:fluidresult1}
    \Vb^\mu=t^{\frac{n-1}{n-2}c_s^2} \betat^{c_s^2}\Bigl(W^0 \eb_0^\mu+W^I \eb_I^\mu\Bigr)
  \end{equation}
  in terms of the physical orthonormal frame $\eb_i^\mu=t^{-1/(n-2)}e_i^\mu$ where 
  \begin{equation*}
  W^0, W^I \in \bigcap_{j=0}^{k}C^j\bigl((0,t_0], H^{k-j}(\Tbb^{n-1})\bigr)
  \end{equation*}
  and there exists a positive function $\Wf^0\in H^{k-1}(\mathbb T^{n-1})$ bounded by
  \begin{equation*}
    \norm{\Wf^0}_{H^{k-1}(\mathbb T^{n-1})}\lesssim \delta+t_0^{\frac{n-1}{n-2}(1-c_s^2)}
  \end{equation*}
  such that $W^0,W^I$ satisfy
  \begin{align}
    \norm{W^0(t)-(V^0_*+\Wf^0)}_{H^{k-1}(\Tbb^{n-1})}&\lesssim t^{\frac 18\frac{n-1}{n-2}(1-c_s^2)-\sigma}\notag \\
    +&t^{2 \bigl(\frac{(n-1)c_s^2-1}{n-2}-\sigma\bigr)}, \label{eq:fluidestimate1}\\
    \norm{W^I(t)}_{H^{k-1}(\Tbb^{n-1})}&\lesssim t^{\frac 18\frac{n-1}{n-2}(1-c_s^2)-\sigma}\notag\\
    &\quad+t^{\frac{(n-1)c_s^2-1}{n-2}},\label{eq:fluidestimate2}\\
   \norm{\Dc_QW^0(t)}_{H^{k-1}(\Tbb^{n-1})}+\norm{\Dc_QW^I(t)}_{H^{k-1}(\Tbb^{n-1})}
&\lesssim t^{-1+\frac 7{16}\frac{n-1}{n-2}(1-c_s^2)-\sigma}\notag \\
+&t^{-\frac {11}{16}\frac{n-1}{n-2}(1-c_s^2)-\sigma} \label{eq:fluidestimate3} 
  \end{align}
  for all $t\in (0,t_0]$. The fluid density $\rho$ satisfies
  \begin{align}
  \Bnorm{ t^{\frac{n-1}{n-2}(1+c_s^2)+\frac{\kf_{J}{}^J(1+c_s^2)}2} \rho(t)- &\frac{P_0}{c_s^2} \mathfrak{b}^{-(1+c_s^2)}(V^0_*+\Wf^0)^{-\frac {1+c_s^2}{c_s^2}}}_{H^{k-1}(\Tbb^{n-1})}\notag \\
  &\lesssim
  t^{\frac 18\frac{n-1}{n-2}(1-c_s^2)-\sigma}+t^{2 \bigl(\frac{(n-1)c_s^2-1}{n-2}-\sigma\bigr)}  \label{eq:fluidpressureestimate}
  \end{align}
  for all $t\in (0,t_0]$. The physical normalised fluid $n$-velocity is
  \begin{equation}
    \label{eq:physical4vectorfield}
    \ub^\mu=\frac{W^0}{w} \eb_0^\mu+ \frac{W^I}{w} \eb_I^\mu,\quad w=\sqrt{(W^0)^2-W^IW_I},
  \end{equation}
  where
  \begin{align}
  \biggl\|\frac{W^0(t)}{w(t)}-1\biggr\|_{H^{k-1}(\Tbb^{n-1})}
  +\biggl\|&\frac{W^I(t)}{w(t)}\biggr\|_{H^{k-1}(\Tbb^{n-1})} \notag\\
 & \lesssim t^{\frac 18\frac{n-1}{n-2}(1-c_s^2)-\sigma}
  +t^{\frac{(n-1)c_s^2-1}{n-2}}, \label{eq:physical4vectorfield.est}
\end{align}
for all $t\in (0,t_0]$, and therefore $\ub^\mu$ agrees with $\eb_0^\mu$ asymptotically at $t=0$.
\end{enumerate}

\noindent The implicit and explicit constants in the above estimates are all independent of the choice of $\delta\in (0,\delta_0]$.
\end{thm}

Before considering the proof of this theorem, we make some observations.
\begin{rem}\label{glob-stab-rem} 
$\;$

\begin{enumerate}[(i)] 
\item The constants $t_0,\delta_0$ that appear in the theorem depend on the the constants $V_*^0,P_0$ from the FLRW solution \eqref{eq:FLRWEulerSFExplSol1}-\eqref{eq:FLRWEulerSFExplSol4}. Since we do not track the explicit dependence of $t_0,\delta_0$ on the choice of these constants, all we know from Theorem \ref{glob-stab-thm} is that for a particular choice of $V_*^0,P_0$ there exist sufficiently small positive constants  $t_0,\delta_0$ that guarantee the past stability of the FLRW solutions. While this may seem to be a restriction on the choice of the time interval $(0,t_0]$, it is easy to see with the help of a Cauchy stability argument, see Proposition~\ref{lag-exist-prop}.(a), that we can take $t_0>0$ as large as we like in Theorem \ref{glob-stab-thm} as long as the initial data satisfies \eqref{glob-stab-thm-idata-A} for a suitably small choice of $\delta_0>0$. 
\item The decay rates, i.e., the $t$-exponents in the estimates from Theorem~\ref{glob-stab-thm}, are most likely not optimal. Improved decay rates can be obtained through an adaptive choice of the parameters $\ep_0,\ldots,\ep_4$  as was carried out in \cite{BeyerOliynyk:2021}, but we have decided against doing this here because is would significantly increase the complexity of the proof and to achieve optimal results would require additional arguments along the lines of those employed in \cite{BeyerOliynyk:2020}.

\item It is worth noticing that the estimate \eqref{eq:fluidpressureestimate} can be used to distinguish the blow-up of the FLRW background fluid from that of a generic perturbation. The relations \eqref{eq:fluidBG2} and \eqref{eq:FLRWexpansions} yield
    \[\rho_{bg}(t)=\rho_* t^{-\frac{n-1}{n-2}(1+c_s^2)}
      +\Ord\bigl(t^{-2\frac{n-1}{n-2}c_s^2}\bigr)\]
for the FLRW background,
where $\rho_*$ can be interpreted as the background fluid density parameter determined by $V_*^0$ via \eqref{eq:fluidBG1}. Now by\eqref{eq:fluidpressureestimate}, the density of a generic perturbation behaves like
\[\rho(t)=\tilde\rho_* t^{-\frac{n-1}{n-2}(1+c_s^2)-\frac{\kf_{J}{}^J(1+c_s^2)}2}+\Ord\Bigl(t^{\frac 18\frac{n-1}{n-2}(1-c_s^2)-\sigma}+t^{2 \bigl(\frac{(n-1)c_s^2-1}{n-2}-\sigma\bigr)}\Bigr),\]
with $\tilde\rho_*$ close to $\rho_*$. By Theorem \ref{glob-stab-thm}.(d), the trace $\kf_{J}{}^J(x)$ is a non-negative function which is zero at a point $x$ if and only if the full matrix $\kf_{I}{}^J(x)$ is zero, and since we expect that this matrix is non-zero at least somewhere for a generic perturbation of the FLRW solution, it follows that the blow-up profile of the perturbed fluid density must differ significantly from that of the FLRW density. 
\item Given that the vector field $e_0$, see Theorem \ref{glob-stab-thm}.(a), and therefore $\eb_0$, see Theorem \ref{glob-stab-thm}.(f), is orthogonal to the scalar-field synchronised $t=const$ surfaces, the formula \eqref{eq:physical4vectorfield} and the estimate \eqref{eq:physical4vectorfield.est} imply, in particular, that the spatial fluid velocity approaches zero relative to observers that are at rest with respect to $t=const$ foliation (i.e. whose wordlines are integral curves of $\eb_0$). The perturbed solutions of the Einstein-Euler-scalar field system can therefore be interpreted as \emph{asymptotically co-moving}.
\end{enumerate}
\end{rem}

\subsection{Fuchsian stability}
\label{sec:proof_globstab_Fuchsian}
As discussed above, the first step in the proof Theorem \ref{glob-stab-thm} is to establish stability of the trivial solution $u\equiv 0$ to the Fuchsian equation \eqref{Fuch-ev-A2}. Thus, we need to solve the Fuchsian global initial value problem (GIVP) 
\begin{align}
  \label{eq:givp1}
  A^0(u)\del{t}u + \frac{1}{t^{\ep_0+\ep_2}}A^\Lambda(t,u) \del{\Lambda} u &=
\frac{1}{t}\Ac\Pbb u +\frac{1}{t^{\tilde\ep}}\Ht(t) +\frac{1}{t}\Htt(u)\Pbb u \notag \\
&\hspace{0.2cm} +\frac 1t\Pbb^{\perp} \hat\Htt(u)+\frac{1}{t^{\tilde\epsilon}} H(t,u)\hspace{0.5cm}\text{in $M_{0,t_0}$,}
%=(0,t_0]\times \Tbb^{n-1}$,} 
\\ 
  \label{eq:givp2}
  u&=u_0 \hspace{3.65cm} \text{
in $\Sigma_{t_0}$,}%$=\{t_0\}\times \Tbb^{n-1}$,}
\end{align}
for initial data $u_0$ that is sufficiently small.
Existence of solutions to this GIVP is obtained in the following proposition. Its proof follows from an application of the Fuchsian global existence theory established in \cite{BOOS:2021}. The actual existence result we employ is Theorem~A.2 from \cite{BeyerOliynyk:2020} together with Remark~A.3 from that same article, which, together, amount to a slight generalization of the Fuchsian global existence theory from \cite{BOOS:2021}.

\begin{prop}
  \label{prop:globalstability}
Suppose that $n\in\Zbb_{\ge 3}$,  $k \in \Zbb_{>(n+1)/2}$, $\sigma>0$, $T_0>0$, $V_*^0>0$, $c_s^2\in (1/(n-1),1)$, $P_0>0$ and that $\ep_0$, $\ep_1$, $\ep_2$, $\ep_3$ and $\ep_4$ satisfy \eqref{eq:epscond.N}.
  % \begin{equation}
  %   \label{eq:epscond2}
  %   \ep_1<1 \ldots
  %   0<\ep_0,\quad \frac{1}{1+\sqrt{3}}<\ep_1,\quad 3\ep_0+\ep_1<1, \quad  0<\ep_2<1-\ep_0.
  % \end{equation}
Then there exists
a $\delta_0 > 0$ such that for every $t_0\in (0,T_0]$ and $\delta \in (0,\delta_0]$, if $u_0\in H^k(\mathbb T^{n-1})$ and $\tilde{\Ftt}(t)$, see \eqref{eq:Ftttilde}, satisfy
\begin{equation}
  \label{glob-stab-thm-idata-A-Fuchsian}
 \norm{u_0}_{H^k(\mathbb T^{n-1})}< \delta \AND \int_0^{t_0} \norm{\tilde\Ftt(s)}_{H^k(\mathbb T^{n-1})} ds < \delta_0,
\end{equation}
respectively, then the Fuchsian GIVP \eqref{eq:givp1}-\eqref{eq:givp2}  admits a unique solution 
\begin{equation*}
u \in C^0_b\bigl((0,t_0],H^k(\mathbb T^{n-1})\bigr)\cap C^1\bigl((0,t_0],H^{k-1}(\mathbb T^{n-1})\bigr)
\end{equation*}
such that $\lim_{t\searrow 0} \Pbb^\perp u(t)$, denoted $\Pbb^\perp u(0)$, exists in $H^{k-1}(\mathbb T^{n-1})$. Moreover, the solution $u$ satisfies the energy estimate
\begin{align}
  \norm{u(t)}_{H^k(\mathbb T^{n-1})}^2 + &\int^{t_0}_t \frac{1}{s} \norm{\Pbb u(s)}_{H^k(\mathbb T^{n-1})}^2\, ds\notag \\
  &\lesssim \norm{u_0}^2_{H^k(\mathbb T^{n-1})}
  +\left(\int_t^{t_0} \norm{\tilde\Ftt(s)}_{H^k(\mathbb T^{n-1})} ds\right)^2   \label{eq:resenergy}
\end{align}
and decay estimates
\begin{align}
\label{eq:PbbuDecay}
  \norm{\Pbb u(t)}_{H^{k-1}(\mathbb T^{n-1})} &\lesssim t^p+t^{\kappat-\sigma},\\
  \label{eq:PbbuDecay2}
\norm{\Pbb^\perp u(t) - \Pbb^\perp u(0)}_{H^{k-1}(\mathbb T^{n-1})} &\lesssim
 t^{p}+ t^{2(\kappat-\sigma)},
\end{align}
for all $t\in(0,t_0]$,  
where 
\begin{align}
  p&=\min\biggl\{1-3\ep_0-\ep_1, 1-\ep_0-\ep_2,\ep_0,\ep_2, \notag\\
   & \hspace{1.5cm} \frac{n-1}{n-2}(1-c_s^2)-\ep_3+\ep_1-\ep_0-1, \ep_4-\ep_1, 1-\ep_0-\ep_4\biggr\},   \label{eq:defp} \\
  % \kappat&=\min\biggl\{\ep_0,\ep_2,\ep_3, \frac{(n-1)c_s^2-1}{n-2}\biggr\},
             \kappat&=\min\biggl\{\ep_0,\ep_2,\ep_3, \frac{(n-1)c_s^2-1}{n-2}, \ep_1-\frac{1}{1+\sqrt{3}}\biggr\},
             \label{eq:defkappat}
\end{align}
and $\Pbb^\perp = \id -\Pbb$.

\medskip

\noindent The implicit constants in the energy and decay estimates are all independent of the choice of $t_0\in (0,T_0]$ and $\delta\in (0,\delta_0]$.
\end{prop}

The proof of Proposition~\ref{prop:globalstability}  makes use of  two technical lemmas, Lemmas~\ref{lem-Bc-lbnd} and \ref{lem:posdef1} below, that we will present first. The first lemma is a restatement of Lemma~3.4 from \cite{BeyerOliynyk:2021}, and we refer the reader to that article for its proof. 

\begin{lem}
  \label{lem-Bc-lbnd}
  Suppose $\ep_1>0$ and let 
\begin{align} 
\Nc_{\gac} =\gac_{Pqrs} \delta^{PQ}\delta^{rl}\delta^{sm} \Bigl(&(1+\ep_1)\delta^{qj}\gac_{Qjlm} \notag \\
& + \delta_0^q\delta_0^j  (\gac_{Qljm}+ \gac_{Qmjl}-\gac_{Qjlm})\Bigr). \label{Ncgac-def}
\end{align}
 % be a symmetric bilinear form on the vector space consisting all $g_{Qmjl}\in \Rbb^{(n-1)n^3}$ satisfying the symmetry condition $\gac_{Qjlm}=\gac_{Qjml}$.
Then 
\begin{equation}
  \label{eq:comparenorms}
     \Nc_{\gac} \geq \Bigl(\ep_1-\frac{1}{1+\sqrt{3}}\Bigr)|\gac|^2
\end{equation}
for all $\gac_{Qmjl}\in \Rbb^{(n-1)n^3}$ satisfying $\gac_{Qjlm}=\gac_{Qjml}$ where
\begin{equation*}
|\gac|^2 =  \delta^{PQ}\delta^{qj}\delta^{rl}\delta^{sm} \gac_{Pqrs}\gac_{Qjlm}
\end{equation*}
is the Euclidean norm.
\end{lem}
\noindent It is important to note that \eqref{eq:comparenorms} yields an effective bound on the norm $|\gac|^2$ because $\ep_1-\frac{1}{1+\sqrt{3}}>0$ due to 
 \eqref{eq:epscond.N}.

Before we state the second technical lemma, Lemma~\ref{lem:posdef1}, we first define an alternative formulation of the Fuchsian equation \eqref{eq:givp1} that will be employed in our subsequent analysis.  To this end, we consider so far arbitrary positive constants $\sigma_1,\ldots,\sigma_{10}>0$ and
set
\begin{gather}
  \label{Asc-def}
  \Asc=\small\diag\bigl(\Asc_G, \Asc_F\bigr),\\
  \label{AscG-def}
  \begin{split}
  \Asc_G= \small \diag\Bigl(&\sigma_1\delta^{\Lt L}\delta^{\Mt M},
  \sigma_2\delta^{\Mt M},
  \sigma_3\delta^{\Rt R}\delta^{\Mt M},
  \\
  &\sigma_4\delta^{\Rt R} \delta^{\Lt L} \delta^{\Mt M}
  \sigma_5\delta^{\rt r}\delta^{\lt l},
  \sigma_6,
  \sigma_7\delta^{\It I}\delta_{\Lambdat \Lambda}, \\
  &\sigma_8\delta^{\It I}\delta_{\kt k} \delta^{\Jt J},
  \sigma_9\delta^{\Qt Q}\delta^{\jt j}\delta^{\lt l},
  \sigma_{10}\delta^{\Qt Q}\delta^{\jt j}\delta^{\lt l} \delta^{\mt
    m}\Bigr),
\end{split}\\
\label{AscF-def}
\Asc_F=\small\diag\bigl(1, \delta_{\jt j}, \delta_{\jt j}\delta^{\Qt Q}\bigr).
\end{gather}
Then we let
  \begin{align}    
    \Bsc^0(u)=&\Asc A^0(u), \label{Bsc0-def}\\
    \Bsc^\Lambda(t,u)=&\frac{1}{t^{\ep_0+\ep_2}}\Asc  A^\Lambda(t,u),\label{BscLambda-def}\\
    \Bc(u)=&\Asc\Ac(u)+\Asc\Htt(u) \label{eq:defBc}
    \intertext{and}
    \Hsc(t,u)=&\frac{1}{t^{\tilde\ep}}\Asc H(t,u) \label{Hsc-def}
  \end{align}
  for any $\tilde\ep$ satisfying \eqref{eq:epstildecond}.
  Using these definitions, a short calculation shows that \eqref{eq:givp1} can be expressed as 
  \begin{equation} \label{Fuch-ev-A3}
\Bsc^0(u)\del{t}u + \Bsc^\Lambda(t,u) \del{\Lambda} u =
\frac{1}{t}\Bc(u)\Pbb u +\frac {1}t\Pbb^{\perp} \hat\Htt(u)+\sigma_1 \tilde\Ftt(t)+ \Hsc(t,u),
\end{equation}
where in deriving this we have used \eqref{Fuch-ev-A2} -- \eqref{eq:Ftttilde} and \eqref{eq:SourcetermVan1}.

\begin{lem}
  \label{lem:posdef1}  
  Suppose that $V_*^0>0$, $P_0>0$, $c_s^2\in (1/(n-1),1)$,  $\ep_0,\ldots,\ep_4$ satisfy \eqref{eq:epscond.N} and $\kappat$ is given by \eqref{eq:defkappat}. Then there exist a $\bc>0$ and  for each $\eta>0$, constants $\sigma_i=\sigma_i(\eta)$, $1\leq i\leq 10$, and $R_0=R_0(\kappat,\bc,\eta)>0$ such that
  \begin{equation*}
    \Asc\Ac(u)\geq(\kappat-\bc\eta)\Bsc^0(u),
  \end{equation*}
  for all $u\in B_R(\Rbb^{\udim})$ and $R\in (0,R_0)$. 
\end{lem}

 \begin{rem} The lower bound $c_s^2>1/(n-1)$ on the sound speed is needed to ensure that $\kappat>0$; see  \eqref{eq:epscond.N} and \eqref{eq:defkappat}.
 \end{rem}

\begin{proof}
The proof is to a large part a straightforward consequence of \cite[Lemma~A.1]{BeyerOliynyk:2021}.
To see why this is the case, we identify the matrix $\Ac$ in \cite[Lemma~A.1]{BeyerOliynyk:2021}  with the matrix $\Ac_G$ in \eqref{AcG-def}. This matrix is then partitioned by choosing $N=10$ and by identifying its diagonal blocks with the blocks $\Ac_{1\,1}$, \ldots, $\Ac_{10\,10}$ from \cite[Lemma~A.1]{BeyerOliynyk:2021} and 
similarly for the off-diagonal blocks. Noticing that $A^0_G$ in \eqref{A0G-def} is the identity matrix, the 
inequality
\begin{equation*}%\label{AcG-lbnd}
    \Asc_G\Ac_G\geq (\kappat_G-\bc\eta) \Asc_G A^0_G,
\end{equation*}
where
\begin{equation*}
\Asc_G = \diag\bigl(\sigma_1\id,\sigma_2\id,\ldots,\sigma_{10}\id\bigr)
\end{equation*}
and $\sigma_1$,\ldots, $\sigma_{10}$ are positive constants that depend on the choice of $\eta>0$,
is then a direct consequence of that lemma, \eqref{Ac-diag-1}-\eqref{Ac-offdiag-3} and Lemma~\ref{lem-Bc-lbnd}  provided $\kappat_G$ is defined by
\begin{equation*}
  \kappat_G=\min\biggl\{1,2+\ep_1,\ep_1,\ep_1-\ep_0,\ep_0,\ep_2,\ep_0+2\ep_1,
  \ep_1-\frac{1}{1+\sqrt{3}}\biggr\}.
\end{equation*}
Note here that $\bc>0$ should be considered as a fixed positive constant while $\eta>0$ can be chosen arbitrarily small. We also observe that $\kappat_G$ simplifies to
\begin{equation*}
  \kappat_G=\min\biggl\{\ep_0,\ep_2,
  \ep_1-\frac{1}{1+\sqrt{3}}\biggr\}
\end{equation*}
because of
\begin{equation*}
    \ep_1-\ep_0 - \biggl(\ep_1 - \frac{1}{1+\sqrt{3}}\biggr) = \frac{1}{1+\sqrt{3}} - \ep_0 = \biggl(\frac{1}{1+\sqrt{3}} -\frac{1}{3}\biggr) 
    +\biggl(\frac{1}{3} - \ep_0\biggr)> 0
\end{equation*}
where the final inequality is due to \eqref{eq:epscond.N}.

In the same way we have from
\eqref{Ac-diag-4} -- \eqref{Ac-diag-6} and \eqref{eq:a0exp} that
\begin{align*}
  %\label{eq:1}
  &\Asc_F \Ac_F(u)\\
  =&\diag\biggl\{\ep_3, \frac{(n-1)c_s^2-1}{n-2}W^0 \bigl(c_s^{-2}\delta_{\jt}^0\delta_j^0+\delta_{\Jt J}\delta_{\jt}^{\Jt}\delta_{j}^J\bigr),\\ 
  &\qquad\quad\Bigl(\ep_4a^0_{s\tilde j} +
  W^0\frac{(n-1)c_s^2-1}{n-2} \delta_{JK}\delta_s^{J}\delta_{\jt}^{K}      
                    \Bigr)\delta^{\Qt Q}\biggr\}\\
  =&\diag\biggl\{\ep_3, \frac{(n-1)c_s^2-1}{n-2}W^0 \bigl(c_s^{-2}\delta_{\jt}^0\delta_j^0+\delta_{\Jt J}\delta_{\jt}^{\Jt}\delta_{j}^J\bigr),
     \\
     &\qquad\quad W^0\Bigl(\frac{\ep_4}{c_s^2}\delta_s^{0}\delta_{\jt}^{0} 
  +\frac{\ep_4(n-2)+(n-1)c_s^2-1}{n-2}\delta_{JK}\delta_s^{J}\delta_{\jt}^{K}
  \Bigr) \delta^{\Qt Q}
 \biggr\} +\Ord(\Pbb u),
\end{align*}
and from \eqref{A0F-def} with \eqref{eq:a0exp} that
\begin{align*}
  %\label{eq:3}
  &\Asc_F A^0_F(u)\\
  =&\diag\bigl\{1, a^0_{\jt j}, a^0_{\jt j}\delta^{\Qt Q}\bigr\}\\
  =&\diag\biggl\{1, W^0\Bigl(\frac{1}{c_s^2}\delta_{\jt}^{0}\delta_j^{0} 
  +\delta_{JK}\delta_{\jt}^{J}\delta_j^{K}
  \Bigr), W^0\Bigl(\frac{1}{c_s^2}\delta_{\jt}^{0}\delta_j^{0} 
  +\delta_{JK}\delta_{\jt}^{J}\delta_j^{K}
  \Bigr) \delta^{\Qt Q}
 \biggr\}\\
 &+\Ord(\Pbb u).
\end{align*}
Hence, by decreasing the size of $R$ if necessary, it follows that
\[  \Asc_F \Ac_F(u)\geq (\kappat_F-\bc\eta)\Asc_F A^0_F(u)\]
holds for all $u\in B_R(\Rbb^{\udim})$
for the same $\bc$ and $\eta$ as before provided
\[\kappat_F=\min\biggl\{\ep_3,\frac{(n-1)c_s^2-1}{n-2}, \ep_4\biggr\}.\]
The quantity $\kappat$ in \eqref{eq:defkappat} is then determined from $\min\{\kappat_G,\kappat_F\}$ given that $\ep_4>\ep_1$ as a consequence of \eqref{eq:epscond.N}.

% \bigskip
% \cnote{red}{\hrule}
% \bigskip

% \cnote{red}{Florian, I don't understand this part of the proof. I would like to discuss it with you.}

% \bigskip

% Next we notice from \eqref{eq:B0fluid}, \eqref{A0-def}, \eqref{A0F-def}, \eqref{AcF-def} and \eqref{Ac-diag-4} that given $\bc$ and $\eta$ as before, we can find $R>0$ so that
% \[  \Asc_F \Ac_F(u)\geq (\kappat_F-\bc\eta)\Asc_F A^0_F(u)\]
% holds for all $u\in B_R(\Rbb^{\udim})$ where
% \[\kappat_F=\min\{\ep_3,\ep_4,\frac{(n-1)c_s^2-1}{n-2}\}.\]

% It is then a consequence of  \eqref{eq:epscond.N} (notice especially that $\ep_4>\ep_1$) 
% that the minimum of 
% $\kappat_G$ and $\kappat_F$ agrees with $\kappat$ in \eqref{eq:defkappat}, observing that
% the assumption $c_s^2\in (1/(n-1),1)$ implies that
% \begin{align*} \Bigl(\ep_1-\frac{1}{1+\sqrt{3}}\Bigr)-\frac{(n-1)c_s^2-1}{n-2}&>1- \frac{n-1}{2(n-2)}(1-c_s^2) -\frac{(n-1)c_s^2-1}{n-2}-\frac{1}{1+\sqrt{3}}
%   \\
%   &=\frac{n-1}{2(n-2)}(1-c_s^2) -\frac{1}{1+\sqrt{3}}>\frac 12 -\frac{1}{1+\sqrt{3}}>0.
% \end{align*}

% \bigskip
% \cnote{red}{\hrule}
% \bigskip

\end{proof}

We now turn to the proof of Proposition~\ref{prop:globalstability}. As will become apparent, the proof of this proposition will follow from an application of \cite[Theorem~A.2]{BeyerOliynyk:2020}, which requires us to verify that  \eqref{Fuch-ev-A3} satisfies the \emph{coefficient assumptions} from \cite[Section~A.1]{BeyerOliynyk:2020}. The  proof of the proposition then amounts to verifying these coefficient assumptions.

\begin{proof}[Proof of Proposition~\ref{prop:globalstability}]
Suppose $n$, $k$, $\sigma$, $T_0$, $V_*^0$, $c_s$, $\ep_0$, $\ep_1$, $\ep_2$, $\ep_3$, $\ep_4$, $p$ and $\kappat$ are chosen as in the statement of Proposition~\ref{prop:globalstability}. Then by \eqref{Pbb-def}, we observe that the projection matrix $\Pbb$ satisfies
\begin{equation*}
\Pbb^2 = \Pbb,  \quad  \Pbb^{\tr} = \Pbb, \quad \del{t}\Pbb =0 \AND \del{\Lambda} \Pbb =0
\end{equation*}
while we note from \eqref{U-def}, \eqref{eq:Ubgdef} and \eqref{u-def} that kernel of $\Pbb$ is spanned by the variables $u_{1,IJ}=k_{IJ}$ and $u_{12}^0=W^0-V_*^0$.
By \eqref{eq:B0fluid}, \eqref{A0-def}-\eqref{A0F-def}, \eqref{Bsc0-def}, it is clear that the matrix-valued map $\Bsc^0(u)$ from the Fuchsian equation \eqref{Fuch-ev-A3} depends smoothly on $u$ near $u=0$, i.e.
$\Bc^0 \in C^\infty(B_{R}(\Rbb^{\udim}),\Mbb{\udim})$ for $R>0$ sufficiently small, is positive, symmetric and satisfies
\[\Pbb \Bsc^0(u) \Pbb^\perp=\Pbb^\perp \Bsc^0(u) \Pbb=\Ord(\Pbb u).\]
In addition, we note from \eqref{Pbb-def}, \eqref{eq:SourcetermVan1}-\eqref{eq:Gtprops}, \eqref{Asc-def}-\eqref{AscF-def}  and \eqref{eq:defBc} that
the matrix-valued map $\Bc(u)$ has the property
$\Bc \in C^\infty(B_{R}(\Rbb^{\udim}),\Mbb{\udim})$ for $R>0$ sufficiently small, and
satisfies 
\begin{equation} \label{Bc(0)}
    [\Pbb,\Bc]=0 \AND 
    \Bc(0) = \Asc \Ac(0).
  \end{equation}
The matrix $\Ac(0)$ can be constructed from \eqref{Ac-def}-\eqref{Ac-offdiag-3} by observing that the
only $u$-dependent blocks of $\Ac$ are $\Ac_{12\, 12}$ and $\Ac_{13\,
  13}$, which when evaluated at $u=0$, take the form 
\begin{align*}
  \Ac_{12\, 12}(0)&=\frac{(n-1)c_s^2-1}{n-2}V^0_* \bigl(c_s^{-2}\delta_{\jt}^0\delta_j^0+\delta_{\Jt J}\delta_{\jt}^{\Jt}\delta_{j}^J\bigr),\\
\Ac_{13\, 13}(0)&=\Bigr(\frac{\ep_4 V_*^0}{c_s^2}\delta_{\jt}^{0}\delta_s^{0} 
  +
  V^0_*\frac{\ep_4(n-2)+(n-1)c_s^2-1}{n-2}  \delta_{\Jt J}\delta_{\jt}^{\Jt}\delta_{s}^J     
  \Bigr)\delta^{\Qt Q},
\end{align*}
as consequence of \eqref{eq:a0exp}.

Now, from \eqref{Bc(0)} and the smooth dependence of $\Bc(u)$ on $u$, we have that
\begin{equation} \label{Bc-exp}
    \Bc(u) = \Asc \Ac(0) + \Ord(u),
\end{equation}
and for similar reasons, that 
\begin{equation} \label{Bsc0-exp}
\Bsc^0(u) = \Asc A^0(0)+\Ord(u),  
\end{equation}
where, as a consequence of \eqref{A0-def}, \eqref{A0G-def}, \eqref{A0F-def} and \eqref{eq:a0exp}, we have
\begin{align*}
A^0(0)=\diag\Biggl(&\delta^{\Lt L}\delta^{\Mt M},\delta^{\Mt M},\delta^{\Rt R}\delta^{\Mt M},\delta^{\Rt R} \delta^{\Lt L} \delta^{\Mt M},\delta^{\rt r}\delta^{\lt l},1,\delta^{\It I}\delta_{\Lambdat \Lambda},\\
&
\delta^{\It I}\delta_{\kt k} \delta^{\Jt J},\delta^{\Qt Q}\delta^{\jt j}\delta^{\lt l},\delta^{\Qt Q}\delta^{\jt j}\delta^{\lt l} \delta^{\mt m},1,\\
 &V_*^0\Bigl(\frac{1}{c_s^2}\delta_j^{0}\delta_{\jt}^{0} 
  +\delta_{JK}\delta_j^{J}\delta_{\jt}^{K}
  \Bigr), V_*^0\Bigl(\frac{1}{c_s^2}\delta_j^{0}\delta_{\jt}^{0} 
  +\delta_{JK}\delta_j^{J}\delta_{\jt}^{K}
  \Bigr)\delta^{\Qt Q}\Biggr).
\end{align*}
Then fixing $\kappa \in (0,\kappat)$ and choosing $\eta$ sufficiently small, we deduce
from the expansions \eqref{Bc-exp}-\eqref{Bsc0-exp} and Lemma \ref{lem:posdef1}  the existence of constants 
$\gamma_1,\gamma_2,R,\sigma_1,\ldots,\sigma_{10}>0$ such that
\begin{equation*}
\frac{1}{\gamma_1}\id \leq \Bsc^0(u) \leq \frac{1}{\kappa} \Bc(u) \leq \gamma_2 \id
\end{equation*}
for all $u\in B_R(\Rbb^{\udim})$.

Next, due to \eqref{U-def}, \eqref{ALambda-def}-\eqref{ALambdaF-def}, \eqref{eq:Ubgdef}-\eqref{u-def}, \eqref{Asc-def} and \eqref{BscLambda-def}, we note that 
\begin{equation*}
    t^{1-p}\Bsc^\Lambda\in C^0\bigl([0,T_0],C^\infty(\Rbb^{\udim})\bigr)
\end{equation*}
since
\begin{equation} \label{p-rest-A}
   p\leq \min\{ 1-\ep_0 -\ep_2,1-\tilde{\ep}\}
 \end{equation}
as a consequence of \eqref{eq:epscond.N}-\eqref{eq:epstildecond} and \eqref{eq:defp}.
Further, it is clear from \eqref{Hsc-def}, \eqref{p-rest-A}, \eqref{eq:Ftttilde} and  Lemma~\ref{lem:sourceterm} that 
 \begin{gather*}
     t^{1-p}\Hsc\in C^0\bigl([0,T_0],C^\infty(B_R(\Rbb^\udim),\Rbb^{\udim})\bigr),  \qquad t^{1-p}\Hsc(t,0) =0,\\
    t^{1-p}\tilde\Ftt\in C^\infty\bigl((0,T_0],\Rbb^{\udim}\bigr)\cap C^0 \bigl([0,T_0],\Rbb^{\udim}\bigr),\\
    \Pbb^{\perp} \hat\Htt \in C^0\bigl([0,T_0],C^\infty(B_R(\Rbb^\udim),\Rbb^{\udim})\bigr) \AND \Pbb^\perp \Hh(u)  = \Ordc\biggl(\frac{\lambda}{R}\Pbb u\otimes\Pbb u \biggr),
 \end{gather*}
where the constant $\lambda>0$ can be chosen arbitrarily small by further shrinking $R>0$ if necessary.

Thus far, we have verified that the coefficients of the Fuchsian equation \eqref{Fuch-ev-A3} satisfy all of the assumptions from \cite[Section~A.1]{BeyerOliynyk:2020} except for the assumptions regarding the divergence map $\Div\Bsc(t,u,w)$, which we now consider.
According to item (4) of Definition~2.1 from \cite{BeyerOliynyk:2020}, the divergence map $\Div\Bsc(t,u,w)$ is defined by first computing
 \begin{align}
  \del{t}&(\Bsc^0(u))+\del{\Lambda}(\Bsc^\Lambda(t,u))
  =D_{W^s}\Bsc^0(u) \del{t} W^s\notag\\
  &+t^{-\ep_0-\ep_2} D_\beta \Bsc^\Lambda(t,u) \partial_\Lambda\beta+t^{-\ep_0-\ep_2} D_{f_I^\Omega} \Bsc^\Lambda(t,u) \partial_\Lambda f_I^\Omega\notag\\
  &+t^{-\ep_0-\ep_2} D_{W^s} \Bsc^\Lambda(t,u) \partial_\Lambda W^s,
  \label{eq:divergencemap}
\end{align}
where in deriving this we have used \eqref{U-def},  \eqref{A0-def}--\eqref{ALambdaF-def}, \eqref{eq:Ubgdef}--\eqref{u-def}, \eqref{Asc-def}--\eqref{BscLambda-def} and the definitions for $A^0(u)$ and $A^\Lambda(t,u)$ that were introduced directly after \eqref{Fuch-ev-A2}. The variable $w$ is taken to be the spatial derivative of $u$, that is,
\begin{equation} \label{w-def}
w=(w_\Lambda)=(\del{\Lambda}u).
\end{equation}
The divergence map $\Div\Bsc(t,u,w)$ is then defined as the right hand side of \eqref{eq:divergencemap} where the time derivative $\del{t}W^s$ is replaced using its  evolution equation, see \eqref{eq:20934klsjd} below,
%, i.e. selecting the appropriate component of
%\begin{equation*}
 %   \del{t}u = (\Bsc^0(u))^{-1}\biggl(-\Bsc^\Lambda(t,u)\del{\Lambda} u+\frac{1}{t}\Bc(u)\Pbb u+\frac{1}{t}\Pbb^\perp \hat{\Htt}(u)+\sigma_1 \tilde{\Ftt}(t)+\Hsc(t,u)\biggr),
%\end{equation*}
and after that replacing all remaining spatial derivatives $\del{\Lambda}u$ with $w_\Lambda$.  

To verify the coefficient assumption for $\Div\Bsc(t,u,w)$, we need to show, for some $R>0$, that $\Div\Bsc(t,u,w)$ depends continuously on $t$ and smoothly on $(u,w)$ for $(t,u,w)\in (0,T_0]\times B_R(\Rbb^\udim)\times B_R(\Rbb^{(n-1)\udim})$ and that there exist positive constants $\theta$, $\beta_1$, $\beta_2$ and $\beta_3$ such that\footnote{The order notation $\Ord(\cdot)$ and $\Ordc(\cdot)$ is defined in Section~\ref{prelim}.}
\begin{gather}
  \label{eq:djf9sdfj1}
\Pbb \Div\Bsc(t,u,w)\Pbb = 
\Ordc\bigl(t^{-(1-p)}\theta + t^{-1}\beta_1\bigr), \\
\label{eq:djf9sdfj2}
\Pbb  \Div\Bsc(t,u,w) \Pbb^\perp = \Pbb^\perp  \Div\Bsc(t,u,w) \Pbb\notag\\
 =
\Ordc\biggl(t^{-(1-p)}\theta
+ \frac{t^{-1}\beta_2}{R}\Pbb u\biggr)
\intertext{and}
\Pbb^\perp \Div\Bsc(t,u,w) \Pbb^\perp =\Ordc\biggl(t^{-(1-p)}\theta+ \frac{t^{-1}\beta_3}{R^2}\Pbb u\otimes\Pbb u \biggr).\label{eq:djf9sdfj}
\end{gather}
To this end, we note, since $p$ satisfies \eqref{p-rest-A}, that the second line of \eqref{eq:divergencemap}, that is,
\begin{gather*}
t^{-\ep_0-\ep_2} D_\beta \Bsc^\Lambda(t,u) \partial_\Lambda\beta+t^{-\ep_0-\ep_2} D_{f_I^\Omega} \Bsc^\Lambda(t,u) \partial_\Lambda f_I^\Omega\\
+t^{-\ep_0-\ep_2} D_{W^s} \Bsc^\Lambda(t,u) \partial_\Lambda W^s,
\end{gather*}
defines, for $R>0$ small enough, a map $\Csc(t,u,w)$ satisfying
\begin{equation}\label{Csc-map}
 t^{1-p}\Csc\in C^0\bigl([0,T_0], C^\infty\bigl(B_{R}(\Rbb^{\udim}) \times B_{R}(\Rbb^{(n-1)\udim}),\Rbb^\udim\bigr)\bigr).
\end{equation}
Considering now the first line of \eqref{eq:divergencemap}, we observe, with the help of \eqref{eq:B0fluid}, \eqref{for-Euler.1.2} and \eqref{F12-def}, that 
\begin{align}
  &D_{W^s}\Bsc^0(u) \del{t} W^s\\
  = &D_{W^s}\Bsc^0(u) (a^{0})^{-1}{}^{sj}\Bigl(\betat G_{js'l}W^{s'}W^l- \betat {a^i}_{j k}\gamma_i{}^k{}_lW^l-{a^I}_{j k}\betat e_I^\Lambda\del{\Lambda} W^k\Bigr)\notag\\
  = &\Asc D_{W^s}A^0(u) (a^{0})^{-1}{}^{s\jt}\Bigl(F_{12\jt}+t^{-1}\frac{(n-1)c_s^2-1}{n-2} W^0\delta_{\jt}^{\Jt}\delta_j^J\delta_{\Jt J}W^j\notag\\
  &\hspace{2cm}-t^{-\ep_0-\ep_2}{a^I}_{j k}\beta f_I^\Lambda\del{\Lambda} W^k\Bigr),\notag
%  &=: M_s(v) (a^{0})^{-1}{}^{s\jt}(v) I_{sj} (t,v,w)
%\label{eq:20934klsjd},
\end{align}
which we express more compactly as
\begin{equation}
D_{W^s}\Bsc^0(u) \del{t} W^s= M_s(u)(a^{0})^{-1}{}^{s\jt}(u) I_{\jt} (t,u,w)\label{eq:20934klsjd}
\end{equation}
where 
\begin{equation}\label{M-def}
    M_s(u)=\Asc D_{W^s}A^0(u)
\end{equation}
and $I_{\jt} (t,u,w)$ is defined by
\begin{equation*}
F_{12\jt}+t^{-1}\frac{(n-1)c_s^2-1}{n-2} W^0\delta_{\jt}^{\Jt}\delta_j^J\delta_{\Jt J}W^j-t^{-\ep_0-\ep_2}{a^I}_{j k}\beta f_I^\Lambda\del{\Lambda} W^k.
\end{equation*}
With the help of \eqref{eq:BIfluid}, \eqref{U-def}, \eqref{eq:Ubgdef}-\eqref{u-def}, \eqref{w-def} and Lemma \ref{lem:sourceterm}, it is then not difficult to verify that  $I_{\jt}$ can be decomposed as
\begin{equation} \label{I-decomp}
  %\label{eq:40}
  I_{\jt}(t,u,w)=t^{-1} I_{\jt}^{(0)}(u)+t^{-1+p} I^{(1)}_{\jt}(t,u,w)
\end{equation}
where 
\begin{equation} \label{I-reg}
    I_{\jt}^{(0)}\in C^\infty(B_R(\Rbb^\udim),\Rbb^\udim) \AND I^{(1)}_{\jt}\in C^0\bigl([0,T_0],C^\infty(B_R(\Rbb^\udim),\Rbb^\udim)\bigr).
\end{equation}
We futher observe from \eqref{eq:B0fluid} and \eqref{A0G-def}-\eqref{A0F-def} that $M_s$, defined above by \eqref{M-def}, satisfies
\begin{gather}
    \Pbb  M_0(u) \Pbb=\Ord(1), \quad 
    \Pbb^\perp  M_0(u) \Pbb=\Pbb  M_0(u) \Pbb^\perp=\Ord(\Pbb v\otimes\Pbb v\otimes\Pbb v),\notag\\
  \Pbb^\perp   M_0(u) \Pbb^\perp =\Ord(1), \label{M-exp.1}\\
    \Pbb  M_Q(u) \Pbb=\Ord(\Pbb u), \quad \Pbb^\perp  M_Q(u) \Pbb=\Pbb  M_Q(u) \Pbb^\perp=\Ord(1)\notag\\
    \AND \Pbb^\perp   M_Q(u) \Pbb^\perp =\Ord(\Pbb u), \label{M-exp.2}
  \end{gather}
and note from \eqref{eq:a0exp}
that $(a^{0})^{-1}$ can
be expanded as
\begin{equation}\label{a0inv-exp}
(a^{0})^{-1}{}^{\,jk}(u)
=\frac1{W^0}\bigl({c_s^2}   \delta^j_{0}\delta^k_{0}
  +\delta^{JK}\delta^j_{J}\delta^k_{K}
  \bigr)
  +\Ord(\Pbb u).
\end{equation}
From the analysis of $F_{12}$ in the proof of Lemma~\ref{lem:sourceterm}, we deduce that  $I^{(0)}$, defined above by \eqref{I-decomp}, satisfies
\begin{equation}\label{I(0)-exp}
I_0^{(0)}(u)=\Ord(\Pbb u\otimes\Pbb u) \AND I_Q^{(0)}(u)=\Ord(\Pbb u).
\end{equation}

Taken together, the expansions \eqref{Csc-map}, \eqref{I-decomp}-\eqref{M-exp.2} and
\eqref{I(0)-exp} along with the fact that $\Asc$, defined by \eqref{Asc-def}-\eqref{AscF-def}, and $\Pbb$, defined by \eqref{Pbb-def}, commute, imply  the existence of the constants $\theta$, $\beta_1$, $\beta_2$ and $\beta_3$ such that \eqref{eq:djf9sdfj1}-\eqref{eq:djf9sdfj} hold. Moreover, 
because the constants $\sigma_1$,\ldots, $\sigma_{10}$ used to define $\Asc$, see \eqref{Asc-def}-\eqref{AscF-def}, can be chosen independently  of $R>0$ for $R$ sufficiently small, see Lemma \ref{lem:posdef1}, we can arrange that the constants $\beta_1$, $\beta_2$ and $\beta_3$ are as small as we like by choosing $R$ small enough. This completes the verification of the coefficient assumptions from \cite[Section~A.1]{BeyerOliynyk:2020}.

Given that the Fuchsian equation \eqref{Fuch-ev-A3} satisfies all of the coefficient assumption from \cite[Section~A.1]{BeyerOliynyk:2020}, we conclude directly from Theorem~A.2 and Remark~A.3 of \cite{BeyerOliynyk:2020}
the existence of a constant $\delta_0 > 0$ such that if $t_0\in (0,T_0]$,
 $\int_0^{t_0} \norm{\tilde\Ftt(s)} ds < \delta_0$
and $u_0 \in H^k(\Tbb^{n-1})$ satisfies
$\norm{u_0}_{H^k(\Tbb^{n-1})}< \delta$ for any $\delta\in (0,\delta_0]$, then
there exists a unique solution 
\begin{equation*}
%\label{eq:symhypreg}
u \in C^0_b\bigl((0,t_0],H^k(\mathbb \Tbb^{n-1})\bigr)\cap C^1\bigl((0,t_0],H^{k-1}(\Tbb^{n-1})\bigr)
\end{equation*}
of the initial value problem 
 \begin{align*}
\Bsc^0(u)\del{t}u + \Bsc^\Lambda(t,u) \del{\Lambda} u &=
\frac{1}{t}\Bc(u)\Pbb u +\frac {1}t\Pbb^{\perp} \hat\Htt(u)+\sigma_1 \tilde\Ftt(t)+ \Hsc(t,u)\\
&\qquad\qquad \text{in $M_{0,t_0}=(0,t_0]\times \Tbb^{n-1}$,} \\ 
u&=u_0\\ 
&\qquad\qquad\text{in $\Sigma_{t_0}=\{t_0\}\times \Tbb^{n-1}$,}
\end{align*}
such that the limit $\lim_{\searrow 0} \Pbb^\perp u(t)$, denoted $\Pbb^\perp u(0)$, exists in $H^{k-1}(\mathbb \Tbb^{n-1})$ and $u$ satisfies the energy and decay estimates given by
\eqref{eq:resenergy} and \eqref{eq:PbbuDecay}-\eqref{eq:PbbuDecay2}, respectively.
Since the above initial value problem is equivalent to \eqref{eq:givp1}-\eqref{eq:givp2}, the proof of the proposition is complete. 
\end{proof}

\subsection{Proof of Theorem~\ref{glob-stab-thm}}
\label{sec:proof_globstab}

Equipped with Proposition~\ref{prop:globalstability}, we are now in a position to prove  Theorem~\ref{glob-stab-thm}. We begin by fixing $n\in\Zbb_{\ge 3}$, $k \in \Zbb_{>(n+3)/2}$, $V_*^0,P_*>0$, $c_s^2\in (1/(n-1),1)$, $T_0>0$, $\sigma>0$ and choosing constants $\ep_0,\ldots,\ep_4$ that satisfy \eqref{eq:epscond.N}.

\bigskip

\noindent \underline{Small Fuchsian initial data:} Denoting the constant $\delta_0>0$ from Proposition~\ref{prop:globalstability} as $\tilde\delta_0$, we assume that $t_0\in (0,T_0]$ and use $\breve U(t_0)$ to denote the FLRW background solution at $t=t_0$ defined by \eqref{eq:UbreveFirst}-\eqref{eq:UbreveLast}. We also set
\begin{equation*}
\breve u_0 =\breve U(t_0)-\mathring U(t_0),
\end{equation*}
where $\mathring U(t_0)$ is defined above by \eqref{eq:Ubgdef}.
Then, due to \eqref{Ubr-Ur} from Lemma~\ref{lem:bgpertsol}, we can, by choosing $t_0\in (0,T_0]$ sufficiently small, arrange that
 \begin{equation}
 \label{glob-stab-thm-idata-FLRW}
 \norm{\breve u_0}_{H^k(\mathbb T^{n-1})}< \frac{\tilde\delta_0}{2}.
\end{equation}
Additionally, by \eqref{eq:Ftttilde}, we can, by shrinking $t_0$ further if necessary, ensure that
\begin{equation}
  \label{glob-stab-thm-idata-DFtt}
\int_0^{t_0} \norm{\tilde\Ftt(s)}_{H^k(\mathbb T^{n-1})} ds\lesssim t_0^{\frac{(n-1)(1-c_s^2)}{n-2}} < \tilde \delta_0.
\end{equation}

Next, we suppose $\delta_0>0$ and choose $\delta\in (0,\delta_0]$ and initial data $\gr_{\mu\nu}\in H^{k+2}(\Tbb^{n-1},\mathbb{S}_n)$, 
$\ggr_{\mu\nu}\in H^{k+1}(\Tbb^{n-1},\mathbb{S}_n)$, $\taur=t_0$,
$\taugr\in H^{k+2}(\Tbb^{n-1})$ and $\Vr^\mu\in H^{k+1}(\Tbb^{n-1})$ that satisfies \eqref{glob-stab-thm-idata-A} as well as the constraint equations \eqref{grav-constr}-\eqref{wave-constr}. 
We then construct an  orthonormal frame on the initial hypersurface $\Sigma_{t_0}=\{t_0\}\times \Tbb^{n-1}$ as follows: recalling that $\{\gr_{\mu\nu},\ggr_{\mu\nu},\taur=t_0,\taugr, \Vr^\mu\}$ determines initial data $\{g_{\mu\nu}|_{\Sigma_{t_0}},\del{0}g_{\mu\nu}|_{\Sigma_{t_0}},\tau=t_0,\del{0}\tau =1, V^\mu |_{\Sigma_{t_0}}\}$ for the metric $g_{\mu\nu}$, the scalar field $\tau$ and the fluid vector $V^\mu$ in Lagrangian coordinates on $\Sigma_{t_0}$ via \eqref{V-idata}-\eqref{dt-g-idata},
we set 
\begin{equation*}
e_0^\mu=(-|\chi|_g^2)^{-\frac{1}{2}}\chi^\mu    
\end{equation*} 
and note that it can be computed from the Lagrangian initial data on $\Sigma_{t_0}$ by \eqref{chi-idata}. We further fix spatial frame initial data $e^\mu_I|_{\Sigma_{t_0}}=\delta^\mu_I\er^\Lambda_\mu$ where the functions
$\er^\Lambda_I \in H^k(\Tbb^{n-1})$ are chosen to satisfy
\begin{equation*} 
\norm{\er^\Lambda_I-(\omega|_{t=t_0})\delta^\Lambda_I}_{H^{k}(\Tbb^{n-1})} < \delta,
\end{equation*}
with $\omega$ defined by \eqref{eq:FLRWEulerSFExplSol4},
and make the frame $e^\mu_i$ orthonormal on $\Sigma_{t_0}$ with respect to the metric $g_{\mu\nu}$ given there, see \eqref{g-idata}.

Then we use the prescription outlined in Section~\ref{frame-idata}, see also \eqref{k-def}-\eqref{Uac-def}, \eqref{eq:Ubgdef} and \eqref{u-def}, to construct of complete set of initial data $u_0\in H^k(\Tbb^{n-1})$ for the Fuchsian system \eqref{eq:givp1}. It is straightforward to verify, with the help of the Sobolev and Moser inequalities (see Propositions 2.4., 3.7.~and 3.9.~from   \cite[Ch.~13]{TaylorIII:1996}), that the this initial data satisfies $\norm{u_0-\breve u_0}_{H^k(\Tbb^{n-1})} \leq C_0\delta$
%  \begin{equation*} %\label{glob-stab-thm-idata-D}
%    \norm{u_0-\breve u_0}_{H^k(\Tbb^{n-1})} \leq C_0\delta,
%\end{equation*}
where $C_0=C_0(t_0,t_0^{-1},\delta)>0$,
%\begin{equation*}%\label{C0-def}
%  C_0=C_0(t_0,t_0^{-1},\delta)>0,
%\end{equation*}
which allows us, by \eqref{glob-stab-thm-idata-FLRW}, to conclude that
\begin{equation} \label{idata-bnd}
\norm{u_0}_{H^k(\Tbb^{n-1})}\le \norm{\breve u_0}_{H^k(\Tbb^{n-1})} + C_0\delta\le \frac{\tilde\delta_0}2 + C_0\delta. 
\end{equation}
Given the value of $t_0$ found above by imposing \eqref{glob-stab-thm-idata-FLRW} and \eqref{glob-stab-thm-idata-DFtt}, we can now adjust the value  $\delta_0$ to ensure that $\delta<\tilde \delta_0/(2C_0)$ for all $\delta\in (0,\delta_0]$. This guarantees that the hypotheses of Proposition~\ref{prop:globalstability} are satisfied for appropriate choices of $\delta_0$ and $t_0$, and all $\delta\in (0,\delta_0]$.

\bigskip

\noindent \underline{Fuchsian stability:} 
By \eqref{glob-stab-thm-idata-DFtt} and \eqref{idata-bnd}, we obtain from Proposition~\ref{prop:globalstability}
a unique solution    
\begin{equation}\label{eq:symhypreg_appl}
u \in C^0_b\bigl((0,t_0],H^k(\Tbb^{n-1})\bigr)\cap C^1\bigl((0,t_0],H^{k-1}(\Tbb^{n-1})\bigr)
\end{equation}
to the Fuchsian GIVP \eqref{eq:givp1}-\eqref{eq:givp2} that extends continuously to $t=0$ in $H^{k-1}(\Tbb^{n-1})$ 
and
satisfies the energy 
and decay estimates given by \eqref{eq:resenergy} and  \eqref{eq:PbbuDecay}-\eqref{eq:PbbuDecay2}, respectively, 
for all $t\in (0,t_0]$.

By definition $\Pbb^\perp=\id-\Pbb$, and so we have by \eqref{Pbb-def} that
  \begin{equation}\label{Pbb-perp-def}
  \Pbb^\perp = \diag\Bigl(\delta_{\Lt}^L\delta_{\Mt}^M,0,0,0,0,0,0,0,0,0,0,\delta^{\jt}_0,0\Bigr).
\end{equation}
Using this, we can label the components of $\Pbb^\perp u(0) \in H^{k-1}(\Tbb^{n-1})$ as
\begin{equation} \label{Pbb-perp-u(0)}
\Pbb^\perp u(0)=\bigl(\kf_{IJ},0,0,0,0,0,0,0,0,0,\Wf^0 \delta^{\jt}_0,0\bigr)
\end{equation}
where $\kf_{IJ}\in H^{k-1}(\mathbb T^{n-1},\Sbb{n-1})$ and $\Wf^0\in H^{k-1}(\mathbb T^{n-1})$.
Noting, by the triangle inequality, that
\begin{equation*}
\norm{\Pbb^\perp u(0)}_{H^{k-1}(\Tbb^{n-1})}\leq 
\norm{u(t)-\Pbb^\perp u(0)}_{H^{k-1}(\Tbb^{n-1})} + \norm{u(t)}_{H^{k-1}(\Tbb^{n-1})},
\end{equation*}
we conclude from
the initial data bound $\norm{u_0}_{H^k(\Tbb^{n-1})}<\delta$, the energy and decay estimates \eqref{eq:resenergy} and  \eqref{eq:PbbuDecay}-\eqref{eq:PbbuDecay2},
and the estimate \eqref{glob-stab-thm-idata-DFtt}
that 
\begin{equation*}
\norm{\Pbb^\perp u(0)}_{H^{k-1}(\Tbb^{n-1})}\lesssim 
 t^{p}+t^{\kappat-\sigma}+\delta+t_0^{\frac{(n-1)(1-c_s^2)}{n-2}}.
\end{equation*}
Choosing $\sigma$ so that $0<\sigma < \kappat$ and letting $t\searrow 0$ in the above expression yields
\begin{equation}
  \label{eq:kijsmallness}
  \norm{\kf_{IJ}}_{H^{k-1}(\mathbb T^{n-1})}+\norm{\Wf^0}_{H^{k-1}(\mathbb T^{n-1})}\lesssim \delta+t_0^{\frac{(n-1)(1-c_s^2)}{n-2}}.
\end{equation}

\bigskip

\noindent \underline{Limits of the frame variables:} The frame $e_i^\mu$ is obtained from the components of $u$ via \eqref{beta-def}, \eqref{f-def}, \eqref{U-def}, \eqref{eq:Ubgdef}-\eqref{u-def} and the relations 
\begin{equation} \label{e0i-fix}
    e_0^\mu=\betat^{-1}\delta_0^\mu \AND e_I^0=0, 
\end{equation}
which hold by \eqref{e0-mu} and \eqref{e0I-fix}.
Since $\betat$ is initially positive at $t=t_0$
and the evolution equation \eqref{for-O.1.S2} for $\betat$ implies that $\betat$ cannot cross zero, we conclude that $\betat$ is positive on $M_{0,t_0}$. 
While \eqref{eq:PbbuDecay}, which we showed above that $u$ satisfies, yields a decay estimate for $\betat$, a stronger result can be obtained from
noting that \eqref{for-O.1.S2} can be expressed as
  \begin{equation*}
    \del{t}\bigl(t^{-\frac{1}{2}\kf_{J}{}^J}\betat\bigr) = \Bigl(-\beta^2 t^{-\ep_0-\ep_1}\xi_{00}+\frac{1}{2} t^{-1} (k_{J}{}^J-\kf_{J}{}^J)\Bigr) t^{-\frac{1}{2}\kf_{J}{}^J}\betat.
  \end{equation*}
Since $\betat>0$, we can integrate this in time to get
  \begin{align*}
    &\ln\bigl(t^{-\frac{1}{2}\kf_{J}{}^J}\betat(t)\bigr)
    -\ln\bigl({\tilde t}^{-\frac{1}{2}\kf_{J}{}^J}\betat(\tilde t)\bigr)\\
    = &\int_{\tilde t}^t\Bigl(-s^{-\ep_0-\ep_1}
    \beta^2 \xi_{00}(s)
       +\frac{1}{2} s^{-1} (k_{J}{}^J(s)-\kf_{J}{}^J)\Bigr)ds,    
  \end{align*}
which holds for all $t,\tilde t\in (0,t_0]$. From the  Sobolev and Moser inequalities (see Propositions 2.4., 3.7.~and 3.9.~from   \cite[Ch.~13]{TaylorIII:1996}), we have
  \begin{align*}
    &\bnorm{\ln\bigl(t^{-\frac{1}{2}\kf_{J}{}^J}\betat(t)\bigr)
    -\ln\bigl({\tilde t}^{-\frac{1}{2}\kf_{J}{}^J}\betat(\tilde t)\bigr)}_{H^{k-1}(\Tbb^{n-1})}\\
    \le& \int_{\tilde t}^t\Bigl(s^{-\ep_0-\ep_1}\norm{\beta(s)}_{H^{k-1}(\Tbb^{n-1})}^2 \norm{\xi_{00}(s)}_{H^{k-1}(\Tbb^{n-1})}\\
       &+\frac{1}{2} s^{-1} \norm{k_{J}{}^J(s)-\kf_{J}{}^J}_{H^{k-1}(\Tbb^{n-1})}\Bigr)ds
  \end{align*}
  provided $0<\tilde t\le t\le t_0$. From the decay estimates \eqref{eq:PbbuDecay}-\eqref{eq:PbbuDecay2} and \eqref{eq:kijsmallness}, the positivity of the constants $p$ and $\kappat$,
  and the fact that $\sigma>0$ can be chosen arbitrarily small, we deduce from the above estimate that $\ln(t^{-\frac{1}{2}\kf_{J}{}^J}\betat(t))$ converges in $H^{k-1}(\Tbb^{n-1})$ as $t\searrow 0$ to a limit, denoted $\tilde{\bfr}$, and that
\begin{equation*}
    \bnorm{\ln\bigl(t^{-\frac{1}{2}\kf_{J}{}^J}\betat(t)\bigr)
    -\tilde{\bfr}}_{H^{k-1}(\Tbb^{n-1})}
    \lesssim t^{1-\ep_0-\ep_1}(t^p+t^{\kappat-\sigma})^3
       + t^{p}+ t^{2(\kappat-\sigma)}    
\end{equation*}
for $0<t\leq t_0$.
From this inequality, it follows that the limit $\tilde{\bfr}$ is a strictly positive function in $H^{k-1}(\Tbb^{n-1})$, which we express as $\tilde{\bfr}=\ln(\bfr)$ where $\bfr\in H^{k-1}(\Tbb^{n-1})$. Using  $\tilde{\bfr}=\ln(\bfr)$ and noting that $1-\ep_0-\ep_1>0$, we observe that the above estimate simplifies to
\begin{equation}
       \label{eq:improvbetaestimate.pre}
       \bnorm{\ln\bigl(t^{-\frac{1}{2}\kf_{J}{}^J}\betat(t)\bigr)
    -\ln(\mathfrak{b})}_{H^{k-1}(\Tbb^{n-1})}
  \lesssim t^{p}+ t^{2(\kappat-\sigma)}.   
\end{equation}
Estimate \eqref{eq:improvbetaestimate-glob} is then a direct consequence of 
\eqref{eq:epscond.N}, \eqref{eq:defp}, \eqref{eq:defkappat} and the above estimate. 

Considering now the spatial frame components $e^\Lambda_I$, we see, with the help of \eqref{k-def}, \eqref{beta-def} and \eqref{psi-def},
that the evolution equation \eqref{for-M.1.S2} for the spatial frame components can be expressed in matrix form as
\begin{equation*} 
\del{t} e^\Lambda = \frac{1}{t}(\Ksc + t\Lsc)e^\Lambda    
\end{equation*}
where 
\begin{gather} \label{frame-lim-A}
  e^\Lambda=(e^\Lambda_I), \quad    \Ksc=\Bigl(-\frac{1}{2}\delta^{JL}k_{IL}(0)\Bigr),
  \intertext{and}
 \Lsc =-\Bigl(\frac{1}{2}t^{-1}\delta^{JL}(k_{IL}-k_{IL}(0))\notag\\
 \qquad\qquad\qquad\qquad+t^{-\ep_1+k_{J}{}^J(0)/2} \delta^{JL}(t^{-k_{J}{}^J(0)/2}\betat)\psi_{I}{}^0{}_L\Bigr). \label{frame-lim-A2}
\end{gather}
Letting\footnote{Given a square matrix $A$, we frequently use the notation $t^A$ instead of $\exp(\ln(t)A)$.}
\begin{equation} \label{frame-lim-B} 
\ec^\Lambda = (\ec^\Lambda_\mu) :=  t^{-\Ksc} e^\Lambda,
\end{equation}
a short calculation then shows that $\ec^\Lambda$ satisfies
\begin{equation} \label{frame-lim-C}
    \del{t} \ec^\Lambda = \Msc \ec^\Lambda
\end{equation}
where
\begin{equation*} \label{frame-lim-D}
    \Msc = t^{-\Ksc}\Lsc t^{\Ksc}.
\end{equation*}
By differentiating \eqref{frame-lim-C} repeatedly in space, we obtain from standard $L^2$-energy estimates and the Sobolev and Moser inequalities the differential energy inequality
\begin{equation*}
  %\label{frame-lim-AAA-F}
    \del{t}\norm{\ec^\Lambda(t)}_{H^{k-1}(\Tbb^{n-1})}^2 \lesssim
    \norm{\Msc(t)}_{H^{k-1}(\Tbb^{n-1})}\norm{\ec^\Lambda(t)}^2_{H^{k-1}(\Tbb^{n-1})},
  \end{equation*}
 which in turn, yields
 \begin{equation*}
  %\label{frame-lim-AAA-F}
    \del{t}\norm{\ec^\Lambda(t)}_{H^{k-1}(\Tbb^{n-1})} \lesssim
    \norm{\Msc(t)}_{H^{k-1}(\Tbb^{n-1})}\norm{\ec^\Lambda(t)}_{H^{k-1}(\Tbb^{n-1})}.
  \end{equation*}
Applying  Gr\"onwall's lemma to this differential inequality gives
  \[
    \norm{\ec^\Lambda(t)}_{H^{k-1}(\Tbb^{n-1})}\lesssim \norm{\ec^\Lambda(t_0)}_{H^{k-1}(\Tbb^{n-1})}e^{-\frac 12\int_{t_0}^t \norm{\Msc(s)}_{H^{k-1}(\Tbb^{n-1})}ds},\quad 0<t\leq t_0.
  \]
Integrating \eqref{frame-lim-C} in time, we see, with the help of the above inequality, that
\begin{align}
  &\norm{\ec^\Lambda(t)-\ec^\Lambda(\tilde t)}_{H^{k-1}(\Tbb^{n-1})}  \notag\\  
  \label{eq:frame_convest}
  \lesssim &\norm{\ec^\Lambda(t_0)}_{H^{k-1}(\Tbb^{n-1})}e^{\frac
    12\int_{t_0}^0
    \norm{\Msc(s)}_{H^{k-1}(\Tbb^{n-1})}ds}\int_{\tilde t}^t
  \norm{\Msc(s)}_{H^{k-1}(\Tbb^{n-1})}ds
\end{align}
for all $0<\tilde t\le t\le t_0$.

Next, by \eqref{eq:PbbuDecay}, \eqref{eq:PbbuDecay2} and \eqref{frame-lim-A2}, we observe that the matrix $\Lsc$ is bounded by
\begin{align} 
    &\norm{\Lsc(t)}_{H^{k-1}(\Tbb^{n-1})}\notag\\
    \label{frame-lim-G}
     &\lesssim t^{-1}(t^{p}+ t^{2(\kappat-\sigma)})+t^{-\ep_1}(t^p+t^{\kappat-\sigma}) 
  \end{align}
for all $0<t\leq t_0$. Also, since $t^{\pm \Ksc}=e^{\pm \ln(t)\Ksc}$, the estimate
\begin{equation} \label{frame-lim-H}
\norm{t^{\pm\Ksc}}_{H^{k-1}(\Tbb^{n-1})}\le C e^{C \norm{\ln(t)\Ksc}_{H^{k-1}(\Tbb^{n-1})}} \lesssim t^{-C \norm{\kf_{IJ}}_{H^{k-1}(\Tbb^{n-1})}}, 
\end{equation}
$0<t\leq t_0$, is a direct consequence of the analyticity of the exponential $e^{X}$, the definition \eqref{frame-lim-A}
of $\Ksc$, and the fact that $H^{k-1}(\Tbb^{n-1})$ is a Banach algebra by virtue of the assumption that $k-1>(n-1)/2$.  
By \eqref{frame-lim-D}, \eqref{frame-lim-G}, and \eqref{frame-lim-H}, we can then bound the matrix
$\Msc$ by
\begin{align}
  \label{frame-lim-Last}
&\norm{\Msc(t)}_{H^{k-1}(\Tbb^{n-1})} \notag\\
\lesssim &\Bigl(t^{-1}(t^{p}+ t^{2(\kappat-\sigma)})+t^{-\ep_1}(t^p+t^{\kappat-\sigma})\Bigr) t^{-C \norm{\kf_{IJ}}_{H^{k-1}(\Tbb^{n-1})}}
\end{align}
for all $0<t\leq t_0$. By choosing $\delta$ and $\sigma$ sufficiently small, it follows from \eqref{eq:epscond.N}, \eqref{eq:defkappat} and \eqref{frame-lim-Last} that $\Msc$ is $H^{k-1}(\Tbb^{n-1})$-integrable in time, that is, $\int_{0}^{t_0}\norm{\Msc(s)}_{H^{k-1}(\Tbb^{n-1})}\, ds <\infty$. We then deduce from this integrability and the inequality \eqref{eq:frame_convest} that $\ec^\Lambda(t)$ converges as $t\searrow 0$ in $H^{k-1}(\Tbb^{n-1})$ to a limit, which we will denote by $\ef^\Lambda_I(0)$. In order to derive the estimate \eqref{eq:improvframeestimate}, we let $\tilde t\searrow 0$ in \eqref{eq:frame_convest}. Doing so, we see, with the help of \eqref{frame-lim-B} and \eqref{frame-lim-Last}, that, for $\delta$ and $\sigma$ chosen sufficiently small, that the estimate 
\begin{align*}
  %\label{eq:improvframeestimate.pre}
  &\norm{e^{-\ln(t)\Ksc}e^\Lambda(t)-\ec^\Lambda(0)}_{H^{k-1}(\Tbb^{n-1})}\\    
  \lesssim &\Bigl(t^{p}+ t^{2(\kappat-\sigma)}+t^{1-\ep_1}(t^p+t^{\kappat-\sigma})\Bigr) t^{-C \norm{\kf_{JM}}_{H^{k-1}(\Tbb^{n-1})}},
\end{align*}
holds for all $0<t\leq t_0$. Choosing $\ep_0$,\ldots,$\ep_4$ according to \eqref{eq:epsfluidchoice} then yields the estimate \eqref{eq:improvframeestimate}.

\bigskip

\noindent \underline{Conformal Einstein-Euler-scalar field past stability:} Proposition~\ref{lag-exist-prop} implies, for some $t_1\in (0,t_0]$, the existence of a unique solution $\Wsc$ with regularity \eqref{eq:Wreg} on $M_{t_1,t_0}=(t_1,t_0]\times \Tbb^{n-1}$ of 
the system
\eqref{tconf-ford-C.1}-\eqref{tconf-ford-C.8} that satisfies
the initial conditions \eqref{l-idata}-\eqref{hhu-idata}. Moreover, since the conformal Einstein-Euler-scalar field initial data is assumed to satisfy the gravitational and wave gauge constraint equations, it follows from  Proposition~\ref{lag-exist-prop} that this solution satisfies the constraints \eqref{eq:Lag-constraints}, and determines a solution
$\{g_{\mu\nu},\tau,V^\mu\}$
of the conformal Einstein-Euler-scalar field equations \eqref{lag-confeqns} in Lagrangian coordinates that satisfies the wave gauge constraint
\eqref{lag-wave-gauge} and where $\tau$ is
given by
\begin{equation} \label{tau=t}
\tau=t.
\end{equation}
By construction of the Fuchsian system \eqref{eq:givp1}, the solution $\Wsc$ determines a solution $\ut$ on $M_{t_1,t_0}$ of the Fuchsian IVP \eqref{eq:givp1}-\eqref{eq:givp2} with the same  initial data $u_0$ as above. By the uniqueness statement of Proposition~\ref{prop:globalstability}, we conclude
that 
\begin{equation} \label{ut=u}
    \ut = u|_{M_{t_1,t_0}}
\end{equation}
provided the parameters $\ep_0$, $\ep_1$, $\ep_2$, $\ep_3$ and $\ep_4$ are chosen to be the same for both solutions.

From the equality \eqref{ut=u}, it follows that the energy estimate \eqref{eq:resenergy} together with the Sobolev inequality yields the bound
\begin{equation*} %\label{ut-bnd}
    \sup_{t_1<t<t_0}\norm{\ut(t)}_{W^{2,\infty}(\Tbb^{n-1})} < \infty.
\end{equation*}
From this bound on $\ut$, we see, with the help of  \eqref{p-fields}, \eqref{gi00},   \eqref{kt-def}-\eqref{tau-def}, \eqref{psit-def}, \eqref{k-def}-\eqref{Uac-def}, \eqref{eq:Ubgdef}-\eqref{u-def} and \eqref{e0i-fix}, that
\begin{align}
&\sup_{t_1<t<t_0}\Bigl(\norm{e_j^\mu(t)}_{W^{2,\infty}(\Tbb^{n-1})}+\norm{\Dc_i g_{jk}(t)}_{W^{2,\infty}(\Tbb^{n-1})}+
\norm{\betat(t)}_{W^{2,\infty}(\Tbb^{n-1})}
\notag \\
&\hspace{0.7cm}
+ 
\norm{ \kt_{IJ}(t)}_{W^{2,\infty}(\Tbb^{n-1})}
+\norm{\psit_I{}^k{}_J(t)}_{W^{2,\infty}(\Tbb^{n-1})}\notag \\
&\hspace{0.7cm}
+
\norm{\gamma_I{}^k{}_J(t)}_{W^{2,\infty}(\Tbb^{n-1})}
+\norm{\Dc_i\Dc_j\tau}_{W^{2,\infty}(\Tbb^{n-1})}\notag \\
&\hspace{0.7cm}
+\norm{W^s}_{W^{2,\infty}(\Tbb^{n-1})}\Bigr)< \infty.
\label{glob-cont-bnd-A}
\end{align}
With the help of this bound, it then follows from the evolution equations
\eqref{for-O.1.S2}-\eqref{for-M.1.S2} for $\betat$ and $e^\Lambda_I$, and \cite[Lemma~A.2]{BeyerOliynyk:2021} that
\begin{equation} \label{glob-cont-bnd-B}
    \inf_{M_{t_1,t_0}}\bigr\{\betat,\det(e^\Lambda_I)\bigl\} > 0,
\end{equation}
which in turn, implies via
 \eqref{e0i-fix} that
\begin{equation}\label{glob-cont-bnd-C}
    \inf_{M_{t_1,t_0}}\det(e^\mu_j) > 0 \AND \sup_{M_{t_1,t_0}}\det(g_{\mu\nu})<0.
\end{equation}
In addition to this,  we observe from \eqref{p-fields}, \eqref{eq:Eulerfchoice}, \eqref{eq:defw1234}, \eqref{k-def}-\eqref{Uac-def}, \eqref{eq:Ubgdef} and \eqref{u-def} that
\begin{align}
    &\sup_{M_{t_1,t_0}}\norm{V^\mu(t)}_{W^{1,\infty}(\Tbb^{n-1})}  =\sup_{M_{t_1,t_0}}\norm{V^i(t)e_i^\mu(t)}_{W^{1,\infty}(\Tbb^{n-1})}\notag\\
    &\lesssim \sup_{M_{t_1,t_0}} t^{\frac{(n-1)c_s^2-1}{n-2}-\ep_2}\norm{\betat(t)^{c_s^2}}_{H^{k}(\Tbb^{n-1})}\norm{W^i(t)}_{H^{k}(\Tbb^{n-1})}\norm{f_i^\mu(t)}_{H^{k}(\Tbb^{n-1})}\notag\\
    \label{glob-cont-bnd-C.1}
    &<\infty, \\
    &\sup_{M_{t_1,t_0}}|V|_{g}^2\notag\\
    \label{glob-cont-bnd-C.2}
    &=\sup_{M_{t_1,t_0}}(-\mathtt f^2 w^2)
    =-\inf_{M_{t_1,t_0}}t^{2\frac{(n-1)c_s^2-1}{n-2}}\bigl(\betat^{2c_s^2}((V_*^0)^2+\Ord(u))\bigr)<0,
\end{align}
where in deriving the above inequalities we have used the fact that $V_*^0>0$, the bounds \eqref{eq:improvbetaestimate-glob}, \eqref{eq:resenergy}, and \eqref{glob-cont-bnd-B}, and the Sobolev and the Moser inequalities.

Since the frame $e_i^\mu$ is orthonormal by construction, the components of the conformal metric in the Lagrangian coordinates are determined by 
\begin{equation} \label{g-Lag-components}
    g_{\mu\nu}=e_\mu^i \eta_{ij} e_\nu^j,
\end{equation}
and so by \eqref{glob-cont-bnd-A}, we have
\begin{equation} \label{glob-cont-bnd-D}
\sup_{t_1<t<t_0}\norm{g_{\mu\nu}(t)}_{W^{2,\infty}(\Tbb^{n-1})} < \infty.
\end{equation}
From the calculation
\begin{align*}
    e^\mu_i e^\nu_j \del{t}g_{\mu\nu} &= \Ld_{\del{t}}g_{ij}
    \oset{\eqref{e0i-fix}}{=}\Ld_{\betat e_0}g_{ij}\\
    &=\betat e_0^k \Dc_{k}g_{ij}+\Dc_i (\betat e_0^k)\eta_{kj}
    + \Dc_j (\betat e_0^k)\eta_{jk} \\
    & =\betat \delta_0^k \Dc_{k}g_{ij}+ e_i(\betat)\eta_{0j}+\betat \gamma_i{}^k{}_0 \eta_{kj}
    + e_j(\betat)\eta_{0i}+\betat \gamma_j{}^k{}_0\eta_{ik},
\end{align*}
we also observe that the bound
\begin{equation} \label{glob-cont-bnd-E}
\sup_{t_1<t<t_0}\norm{\del{t}g_{\mu\nu}(t)}_{W^{1,\infty}(\Tbb^{n-1})} < \infty
\end{equation}
is a direct consequence of
\eqref{glob-cont-bnd-A}-\eqref{glob-cont-bnd-C},
the relations  \eqref{p-fields}, \eqref{gamma-000}-\eqref{gamma-0K0}, \eqref{gamma-I00}-\eqref{gamma-IJ0} and \eqref{e0i-fix}, and the evolution
equation \eqref{for-O.1.S2}. Moreover, by employing similar arguments, it is also not difficult to 
verify that $\Dc_\nu \chi^\mu$, where $\chi^\mu$ is defined by \eqref{chi-def}, satisfies
\begin{equation} \label{glob-cont-bnd-F}
\sup_{t_1<t<t_0}\Bigl(\norm{\Dc_\nu \chi^\mu(t)}_{W^{2,\infty}(\Tbb^{n-1})}+\norm{\del{t}(\Dc_\nu \chi^\mu)(t)}_{W^{1,\infty}(\Tbb^{n-1})}\Bigr) < \infty.
\end{equation}
We conclude from \eqref{alpha-def} and the bounds \eqref{glob-cont-bnd-A}, \eqref{glob-cont-bnd-C}, \eqref{glob-cont-bnd-C.1}, \eqref{glob-cont-bnd-C.2}, \eqref{glob-cont-bnd-D}, \eqref{glob-cont-bnd-E} and \eqref{glob-cont-bnd-F} that the solution $\Wsc$ satisfies the continuation criteria \eqref{eq:cont_crit1}-\eqref{eq:cont_crit2}. Hence, by Proposition \ref{lag-exist-prop} the solution $\Wsc$ can be continued beyond $t_1$, and consequently, the solution $\Wsc$ exists on $M_{0,t_0}$. This solution continues to satisfy the constraints \eqref{eq:Lag-constraints} and determine a solution
$\{g_{\mu\nu},\tau,V^\mu\}$
of the conformal Einstein-Euler-scalar field equations \eqref{lag-confeqns} in Lagrangian coordinates that verifies the wave gauge constraint \eqref{lag-wave-gauge}.

\bigskip

\noindent \underline{Second fundamental form estimate:} Using \eqref{alpha-def}, \eqref{e0i-fix} and \eqref{tau=t},  it follows easily from the formula \eqref{g-Lag-components} for the conformal metric that on the  $t=const$-surfaces the lapse $\Ntt$ is given by
\begin{equation} \label{Ntt=betat}
    \Ntt=\betat,
\end{equation} the shift vanishes, and the induced spatial metric is $\gtt_{\Lambda\Omega}=g_{\Lambda\Omega}$. Furthermore, by \eqref{Ccdef}, \eqref{gamma-IJ0}, \eqref{k-def}, \eqref{psi-def} and \eqref{tau=t}, the 
second fundamental form induced on the $t=const$-surfaces by the conformal metric is 
\begin{equation*}
  \Ktt_{LJ}
  {=} \frac 12\kt_{LJ}           +\gamma_{(L}{}^0{}_{J)}
  =\frac 12t^{-1}\betat^{-1} k_{LJ}+t^{-\ep_1} \psi_{(L}{}^0{}_{J)},
\end{equation*}       
which using \eqref{beta-def},  we note can be expressed as
\begin{equation*}
  2t\betat \Ktt_{LJ}=k_{LJ}
  +2t^{1-\ep_1-\ep_0}\beta \psi_{(L}{}^0{}_{J)}.
\end{equation*}
The bound
\begin{align*}
  \norm{ 2t\betat \Ktt_{LJ}(t)-\kf_{LJ}}_{H^{k-1}(\mathbb T^{n-1})} 
  &\lesssim t^{p}+ t^{2(\kappat-\sigma)} 
  +t^{1-\ep_0-\ep_1}(t^p+t^{\kappat-\sigma})^2 \\
  &\lesssim t^{p}+ t^{2(\kappat-\sigma)}
\end{align*}
is then a direct consequence of the above expression,  the energy and decay estimates \eqref{eq:resenergy}-\eqref{eq:PbbuDecay2}, and the positivity of the constants $p$, $\kappat-\sigma$ and $1-\ep_0-\ep_1$.
Since the right hand side here  equals that of \eqref{eq:improvbetaestimate.pre}, the same choice \eqref{eq:epsfluidchoice} of the parameters $\ep_0,\ldots,\ep_4$
yields the estimate \eqref{eq:2ndFF.est}.

\bigskip

\noindent \underline{Fluid estimates:} We begin by observing that the estimate
\eqref{eq:fluidestimate1} follows from \eqref{U-def}, \eqref{eq:Ubgdef}, \eqref{u-def}, \eqref{Pbb-perp-def}-\eqref{Pbb-perp-u(0)} and the decay estimates \eqref{eq:PbbuDecay2} using the choice \eqref{eq:epsfluidchoice} for the parameters $\ep_0$, \ldots, $\ep_4$ while \eqref{eq:fluidestimate2} follows from \eqref{Pbb-def} and the decay estimate \eqref{eq:PbbuDecay} in a similar fashion. Also, we note that \eqref{eq:fluidestimate3} is a direct consequence of \eqref{d-fields}, \eqref{Uac-def}, \eqref{Pbb-def}  and \eqref{eq:PbbuDecay}. We further observe from \eqref{eq:defw1234} and the fact that $H^{k-1}(\Tbb^{n-1})$ is a Banach algebra (since $k-1>(n-1)/2$) that 
\begin{align}
  &\norm{w^2-(V_*^0+\Wf^0)^2}_{H^{k-1}(\Tbb^{n-1})}\notag\\
  =&\norm{(W^0)^2-\delta_{IJ}W^I W^J-(V_*^0+\Wf^0)^2}_{H^{k-1}(\Tbb^{n-1})}\notag\\
  \lesssim& \norm{W^0-(V_*^0+\Wf^0)}_{H^{k-1}(\Tbb^{n-1})}(V^0_*+\norm{\Wf^0}_{H^{k-1}(\Tbb^{n-1})}+\norm{W^0}_{H^{k-1}(\Tbb^{n-1})})\notag\\
  &\hspace{1.4cm}+\norm{W^I}_{H^{k-1}(\Tbb^{n-1})}^2\notag\\
  \lesssim& t^{\frac 18\frac{n-1}{n-2}(1-c_s^2)-\sigma}+t^{2 \bigl(\frac{(n-1)c_s^2-1}{n-2}-\sigma\bigr)},
    \label{eq:w2estimate}
\end{align}
where the last inequality holds by \eqref{eq:fluidestimate1} and \eqref{eq:fluidestimate2}.

To proceed, we define, for any $\gamma\in\Rbb$, a function $\Phi_\gamma(x,y)$ via
\begin{equation}
\label{eq:defPhi}
  \Phi_\gamma: I_1\times I_2\longrightarrow\Rbb\: :\: (x,y)\longmapsto (x+y)^{\gamma}-y^{\gamma},
\end{equation}
where $I_2\ssubset (0,\infty)$ is an open interval and $I_1$ is an open interval around $0$ that is sufficiently small to ensure that $\Phi_\gamma$ is well-defined and smooth. Since $\Phi_\gamma(0,y)=0$ for all $y\in I_2$, we deduce from Moser's inequality that 
\begin{equation}
  \label{eq:EstPhi}
  \norm{\Phi_\gamma(v_1,v_2)}_{H^\ell(\Tbb^{n-1})}\le C(\norm{v_1}_{H^\ell(\Tbb^{n-1})}, \norm{v_2}_{H^\ell(\Tbb^{n-1})}) \norm{v_1}_{H^\ell(\Tbb^{n-1})}
\end{equation}
for any $v_1, v_2\in H^\ell(\Tbb^{n-1})$
satisfying $v_1(x)\in I_1$ and $v_2(x)\in I_2$ for all $x\in \Tbb^{n-1}$ provided that $\ell>(n-1)/2$.
Next, we use $\Phi_\gamma$ along with \eqref{eq:physicsquantitiesfluid}, \eqref{eq:conffluidrel}, \eqref{eq:alphaByRtt}, \eqref{eq:Eulerfchoice} and \eqref{eq:defw1234} to express the fluid density as 
\begin{align*}
  %\label{eq:36}
  \rho =&\frac{P_0}{c_s^2} \vb^{-\frac {c_s^2+1}{c_s^2}}
  =\frac{P_0}{c_s^2} t^{-\frac{n-1}{n-2}(1+c_s^2)-\kf_{J}{}^J(1+c_s^2)/2} (t^{-\kf_{J}{}^J/2}\betat)^{-(1+c_s^2)} (w^2)^{-\frac {1+c_s^2}{2c_s^2}}\\
       &=\frac{P_0}{c_s^2}t^{-\frac{n-1}{n-2}(1+c_s^2)-\kf_{J}{}^J(1+c_s^2)/2} \Bigl(\Phi_{-(1+c_s^2)} \bigl(t^{-\kf_{J}{}^J/2}\betat-\mathfrak{b}, \mathfrak{b}\bigr)+\mathfrak{b}^{-(1+c_s^2)}\Bigr)\\
         &\qquad\times\Bigl(\Phi_{-\frac {1+c_s^2}{2c_s^2}} \bigl(w^2-(V_*^0+\Wf^0)^2, (V_*^0+\Wf^0)^2\bigr)+(V_*^0+\Wf^0)^{-\frac {1+c_s^2}{c_s^2}}\Bigr).
\end{align*}
From this expression and the estimates \eqref{eq:improvbetaestimate-glob}, \eqref{eq:w2estimate},
\eqref{eq:EstPhi}, and the positivity of $\bfr$ and $V_*+\Wf^0$, we see that estimate \eqref{eq:fluidpressureestimate} holds provided $\delta_0$ and $t_0$ are chosen sufficiently small. We also observe that the fluid variable $\bar{V}{}^\mu$ can be represented as \eqref{eq:fluidresult1} due to \eqref{eq:conffluidrel}, \eqref{p-fields} and \eqref{eq:Eulerfchoice} where $\eb_i^\mu$ is the $\gb$-orthonormal frame $\eb_i^\mu=t^{-1/(n-2)}e_i^\mu$, and, with the help of  \eqref{eq:18jsdkfjkdsjfksdf}, \eqref{eq:physicsquantitiesfluid}, \eqref{eq:conffluidrel}, \eqref{eq:alphaByRtt} and \eqref{eq:defw1234},
that the physical normalised fluid $n$-velocity field is given by \eqref{eq:physical4vectorfield}.
Finally, as $\mathfrak{b}$ and $V_*^0+\Wf^0$ are strictly positive on $\Tbb^{n-1}$ for $\delta_0$ and $t_0$ chosen sufficiently small, the estimates
\eqref{eq:fluidestimate1}, \eqref{eq:fluidestimate2}, \eqref{eq:w2estimate} and \eqref{eq:EstPhi} 
imply that \eqref{eq:physical4vectorfield.est} holds.

\bigskip

\noindent \underline{Asymptotic pointwise Kasner property:} Since $\{g_{\mu\nu},\tau, V^\mu\}$ is a solution of the conformal Einstein-Euler-scalar field equations, it follows from \eqref{eq:conf2phys} and \eqref{tau=t} that the triple
\begin{equation}
  \label{eq:conf2phys-a}
  \biggl\{\gb_{ij}=t^{\frac {2}{n -2}}g_{ij},\,\phi=\sqrt{\frac{n-1}{2(n-2)}}\ln(t), \Vb^\mu=V^\mu\biggr\}
\end{equation}
 determines a solution of the physical Einstein-Euler-scalar field system \eqref{eq:EFE.0}-\eqref{eq:AAA2}.  As a consequence, the spatial metric
  \begin{equation}
    \label{eq:spatmetrrel}
    \bar{\gtt}_{\Lambda\Omega}=t^{\frac {2}{n -2}}\gtt_{\Lambda\Omega}
  \end{equation}
  and second fundamental form 
  \begin{equation}
    \label{eq:phys2ndff}
    \bar{\Ktt}_{\Lambda\Omega}=t^{\frac {1}{n -2}}\Bigl(\Ktt_{\Lambda\Omega}
    +\frac {1}{n -2}t^{-1}\betat^{-1}{\gtt}_{\Lambda\Omega}\Bigr)
  \end{equation}
  induced by the physical metric $\gb_{\mu\nu}$ on the $\tau=t=const$-surfaces must satisfy the Hamiltonian constraint, which we write in the rescaled form
  \begin{equation}
         \label{eq:Hamilton1}
         t^{\frac {2}{n -2}}\betat^{2}t^{2}\overline\Rtt+t^{\frac {2}{n -2}}\betat^{2}t^{2}\bigl((\bar{\Ktt}_\Lambda{}^\Lambda)^2-\bar{\Ktt}_\Lambda{}^\Sigma \bar{\Ktt}_\Sigma{}^\Lambda\bigr)-\Tb=0
  \end{equation}
where $\overline\Rtt$ is the scalar curvature of the spatial metric $\bar{\gtt}_{\Lambda\Omega}$ and
  \begin{align*}
    \Tb= -2t^{\frac{2}{n-2}}\betat^2 t^2 \frac{\nablab_i t \nablab_j t }{|\nablab t|^2_{\gb}} \Tb^{ij}     
\end{align*}
with $\Tb^{ij}=\Tb_{ij}^{\text{SF}}+\Tb_{ij}^{\text{Fl}}$ the Euler-scalar field energy momentum tensor, see \eqref{Tb-ij-def}. Using \eqref{p-fields}-\eqref{eq:Eulerfchoice}, \eqref{eq:defw1234}, \eqref{g-Lag-components},  \eqref{eq:conf2phys-a} and \eqref{e0i-fix}, we can express $\Tb$ as
\begin{equation*}
\Tb=\frac {n-1}{n-2}
         +2P_0 t^{\frac{(n-1)(1-c_s^2)}{n-2}}\betat^{1-c_s^2} w^{-\frac{1+c_s^2}{c_s^2}}\Bigl(\frac{1+c_s^2}{c_s^2} \frac{W^IW_I}{w^{2}}
         +\frac{1}{c_s^2}\Bigr).
\end{equation*}
Taking the pointwise limit of the above expression as $t\searrow 0$, we see from \eqref{eq:improvbetaestimate-glob}, \eqref{eq:fluidestimate2}, \eqref{eq:w2estimate}, \eqref{eq:EstPhi} and the Sobolev inequalities that
\begin{equation}
  \label{eq:skcn33331}
  \lim_{t\searrow 0} \Tb(t,x)=\frac {n-1}{n-2}
\end{equation}
for each $x\in \Tbb^{n-1}$.

Next, since the conformal factor $t^{\frac {2}{n -2}}$ in \eqref{eq:spatmetrrel} is constant on the $t$=const-surfaces, it follows from  \eqref{tau=t} and \eqref{eq:spatmetrrel} that 
  \begin{equation*}
  \tau^{\frac {2}{n -2}}\overline\Rtt=\Rtt,
  \end{equation*} where $\Rtt$ is the scalar curvature of the spatial conformal metric $\gtt_{\Lambda\Omega}$.
  Noting that 
 \begin{align*} %\label{Rtt-formula}
 \Rtt
    =& e_J(\Gamma_I{}^J{}_K)\delta^{IK}
    - e_I(\Gamma_J{}^J{}_K) \delta^{IK}
    +\delta^{IK}\Gamma_I{}^M{}_K\Gamma_J{}^J{}_M
    -\Gamma_J{}^M{}_K \delta^{IK}\Gamma_I{}^J{}_M\\
&+2 \delta^{IK}\Gamma_{[I}{}^M{}_{J]}\Gamma_M{}^J{}_K,
\end{align*}
where the $\Gamma_M{}^J{}_K$ are the spatial components of the connection coefficients of the conformal metric $g_{\mu\nu}$ with respect to the frame $e_i^\mu$,
we find from \eqref{Ccdef}, \eqref{p-fields}-\eqref{d-fields},  \eqref{kt-def}-\eqref{mt-def}, \eqref{gt-def}, \eqref{psit-def}, \eqref{k-def}-\eqref{m-def}, \eqref{psi-def}-\eqref{gac-def}, \eqref{e0i-fix}, the formula
\begin{align*}
g_{IJKL} =& e_I(g_{JKL})-\gamma_I{}^0{}_Jg_{0KL}-\gamma_I{}^M{}_J g_{MKL}-\gamma_I{}^0{}_Kg_{J0L}-\gamma_I{}^M{}_K g_{JML}\\
&-\gamma_I{}^0{}_Lg_{JK0}-\gamma_I{}^M{}_K g_{JKM}
\end{align*}
for the covariant derivative
$g_{IJKL}=\Dc_I g_{JKL}=\Dc_I\Dc_J g_{KL}$, and a straightforward calculation that
\begin{align*}
  \Rtt
  =& e*\del{}\psit + \gt +(\ellt+\psit)* (\kt+\psit+\ellt)\\ 
  =& t^{-\ep_2-\ep_1}f*\del{}\psi + t^{-1-\ep_1}\betat^{-1}\gac\\
  &+t^{-1-\ep_1}\betat^{-1} (\ell+\psi)* (k+t^{1-\ep_1}\betat\psi+t^{1-\ep_1}\betat\ell),
\end{align*}
where $e=(e^\Lambda_I)$, $\del{}=(\del{\Lambda})$ and we are employing the $*$-notation from Section \ref{sec:NonlinDecomp}.
Multiplying the above expression by $\betat^2 \tau^2$ yields
\begin{align*}
  \betat^{2}\tau^{2}\Rtt
  =& t^{2-\ep_2-\ep_1-2\ep_0}\beta^{2} f*\del{}\psi
      + t^{1-\ep_1-\ep_0}\beta\gac\\
   &+t^{1-\ep_1-\ep_0}\beta (\ell+\psi)* (k+t^{1-\ep_1-\ep_0}\beta\psi+t^{1-\ep_1-\ep_0}\beta\ell)
\end{align*}
where in deriving this we have used \eqref{beta-def}.
It is then a straightforward consequence of \eqref{eq:epscond.N}, \eqref{eq:symhypreg_appl} and the Sobolev inequalities that 
\begin{equation}
\label{eq:skcn3333}
\lim_{t\searrow 0}t^{2/(n-2)}\betat(t,x)^{2}t^{2}\Rtt(t,x)= 0
\end{equation} 
for each $x\in \Tbb^{n-1}$.

Given \eqref{eq:skcn33331} and \eqref{eq:skcn3333}, the Hamiltonian constraint \eqref{eq:Hamilton1} simplifies to 
\begin{equation*}
  \lim_{t\searrow 0} t^{\frac {2}{n -2}}\betat^{2}t^{2}\bigl((\bar{\Ktt}_\Lambda{}^\Lambda)^2-\bar{\Ktt}_\Lambda{}^\Sigma \bar{\Ktt}_\Sigma{}^\Lambda\bigr)=0.
\end{equation*}
It is now a straightforward consequence of the above expression, 
\eqref{Ntt=betat}, \eqref{eq:phys2ndff}, and the fact that
\begin{equation}
  \label{eq:asymptptwKasner2-a}
  \lim_{t\searrow 0} 2t\,\Ntt(t,x)\, \Ktt_{I}{}^J (t,x)=\kf_{I}{}^J(x),
\end{equation}
for each $x\in \Tbb^{n-1}$, which follows from \eqref{eq:2ndFF.est} and the Sobolev inequality, that
\begin{equation}
    \label{eq:HamSecFFTrafo}
    (\bar{\Ktt}_\Lambda{}^\Lambda)^2-\bar{\Ktt}_\Lambda{}^\Sigma \bar{\Ktt}_\Sigma{}^\Lambda 
    =
       t^{-\frac {2}{n -2}}\Bigl( (\Ktt _{I}{}^{I})^2
       -\Ktt _{I}{}^{J}\Ktt _{J}{}^{I}
       +2\Ktt _{I}{}^{I}t^{-1}\betat^{-1}
       +\frac {n-1}{n -2}t^{-2}\betat^{-2}                                                   
       \Bigr),
     \end{equation}
     and therefore that
\begin{equation}
  \label{eq:asymptptwKasner-a}
  (\kf_{I}{}^{I})^2 -
  \kf_{I}{}^J \kf_{J}{}^I 
  +4\kf_{I}{}^{I}=0 \quad \text{in $\Tbb^{n-1}$}.
\end{equation}
Solving \eqref{eq:asymptptwKasner-a} for $\kf_I{}^I$, we obtain two solutions $\kf_I{}^I=\pm\sqrt{4+\kf_{I}{}^J \kf_{J}{}^I}
  -2$.
But, by \eqref{eq:kijsmallness} and the Sobolev inequality, we can choose $\delta$ and $t_0$ small enough to ensure that 
$\norm{\kf_I{}^I}_{L^\infty(\Tbb^{n-1})}< 4$.
Doing so, we conclude
that $\kf_I{}^I$ must satisfy $\kf_I{}^I=\sqrt{4+\kf_{I}{}^J \kf_{J}{}^I}
  -2$, which in particular, implies that
\begin{equation} \label{kfII}
\kf_I{}^I \geq 0 \quad \text{in $\Tbb^{n-1}$}.
\end{equation}
Taken together, \eqref{eq:asymptptwKasner2-a}-\eqref{kfII} imply that the solution $\{g_{\mu\nu},\tau, V^\mu\}$  verifies all the conditions of Definition \ref{def:APKasner}, and hence, is asymptotically pointwise Kasner.

\bigskip

\noindent \underline{Crushing singularity:}
 Taking the trace of \eqref{eq:phys2ndff} with respect to the physical metric \eqref{eq:spatmetrrel}, we observe that the physical mean curvature can be expressed as
  \begin{equation}
    \label{eq:physM}
    \bar\Ktt_\Lambda{}^\Lambda=\frac{1}{2\betat t^{-\frac {n-1}{n -2}}}\Bigl(2t\betat\Ktt_I{}^I
    +\frac{2(n-1)}{n -2}\Bigr).
  \end{equation}
Recalling that $\betat>0$, we deduce from \eqref{eq:improvbetaestimate-glob}, \eqref{eq:2ndFF.est}, \eqref{kfII} and Sobolev's inequality the existence of a constant $C>0$ 
such that the pointwise estimates
\begin{equation*}
0<\betat(t,x) \leq \frac{1}{C} \AND 2t\betat\Ktt_I{}^I + \frac{2(n-1)}{n-2} \geq \frac{n-1}{n-2} 
\end{equation*}
hold for all $(t,x)\in M_{0,t_0}=(0,t_0]\times\Tbb^{n-1}$ provided $t_0$ is chosen sufficiently small. These inequalities together with \eqref{eq:physM} imply the pointwise lower bound
\begin{equation*}
\bar{\Ktt}_\Lambda{}^\Lambda \geq \frac{C(n-1)}{2(n-2)}\frac{1}{t^{\frac{n-1}{n-2}}}
\end{equation*}
on $M_{0,t_1}$,  and hence, that $\bar{\Ktt}_\Lambda{}^\Lambda$ blows up uniformly as $t\searrow 0$. By definition, see  \cite{eardley1979}, this uniform blow up of the physical mean curvature implies that the hypersurface $t=0$ is a \textit{crushing singularity}. 

\bigskip

\noindent \underline{Geometric estimates for the conformal metric and scalar field:} We now turn to deriving the estimates \eqref{eq:glostab-est.First}-\eqref{eq:glostab-est.Last}. Before doing so, we first summarise the relationships between the following  geometric variables 
\begin{gather*}
    \bigl\{ t\betat\Dc_0 g_{00},\Dc_I g_{00},\Dc_0 g_{J0},\Dc_I g_{J0},t\betat\Dc_0 g_{JK},\Dc_I g_{JK}, t \betat\Dc_I\Dc_j g_{kl},\Dc_i\Dc_j\tau,\\
    \Dc_I\Dc_j\Dc_k\tau\bigr\}
\end{gather*}
and the corresponding Fuchsian variables determined from the components of $u$. 
From  \eqref{p-fields}-\eqref{d-fields}, \eqref{gi00}, \eqref{kt-def}-\eqref{psit-def} and \eqref{k-def}-\eqref{tauac-def}, it straightforward to check that
  \begin{align*}    
    t\betat\Dc_0 g_{00}=&-\delta^{JK} k_{JK},\\
    \Dc_I g_{00}=&2t^{-\ep_1}m_{I}-t^{-\ep_1}\delta^{JK}(2\ell_{JKI}-\ell_{IJK}),\\
    \Dc_0 g_{J0}=&t^{-\ep_1}m_{J},\\
    \Dc_I g_{J0}=&t^{-\ep_1}\ell_{I0J},\\
    t\betat\Dc_0 g_{JK}=&k_{JK}+ t^{1-\ep_1+\frac{1}{2}\kf_{J}{}^J}(t^{-\frac{1}{2}\kf_{J}{}^J}\betat)(\ell_{K0J}+\ell_{J0K}),\\
    \Dc_I g_{JK}=&t^{-\ep_1}\ell_{IJK},\\    
    t \betat\Dc_I\Dc_j g_{kl}=&t^{-\ep_1}\gac_{Ijkl},\\        
    \Dc_i\Dc_j\tau=&t^{\ep_0-\ep_1}\xi_{ij},\\    
    \Dc_I\Dc_j\Dc_k\tau=&t^{-\ep_0-2\ep_1}\tauac_{Ijk}.
  \end{align*}  
From these relations, it then follows via the definitions 
\eqref{U-def}, \eqref{Pbb-def}, \eqref{Pbb-perp-def}  as well as the 
estimates \eqref{eq:PbbuDecay}-\eqref{eq:PbbuDecay2} and \eqref{kfII}
that the estimates
  \begin{align}
    \label{eq:glostab-est.First.pre}
      \norm{t\betat\Dc_0 g_{00}+\delta^{JK}\kf_{JK}}_{H^{k-1}(\mathbb T^{n-1})}&\lesssim t^{p}+ t^{2(\kappat-\sigma)}, \\
    \label{eq:glostab-est.2.pre}
      \norm{t\betat\Dc_0 g_{JK}-\kf_{JK}}_{H^{k-1}(\mathbb T^{n-1})}
      &\lesssim t^{p}+ t^{2(\kappat-\sigma)}\notag \\
      &\quad+t^{1-\ep_1}(t^p+t^{\kappat-\sigma}),\\   
    \norm{\Dc_I g_{00}}_{H^{k-1}(\mathbb T^{n-1})}+\norm{\Dc_0 g_{J0}}_{H^{k-1}(\mathbb T^{n-1})} \quad &\notag\\
    \label{eq:glostab-est.3.pre}
                                                                                +\norm{\Dc_I g_{J0}}_{H^{k-1}(\mathbb T^{n-1})}+\norm{\Dc_I g_{JK}}_{H^{k-1}(\mathbb T^{n-1})}&\lesssim  t^{p-\ep_1}+t^{\kappat-\sigma-\ep_1},\\      
\label{eq:glostab-est.4.pre}
    \norm{t \betat\Dc_I\Dc_j g_{kl}}_{H^{k-1}(\mathbb T^{n-1})}&\lesssim t^{p-\ep_1}+t^{\kappat-\sigma-\ep_1},\\    
\label{eq:glostab-est.5.pre}
    \norm{\Dc_i\Dc_j\tau }_{H^{k-1}(\mathbb T^{n-1})}&\lesssim t^{\ep_0-\ep_1}(t^p+t^{\kappat-\sigma}),\\
\label{eq:glostab-est.Last.pre}
    \norm{\Dc_I\Dc_j\Dc_k\tau }_{H^{k-1}(\mathbb T^{n-1})}&\lesssim t^{-\ep_0-2\ep_1}(t^p+t^{\kappat-\sigma}),
  \end{align}
hold for $0<t\leq t_0$.  The estimates \eqref{eq:glostab-est.First}-\eqref{eq:glostab-est.Last} then follow from \eqref{eq:glostab-est.First.pre}-\eqref{eq:glostab-est.Last.pre} by fixing the parameters $\ep_0,\ldots,\ep_4$ as in \eqref{eq:epsfluidchoice} and choosing $\sigma$ sufficiently small.

\bigskip

\noindent\underline{AVTD property:} The solution $\{g_{\mu\nu},\tau, V^\mu\}$ to the conformal Einstein-Euler-scalar field equations on $M_{0,t_0}$ constructed above %and hence the corresponding physical solution $\{\gb_{\mu\nu},\phi\}$ 
satisfies the AVTD property in the sense of Section~\ref{sec:AVTDAPK}. To see why, we observe that the VTD equation corresponding to the Fuchsian equation \eqref{eq:givp1} is given by
\[A^0\del{t}u =
 \frac{1}{t}\Ac\Pbb u +\tilde\Ftt(t)+\frac{1}{t}\Htt(u)\Pbb u+\frac 1t\Pbb^{\perp} \hat\Htt(u)+\frac{1}{t^{\tilde\epsilon}} H(t,u),\]
and hence, any solution $u$ of \eqref{eq:givp1} is a solution of this VTD equation up to the error term ${t^{-\ep_0-\ep_2}}A^\Lambda(t,u) \del{\Lambda} u$. Given the definition \eqref{ALambda-def}, it then follows immediately from \eqref{eq:symhypreg_appl}, $\ep_0+\ep_2<1$, see \eqref{eq:epscond.N}, and the Sobolev and Moser inequalities (Propositions 2.4, 3.7~and 3.9~from   \cite[Ch.~13]{TaylorIII:1996}) that
\begin{equation*}
\int^{t_0}_0\bnorm{{s^{-\ep_0-\ep_2}}A^\Lambda(s,u(s)) \del{\Lambda} u(s)}_{H^{k-1}(\Tbb^{n-1})}\,ds <\infty,
\end{equation*}
which establishes the AVTD property. 

\bigskip

\noindent \underline{$C^2$-inextendibility of the physical metric:} We now establish the $C^2$-in\-ex\-tend\-i\-bil\-i\-ty of the physical metric by  verifying that the scalar curvature $\Rb=\gb^{\mu\nu}\Rb_{\mu\nu}$ of the physical spacetime metric $\gb_{\mu\nu}$ blows up as $t\searrow 0$. To this end, we observe from \eqref{ESF.1}-\eqref{bTtt-ij-def}, \eqref{eq:Tttconfphys}, \eqref{eq:TttconfphysW}-\eqref{eq:defw1234}, \eqref{tau=t},  and the properties \eqref{e0i-fix} of the orthonormal frame $e_i^\mu$ that the scalar curvature can be expressed as
\begin{align*}
  \Rb&=2\gb^{ij}\nablab_i\phi\nablab_j \phi+\gb^{ij}\bar\Ttt_{ij}
       =2t^{\frac {-2}{n -2}}\eta^{ij}\nabla_i\phi\nabla_j \phi
       +t^{-\frac {2}{n -2}}\eta^{ij}\Ttt_{ij}\\
       &=-\frac{n-1}{n-2}(t^{-\frac{1}{2}\kf_{J}{}^J}\betat)^{-2}t^{-2\frac {n-1}{n -2}-\kf_{J}{}^J}\\
       &-4P_0\frac{(n-1)c_s^2+1}{(n-2)c_s^2} t^{-(1+c_s^2)\frac{n-1}{n-2}}\betat^{-(1+c_s^2)} w^{-\frac{1+c_s^2}{c_s^2}},
\end{align*}
from which we get that
\begin{align*}
  &t^{2\frac {n-1}{n -2}+\kf_{J}{}^J}\Rb+\frac{n-1}{n-2}\mathfrak{b}^{-2}\\
  =&-\frac{n-1}{n-2}\bigl((t^{-\frac{1}{2}\kf_{J}{}^J}\betat)^{-2}-\mathfrak{b}^{-2}\bigr)\\
       &-4P_0\frac{(n+1)c_s^2+1}{(n-2)c_s^2} t^{(1-c_s^2)\frac {n-1}{n -2}+(1-c_s^2)\kf_{J}{}^J} (t^{-\frac{1}{2}\kf_{J}{}^J}\betat)^{-(1+c_s^2)} w^{-\frac{1+c_s^2}{c_s^2}}.
\end{align*}
Estimate \eqref{eq:SptRicciEstimat-glob} then follows from a combination of  \eqref{eq:improvbetaestimate-glob} and \eqref{eq:w2estimate}-\eqref{eq:EstPhi} together with the calculus inequalities that, so long as $t_0$ and $\delta_0$ are sufficiently small  and $V^0_*>0$.

\bigskip

\noindent \underline{Past timelike geodesic incompleteness:}
Recalling that the physical mean curvature $\bar\Ktt_\Lambda{}^\Lambda$ is related to the conformal mean curvature $\Ktt_\Lambda{}^\Lambda$
via \eqref{eq:physM}, we note that when the initial data agrees with FLRW initial data (Section~\ref{sec:FLRWEulerSFExplSol})  the conformal mean curvature $\Ktt_\Lambda{}^\Lambda$ can be made arbitrarily close to zero on the inital hypersurface $\Sigma_{t_0}=\{t_0\}\times \Tbb^{n-1}$ by choosing $t_0$ sufficiently small, which in turn, implies that $\bar\Ktt_\Lambda{}^\Lambda$ can be made close to $\frac {n-1}{n -2} t_0^{-\frac {1}{n -2}-1}$ there.  Thus by choosing $\delta>0$ sufficiently small, we can by \eqref{glob-stab-thm-idata-A} and the Sobolev's inequality ensure that $\bar\Ktt_\Lambda{}^\Lambda$ is close, in a pointwise sense,
as we like to  $\frac {n-1}{n -2} t_0^{-\frac {1}{n -2}-1}$ everywhere on $\Sigma_{t_0}$ for the perturbed initial data. In particular, for $\delta>0$ small enough, we have that $\Ktt_\Lambda{}^\Lambda>{n-1}{2(n -2)} t_0^{-\frac {1}{n -2}-1}$ on $\Sigma_{t_0}$.
Past timelike geodesic incompleteness is then a consequence of Hawking's singularity theorem \cite[Chapter~14, Theorem~55A]{oneill1983a}, that is, all past directed timelike geodesics starting on the $\Sigma_{t_0}$ reach $\{0\}\times\Tbb^{n-1}$ in finite proper time.

This concludes the proof of  Theorem~\ref{glob-stab-thm}.

\bibliographystyle{imsart-number}
\bibliography{refs}

\end{document}